\documentclass[12pt,a4]{article}
\usepackage[dvips]{color,graphicx}
\usepackage{algorithm,bm,psfrag,afterpage,color,amsmath,amssymb,latexsym,amsthm,url,lscape,natbib}
\DeclareMathOperator*{\argmin}{arg\,min}

\newtheorem{thm}{Theorem}

\usepackage[subrefformat=parens]{subcaption}
\usepackage[subrefformat=parens]{caption}

\graphicspath{{./figures/}}

\begin{document}

{\baselineskip = 8mm 

\begin{center}
\textbf{\LARGE Multivariate regression modeling in integrative analysis via sparse regularization} 
\end{center}

\begin{center}
{\large Shuichi Kawano$^{1}$, \ Toshikazu Fukushima$^{2}$, \\
Junichi Nakagawa$^{2}$ \ and \ Mamoru Oshiki$^{3}$}
\end{center}

\begin{center}
\begin{minipage}{14cm}
{
\begin{center}
{\it {\footnotesize 


$^1$Faculty of Mathematics, Kyushu University, 
744 Motooka Nishi-ku Fukuoka 819-0395, Japan. \\

skawano@math.kyushu-u.ac.jp

\vspace{1.2mm}

$^2$Advanced Technology Research Laboratories, Research \& Development, Nippon Steel Corporation, 
Futtsu, Chiba 293-8511, Japan. \\

\vspace{1.2mm}

$^3$Division of Environmental Engineering, Faculty of Engineering, Hokkaido University, 
Sapporo, Hokkaido 060-8628, Japan. \\


}}

\end{center}


}
\end{minipage}
\end{center}

\vspace{1mm} 

\begin{abstract}
\noindent 
The multivariate regression model basically offers the analysis of a single dataset with multiple responses. 
However, such a single-dataset analysis often leads to unsatisfactory results. 
Integrative analysis is an effective method to pool useful information from multiple independent datasets and provides better performance than single-dataset analysis. 
In this study, we propose a multivariate regression modeling in integrative analysis. 
The integration is achieved by sparse estimation that performs variable and group selection. 
Based on the idea of alternating direction method of multipliers, we develop its computational algorithm that enjoys the convergence property. 
The performance of the proposed method is demonstrated through Monte Carlo simulation and analyzing wastewater treatment data with microbe measurements.

\end{abstract}

\begin{center}
\begin{minipage}{14cm}
{
\vspace{3mm}

{\small \noindent {\bf Key Words and Phrases:} 
Group selection,
Integrative analysis,
Regularization,
Sparsity,
Wastewater treatment.
}


}
\end{minipage}
\end{center}

\baselineskip = 8mm


\section{Introduction}

Multivariate regression models are widely used for analyzing data with multiple continuous responses and have been studied exhaustively \citep{bedrick1994model,liu1997segmented,rousseeuw2004robust,peng2010regularized,obozinski2011support,qian2022large}. 
The existing multivariate regression methods basically offer the analysis of a single dataset. 
However, single-dataset analysis causes models with low prediction accuracy and results with poor reproducibility \citep{tseng2015integrating,zhao2015integrative}. 
If there are multiple datasets from multiple independent studies with comparable designs, multi-datasets analysis can be used to extract useful information and increase sample size. 
Due to the property, multi-datasets analysis generally provides better performance than single-dataset analysis. 
Among multi-datasets analysis methods, integrative analysis has received considerable attention over the past decade \citep{zhao2015integrative}. 
A characteristic of integrative analysis is to analyze raw data from multiple datasets jointly and can outperform classical multi-datasets analysis methods such as meta-analysis, which pools summary statistic obtained by analyzing multiple datasets separately. 
Thus far, there have been many researches about integrative analysis: multiple regression \citep{liu2014integrative,huang2017promotingJASA,chang2022integrative}, logistic regression \citep{ma2011integrative,tang2016fused}, survival data analysis \citep{liu2011high,ma2012integrative,liu2013incorporating,liu2013estimation,cheng2015identification,zhang2016penalized,deng2021integrative,tang2021poststratification,ventz2022integration}, boosting \citep{huang2017promoting,sun2020integrative}, and multivariate analysis \citep{fang2018integrative,dondelinger2020joint,fan2020integrating}. 
Meanwhile, integrative analysis for multivariate regression models has not been fully explored. 

In this article, we propose a multivariate regression modeling in integrative analysis. 
Multiple datasets are integrated by performing group selection across each dataset. 
The group selection is achieved by group lasso \citep{yuan2006model}. 
High-dimensional and low sample size data are becoming common in the current statistical context.
To deal with such data, we perform model estimation and covariate selection simultaneously by lasso \citep{tibshirani1996regression}.
The computational algorithm of the proposed multivariate regression method is established by the technique of alternating direction method of multipliers \citep{boyd2011distributed}. 
We also show the convergence property of the algorithm.

The rest of this article is organized as follows. 
In Section \ref{sec:model}, the proposed multivariate regression method is described. 
The computational algorithm and its theoretical property are described in Section \ref{sec:computation}. 
Simulation studies and analyzing wastewater treatment data with microbiome measurement are contained in Section \ref{sec:NumericalStudy}. 
Conclusions are given in Section \ref{sec:Conclusions}.




\section{Model}
\label{sec:model}

Let $\bm y$ be a $q$-dimensional vector of response variables and $\bm x$ be a $p$-dimensional vector of covariates. 
Suppose that we have $M$ datasets for the variables: $\{ ({\bm y}_i^m, {\bm x}_i^m); i=1,\ldots,n_m \} \ (m=1,\ldots,M)$,  where $n_m$ is a sample size of the $m$-th dataset. 
In addition, let $\bm z^m$ be an $r_m$-dimensional vector of explanatory variables that are included in only the $m$-th dataset. 
For the variables, suppose that we obtain $M$ datasets $\{ \bm z^m_i; i=1,\ldots,n_m \} \ (m=1,\ldots,M)$.

We consider a multivariate regression model in the $m$-th dataset in the form
\begin{equation*}
Y^m = \bm 1_{n_m} (\bm \alpha^m)^\top +  X^m B^m + Z^m C^m+ E^m,
\end{equation*}
where $\bm 1_{n_m}$ is an $n_m$-dimensional vector of which all components are one, $\bm \alpha^m$ is a $q$-dimensional vector of intercepts, $Y^m = ({\bm y}_1^m,\ldots,{\bm y}_{n_m}^m)^\top$ is an $n_m \times q$ matrix of response variables, $X^m = ({\bm x}_1^m,\ldots,{\bm x}_{n_m}^m)^\top$ is an $n_m \times p$ matrix of explanatory variables, $B^m$ is a $p \times q$ matrix of coefficients, $Z^m = ({\bm z}_1^m,\ldots,{\bm z}_{n_m}^m)^\top$ is an $n_m \times r_m$ matrix of explanatory variables included in only the $m$-th dataset, $C^m$ is an $r_m \times q$ matrix of coefficients, and $E^m$ is an $n_m \times q$ matrix of errors with mean $O_{n_m \times q}$ and variance-covariance matrix $I_{n_m} \otimes \Sigma$. 
Here, $O_{n_m \times q}$ is an $n_m \times q$ zero matrix, $I_{n_m}$ is an $n_m \times n_m$ identity matrix, and $\Sigma$ is a $q \times q$ positive definite matrix. 
We denote the $(j,k)$-th element of $B^m$ as $\beta_{jk}^m$ and set ${\bm \beta}_{jk}=(\beta_{jk}^1,\ldots,\beta_{jk}^M)^\top \ (j=1,\ldots,p; \ k=1,\ldots,q)$. 
We assume homogeneity models: $B^m$'s have the same sparse structure across datasets such that $I(\beta_{jk}^1=0)=\cdots=I(\beta_{jk}^M=0)$ for all $(j,k)$.

To estimate the parameters $\bm \alpha^m, B^m$, $C^m$ under the homogeneity model, we consider the following minimization problem
\begin{equation}
\min_{\substack{\bm \alpha^1,\ldots,\bm \alpha^M \\ B^1,\ldots,B^M \\ C^1,\ldots,C^M}} \left[ \sum_{m=1}^M \frac{1}{2 n_m} \left\| Y^m - \bm 1_{n_m} (\bm \alpha^m)^\top -  X^m B^m - Z^m C^m \right\|_F^2 + \lambda \sum_{j=1}^p \sum_{k=1}^q \| {\bm \beta}_{jk} \|_2 + \gamma \sum_{m=1}^M \| C^m \|_1 \right],
\label{eq:mini1_lasso}
\end{equation}
where $\| \cdot \|_F$ is the Frobenous norm, $\lambda$ and $\gamma$ are regularization parameters with non-negative value, and $\| \cdot \|_q$ is the $L_q \ (q=1,2)$ norm of a vector. 
The second term is the group lasso penalty, which guarantees the homogeneity structure of a model. 
The third term is the lasso penalty. 
This term induces variable selection for the variables $\bm z^m \ (m=1,\ldots,M)$. 
In general, the group lasso includes the square of the number of dimensions of the parameter vector in the penalty term. 
Note that we omit it in this study, because it is the constant $\sqrt{M}$.

\section{Computation}
\label{sec:computation}

\subsection{Estimation algorithm}
\label{sec:EstimationAlgorithm}

We use the alternating direction method of multipliers (ADMM; \citet{boyd2011distributed}) to obtain an estimate of the parameters $\bm \alpha^1,\ldots,\bm \alpha^M, B^1,\ldots,B^M, C^1,\ldots,C^M$. 
We first rewrite the minimization problem \eqref{eq:mini1_lasso} as
\begin{align}
&\min_{\substack{\bm \alpha^1,\ldots,\bm \alpha^M \\ B^1,\ldots,B^M \\ C^1,\ldots,C^M \\ D^1,\ldots,D^M \\ \bm \eta_{11},\ldots,\bm \eta_{pq}}} \left[ \sum_{m=1}^M \frac{1}{2 n_m} \left\| Y^m - \bm 1_{n_m} (\bm \alpha^m)^\top -  X^m B^m - Z^m C^m \right\|_F^2 + \lambda \sum_{j=1}^p \sum_{k=1}^q \| {\bm \eta}_{jk} \|_2 + \gamma \sum_{m=1}^M  \| D^m \|_1 \right] \nonumber \\
& \text{subject to} \quad {\bm \eta}_{jk} = {\bm \beta}_{jk}, \quad C^m = D^m.
\label{eq:mini2_lasso}
\end{align}
From the problem \eqref{eq:mini2_lasso}, we can obtain the scaled augmented Lagrangian
\begin{align}
\begin{split}
& \sum_{m=1}^M \frac{1}{2 n_m} \left\| Y^m - \bm 1_{n_m} (\bm \alpha^m)^\top -  X^m B^m - Z^m C^m \right\|_F^2 
+ \lambda \sum_{j=1}^p \sum_{k=1}^q \| {\bm \eta}_{jk} \|_2 
+ \gamma \sum_{m=1}^M \| D^m \|_1 \\
& \hspace{3mm} + \frac{\rho}{2} \sum_{j=1}^p \sum_{k=1}^q \| {\bm \eta}_{jk} - {\bm \beta}_{jk} + {\bm u}_{jk} \|_2^2
+ \frac{\rho}{2} \sum_{m=1}^M \| C^m - D^m + V^m \|_F^2,
\end{split}
\label{eq:augmentedLagrangian}
\end{align}
where $\bm u_{jk} \ (j=1,\ldots,p;\ k=1,\ldots,q)$ and $V^m \ (m=1,\ldots,M)$ are dual variables and $\rho$ is a penalty parameter with a positive value.

When we set $(\bm \alpha^m)^\ell, (B^m)^\ell, (C^m)^\ell, (D^m)^\ell, \bm \eta_{jk}^\ell, \bm u_{jk}^\ell, (V^m)^\ell \ (m=1,\ldots,M, \ j=1,\ldots, p, \ k=1,\ldots, q)$ as the estimates of $\bm \alpha^m, B^m, C^m, D^m, \bm \eta_{jk}, \bm u_{jk}, V^m$ in the $\ell$-th iteration, respectively, the idea of ADMM algorithm induces the parameter update as follows:
\begin{align*}
(\bm \alpha^m)^{\ell+1} &= \argmin_{\bm \alpha^m} \left\| Y^m - \bm 1_{n_m} (\bm \alpha^m)^\top -  X^m (B^m)^{\ell} - Z^m (C^m)^{\ell} \right\|_F^2, \\
(B^m)^{\ell+1} &= \argmin_{B^m} \Bigg[ \frac{1}{2 n_m} \left\| Y^m - \bm 1_{n_m} \{(\bm \alpha^m)^{\ell+1}\}^\top -  X^m B^m - Z^m (C^m)^{\ell} \right\|_F^2 \\
& \hspace{30pt} + \frac{\rho}{2} \| B^m - (H^m)^{\ell} - (U^m)^{\ell} \|_F^2 \Bigg], \\
(C^m)^{\ell+1} &= \argmin_{C^m} \Bigg[ \frac{1}{2 n_m} \left\| Y^m - \bm 1_{n_m} \{(\bm \alpha^m)^{\ell+1}\}^\top -  X^m (B^m)^{\ell+1} - Z^m (C^m)^{\ell} \right\|_F^2 \\
& \hspace{30pt} + \frac{\rho}{2} \| C^m - (D^m)^{\ell} + (V^m)^{\ell} \|_F^2 \Bigg], \\
(D^m)^{\ell+1} &= \argmin_{D^m} \left[ \frac{\rho}{2} \| (C^m)^{\ell+1} - D^m + (V^m)^{\ell} \|_F^2 
+ \gamma \| D^m \|_1 \right], \\
\bm \eta_{jk}^{\ell+1} &= \argmin_{\bm \eta_{jk}} \left[ \frac{\rho}{2} \| {\bm \eta}_{jk} - {\bm \beta}_{jk}^{\ell+1} + {\bm u}_{jk}^{\ell} \|_2^2
+ \lambda \| {\bm \eta}_{jk} \|_2 \right], \\
\bm u_{jk}^{\ell+1} &= {\bm u}_{jk}^{\ell} + {\bm \eta}_{jk}^{\ell+1} - {\bm \beta}_{jk}^{\ell+1}, \\
(V^m)^{\ell+1} &= (V^m)^{\ell} + (C^m)^{\ell+1} - (D^m)^{\ell+1},
\end{align*}
where $H^m$ and $U^m$ are, respectively, a matrix whose the $(j,k)$-th element is $\eta_{jk}^m$ and $u_{jk}^m$. 
In the update of $B^m$, we note that the equation $\sum_{j=1}^p \sum_{k=1}^q \| {\bm \eta}_{jk} - {\bm \beta}_{jk} + {\bm u}_{jk} \|_2^2 = \sum_{m=1}^M \| B^m - H^m - U^m \|_F^2$ is used. 
The update formula of $\bm \alpha^m, B^m, C^m$ is easy to obtain. 
Meanwhile, minimization of $D^m$ and $\bm \eta_{jk}$ can be done using soft-thresholding operator. 
We use two soft-thresholding operators: the soft-thresholding operator for a scalar and a vector, which is
\begin{equation*}
S(a,b)={\rm sign} (a) (|a|-b)_+, \qquad \mathcal{S} ({\bm c},d) = \left( 1 - \frac{d}{\|\bm c\|_2} \right)_+ {\bm c}, 
\end{equation*}
respectively. 
Here, $a$, $b$, $d$ are scalars and $\bm c$ is a vector. 
Overall, the update is summarized in Algorithm \ref{algorithm:algorithm1}. 

\begin{algorithm}[htb]
  \caption{ADMM algorithm for multivariate regression model in integrative analysis}
  \begin{enumerate}
    \item
    Initialization: $\ell=0$ and $(\bm \alpha^m)^0, (B^m)^0, (C^m)^0, (D^m)^0, \bm \eta_{jk}^0, \bm u_{jk}^0, (V^m)^0 \ (m=1,\ldots,M, \ j=1,\ldots, p, \ k=1,\ldots, q)$.

    \item
    Update $\ell = \ell+1$.
    \begin{enumerate}
    \item 
    $\displaystyle{(\bm \alpha^m)^{\ell+1} = \frac{1}{n_m} \left\{ Y^m - X^m (B^m)^{\ell} - Z^m (C^m)^{\ell} \right\}^\top \bm 1_{n_m}}$
    
    \item 
    $(B^m)^{\ell+1} = \left\{ \left(X^m\right)^\top X^m + n_m \rho I_p \right\}^{-1} \Big[ \left( X^m \right)^\top \left[ Y^m - \bm 1_{n_m} \{(\bm \alpha^m)^{\ell+1}\}^\top - Z^m (C^m)^{\ell} \right]$ \\
    \hspace{10pt} $ + \  n_m  \rho \{(H^m)^{\ell} + (U^m)^{\ell}\} \Big]$
    
    \item
    $(C^m)^{\ell+1} = \left\{ \left(Z^m\right)^\top Z^m + n_m \rho I_{r_m} \right\}^{-1}  \Big[ \left( Z^m \right)^\top \left\{ Y^m - \bm 1_{n_m} \{(\bm \alpha^m)^{\ell+1} \}^\top - X^m (B^m)^{\ell+1} \right\}$ \\
    \hspace{10pt} $+ \ n_m  \rho \{ (D^m)^{\ell} - (V^m)^{\ell} \} \Big]$
    
    \item
    $\displaystyle{((D^m)^{\ell+1})_{ij} = {S} \left( ( (C^m)^{\ell+1} + (V^m)^{\ell} )_{ij} , \frac{\gamma}{\rho} \right), \quad (i=1,\ldots,r_m; \ j=1,\ldots,q)}$
    
    \item
    ${\bm \eta}_{jk}^{\ell+1} = \displaystyle{\mathcal{S} \left( {\bm \beta}_{jk}^{\ell+1} - {\bm u}_{jk}^{\ell} , \frac{\lambda}{\rho} \right)}$
    
    \item
    $\bm u_{jk}^{\ell+1} = {\bm u}_{jk}^{\ell} + {\bm \eta}_{jk}^{\ell+1} - {\bm \beta}_{jk}^{\ell+1}$
    
    \item
    $(V^m)^{\ell+1} = (V^m)^{\ell} + (C^m)^{\ell+1} - (D^m)^{\ell+1}$
    \end{enumerate}

  \item
  Repeat Step 2 until convergence. 
  In our numerical study, the convergence condition is that the $\ell_2$ norm of the difference between two consecutive quantities of z\eqref{eq:augmentedLagrangian} is smaller than a prefixed threshold. 
  \end{enumerate}
  \label{algorithm:algorithm1}
\end{algorithm}

We provide the property of Algorithm \ref{algorithm:algorithm1}. 
We set $\bm \theta=( (\bm \alpha^1)^\top,\ldots, (\bm \alpha^M)^\top, (\mathrm{vec}(B^1))^\top,\ldots, $\\$(\mathrm{vec}(B^M))^\top, (\mathrm{vec}(C^1))^\top,\ldots, (\mathrm{vec}(C^M))^\top)^\top$ and \\$\mathcal L (\bm \theta) = \sum_{m=1}^M \left\| Y^m - \bm 1_{n_m} (\bm \alpha^m)^\top -  X^m B^m - Z^m C^m \right\|_F^2/(2 n_m) + \lambda \sum_{j=1}^p \sum_{k=1}^q \| {\bm \beta}_{jk} \|_2 + $\\$\gamma \sum_{m=1}^M \| C^m \|_1$. 
In addition, let $\bm \theta^\ell$ be the estimate of $\bm \theta$ in the $\ell$-th iteration derived from Algorithm \ref{algorithm:algorithm1}. 
Then Algorithm \ref{algorithm:algorithm1} satisfies the following convergence property. 
\begin{thm}
Assume that there exists at least one solution $\bm \theta^\ast$ of \eqref{eq:mini1_lasso}. 
Then $\lim_{\ell \to \infty} \mathcal L (\bm \theta^\ell) = \mathcal L (\bm \theta^\ast)$ holds.
Furthermore, $\lim_{\ell \to \infty} \| \bm \theta^\ell - \bm \theta^\ast \|_2=0$ holds whenever $\bm \theta^\ast$ is a unique solution. 
\end{thm}
\begin{proof}
See the supplementary material S1.
\end{proof}

In this theorem, we note that the $L_2$ norm in $\| \bm \theta^\ell - \bm \theta^\ast \|_2$ is not essential. 
This theorem also stands for general norms of a vector space.

\subsection{Selection of tuning parameter}
\label{sec:SelectionTuning}

We have three tuning parameters: $\lambda, \gamma, \rho$. 
According to \cite{boyd2011distributed}, the penalty parameter $\rho$ is fixed as one. 
The two regularization parameters $\lambda, \gamma$ are selected by $K$-fold cross-validation. 
When we divide the original $m$-th dataset into the $K$ datasets $(Y^{m}_{(1)}, X^{m}_{(1)}, Z^{m}_{(1)}), \ldots, (Y^{m}_{(K)}, X^{m}_{(K)}, Z^{m}_{(K)})$, the objective function for the $K$-fold cross-validation is
\begin{equation}
{\rm CV} = \frac{1}{K} \sum_{k=1}^K \sum_{m=1}^M \frac{1}{2 n_m^{(k)}} \left\| Y^m_{(k)} - \bm 1_{n_m^{(k)}} (\hat{\bm \alpha}^m_{(-k)})^\top -  X^m_{(k)} \hat{B}^m_{(-k)} - Z^m_{(k)} \hat{C}^m_{(-k)} \right\|_F^2,
\label{eq:CV}
\end{equation}
where $\hat{\bm \alpha}^m_{(-k)}, \hat{B}^m_{(-k)}, \hat{C}^m_{(-k)}$ are the estimates of ${\bm \alpha}^m, {B}^m, {C}^m$, respectively, computed with the data excluding the $k$-th dataset, and $n_m^{(k)}$ means the sample size in $m$-th and $k$-th dataset.

We choose the values of the regularization parameters $\lambda, \gamma$ from the minimizers of CV in \eqref{eq:CV}.

\section{Numerical study}
\label{sec:NumericalStudy}

\subsection{Monte Carlo simulations}
\label{sec:MonteCarlo}

We investigated the usefulness of our proposed method through Monte Carlo simulations. 
Data were generated from the true model
\begin{equation*}
Y^m = X^m B^\ast + Z^m C^{\ast m} + E^m, \quad (m=1,\ldots,M).
\end{equation*}
We considered $M=2, 3$.

Regardless of the number of datasets, we set as follows. 
We set $q=2$ and $B^\ast = (B_1^{\ast \top}, B_2^{\ast \top})^\top$. 
Here
\begin{equation*}
B_1^{\ast} = 
\begin{pmatrix}
1 & 1 & 1 & 1 & 1 & 0 & 0 & 0 & 0 & 0 \\
0 & 0 & 0 & 0 & 0 & 0.5 & 0.5 & 0.5 & 0.5 & 0.5 \\
\end{pmatrix}^\top
\end{equation*}
and $B_2^{\ast} = O_{s \times 2}$. 
We considered $s=5, 50$. 
Each row of the design matrix corresponding to $B_1^\ast$ was independently generated from a multivariate normal distribution having mean zero vector and variance-covariance matrix of which the $(i, j)$-th element is $\rho_x^{|i-j|}$. 
We considered $\rho_x=0.1, 0.9$. 
On the other hand, each element of the design matrix corresponding to $B_2^\ast$ was independently generated from $N(0,1)$. 
Each row of the error matrix $E^m$ was independently generated from a multivariate normal distribution having mean zero vector and variance-covariance matrix of which the $(i, j)$-th element is $\rho_y^{|i-j|}$. 
We considered $\rho_y=0.1, 0.9$. 
The sample size was set to $n=15, 25, 50, 75$.

For the case $M=2$, we set as follows. 
We set $C^{\ast m} = (C_1^{\ast m \top}, C_2^{\ast m \top})^\top$. 
Here
\begin{equation}
\begin{split}
C_1^{\ast 1} &= 
\begin{pmatrix}
1 & 1 & 1 & 1 & 1 & 0 & 0 & 0 & 0 & 0 \\
0 & 0 & 0 & 0 & 0 & 0.5 & 0.5 & 0.5 & 0.5 & 0.5 \\
\end{pmatrix}^\top, \\
C_1^{\ast 2} &= 
\begin{pmatrix}
0 & 0 & 0 & 0 & 0 & 0.5 & 0.5 & 0.5 & 0.5 & 0.5 \\
1 & 1 & 1 & 1 & 1 & 0 & 0 & 0 & 0 & 0 \\
\end{pmatrix}^\top
\end{split}
\label{eq:SimTrueCoef}
\end{equation}
and $C_2^{\ast m} = O_{s \times 2}$. 
Each row of the design matrix for $C_1^{\ast m}$ was generated from similar manner of that of $B_1^\ast$. 
Each element of the design matrix for $C_2^{\ast m}$ was independently generated from $N(0,1)$.

For the case $M=3$, we set as follows. 
We set $C^{\ast m} = (C_1^{\ast m \top}, C_2^{\ast m \top}, C_3^{\ast m \top})^\top$. 
Here $C_1^{\ast 1}$ and $C_1^{\ast 2}$ were the same as in \eqref{eq:SimTrueCoef}, 
\begin{equation*}
C_1^{\ast 3} = 
\begin{pmatrix}
0 & 0 & 0 & 1 & 1 & 1 & 1 & 0.5 & 0.5 & 0.5 \\
1 & 1 & 1 & 0.5 & 0.5 & 0.5 & 0.5 & 0 & 0 & 0 \\
\end{pmatrix}^\top,
\end{equation*}
and $C_2^{\ast m} = O_{s \times 2}$. 
The generation of each row of the design matrix for $C_1^{\ast m}$ and each element of the design matrix for $C_2^{\ast m}$ was the same as in the case $M=2$.

We compared our proposed multivariate regression method (MR) with univariate multiple regression method in integrative analysis (UR), multiple regression method estimated by lasso (lasso), multivariate regression method estimated by group lasso (mglasso), and multivariate regression method estimated by lasso (mlasso). 
For MR and UR, we used $K=5$ in \eqref{eq:CV}. 
The comparative methods lasso and mglasso were computed by the package \textbf{glmnet} in the software \textsf{R}.

The simulation was conducted 100 times. 
The performance was evaluated in terms of mean squared error (MSE) given by $\mathrm{MSE}=E[ (y-\hat{y})^2 ]$, false positive rate (FPR), and false negative rate (FNR). 
MSE was estimated by 1,000 random samples.
FPR and FNR are defined as
\begin{equation*}
\mathrm{FPR}=
\frac{1}{100}
\sum_{k=1}^{100}
\frac{\left|\left\{ 
j:\hat{\zeta}^{(k)}_{j}\neq0~\wedge~\zeta^{\ast}_{j}= 0
\right\}\right|}
{\left|\left\{
j:\zeta^{\ast}_{j}\neq 0 
\right\}\right|}, \quad 
\mathrm{FNR}=
\frac{1}{100}
\sum_{k=1}^{100}
\frac{\left|\left\{
j:\hat{\zeta}^{(k)}_{j}=0~\wedge~\zeta^{\ast}_{j} \neq 0
\right\}\right|}
{\left|\left\{
j:\zeta^{\ast}_{j}= 0 
\right\}\right|}. 
\end{equation*}
Here, ${\zeta}^{*}_{j}$ is the true $j$-th element, $\hat{\zeta}^{(k)}_{j}$ is the estimated $j$-th element for the $k$-th simulation, and $|\{\ast\}|$ is the number of elements included in a set $\{\ast\}$, where we set $\bm \zeta = (\mathrm{vec}(B^1)^\top,\ldots, \mathrm{vec}(B^M)^\top, \mathrm{vec}(C^{1})^\top,\ldots, \mathrm{vec}(C^{M})^\top)^\top$. 

We summarize boxplots of MSE from Figures \ref{fig:SimuM2n15} to \ref{fig:SimuM2n75} for $M=2$ and Figures S.1 to S.4 for $M=3$ in the supplementary material S2. 
In the figures, D1 and D2 indicate, respectively, a first dataset and a second dataset, while R1 and R2 indicate, respectively, a first response variable and a second response variable. 
Therefore, the term ``D1 \& R1" means the result for a first response variable obtained by analyzing a first dataset. 
First, we discuss the results for $M=2$. 
The lasso and mglasso provide relatively larger MSE than MR, UR, and mlasso. 
The mlasso gives the smallest MSE when $n$ is small, while it is as small as or larger than MR and UR when $n$ is large. 
The MR and UR produce similar MSE, but we note that the UR can sometimes have large variances (e.g., see Figures \ref{fig:SimuM2n15s50rx09ry01} and \ref{fig:SimuM2n15s50rx09ry09}). 
Next, we discuss those for $M=3$. 
The overall result is the same as when $M=2$.
The mlasso in $M=3$ is unstable, because it gives the smallest or largest MSE when $n=15, 25$.


The results of FPR and FNR are summarized in Figures \ref{fig:SimuM2n15_FPRFNR} to \ref{fig:SimuM2n75_FPRFNR} for $M=2$ and Figures S.5 to S.8 for $M=3$. 
In the figures, for example, the term ``FPR (D1)" represents the result for FPR obtained by analyzing a first dataset. 
As the overall result for $M=3$ is the same as when $M=2$, we describe the results for $M=2$. 
First, we discuss the results for FPR. 
In many cases, lasso gives the smallest MSE when $n=15, 25$, while lasso and mlasso are smallest when $n=50, 75$. 
The MR and UR are larger than other methods when $n$ is small, but they are as small as lasso and mlasso when $n$ is large. 
The mlasso provides large variances when $n=15, 25$. 
Next, we discuss for FNR. 
The MR and UR often give the smallest FNR, followed by mlasso. 
The lasso and mglasso have relatively large FNR.

\begin{figure}[htbp]
\begin{minipage}[b]{0.5\linewidth}
\centering
\includegraphics[width=8cm,height=4.6cm]{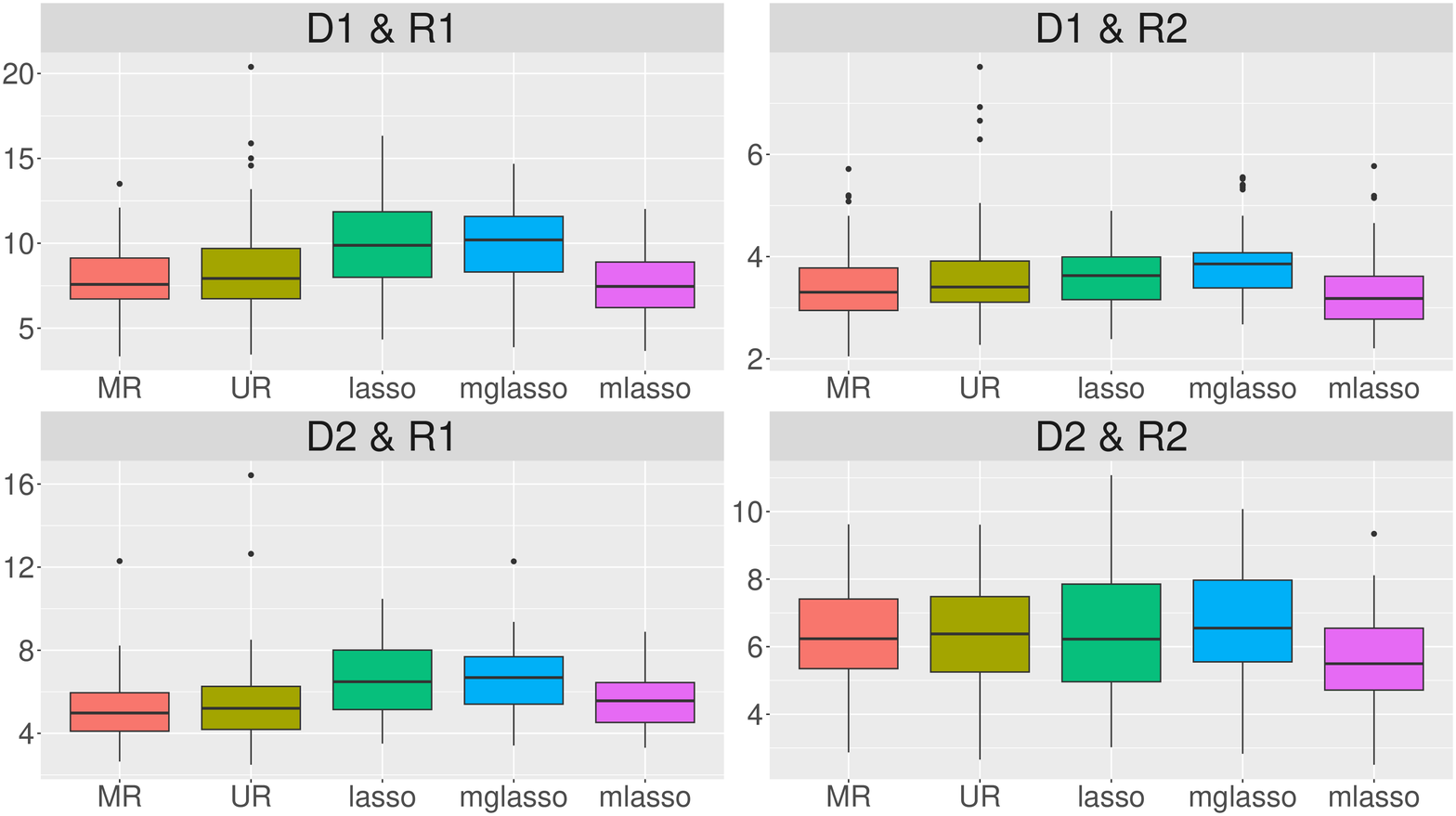}
\vspace{-3.5mm}
\subcaption{$s=5, \rho_x=0.1, \rho_y=0.1$}
\vspace{2.5mm}
\end{minipage}
\begin{minipage}[b]{0.5\linewidth}
\centering
\includegraphics[width=8cm,height=4.6cm]{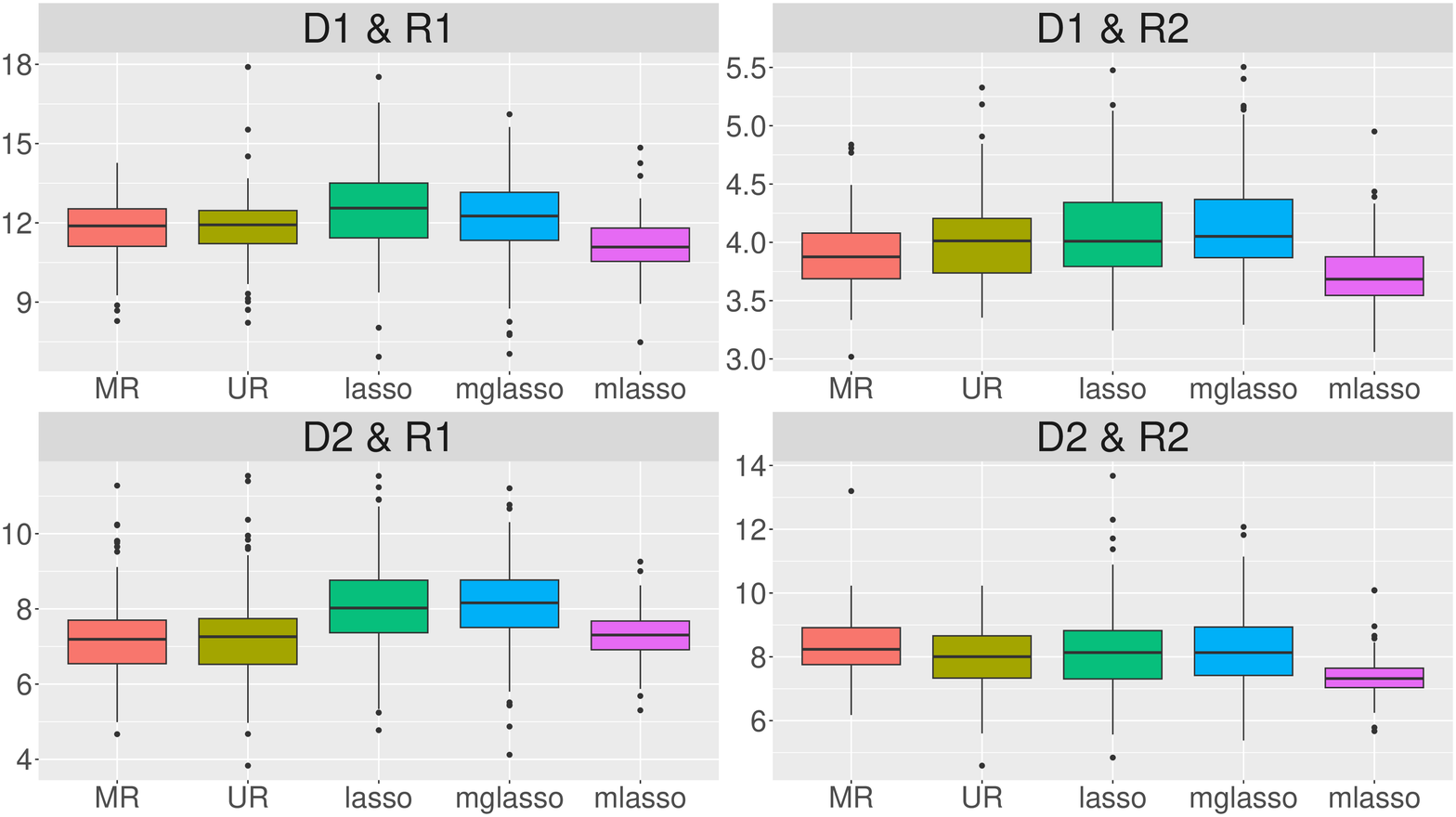} 
\vspace{-3.5mm}
\subcaption{$s=50, \rho_x=0.1, \rho_y=0.1$}
\vspace{2.5mm}
\end{minipage}
\begin{minipage}[b]{0.5\linewidth}
\centering
\includegraphics[width=8cm,height=4.6cm]{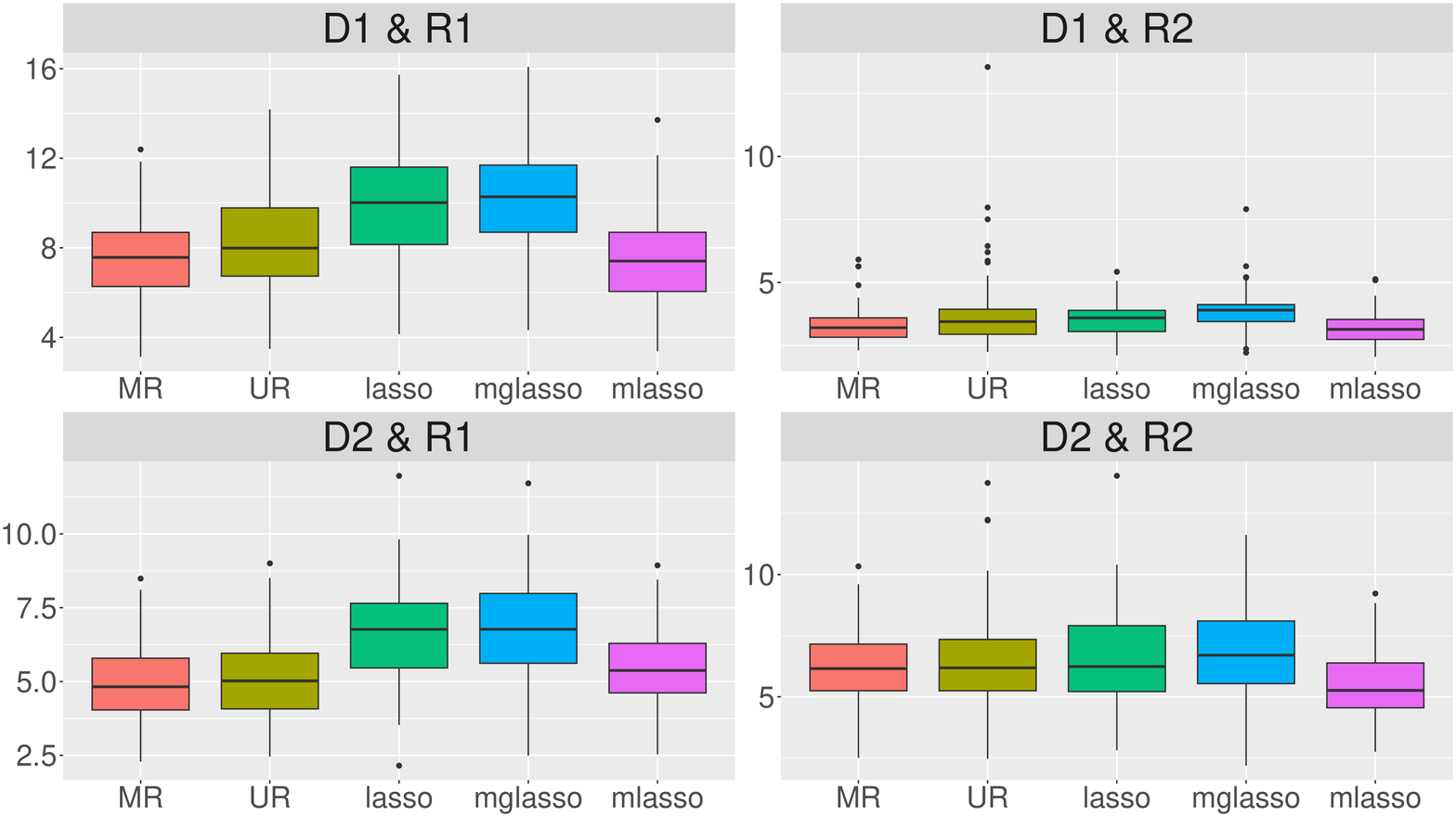} 
\vspace{-3.5mm}
\subcaption{$s=5, \rho_x=0.1, \rho_y=0.9$}
\vspace{2.5mm}
\end{minipage}
\begin{minipage}[b]{0.5\linewidth}
\centering
\includegraphics[width=8cm,height=4.6cm]{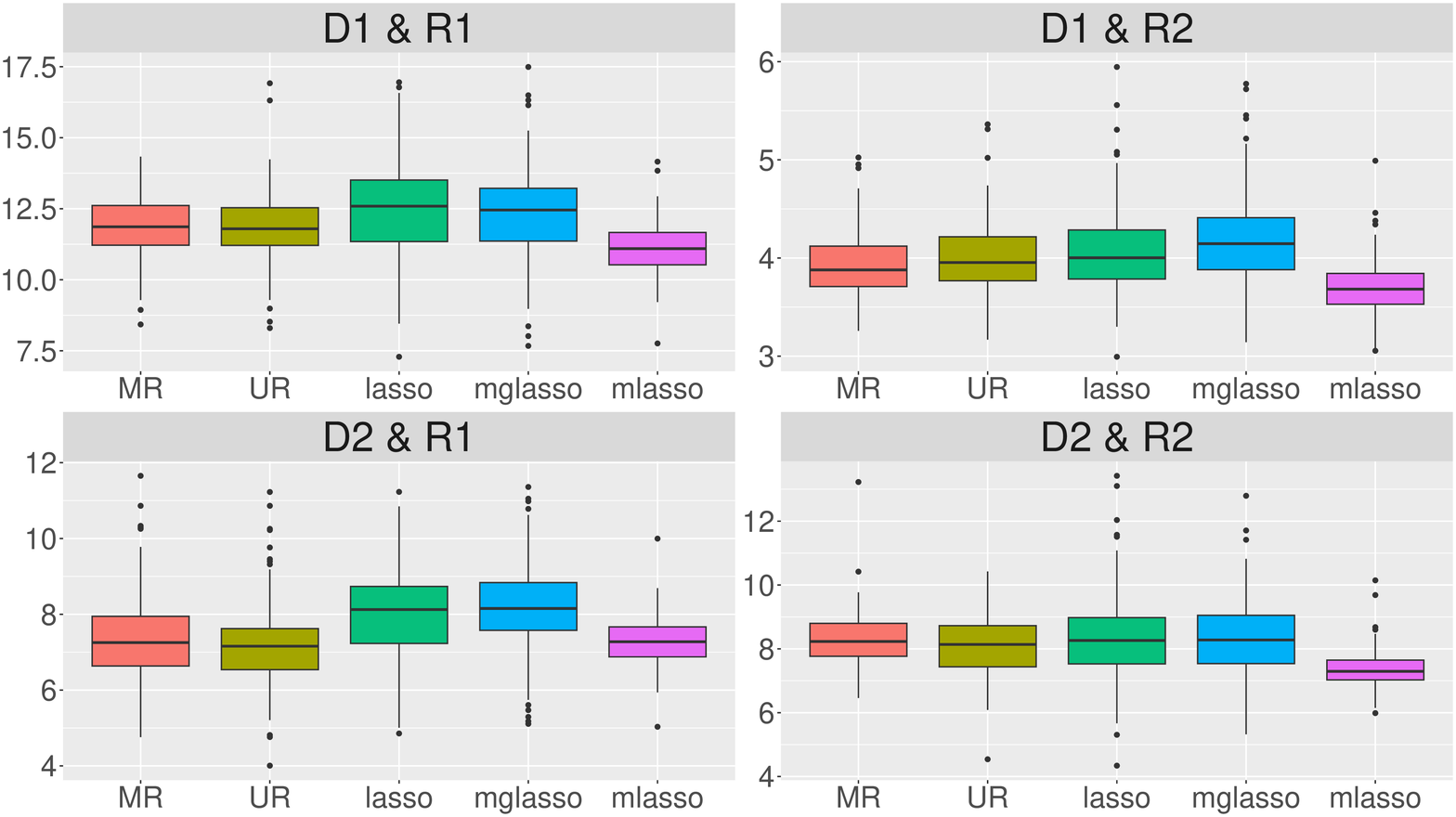}
\vspace{-3.5mm}
\subcaption{$s=50, \rho_x=0.1, \rho_y=0.9$}
\vspace{2.5mm}
\end{minipage}
\begin{minipage}[b]{0.5\linewidth}
\centering
\includegraphics[width=8cm,height=4.6cm]{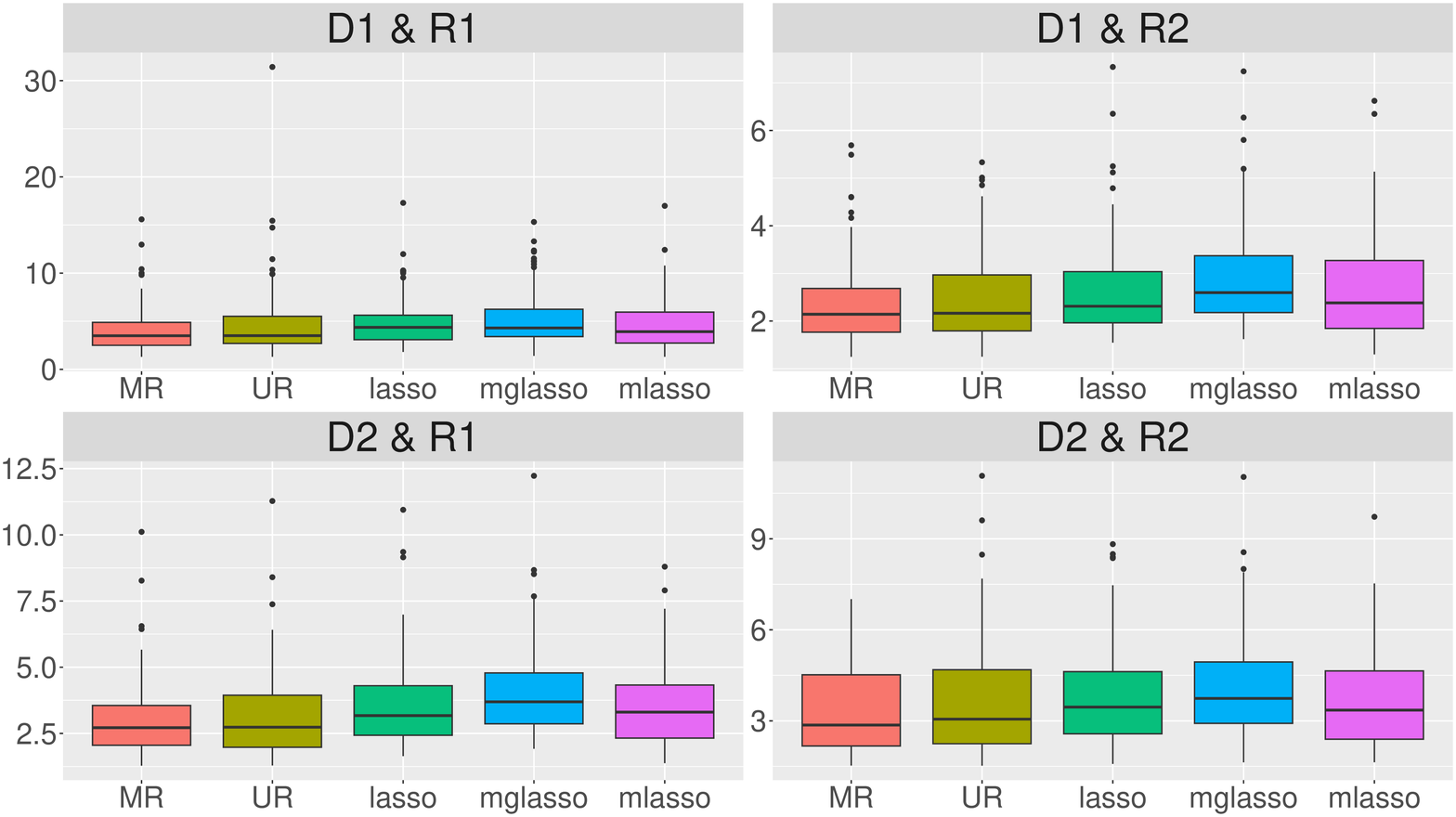}
\vspace{-3.5mm}
\subcaption{$s=5, \rho_x=0.9, \rho_y=0.1$}
\vspace{2.5mm}
\end{minipage}
\begin{minipage}[b]{0.5\linewidth}
\centering
\includegraphics[width=8cm,height=4.6cm]{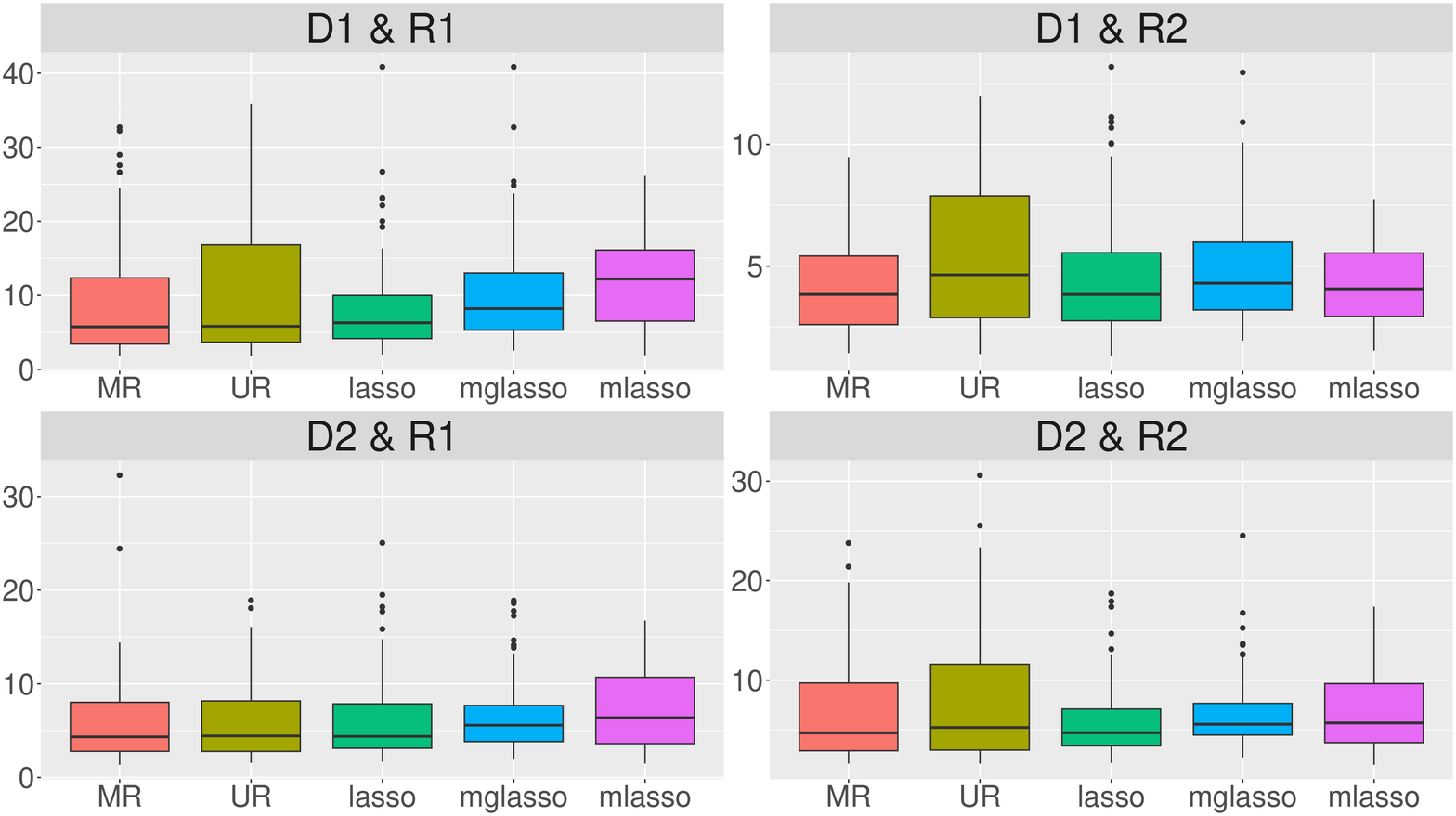}
\vspace{-3.5mm}
\subcaption{$s=50, \rho_x=0.9, \rho_y=0.1$}
\label{fig:SimuM2n15s50rx09ry01}
\vspace{2.5mm}
\end{minipage}
\begin{minipage}[b]{0.5\linewidth}
\centering
\includegraphics[width=8cm,height=4.6cm]{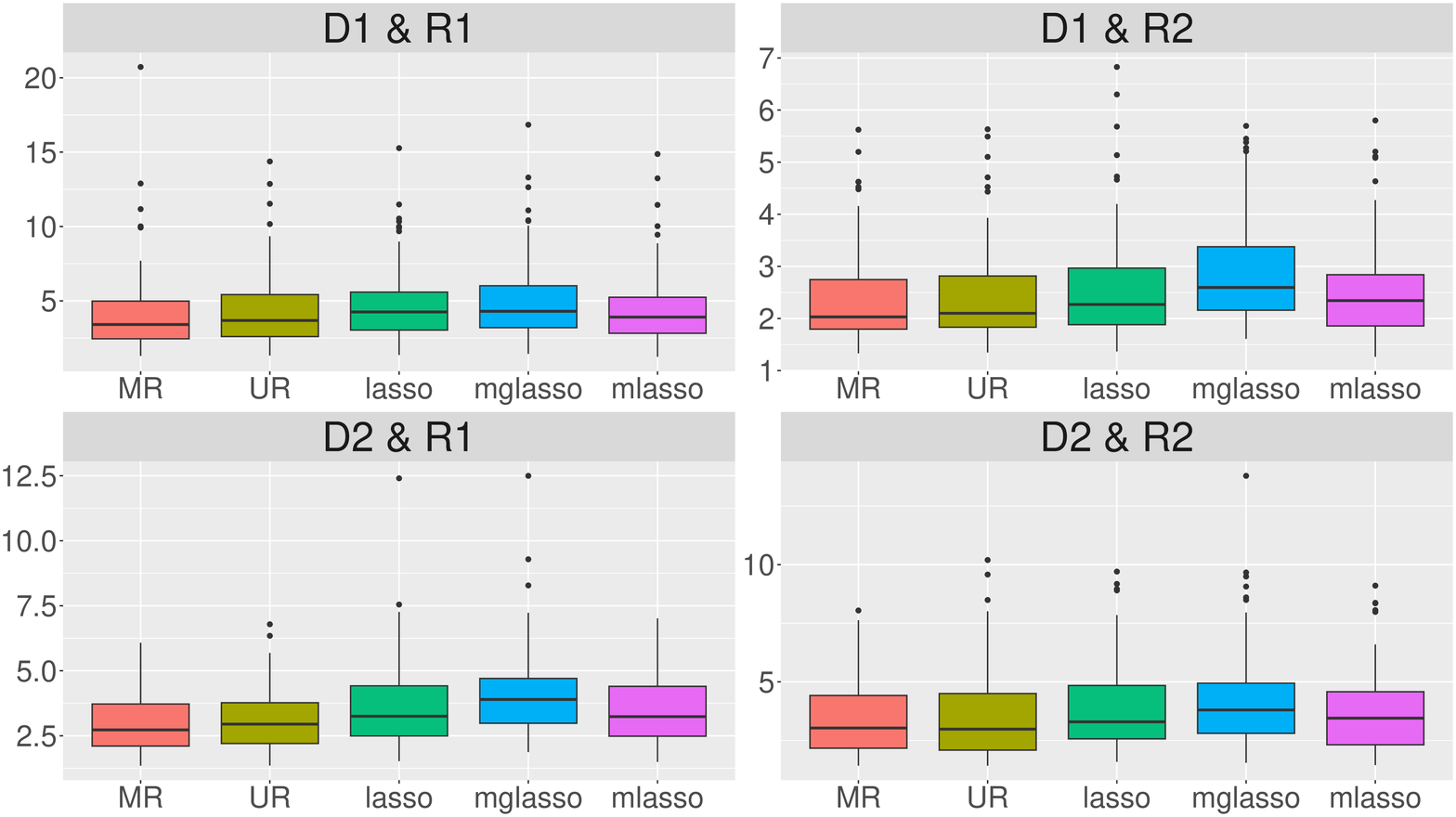}
\vspace{-3.5mm}
\subcaption{$s=5, \rho_x=0.9, \rho_y=0.9$}
\vspace{2.5mm}
\end{minipage}
\begin{minipage}[b]{0.5\linewidth}
\centering
\includegraphics[width=8cm,height=4.6cm]{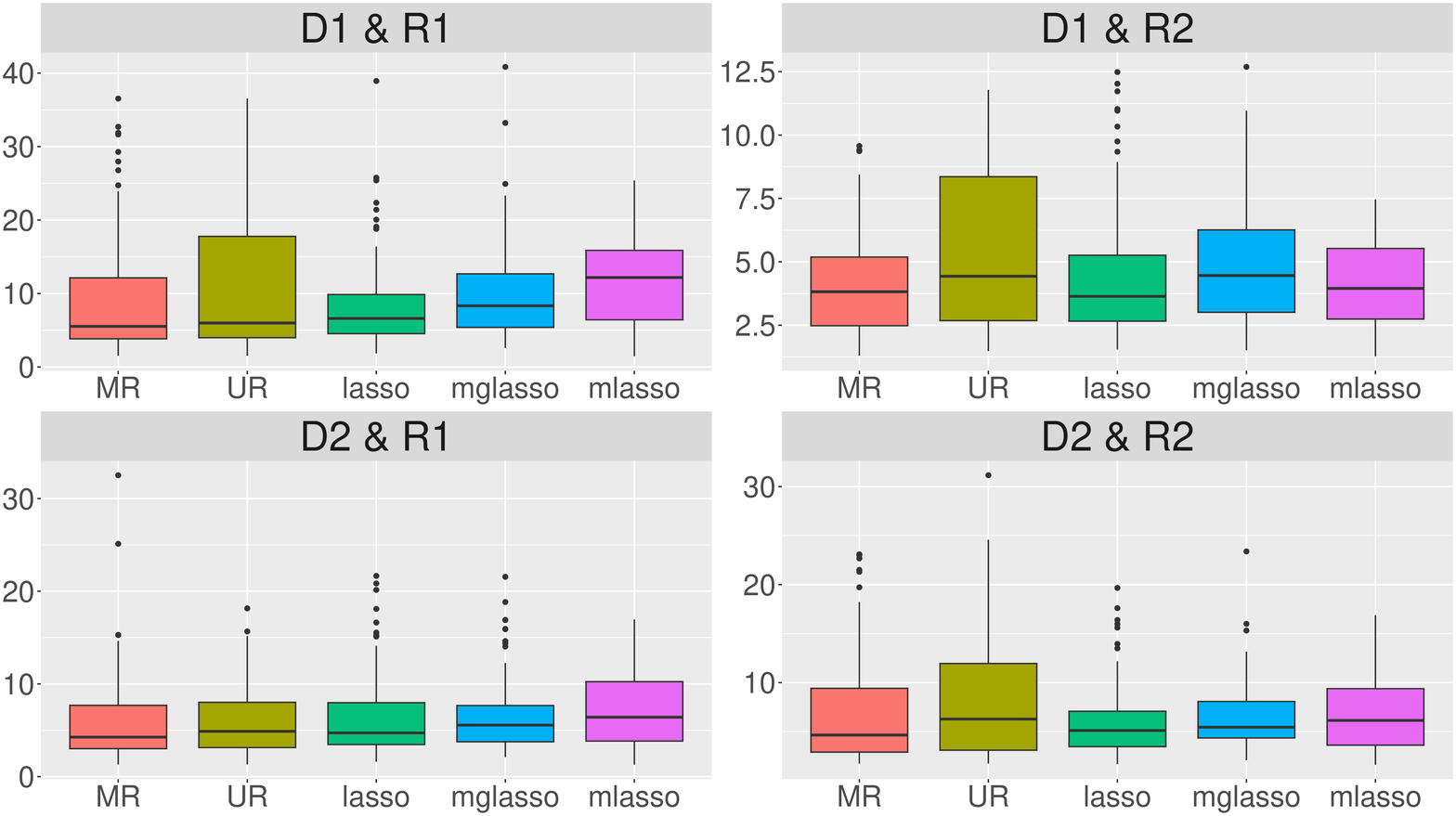}
\vspace{-3.5mm}
\subcaption{$s=50, \rho_x=0.9, \rho_y=0.9$}
\label{fig:SimuM2n15s50rx09ry09}
\vspace{2.5mm}
\end{minipage}
\caption{Boxplots of MSE for $n=15$ when the case $M=2$.
}
\label{fig:SimuM2n15}
\end{figure}

\begin{figure}[htbp]
\begin{minipage}[b]{0.5\linewidth}
\centering
\includegraphics[width=8cm,height=4.6cm]{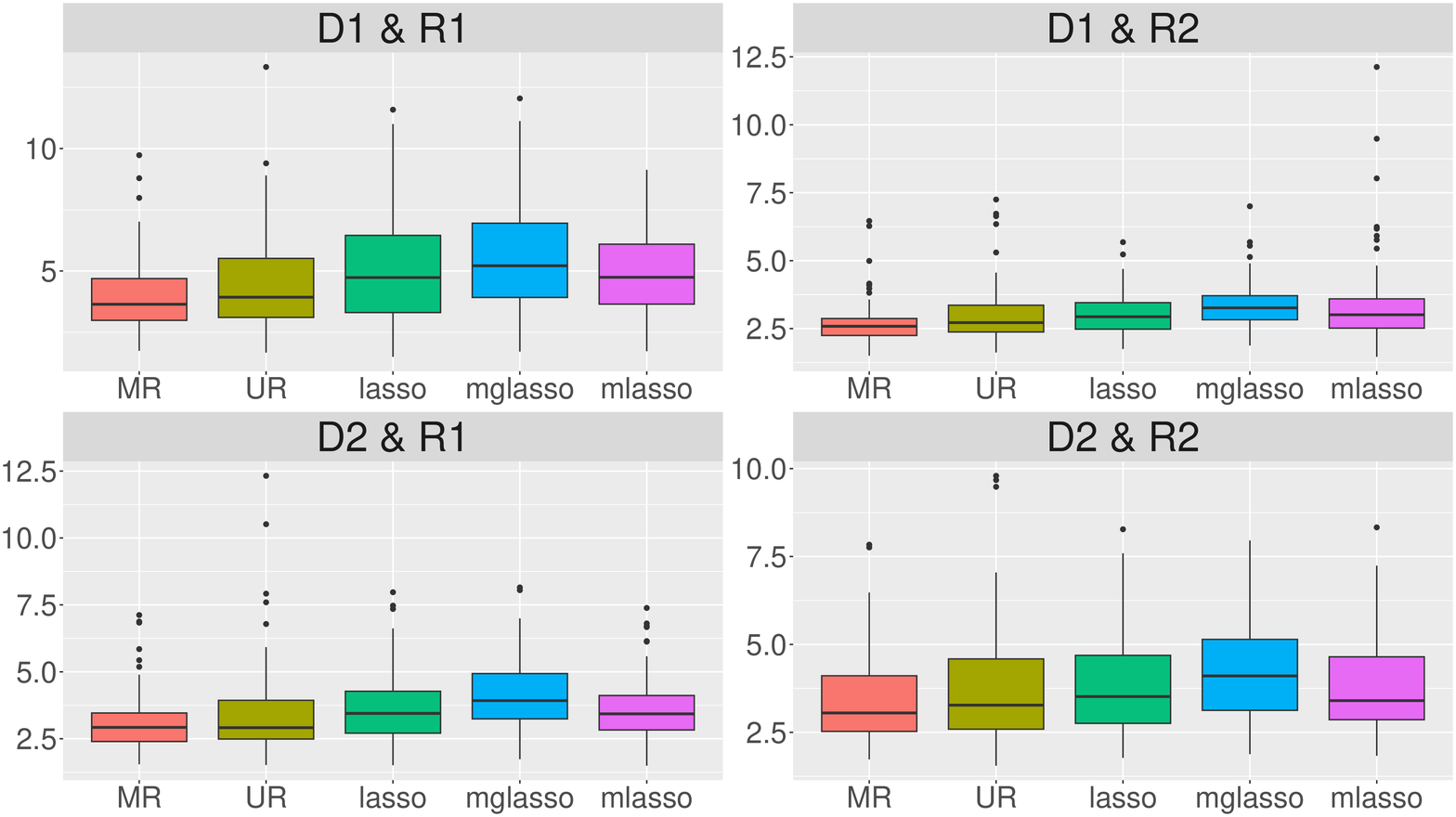}
\vspace{-3.5mm}
\subcaption{$s=5, \rho_x=0.1, \rho_y=0.1$}
\vspace{2.5mm}
\end{minipage}
\begin{minipage}[b]{0.5\linewidth}
\centering
\includegraphics[width=8cm,height=4.6cm]{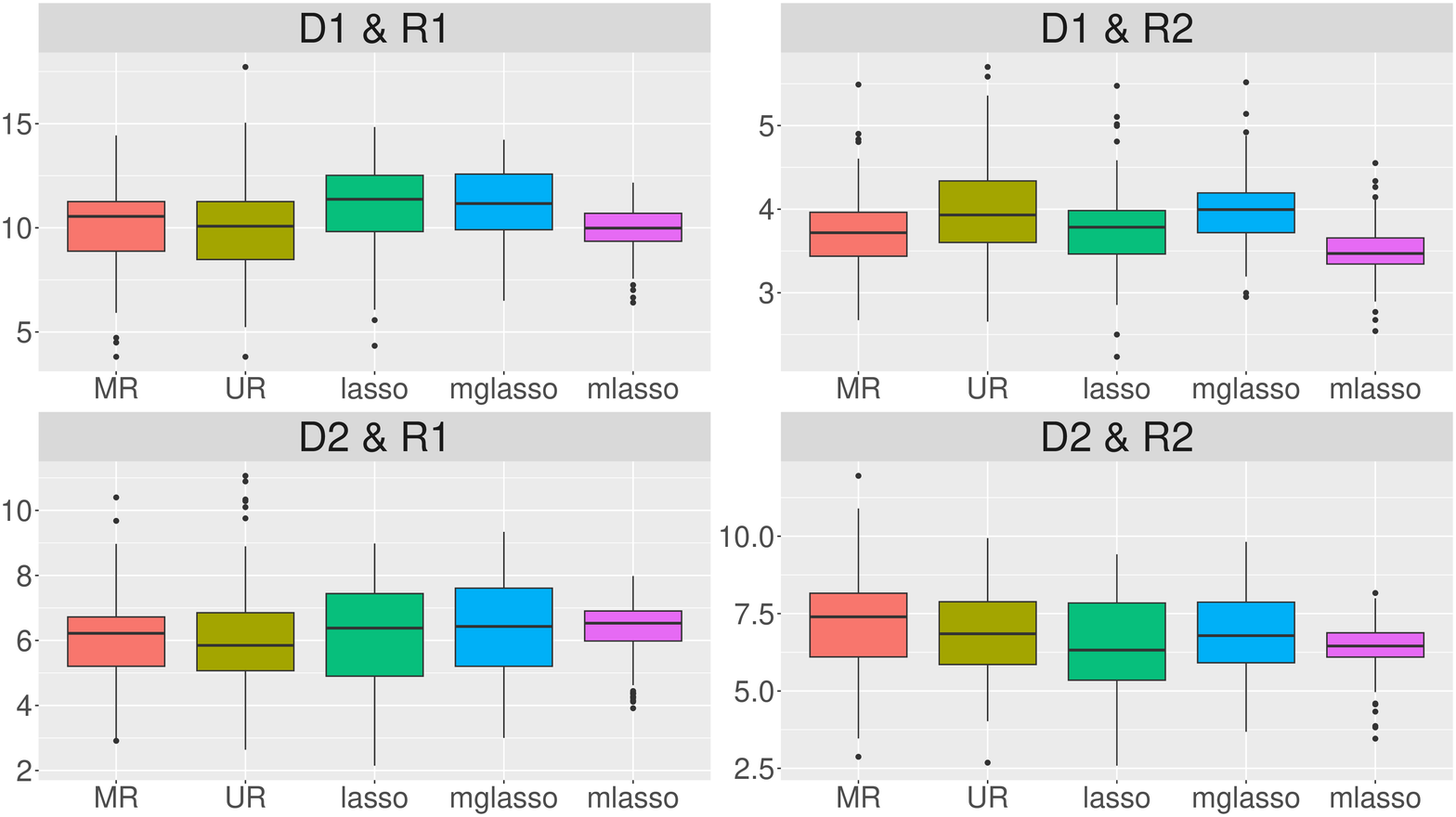} 
\vspace{-3.5mm}
\subcaption{$s=50, \rho_x=0.1, \rho_y=0.1$}
\vspace{2.5mm}
\end{minipage}
\begin{minipage}[b]{0.5\linewidth}
\centering
\includegraphics[width=8cm,height=4.6cm]{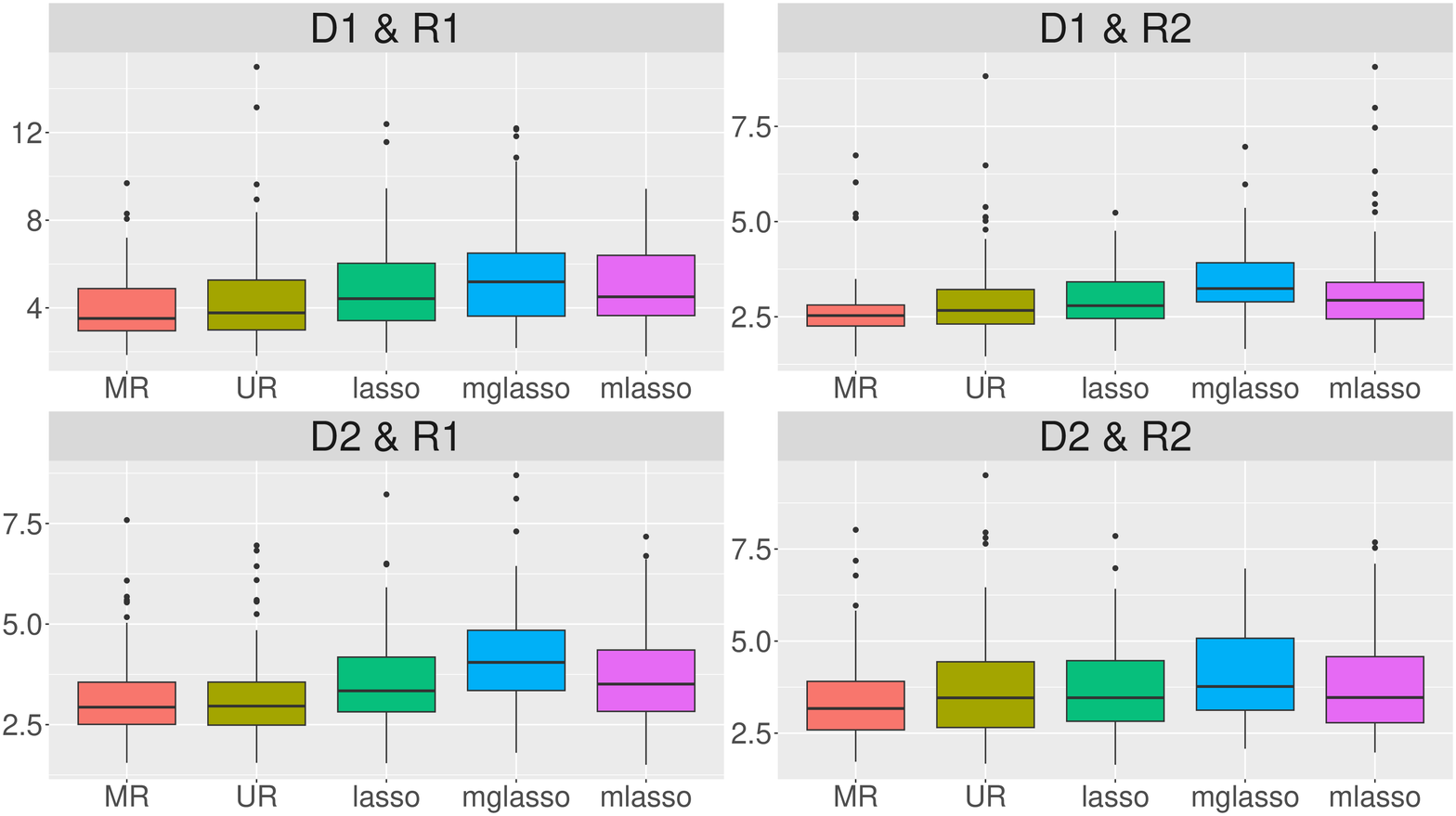} 
\vspace{-3.5mm}
\subcaption{$s=5, \rho_x=0.1, \rho_y=0.9$}
\vspace{2.5mm}
\end{minipage}
\begin{minipage}[b]{0.5\linewidth}
\centering
\includegraphics[width=8cm,height=4.6cm]{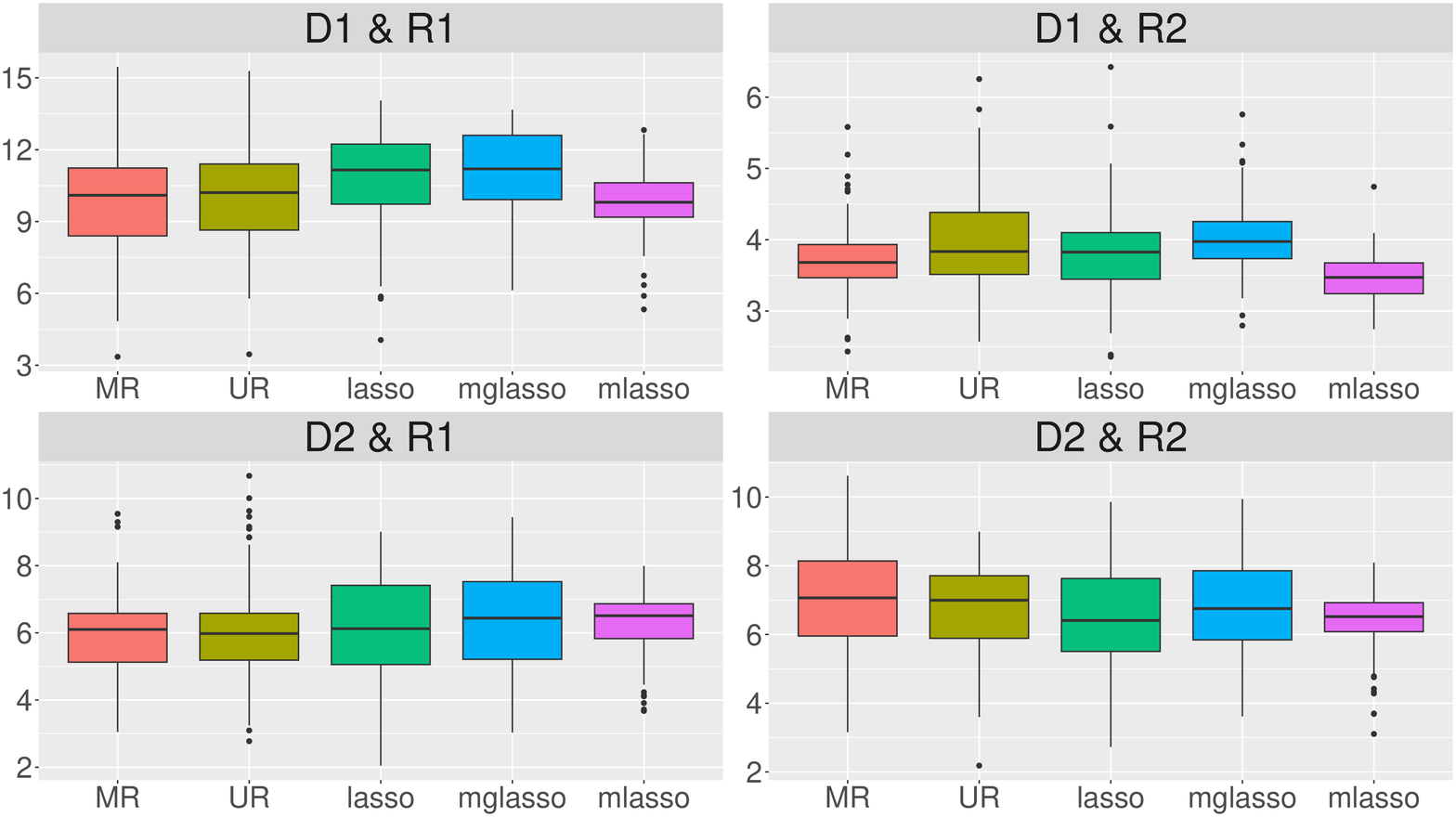}
\vspace{-3.5mm}
\subcaption{$s=50, \rho_x=0.1, \rho_y=0.9$}
\vspace{2.5mm}
\end{minipage}
\begin{minipage}[b]{0.5\linewidth}
\centering
\includegraphics[width=8cm,height=4.6cm]{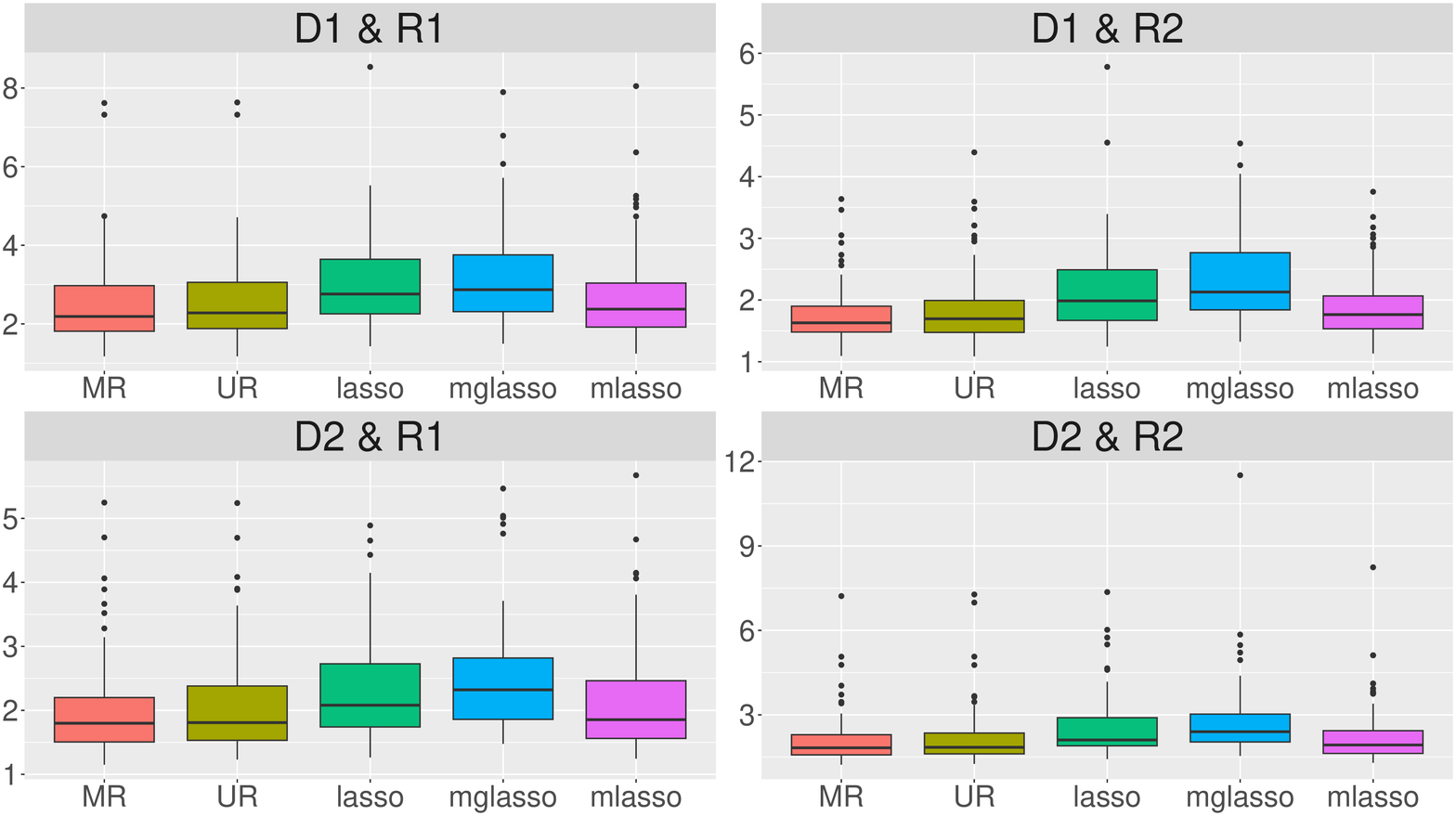}
\vspace{-3.5mm}
\subcaption{$s=5, \rho_x=0.9, \rho_y=0.1$}
\vspace{2.5mm}
\end{minipage}
\begin{minipage}[b]{0.5\linewidth}
\centering
\includegraphics[width=8cm,height=4.6cm]{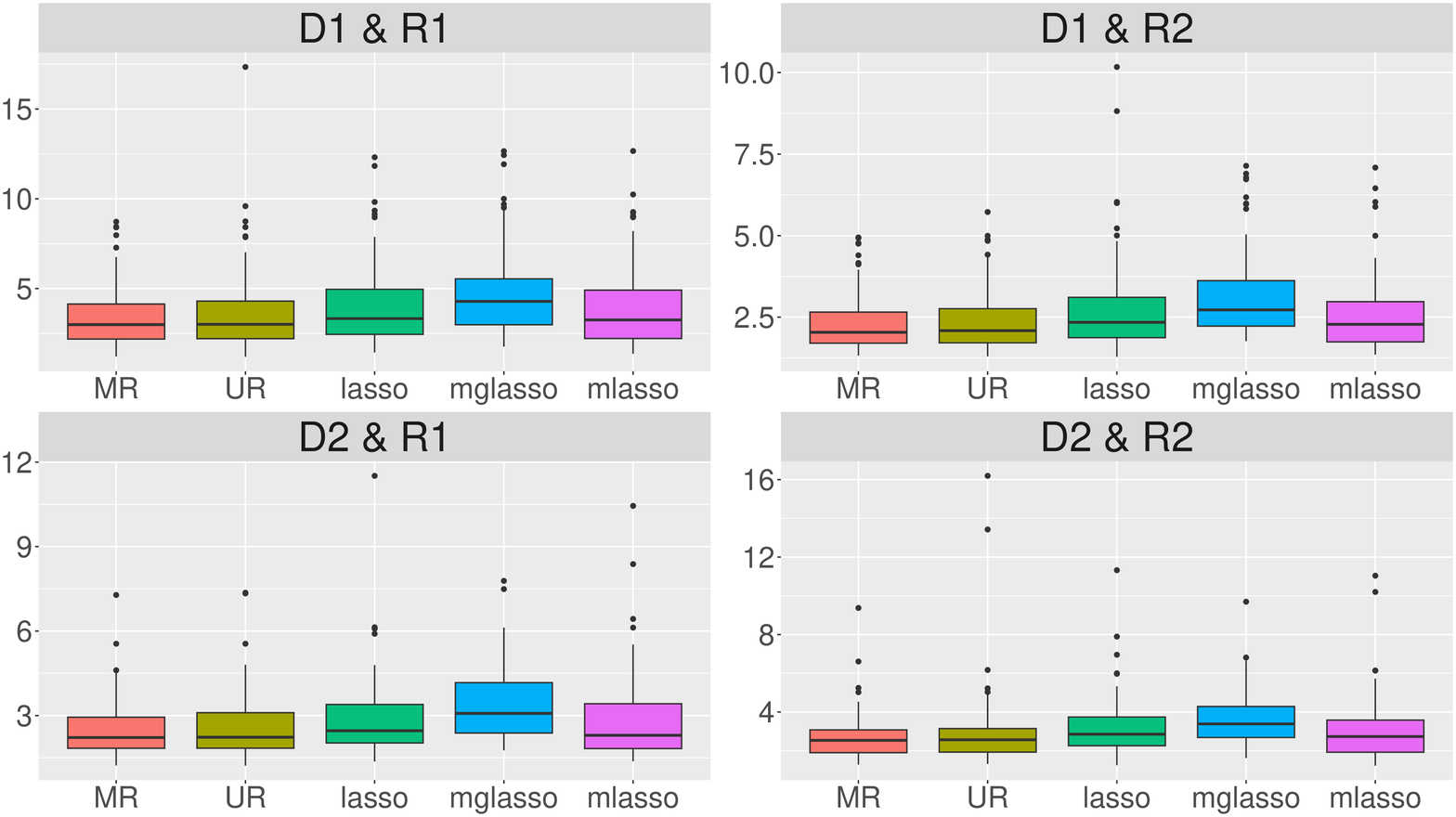}
\vspace{-3.5mm}
\subcaption{$s=50, \rho_x=0.9, \rho_y=0.1$}
\vspace{2.5mm}
\end{minipage}
\begin{minipage}[b]{0.5\linewidth}
\centering
\includegraphics[width=8cm,height=4.6cm]{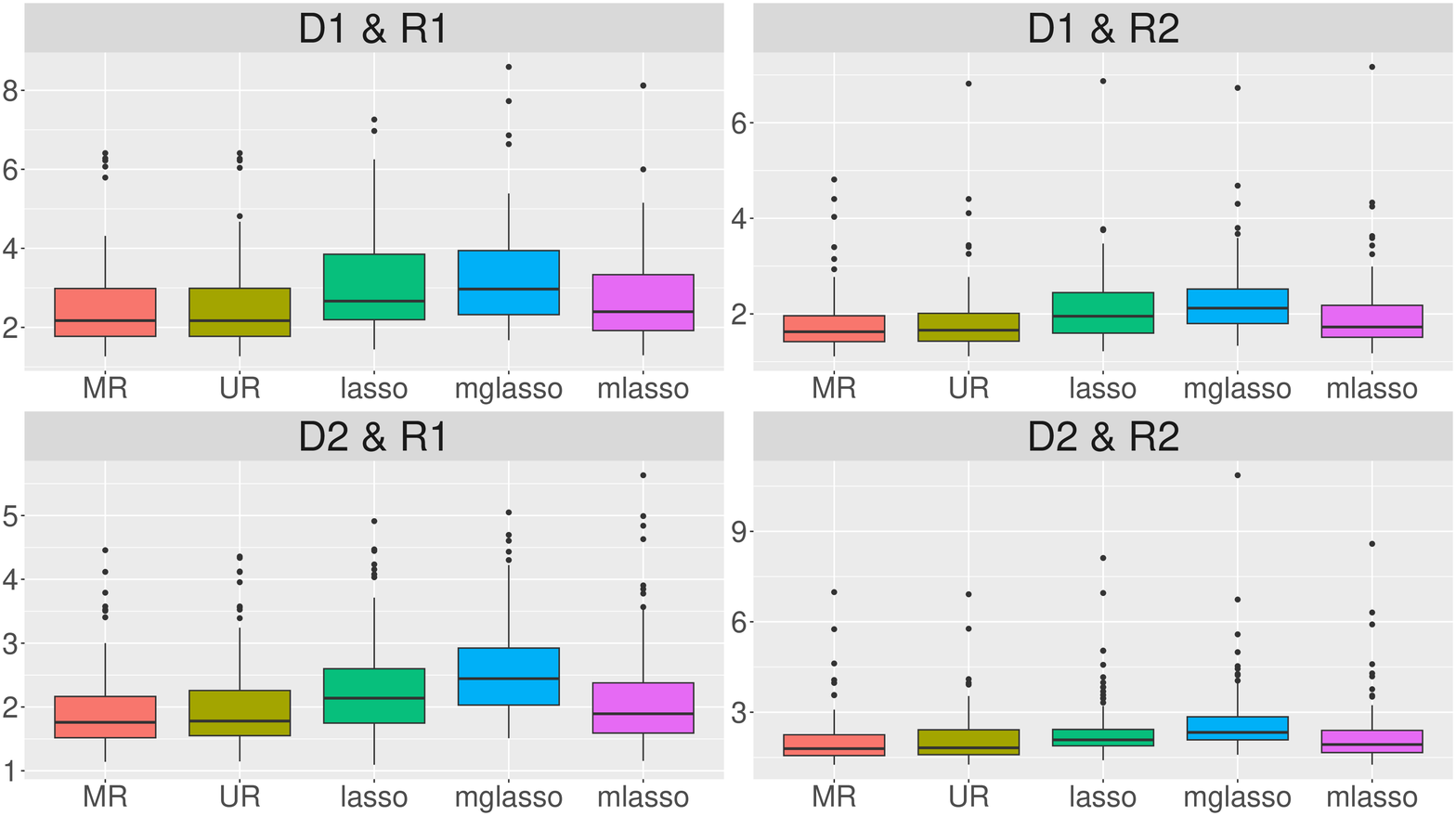}
\vspace{-3.5mm}
\subcaption{$s=5, \rho_x=0.9, \rho_y=0.9$}
\vspace{2.5mm}
\end{minipage}
\begin{minipage}[b]{0.5\linewidth}
\centering
\includegraphics[width=8cm,height=4.6cm]{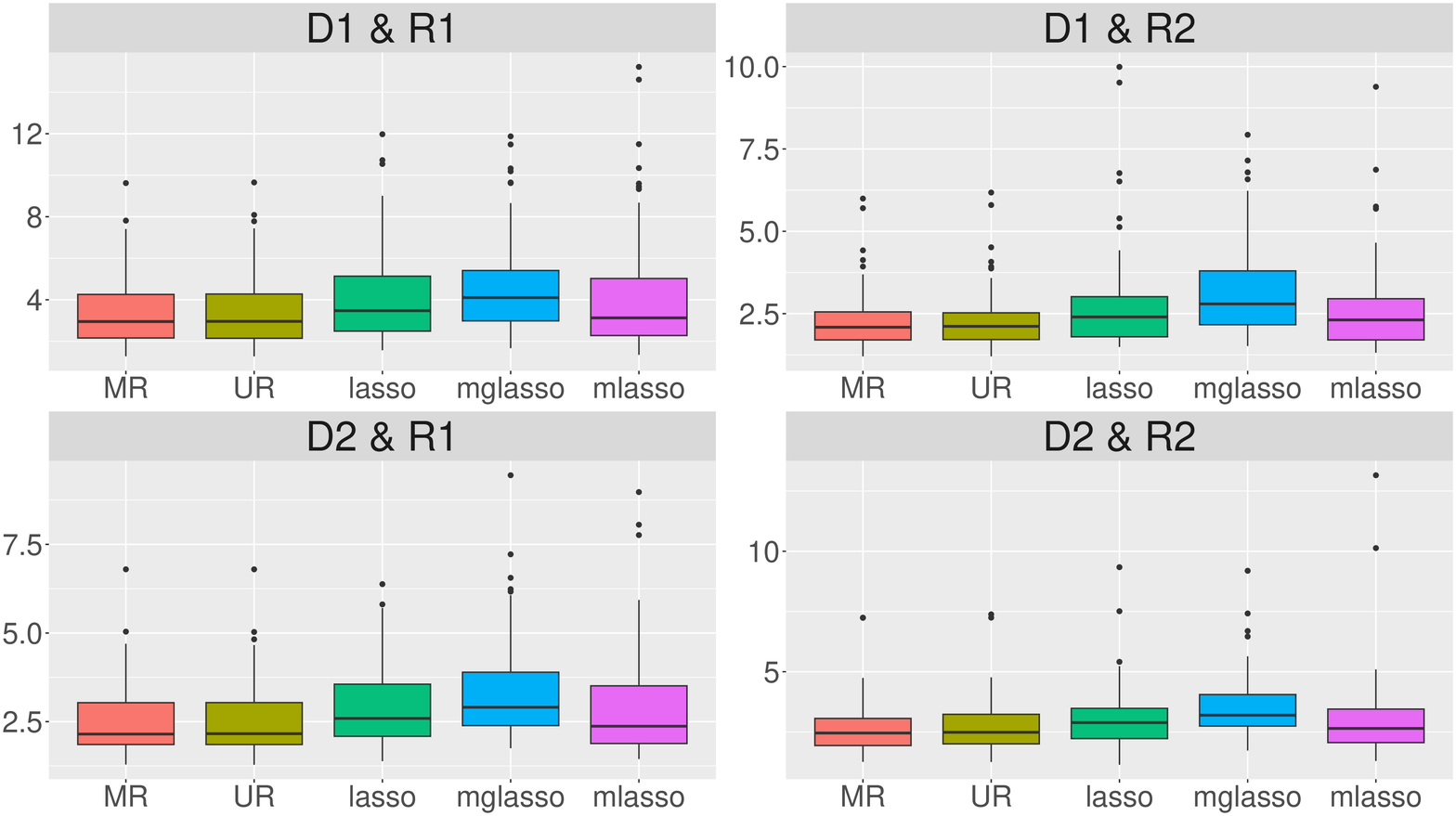}
\vspace{-3.5mm}
\subcaption{$s=50, \rho_x=0.9, \rho_y=0.9$}
\vspace{2.5mm}
\end{minipage}
\caption{Boxplots of MSE for $n=25$ when the case $M=2$.
}
\label{fig:SimuM2n25}
\end{figure}

\begin{figure}[htbp]
\begin{minipage}[b]{0.5\linewidth}
\centering
\includegraphics[width=8cm,height=4.6cm]{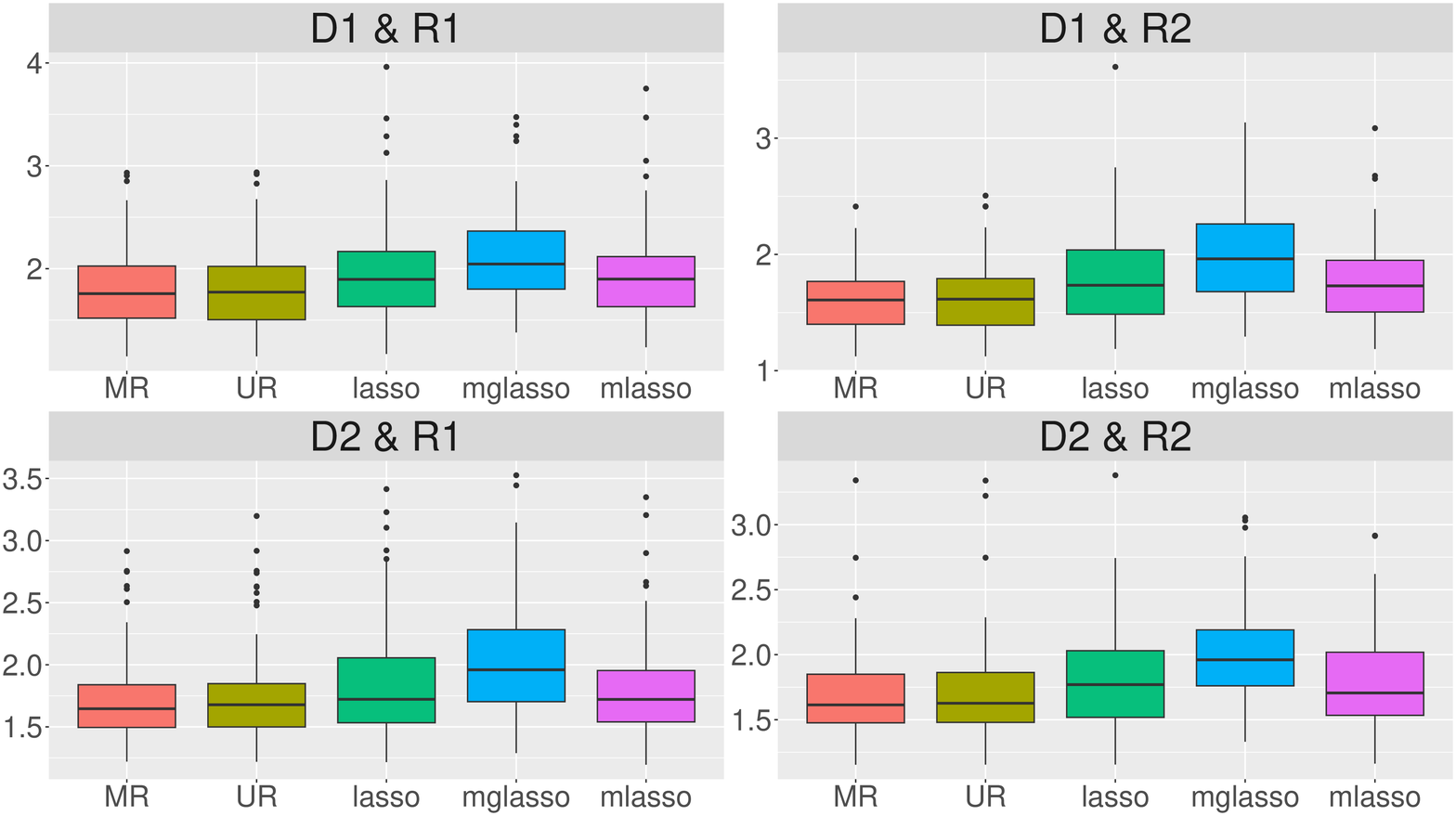}
\vspace{-3.5mm}
\subcaption{$s=5, \rho_x=0.1, \rho_y=0.1$}
\vspace{2.5mm}
\end{minipage}
\begin{minipage}[b]{0.5\linewidth}
\centering
\includegraphics[width=8cm,height=4.6cm]{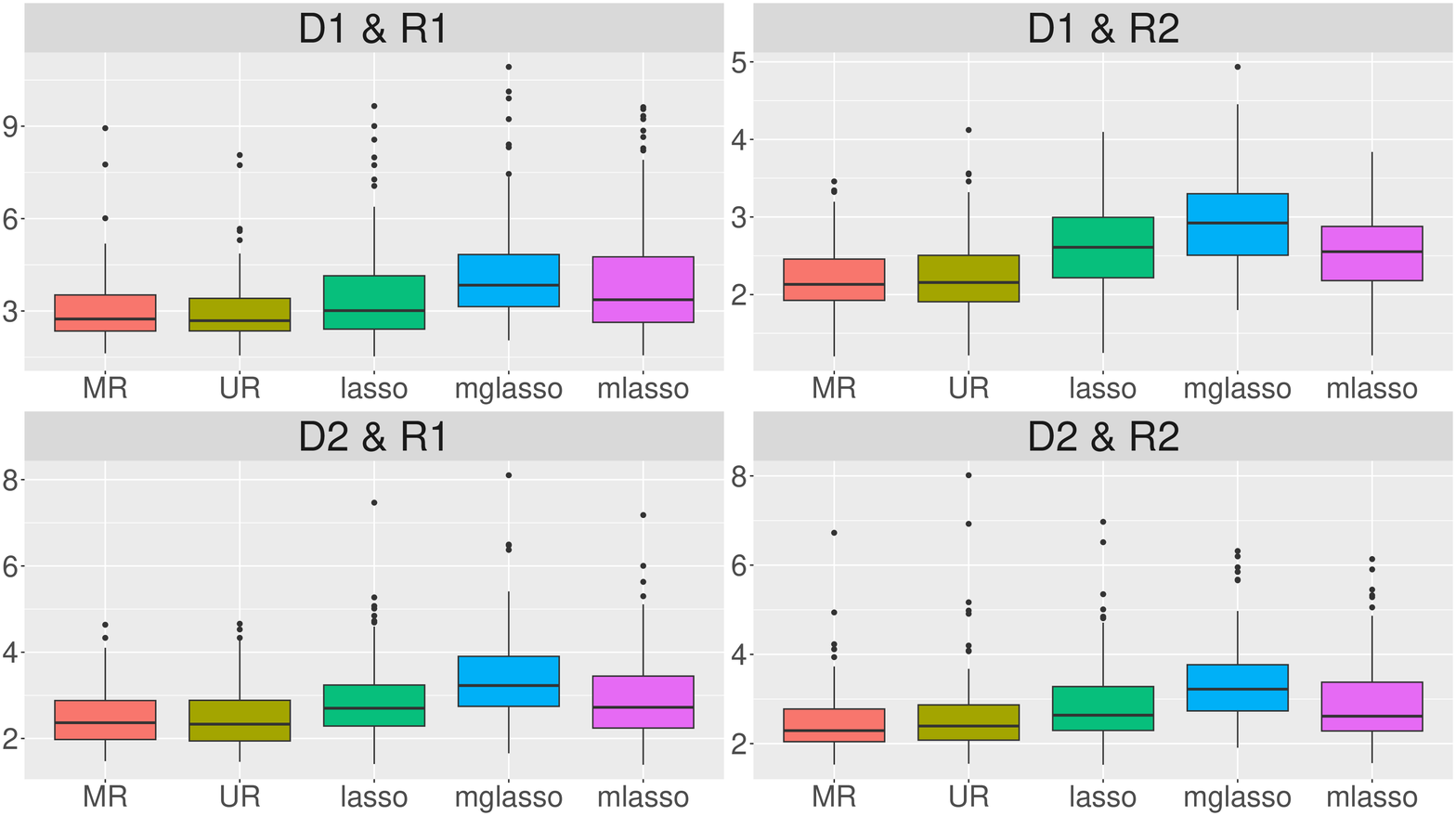} 
\vspace{-3.5mm}
\subcaption{$s=50, \rho_x=0.1, \rho_y=0.1$}
\vspace{2.5mm}
\end{minipage}
\begin{minipage}[b]{0.5\linewidth}
\centering
\includegraphics[width=8cm,height=4.6cm]{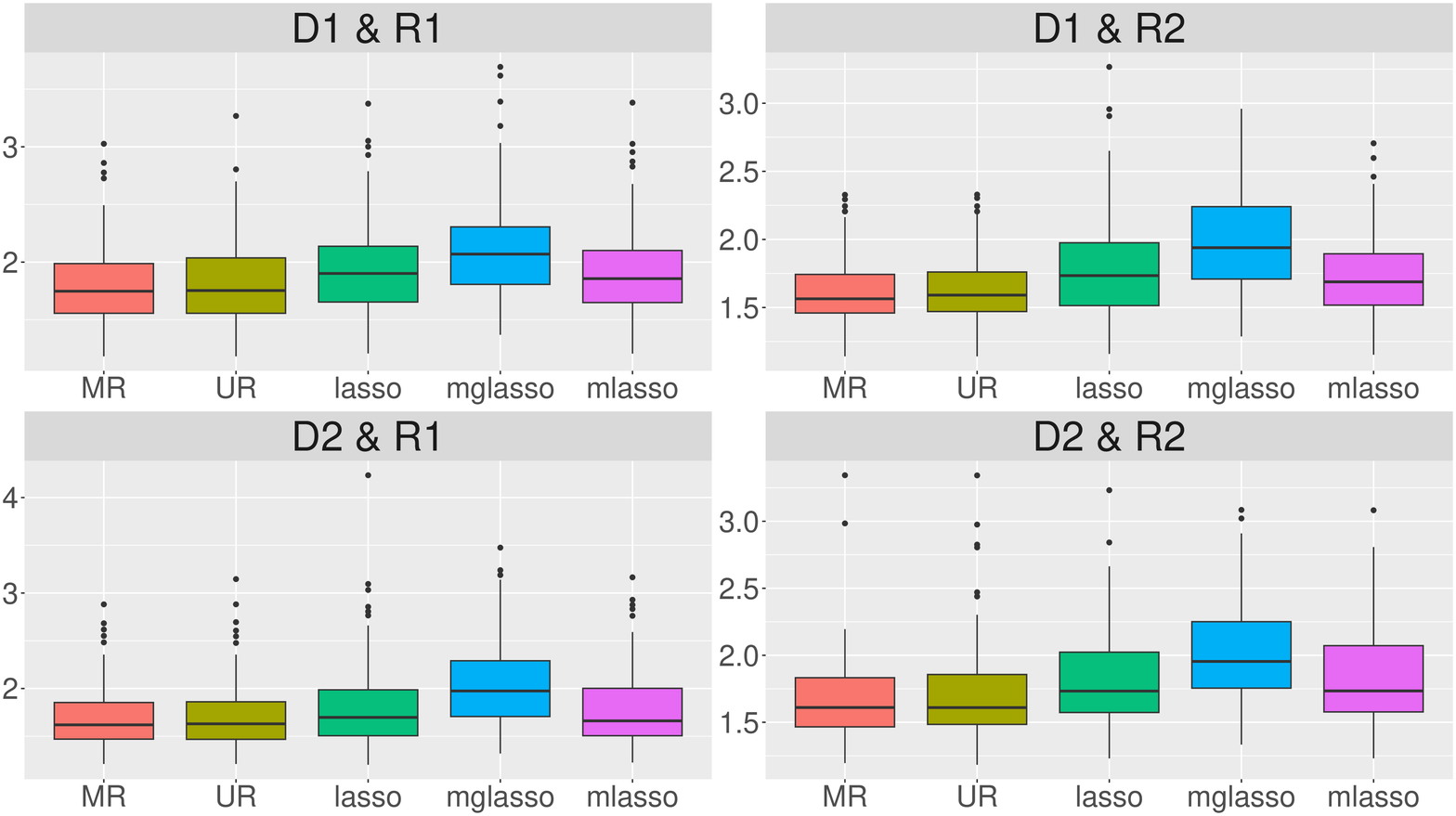} 
\vspace{-3.5mm}
\subcaption{$s=5, \rho_x=0.1, \rho_y=0.9$}
\vspace{2.5mm}
\end{minipage}
\begin{minipage}[b]{0.5\linewidth}
\centering
\includegraphics[width=8cm,height=4.6cm]{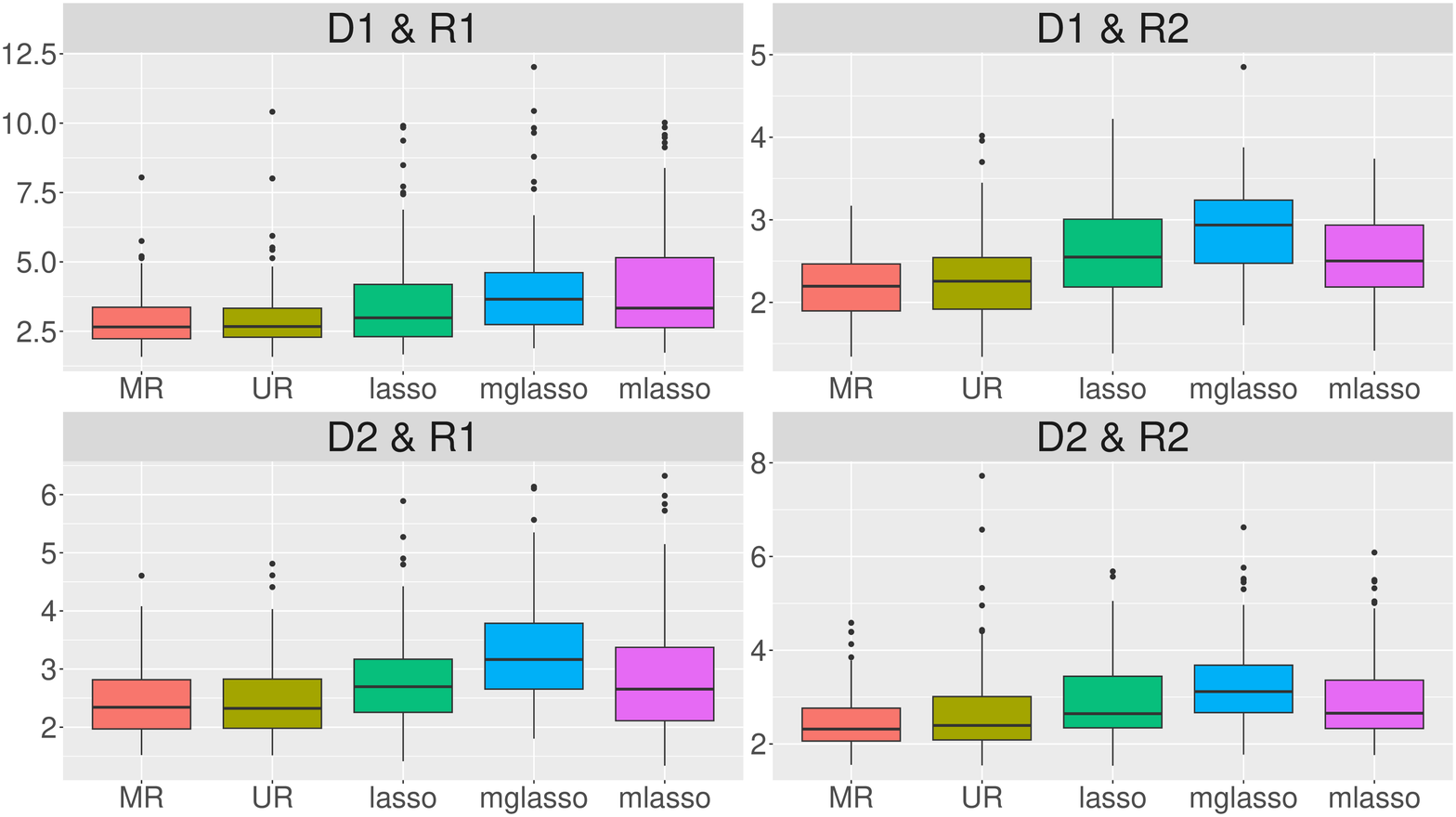}
\vspace{-3.5mm}
\subcaption{$s=50, \rho_x=0.1, \rho_y=0.9$}
\vspace{2.5mm}
\end{minipage}
\begin{minipage}[b]{0.5\linewidth}
\centering
\includegraphics[width=8cm,height=4.6cm]{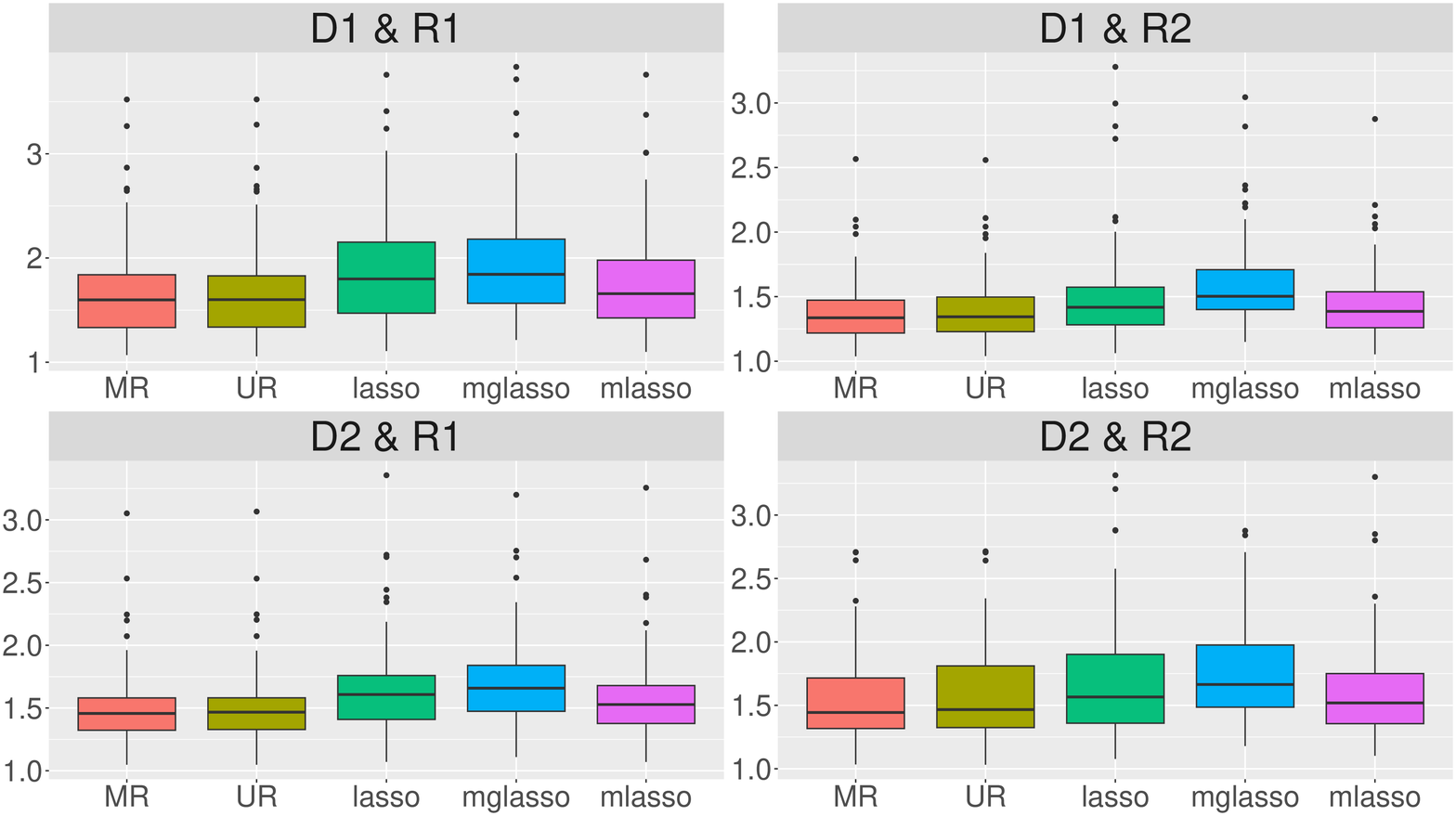}
\vspace{-3.5mm}
\subcaption{$s=5, \rho_x=0.9, \rho_y=0.1$}
\vspace{2.5mm}
\end{minipage}
\begin{minipage}[b]{0.5\linewidth}
\centering
\includegraphics[width=8cm,height=4.6cm]{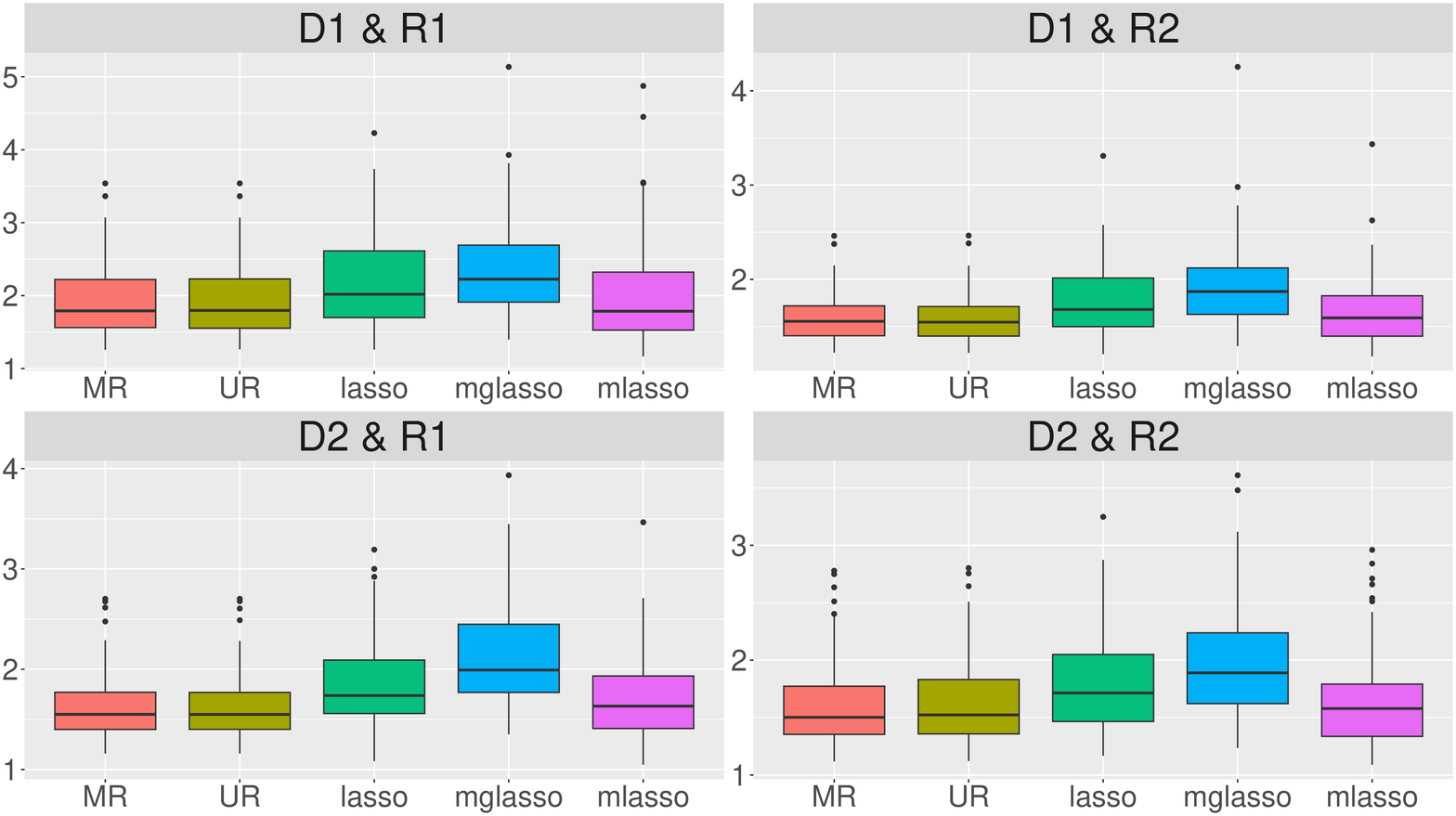}
\vspace{-3.5mm}
\subcaption{$s=50, \rho_x=0.9, \rho_y=0.1$}
\vspace{2.5mm}
\end{minipage}
\begin{minipage}[b]{0.5\linewidth}
\centering
\includegraphics[width=8cm,height=4.6cm]{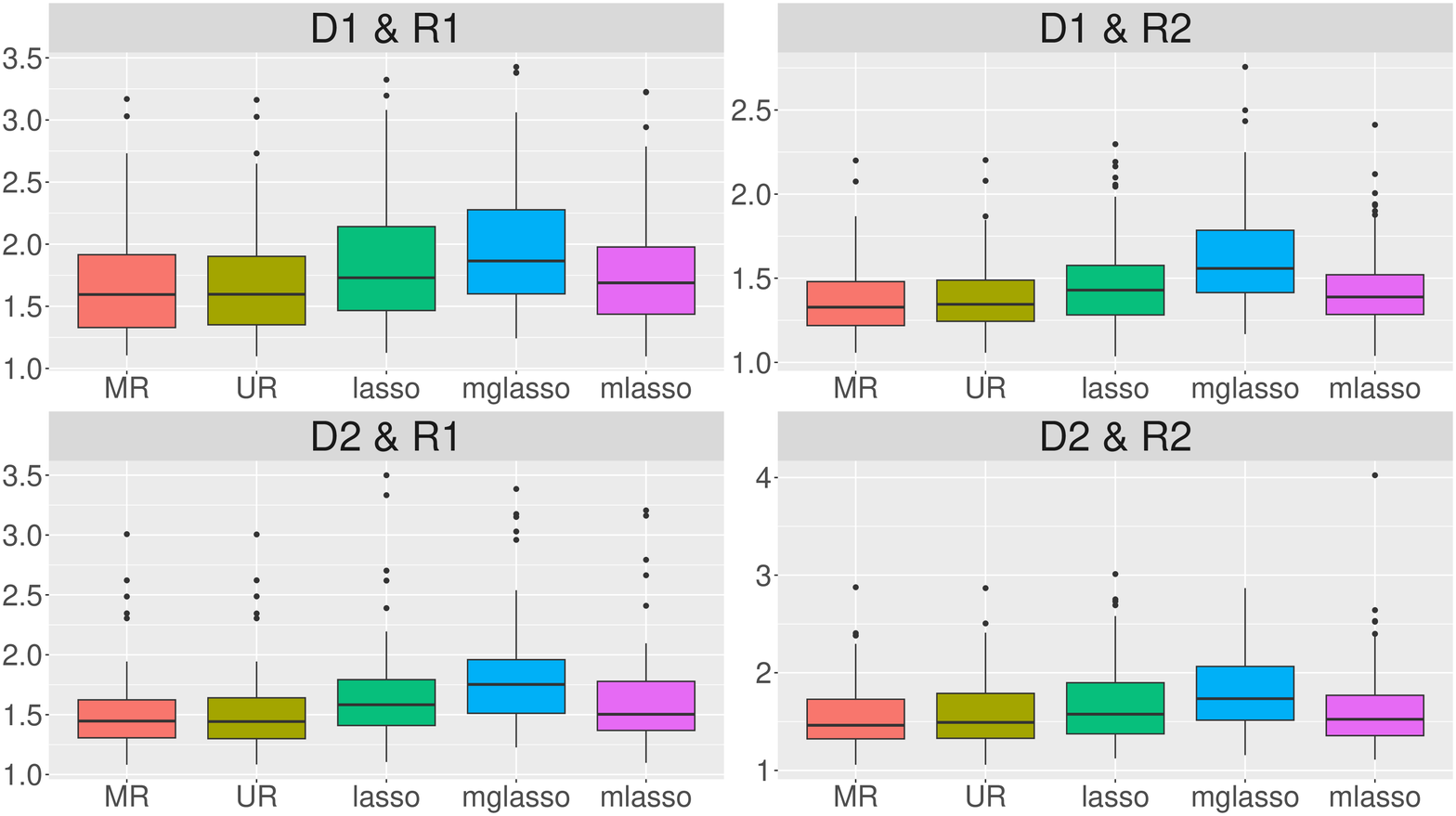}
\vspace{-3.5mm}
\subcaption{$s=5, \rho_x=0.9, \rho_y=0.9$}
\vspace{2.5mm}
\end{minipage}
\begin{minipage}[b]{0.5\linewidth}
\centering
\includegraphics[width=8cm,height=4.6cm]{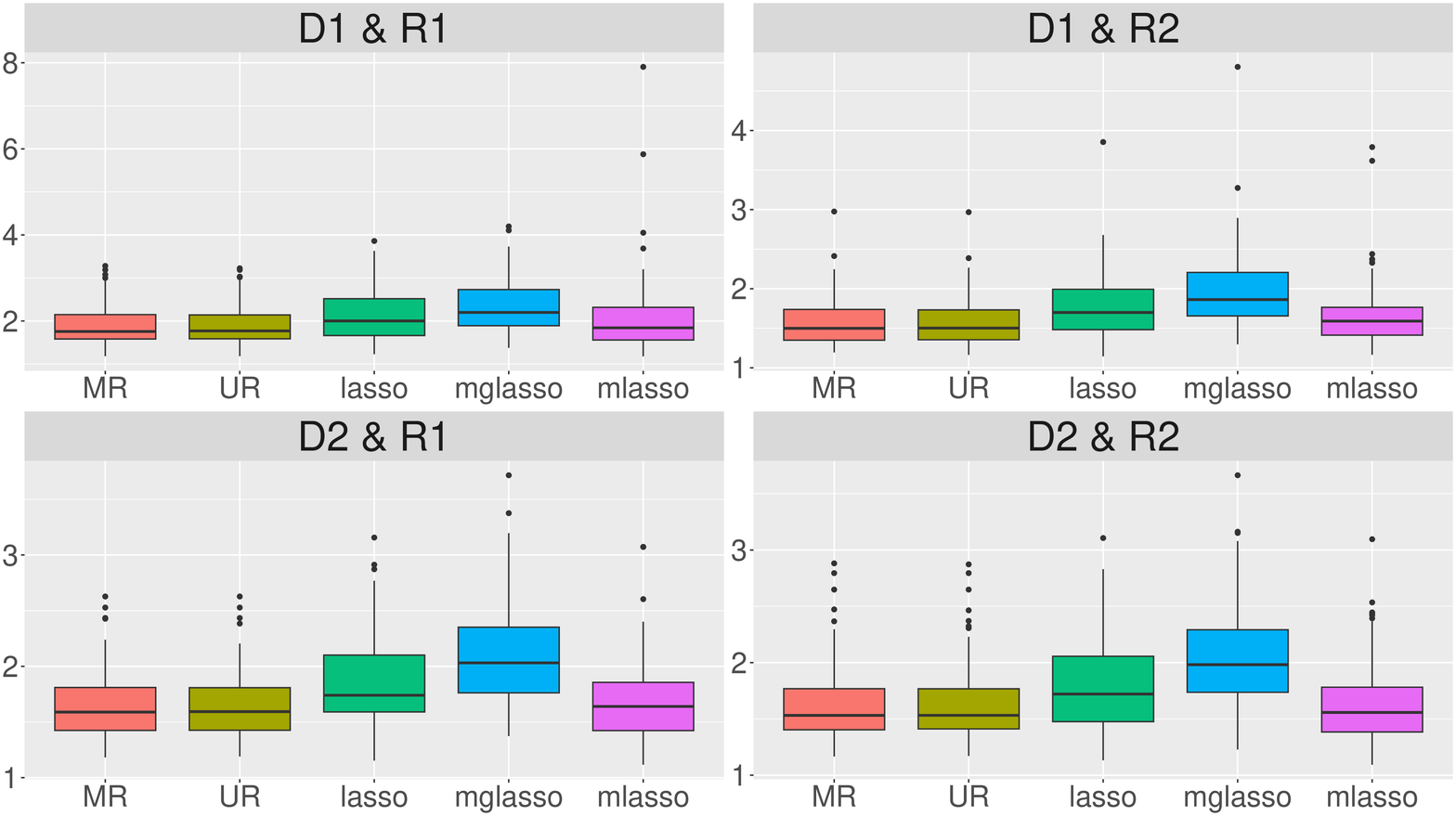}
\vspace{-3.5mm}
\subcaption{$s=50, \rho_x=0.9, \rho_y=0.9$}
\vspace{2.5mm}
\end{minipage}
\caption{Boxplots of MSE for $n=50$ when the case $M=2$.
}
\label{fig:SimuM2n50}
\end{figure}

\begin{figure}[htbp]
\begin{minipage}[b]{0.5\linewidth}
\centering
\includegraphics[width=8cm,height=4.6cm]{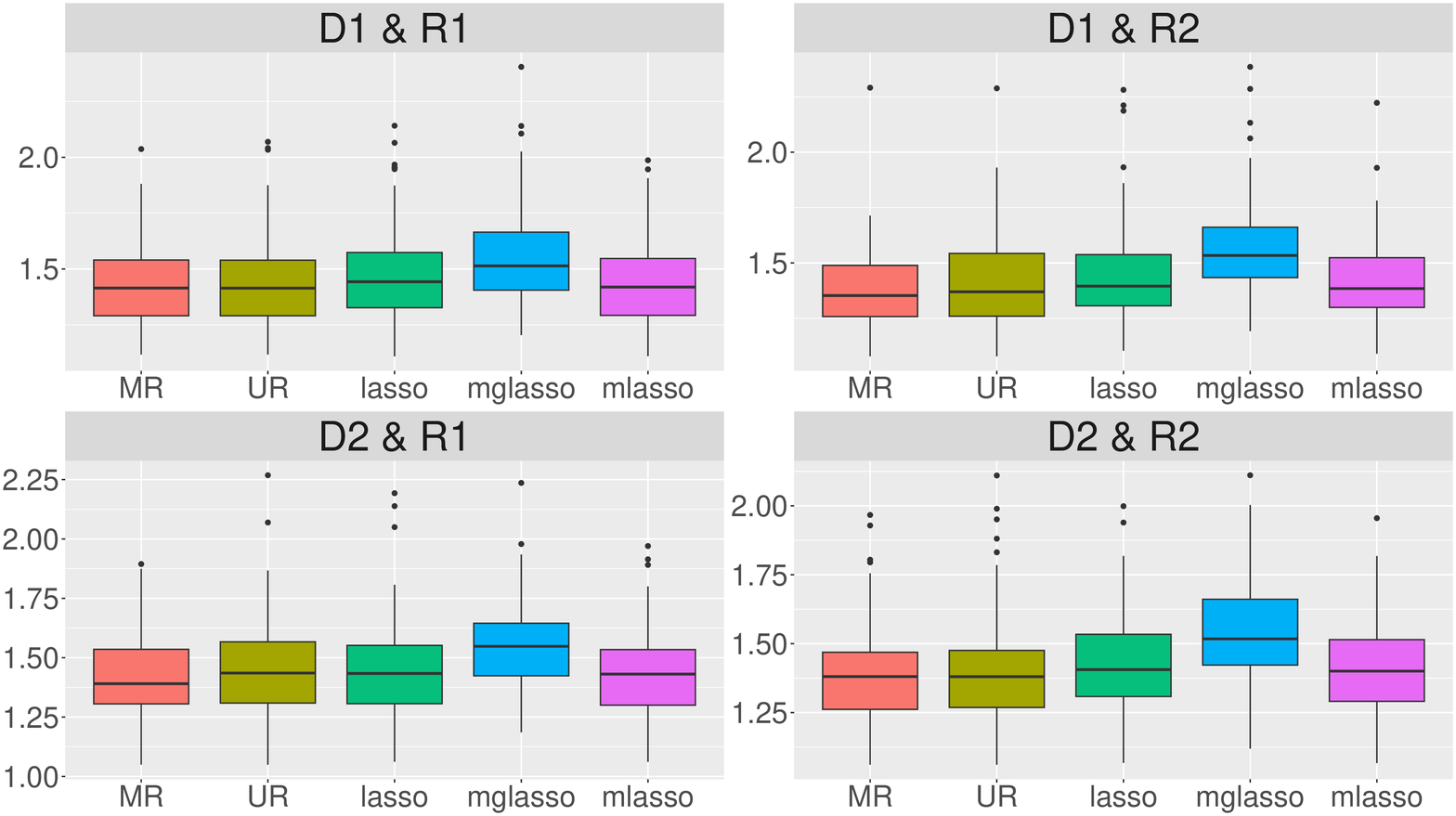}
\vspace{-3.5mm}
\subcaption{$s=5, \rho_x=0.1, \rho_y=0.1$}
\vspace{2.5mm}
\end{minipage}
\begin{minipage}[b]{0.5\linewidth}
\centering
\includegraphics[width=8cm,height=4.6cm]{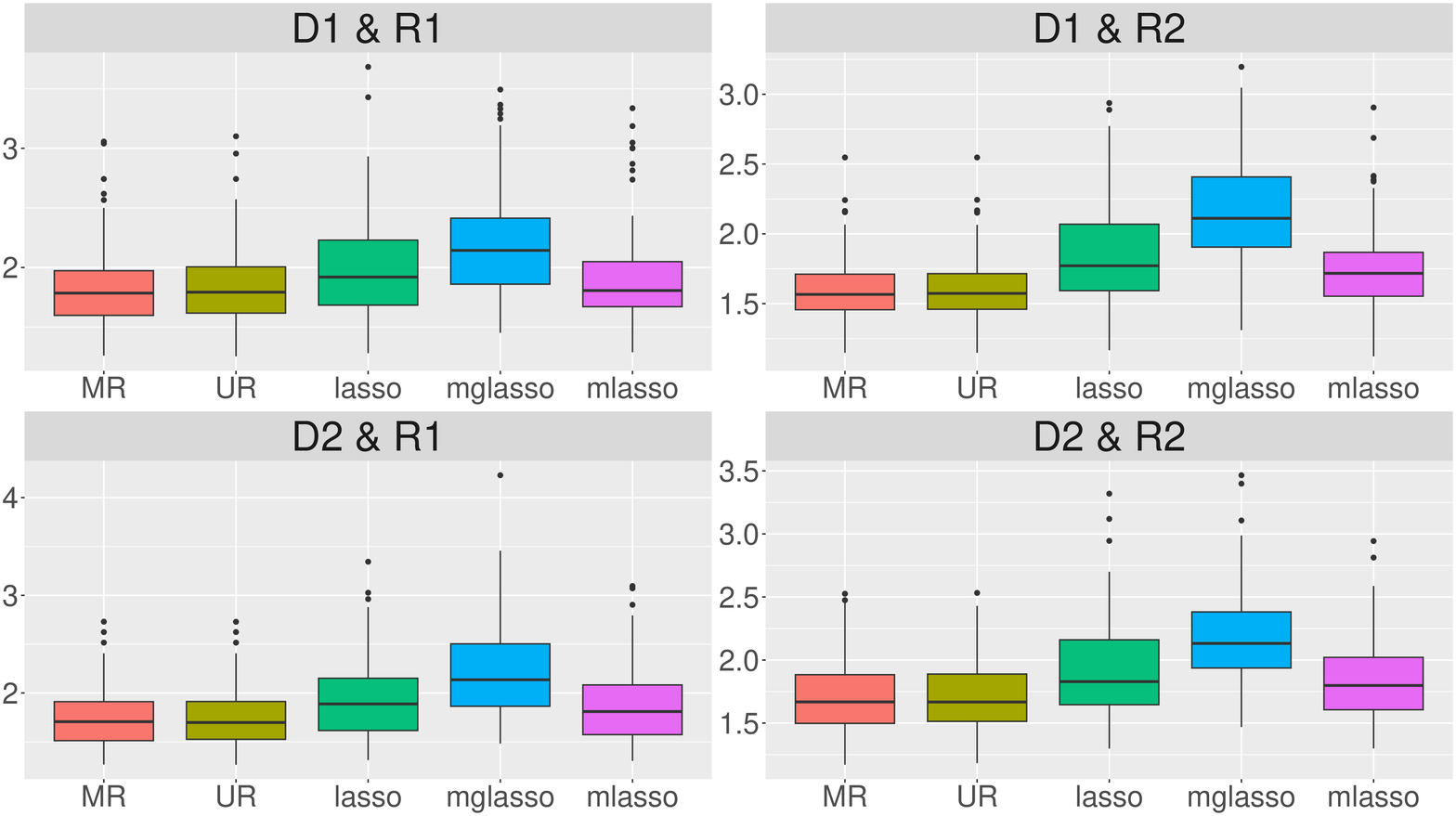} 
\vspace{-3.5mm}
\subcaption{$s=50, \rho_x=0.1, \rho_y=0.1$}
\vspace{2.5mm}
\end{minipage}
\begin{minipage}[b]{0.5\linewidth}
\centering
\includegraphics[width=8cm,height=4.6cm]{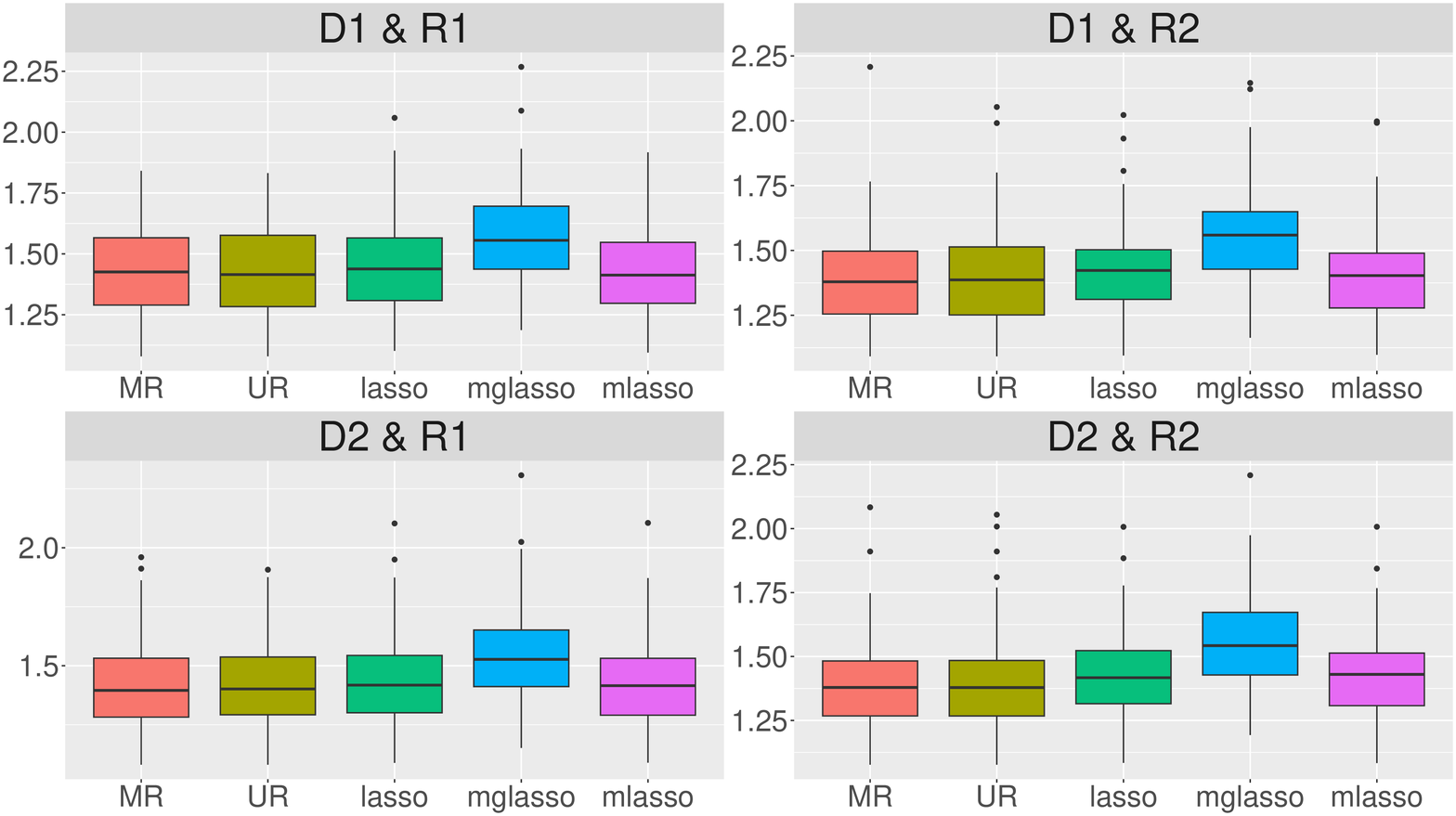} 
\vspace{-3.5mm}
\subcaption{$s=5, \rho_x=0.1, \rho_y=0.9$}
\vspace{2.5mm}
\end{minipage}
\begin{minipage}[b]{0.5\linewidth}
\centering
\includegraphics[width=8cm,height=4.6cm]{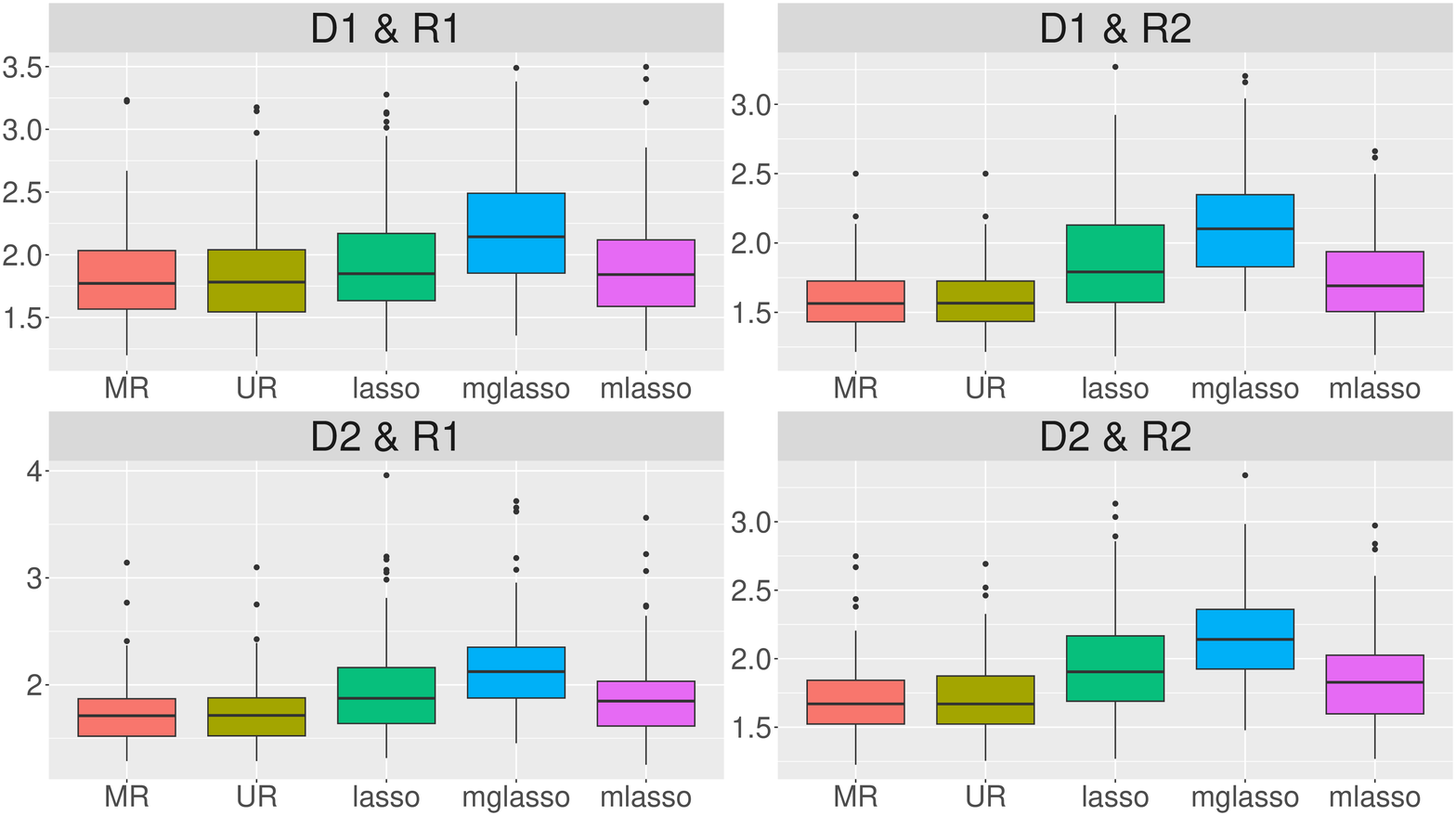}
\vspace{-3.5mm}
\subcaption{$s=50, \rho_x=0.1, \rho_y=0.9$}
\vspace{2.5mm}
\end{minipage}
\begin{minipage}[b]{0.5\linewidth}
\centering
\includegraphics[width=8cm,height=4.6cm]{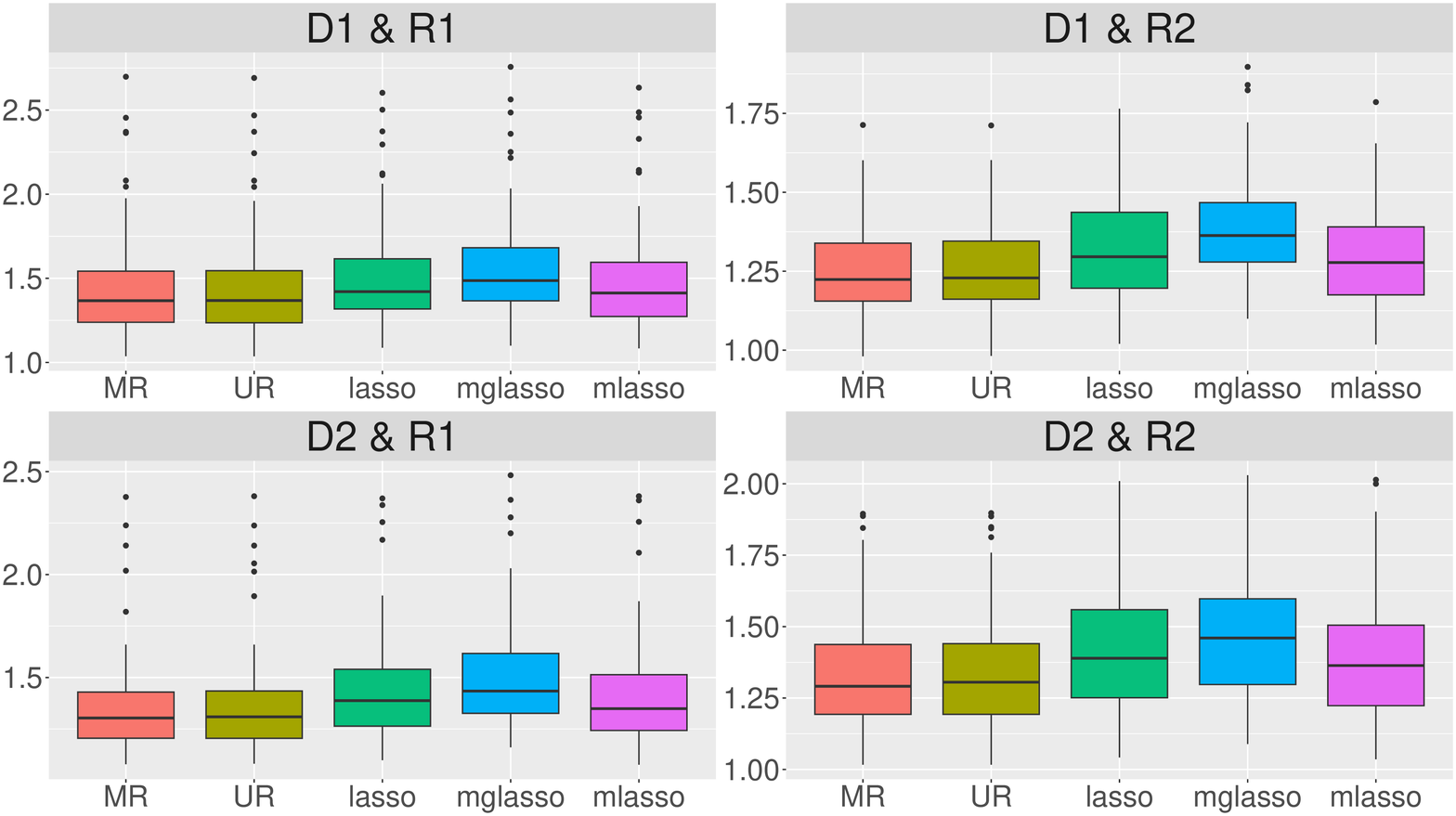}
\vspace{-3.5mm}
\subcaption{$s=5, \rho_x=0.9, \rho_y=0.1$}
\vspace{2.5mm}
\end{minipage}
\begin{minipage}[b]{0.5\linewidth}
\centering
\includegraphics[width=8cm,height=4.6cm]{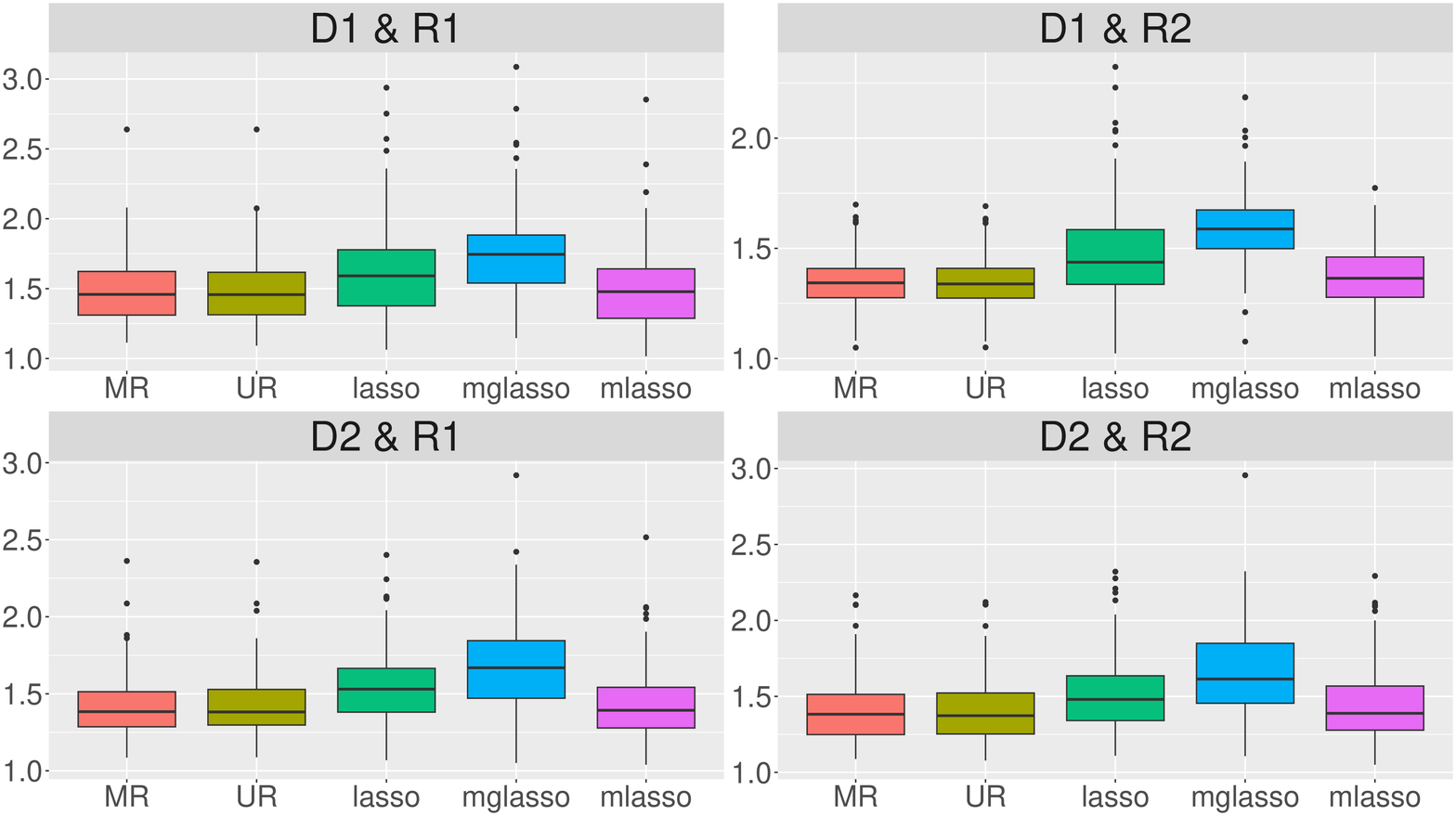}
\vspace{-3.5mm}
\subcaption{$s=50, \rho_x=0.9, \rho_y=0.1$}
\vspace{2.5mm}
\end{minipage}
\begin{minipage}[b]{0.5\linewidth}
\centering
\includegraphics[width=8cm,height=4.6cm]{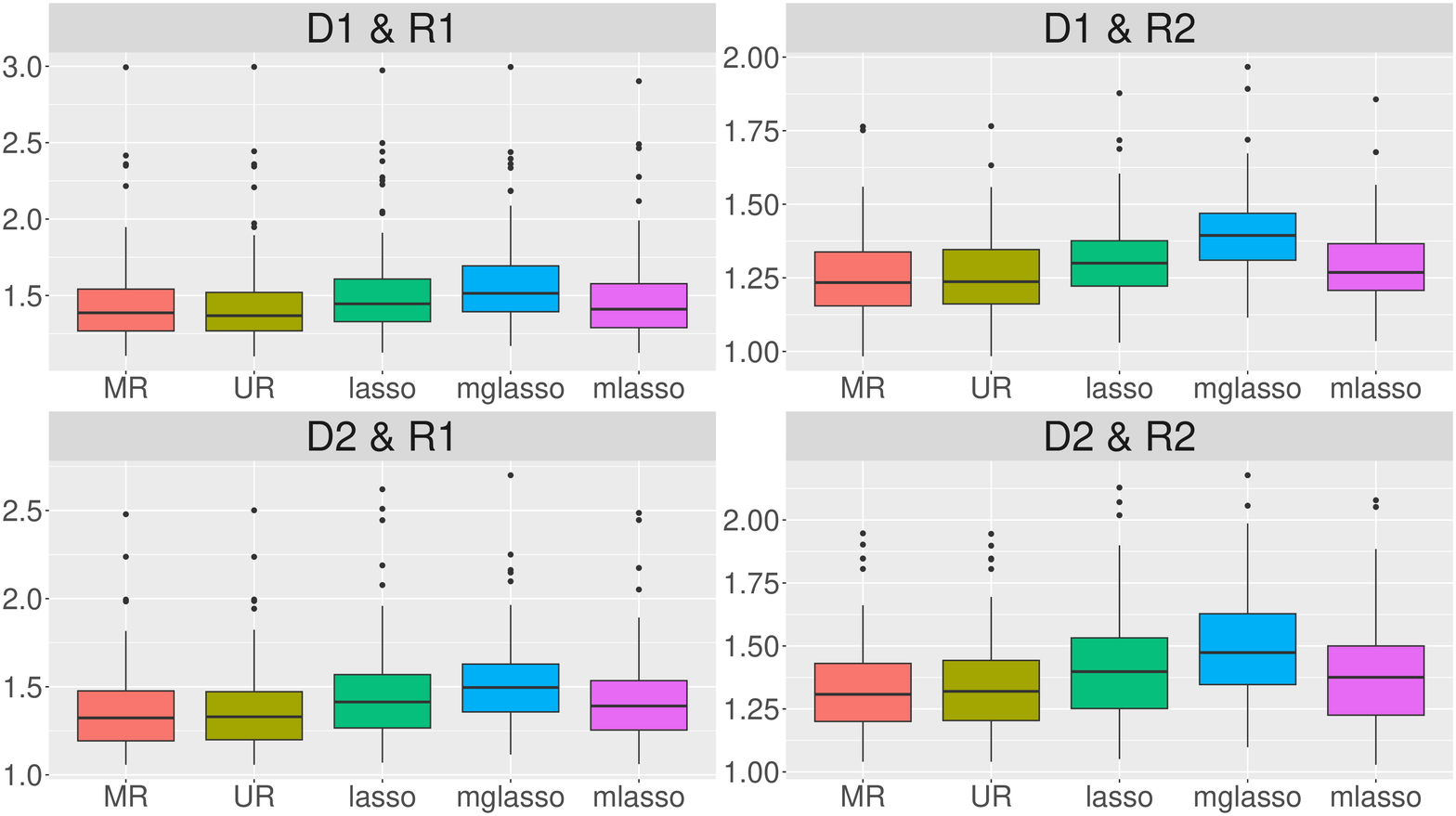}
\vspace{-3.5mm}
\subcaption{$s=5, \rho_x=0.9, \rho_y=0.9$}
\vspace{2.5mm}
\end{minipage}
\begin{minipage}[b]{0.5\linewidth}
\centering
\includegraphics[width=8cm,height=4.6cm]{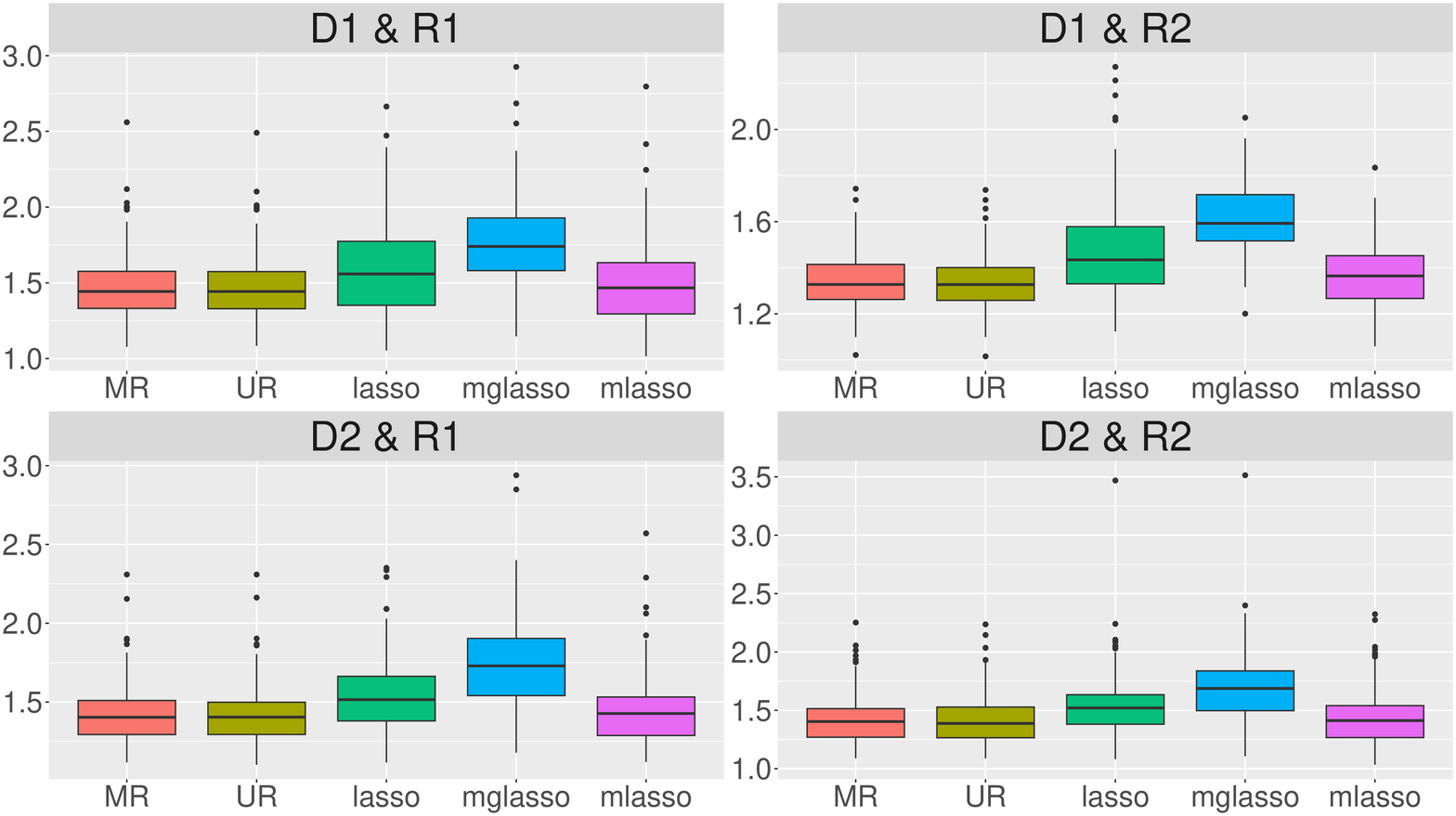}
\vspace{-3.5mm}
\subcaption{$s=50, \rho_x=0.9, \rho_y=0.9$}
\vspace{2.5mm}
\end{minipage}
\caption{Boxplots of MSE for $n=75$ when the case $M=2$.
}
\label{fig:SimuM2n75}
\end{figure}

\begin{figure}[htbp]
\begin{minipage}[b]{0.5\linewidth}
\centering
\includegraphics[width=8cm,height=4.6cm]{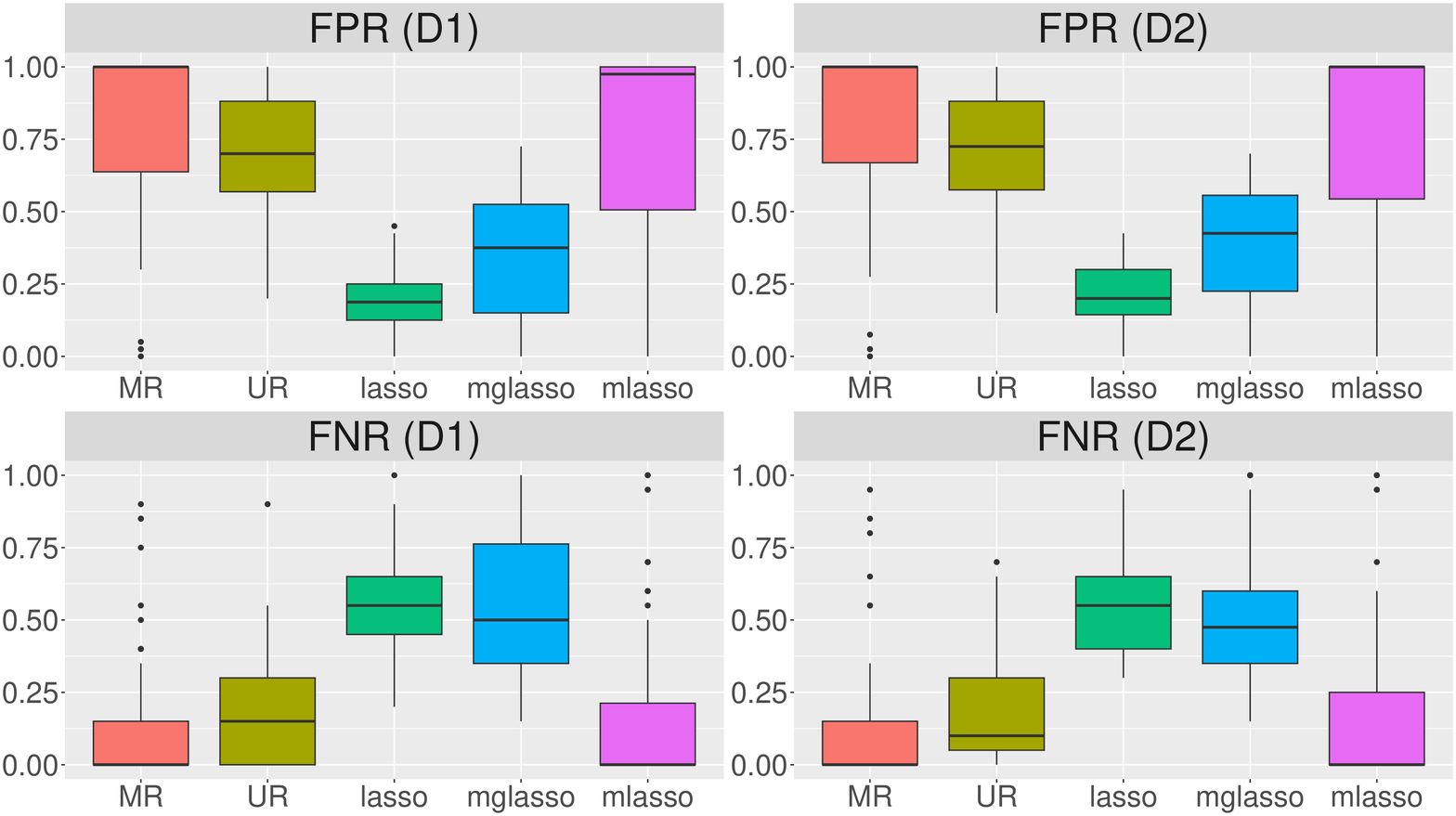}
\vspace{-3.5mm}
\subcaption{$s=5, \rho_x=0.1, \rho_y=0.1$}
\vspace{2.5mm}
\end{minipage}
\begin{minipage}[b]{0.5\linewidth}
\centering
\includegraphics[width=8cm,height=4.6cm]{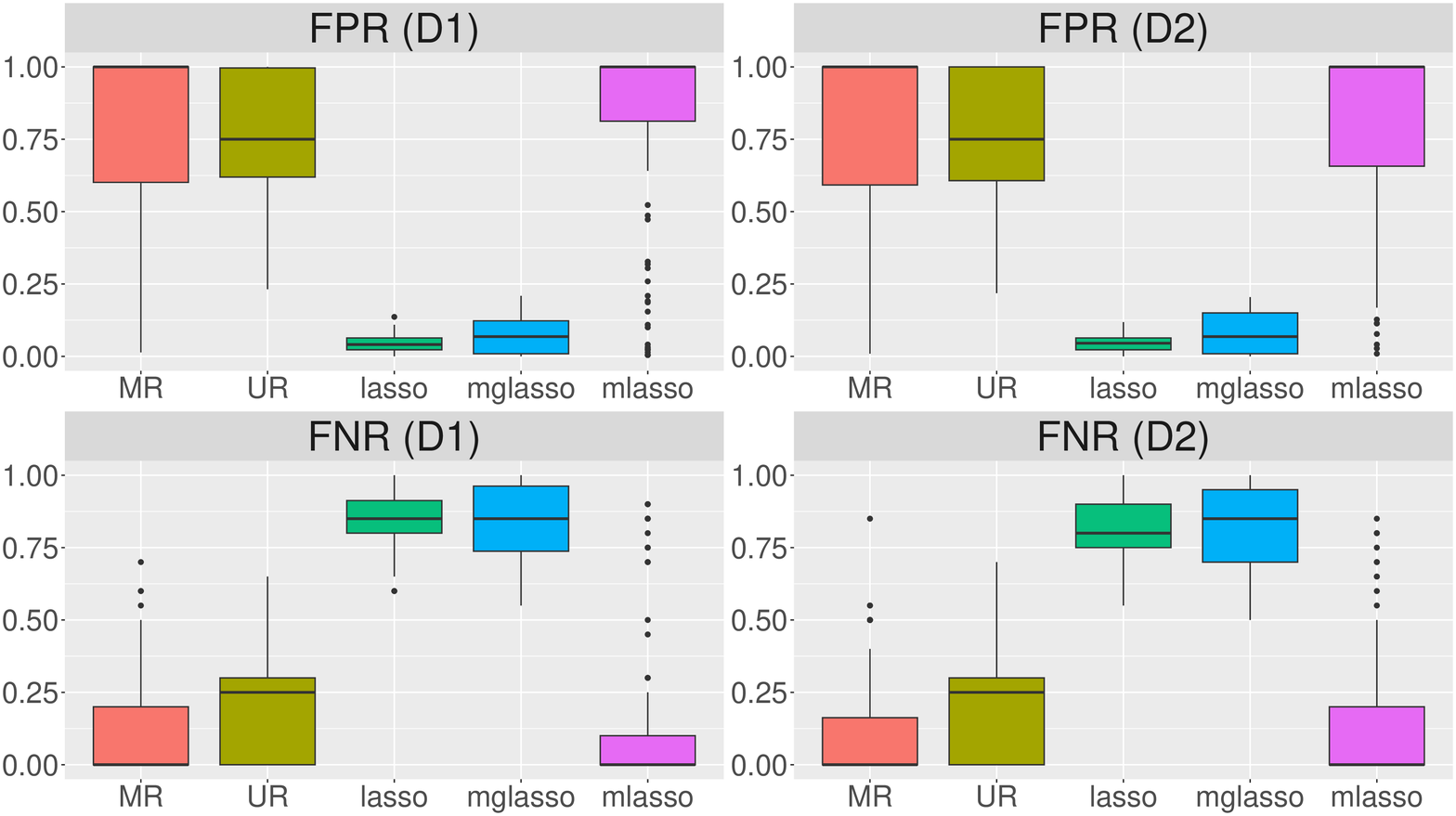} 
\vspace{-3.5mm}
\subcaption{$s=50, \rho_x=0.1, \rho_y=0.1$}
\vspace{2.5mm}
\end{minipage}
\begin{minipage}[b]{0.5\linewidth}
\centering
\includegraphics[width=8cm,height=4.6cm]{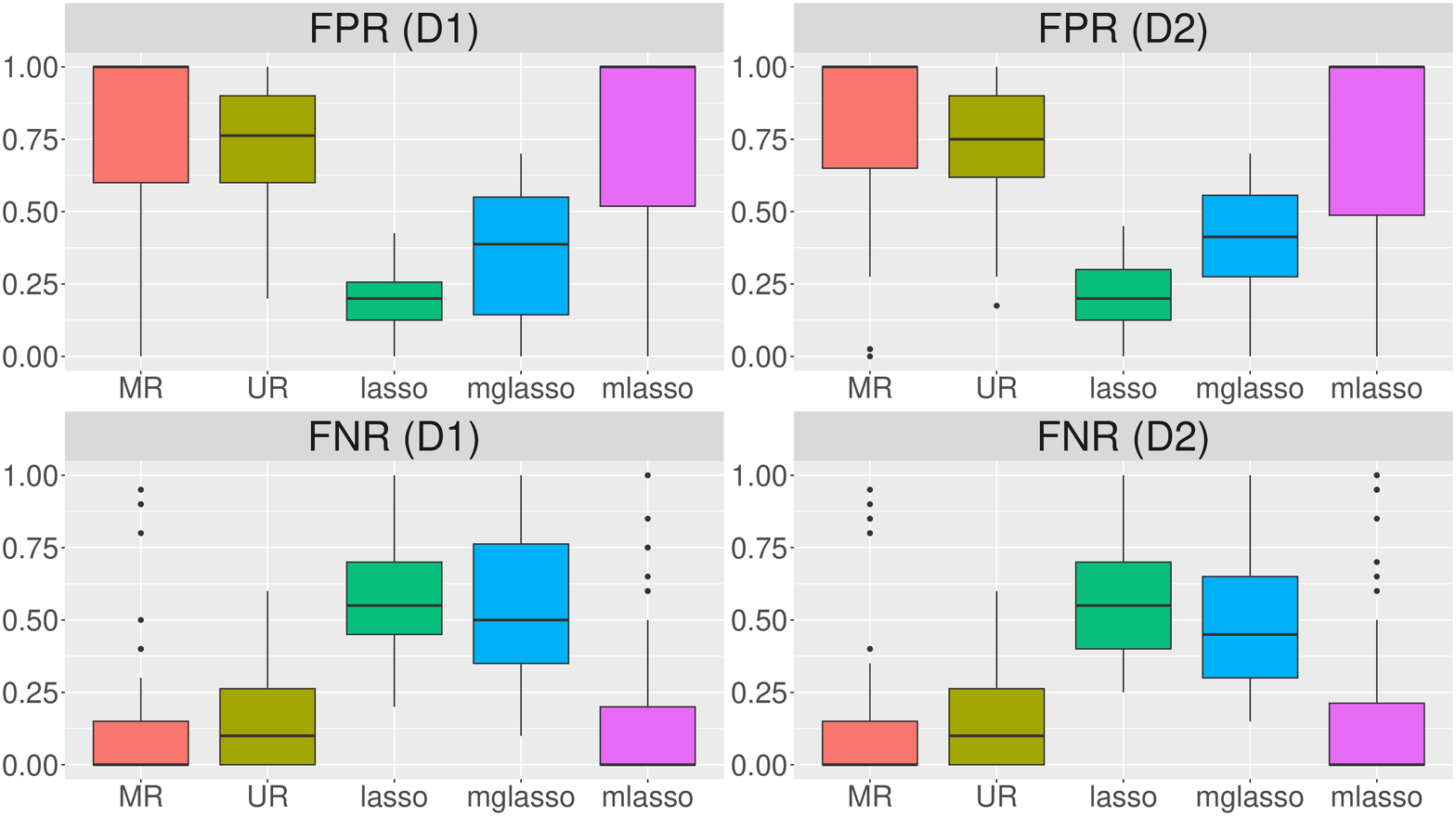} 
\vspace{-3.5mm}
\subcaption{$s=5, \rho_x=0.1, \rho_y=0.9$}
\vspace{2.5mm}
\end{minipage}
\begin{minipage}[b]{0.5\linewidth}
\centering
\includegraphics[width=8cm,height=4.6cm]{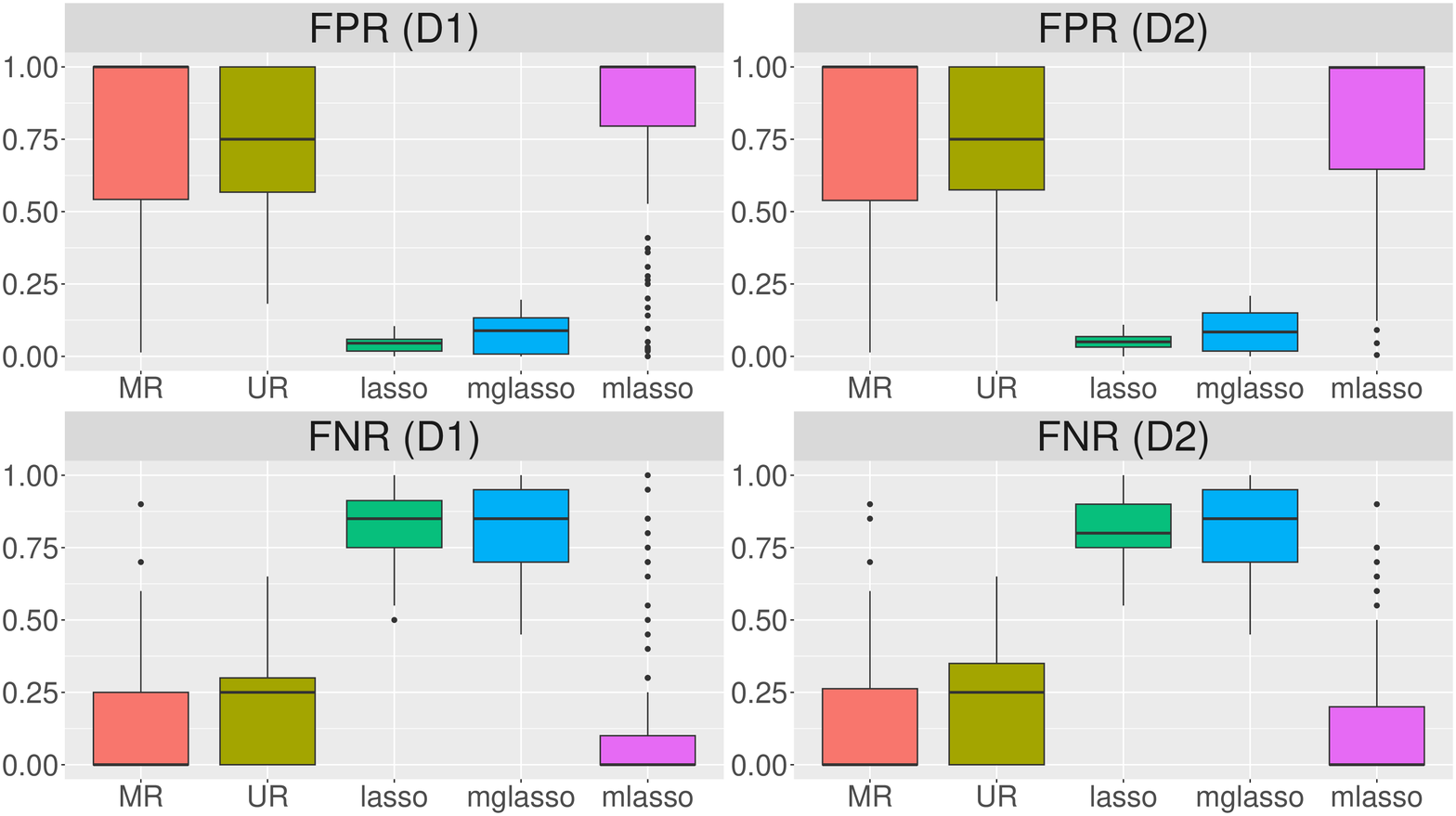}
\vspace{-3.5mm}
\subcaption{$s=50, \rho_x=0.1, \rho_y=0.9$}
\vspace{2.5mm}
\end{minipage}
\begin{minipage}[b]{0.5\linewidth}
\centering
\includegraphics[width=8cm,height=4.6cm]{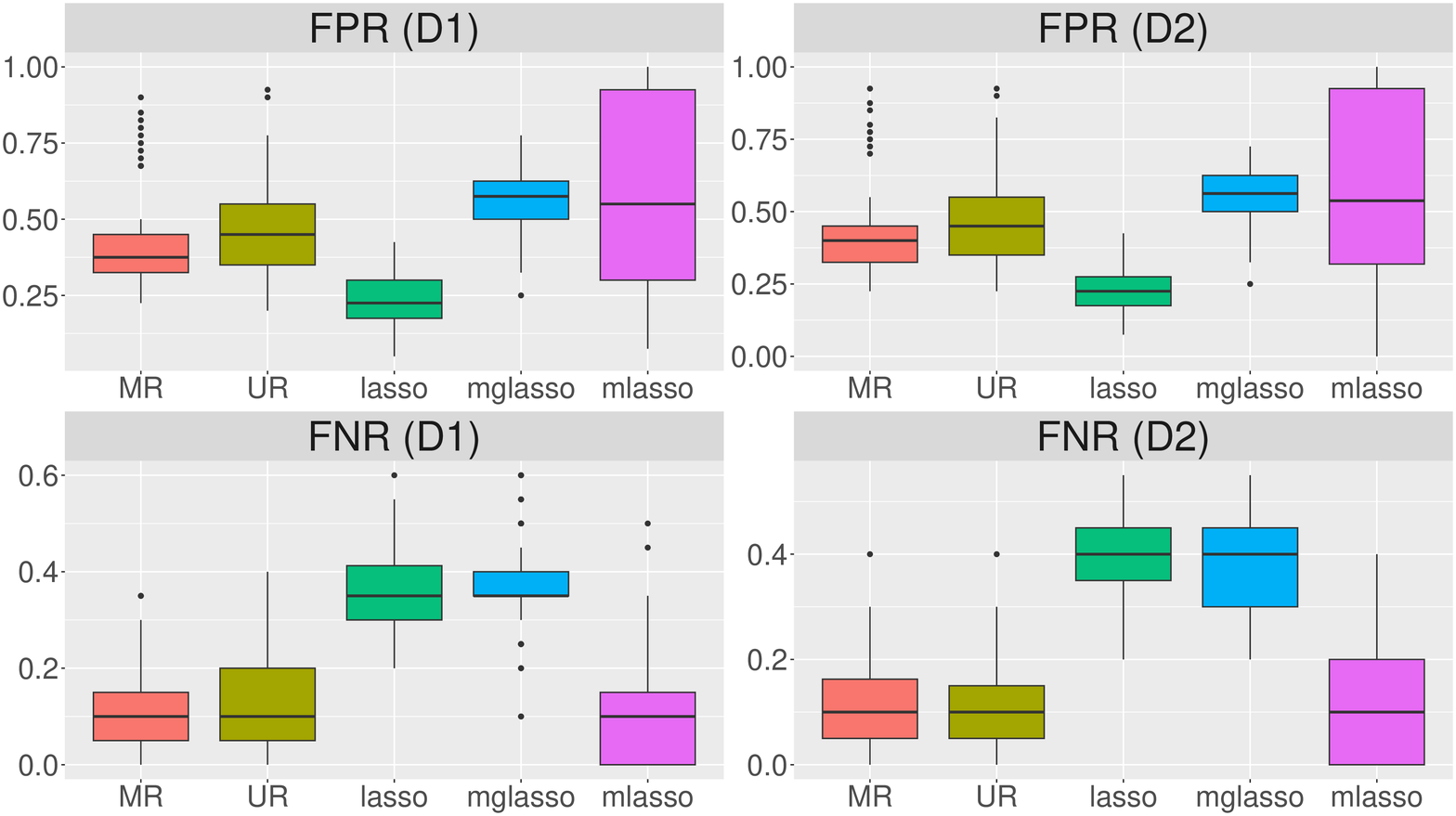}
\vspace{-3.5mm}
\subcaption{$s=5, \rho_x=0.9, \rho_y=0.1$}
\vspace{2.5mm}
\end{minipage}
\begin{minipage}[b]{0.5\linewidth}
\centering
\includegraphics[width=8cm,height=4.6cm]{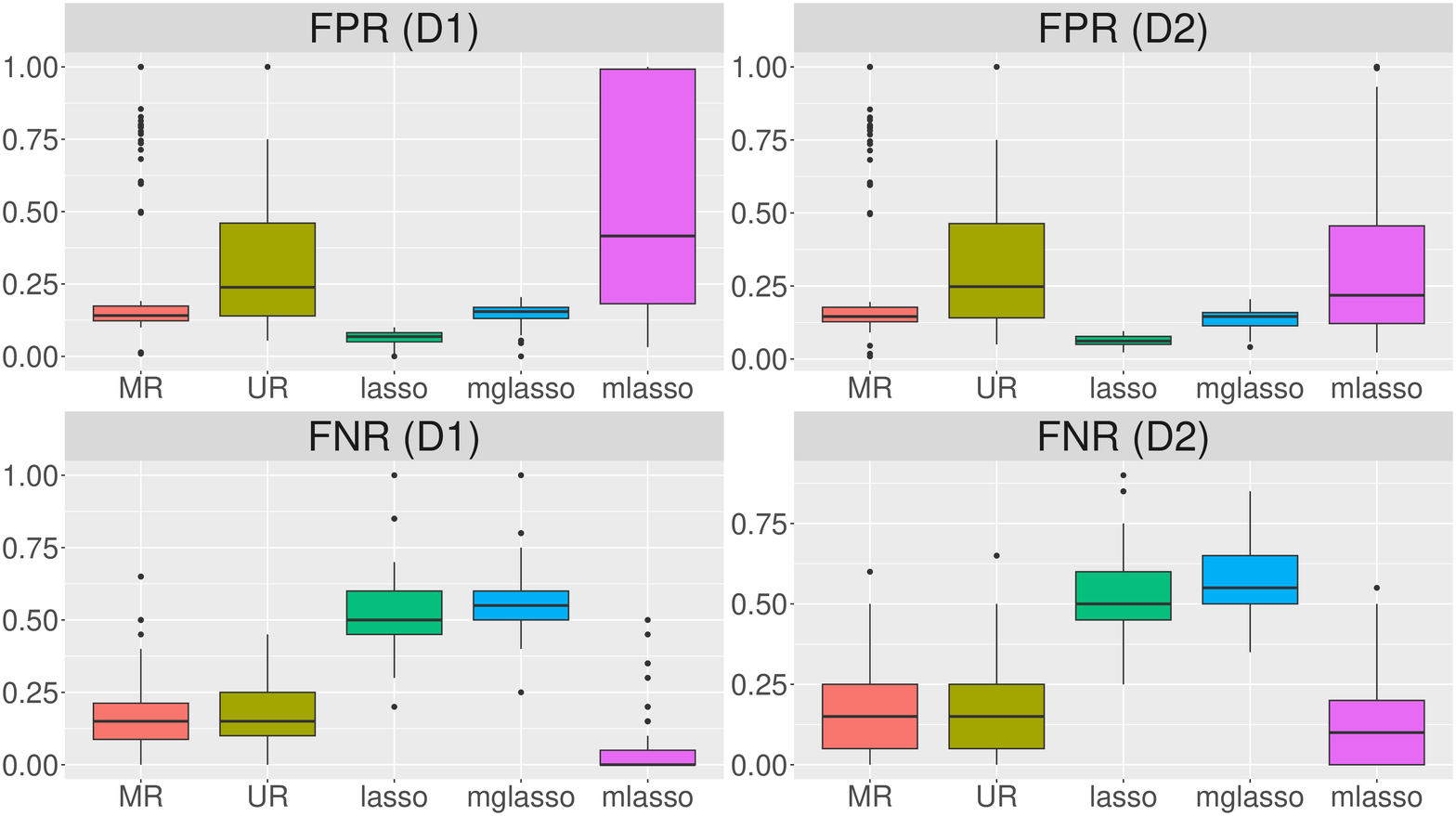}
\vspace{-3.5mm}
\subcaption{$s=50, \rho_x=0.9, \rho_y=0.1$}
\vspace{2.5mm}
\end{minipage}
\begin{minipage}[b]{0.5\linewidth}
\centering
\includegraphics[width=8cm,height=4.6cm]{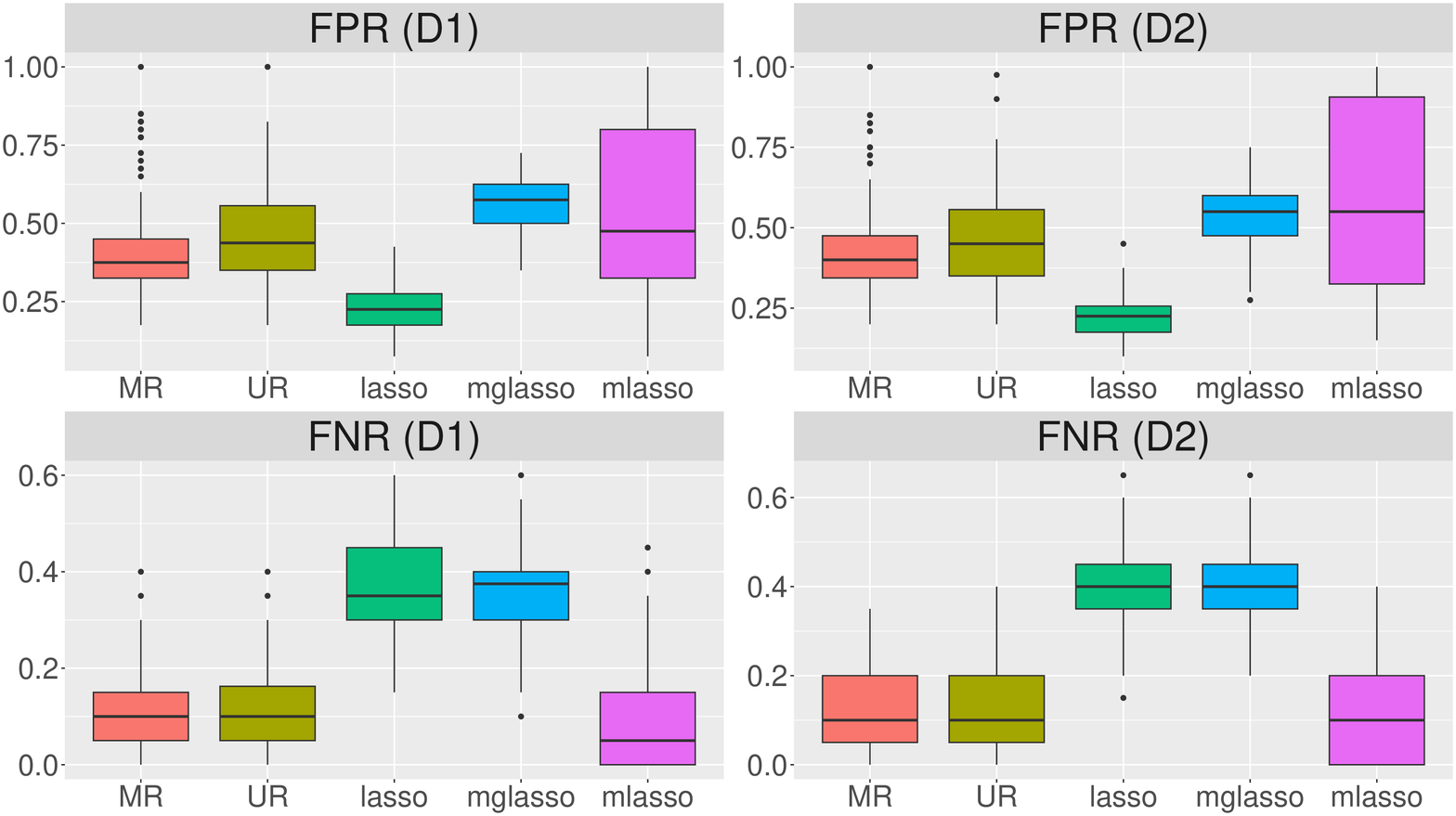}
\vspace{-3.5mm}
\subcaption{$s=5, \rho_x=0.9, \rho_y=0.9$}
\vspace{2.5mm}
\end{minipage}
\begin{minipage}[b]{0.5\linewidth}
\centering
\includegraphics[width=8cm,height=4.6cm]{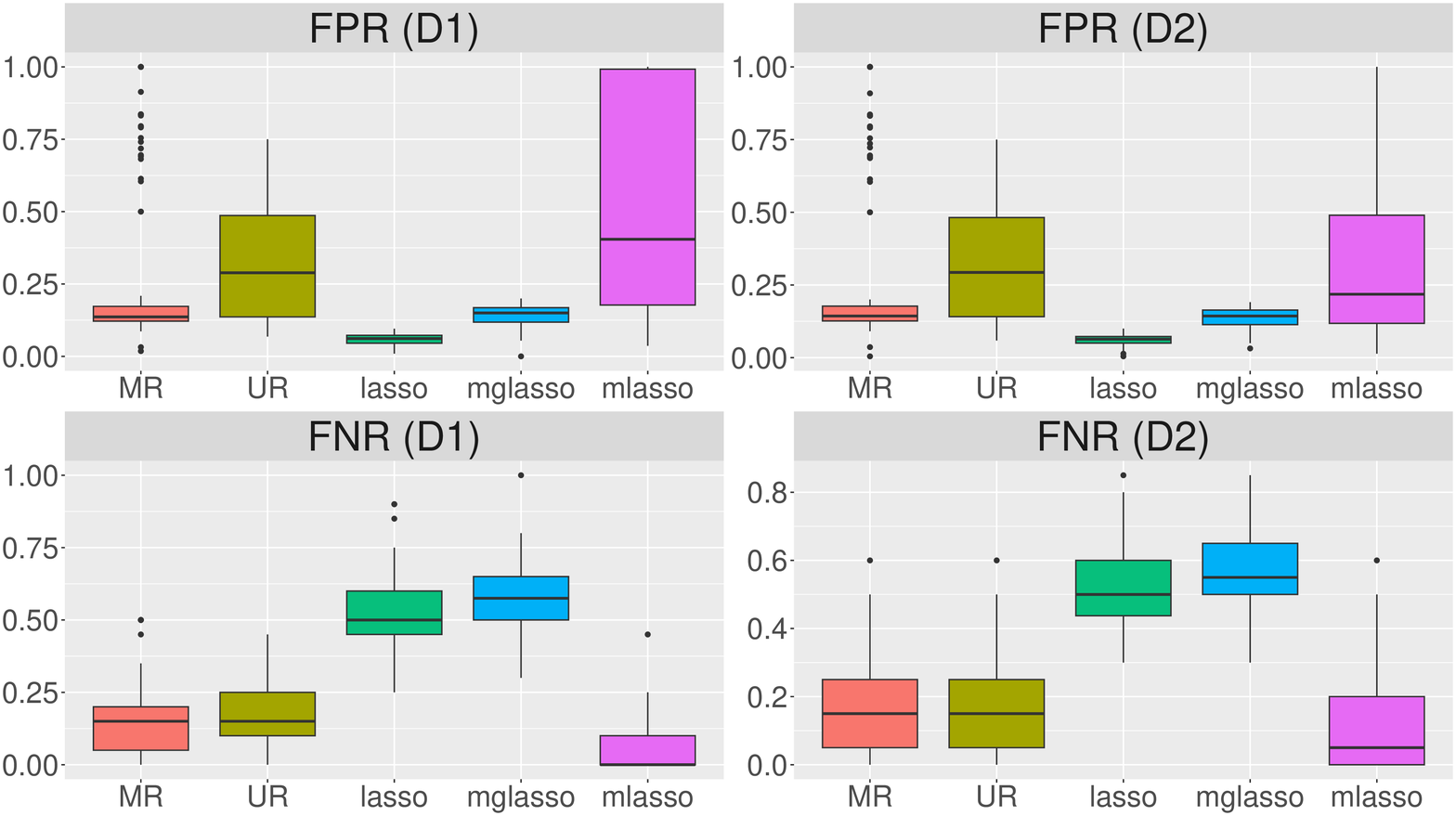}
\vspace{-3.5mm}
\subcaption{$s=50, \rho_x=0.9, \rho_y=0.9$}
\vspace{2.5mm}
\end{minipage}
\caption{Boxplots of FPR and FNR for $n=15$ when the case $M=2$.
}
\label{fig:SimuM2n15_FPRFNR}
\end{figure}

\begin{figure}[htbp]
\begin{minipage}[b]{0.5\linewidth}
\centering
\includegraphics[width=8cm,height=4.6cm]{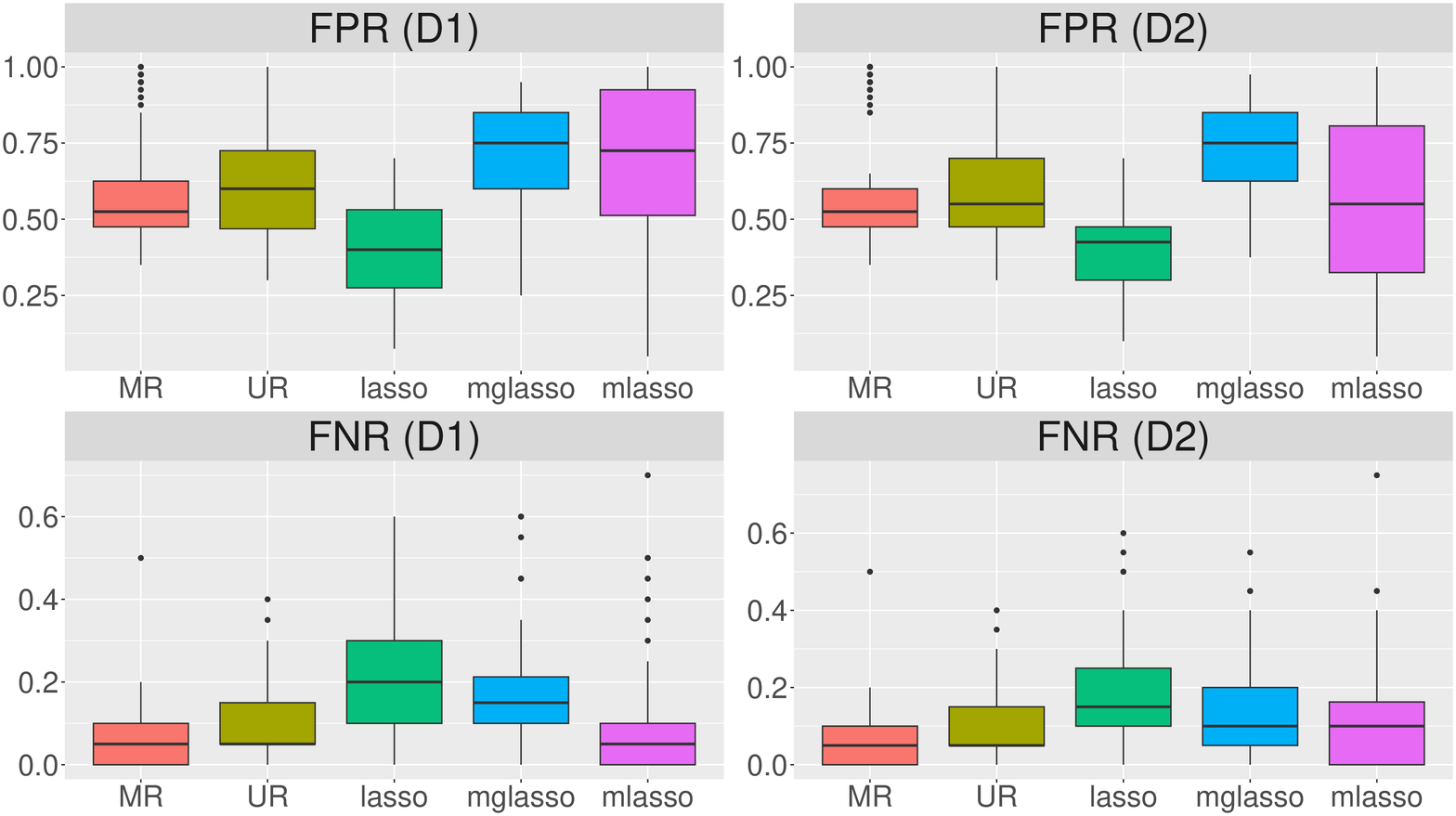}
\vspace{-3.5mm}
\subcaption{$s=5, \rho_x=0.1, \rho_y=0.1$}
\vspace{2.5mm}
\end{minipage}
\begin{minipage}[b]{0.5\linewidth}
\centering
\includegraphics[width=8cm,height=4.6cm]{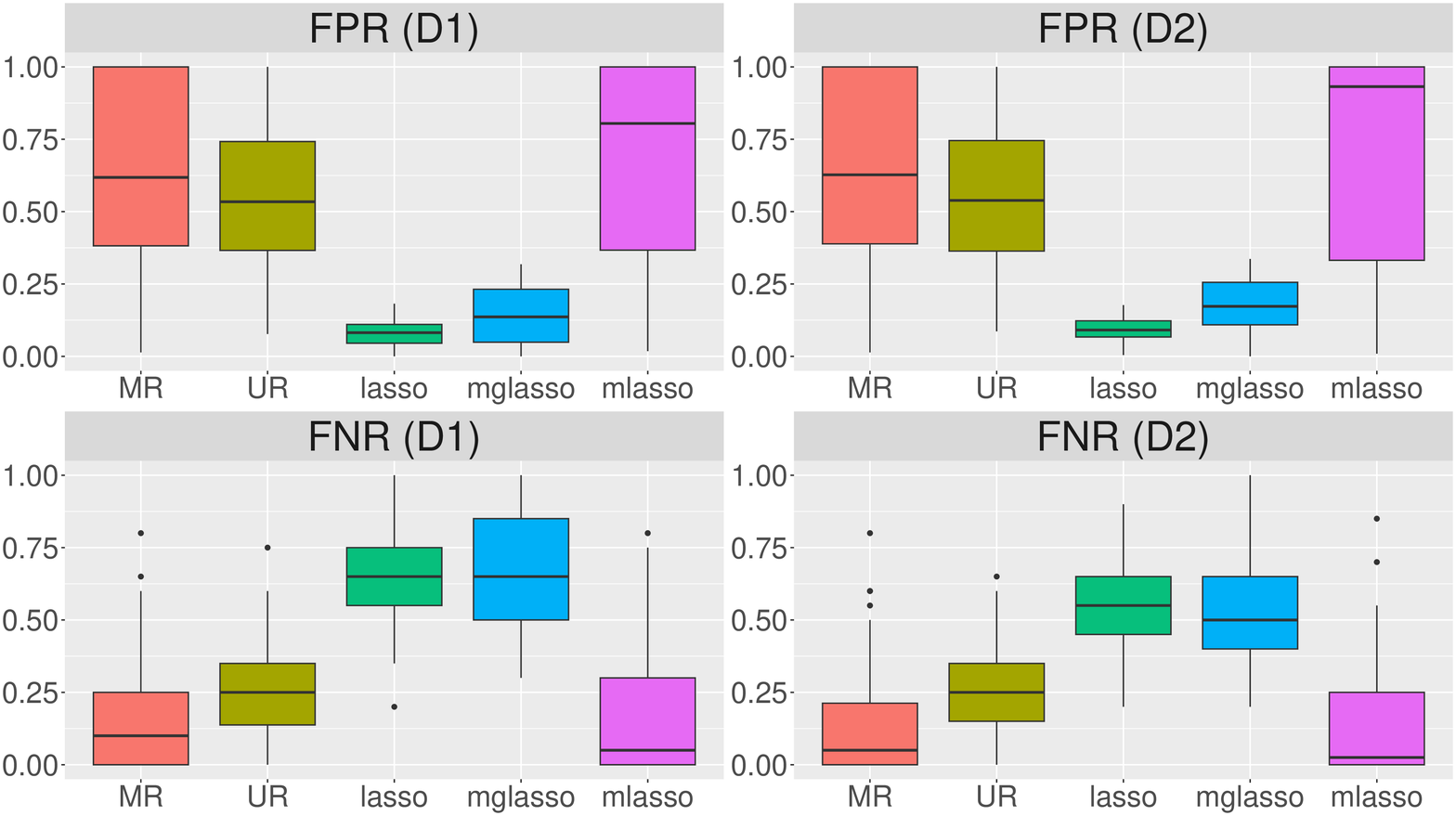} 
\vspace{-3.5mm}
\subcaption{$s=50, \rho_x=0.1, \rho_y=0.1$}
\vspace{2.5mm}
\end{minipage}
\begin{minipage}[b]{0.5\linewidth}
\centering
\includegraphics[width=8cm,height=4.6cm]{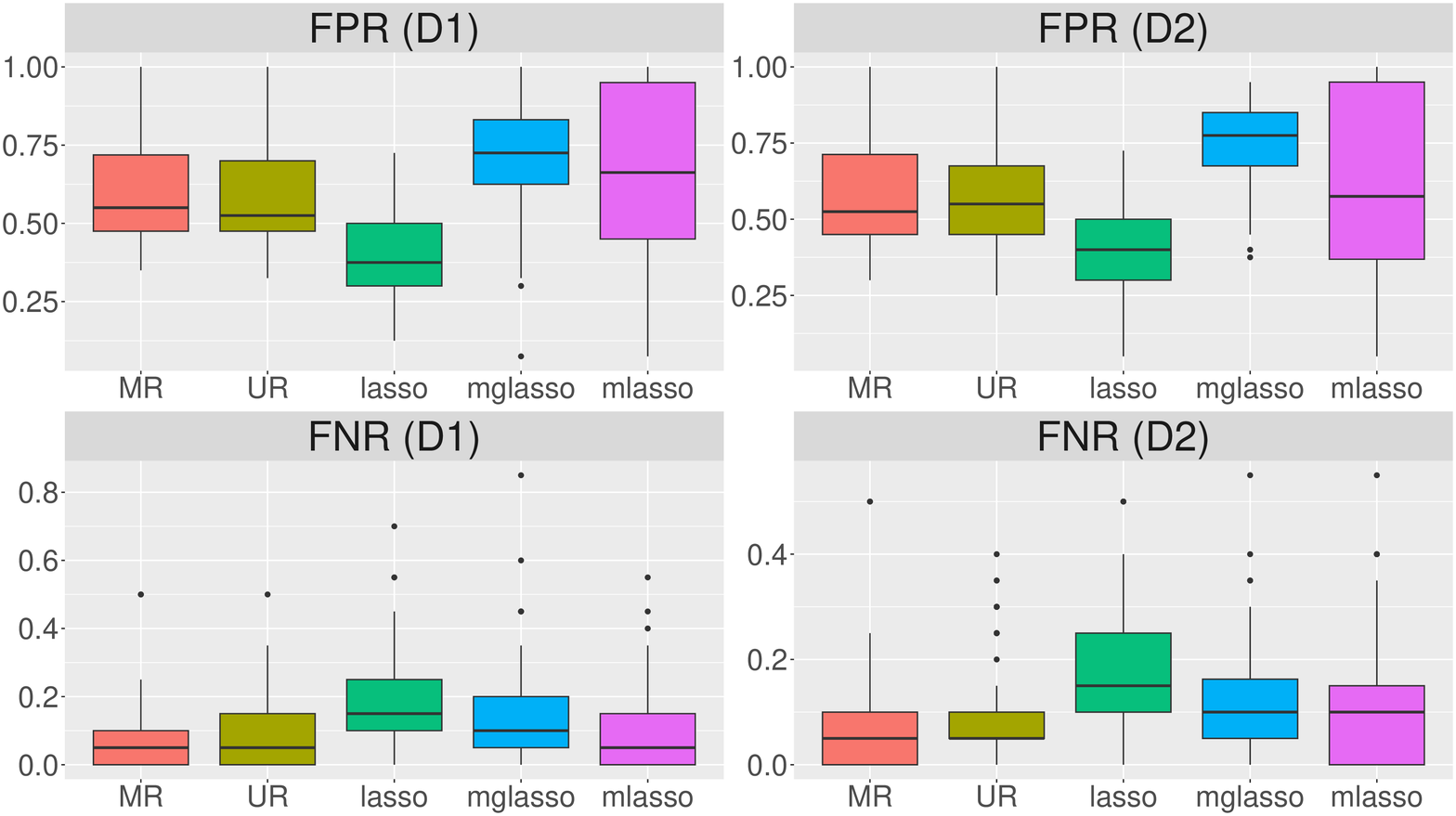} 
\vspace{-3.5mm}
\subcaption{$s=5, \rho_x=0.1, \rho_y=0.9$}
\vspace{2.5mm}
\end{minipage}
\begin{minipage}[b]{0.5\linewidth}
\centering
\includegraphics[width=8cm,height=4.6cm]{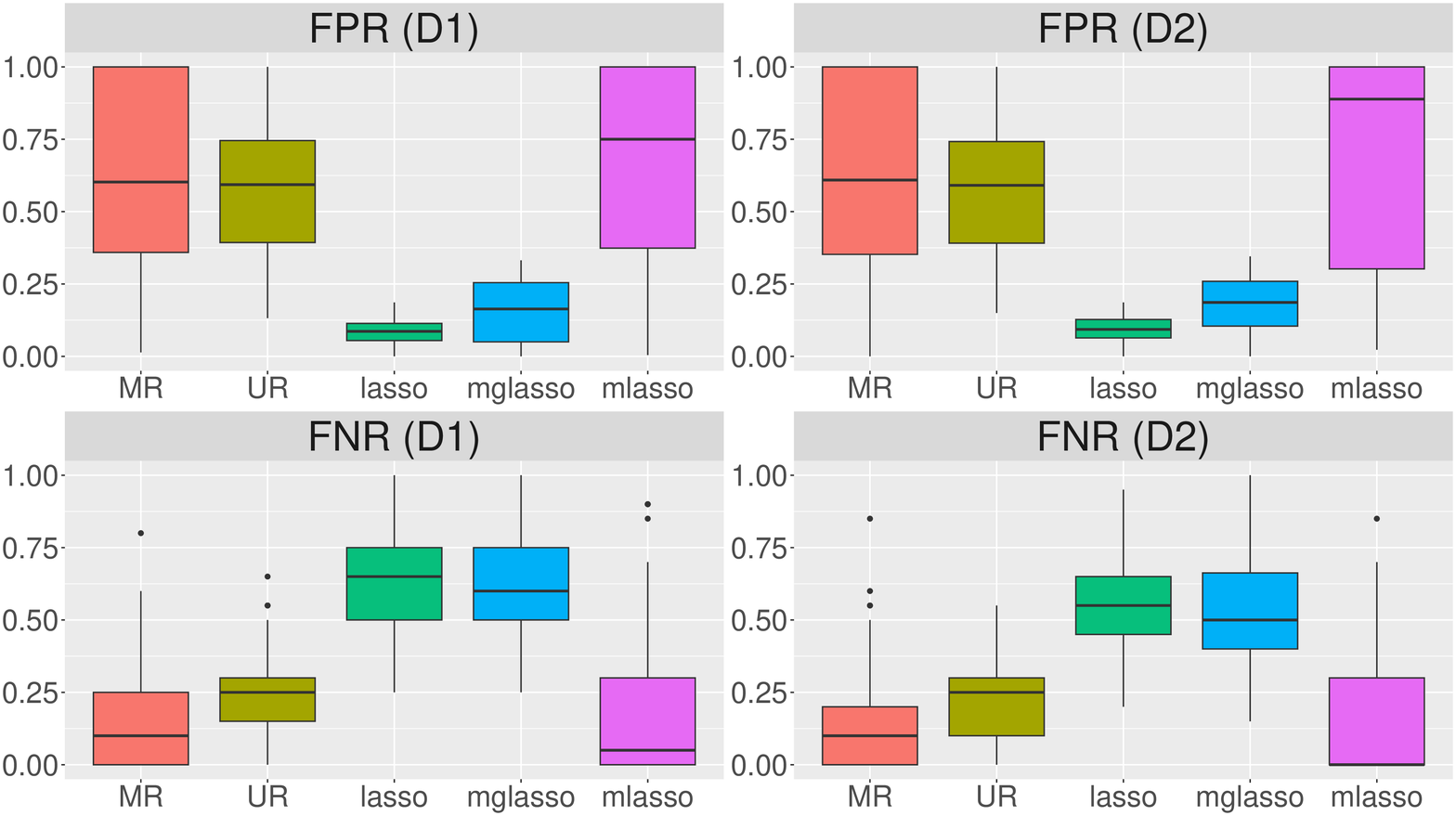}
\vspace{-3.5mm}
\subcaption{$s=50, \rho_x=0.1, \rho_y=0.9$}
\vspace{2.5mm}
\end{minipage}
\begin{minipage}[b]{0.5\linewidth}
\centering
\includegraphics[width=8cm,height=4.6cm]{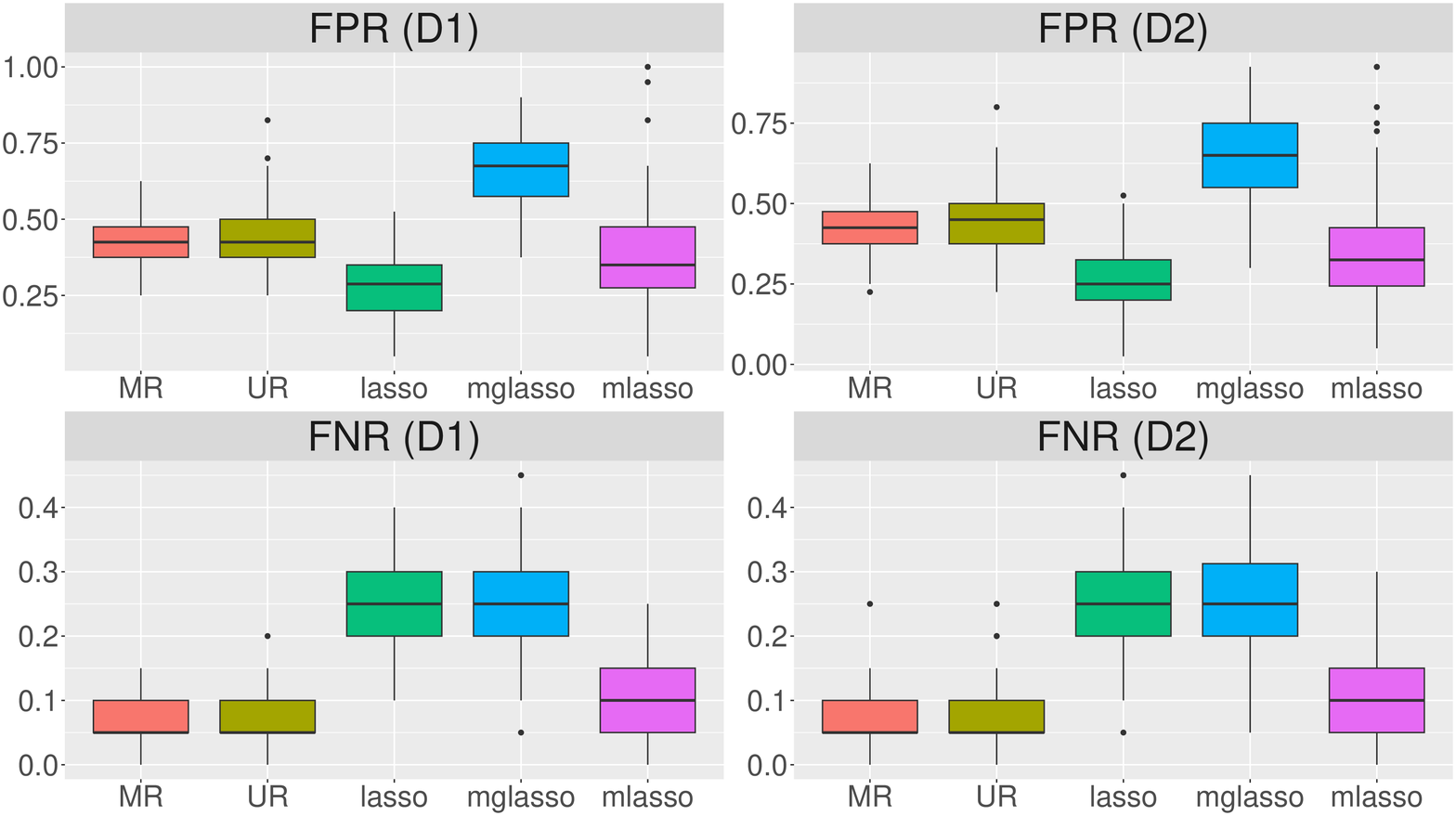}
\vspace{-3.5mm}
\subcaption{$s=5, \rho_x=0.9, \rho_y=0.1$}
\vspace{2.5mm}
\end{minipage}
\begin{minipage}[b]{0.5\linewidth}
\centering
\includegraphics[width=8cm,height=4.6cm]{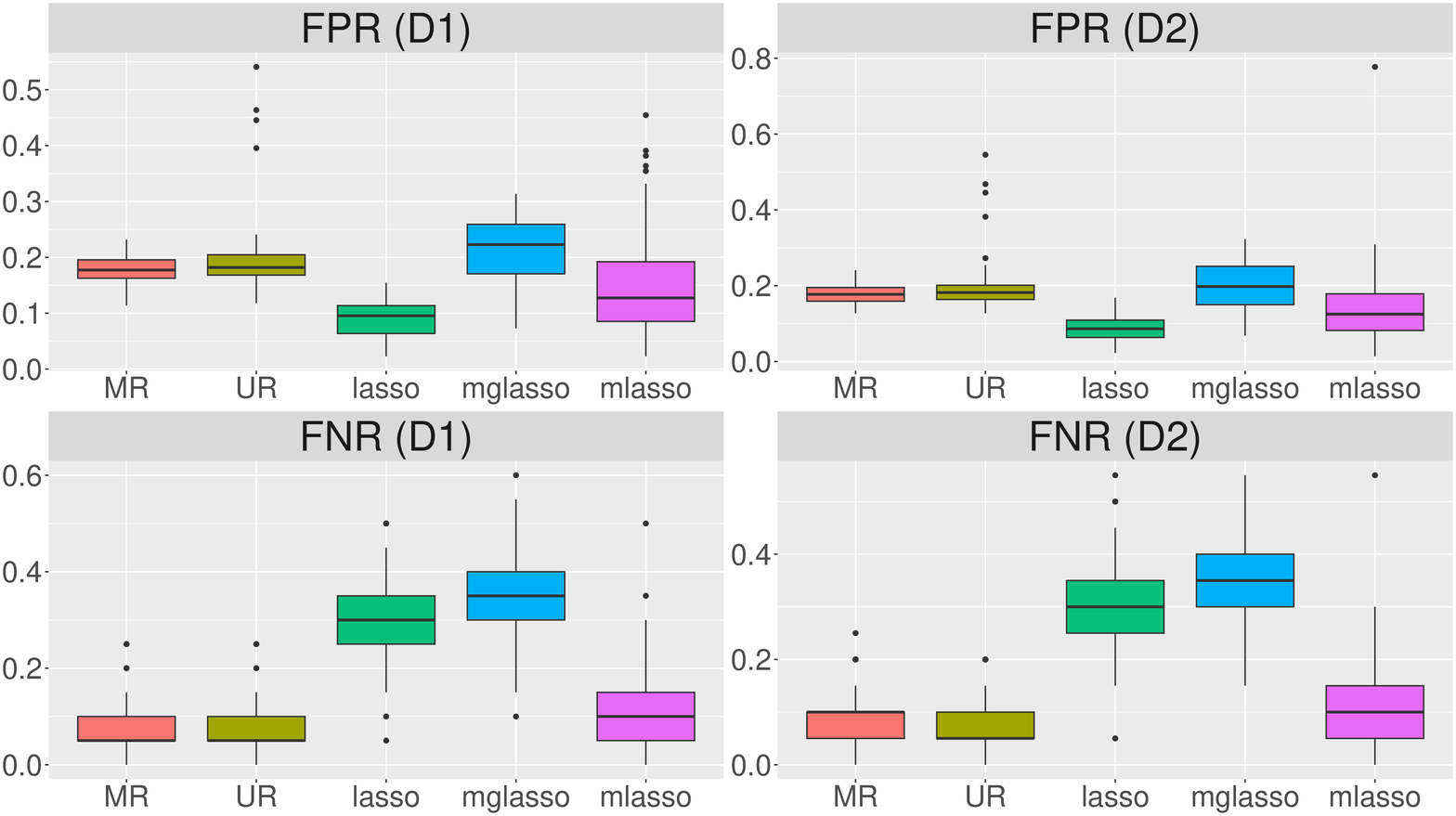}
\vspace{-3.5mm}
\subcaption{$s=50, \rho_x=0.9, \rho_y=0.1$}
\vspace{2.5mm}
\end{minipage}
\begin{minipage}[b]{0.5\linewidth}
\centering
\includegraphics[width=8cm,height=4.6cm]{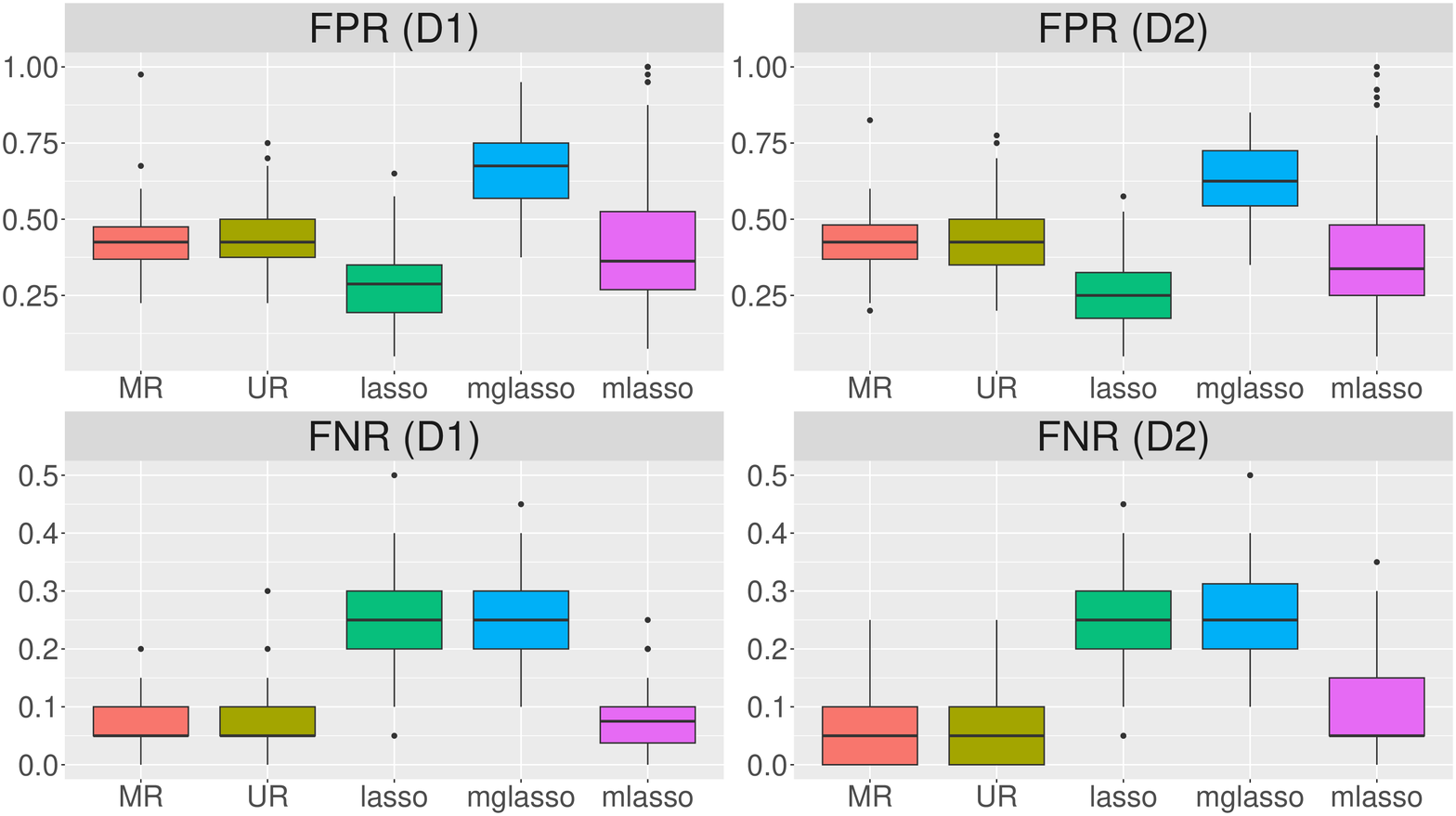}
\vspace{-3.5mm}
\subcaption{$s=5, \rho_x=0.9, \rho_y=0.9$}
\vspace{2.5mm}
\end{minipage}
\begin{minipage}[b]{0.5\linewidth}
\centering
\includegraphics[width=8cm,height=4.6cm]{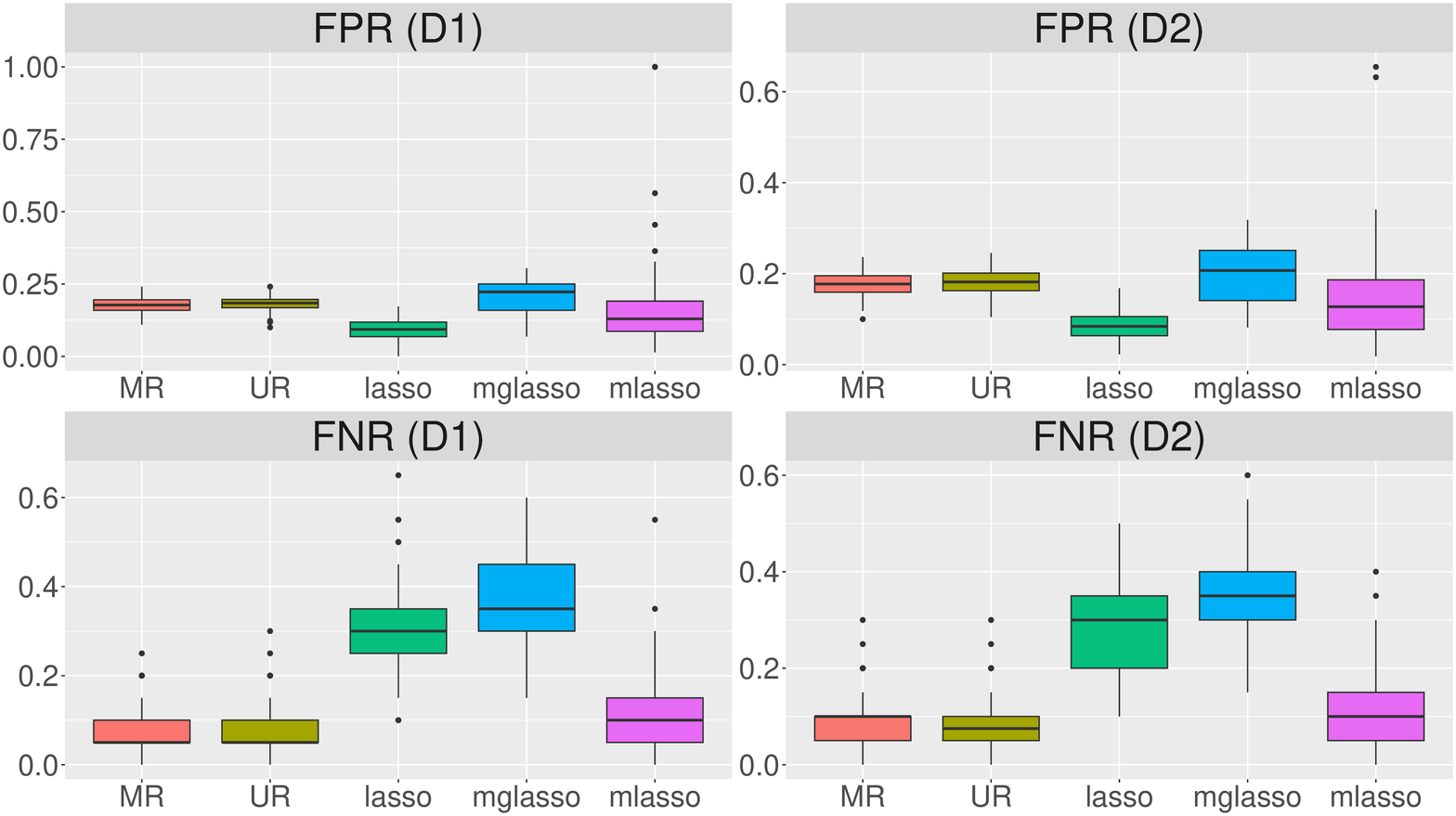}
\vspace{-3.5mm}
\subcaption{$s=50, \rho_x=0.9, \rho_y=0.9$}
\vspace{2.5mm}
\end{minipage}
\caption{Boxplots of FPR and FNR for $n=25$ when the case $M=2$.
}
\label{fig:SimuM2n25_FPRFNR}
\end{figure}

\begin{figure}[htbp]
\begin{minipage}[b]{0.5\linewidth}
\centering
\includegraphics[width=8cm,height=4.6cm]{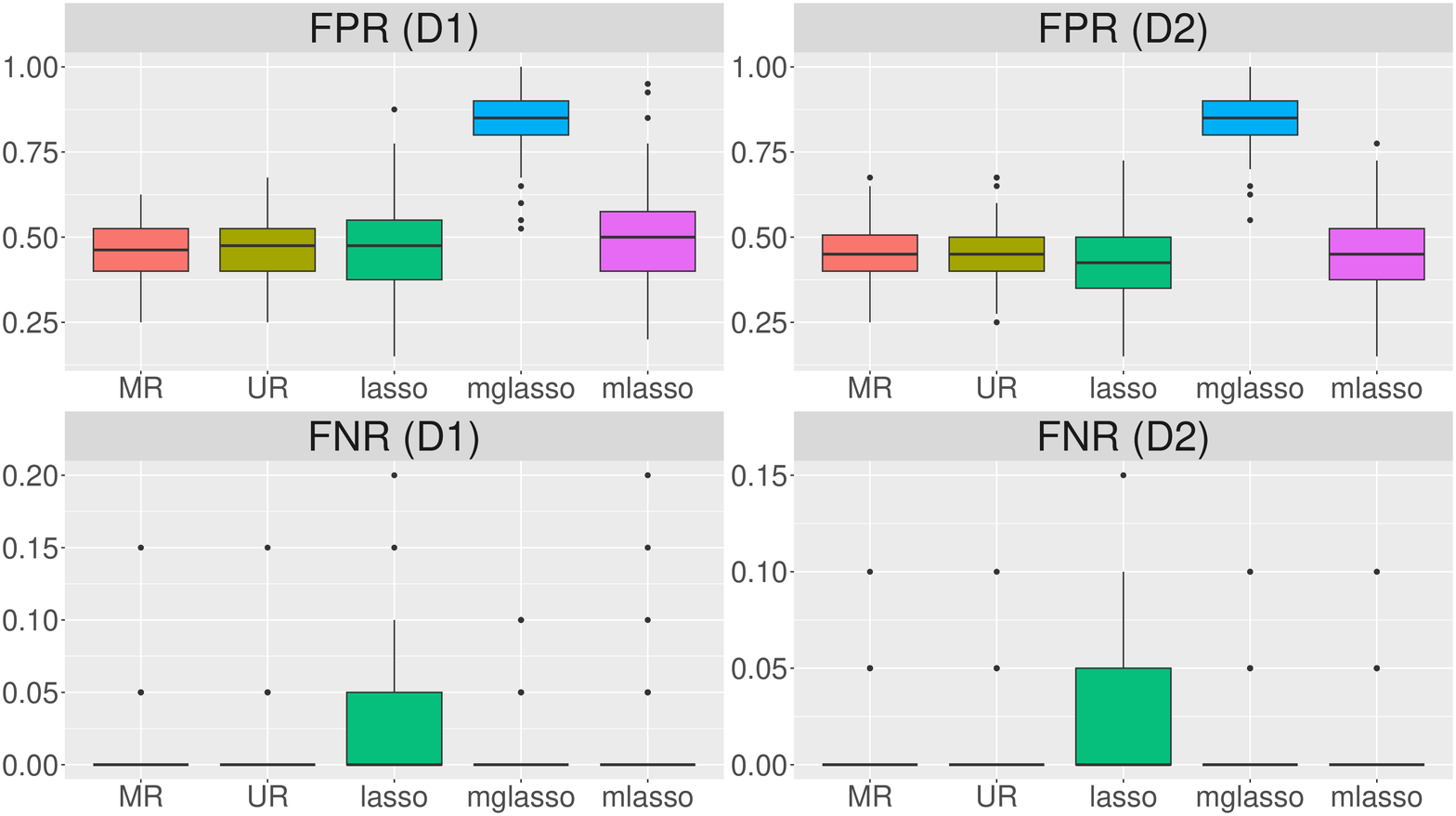}
\vspace{-3.5mm}
\subcaption{$s=5, \rho_x=0.1, \rho_y=0.1$}
\vspace{2.5mm}
\end{minipage}
\begin{minipage}[b]{0.5\linewidth}
\centering
\includegraphics[width=8cm,height=4.6cm]{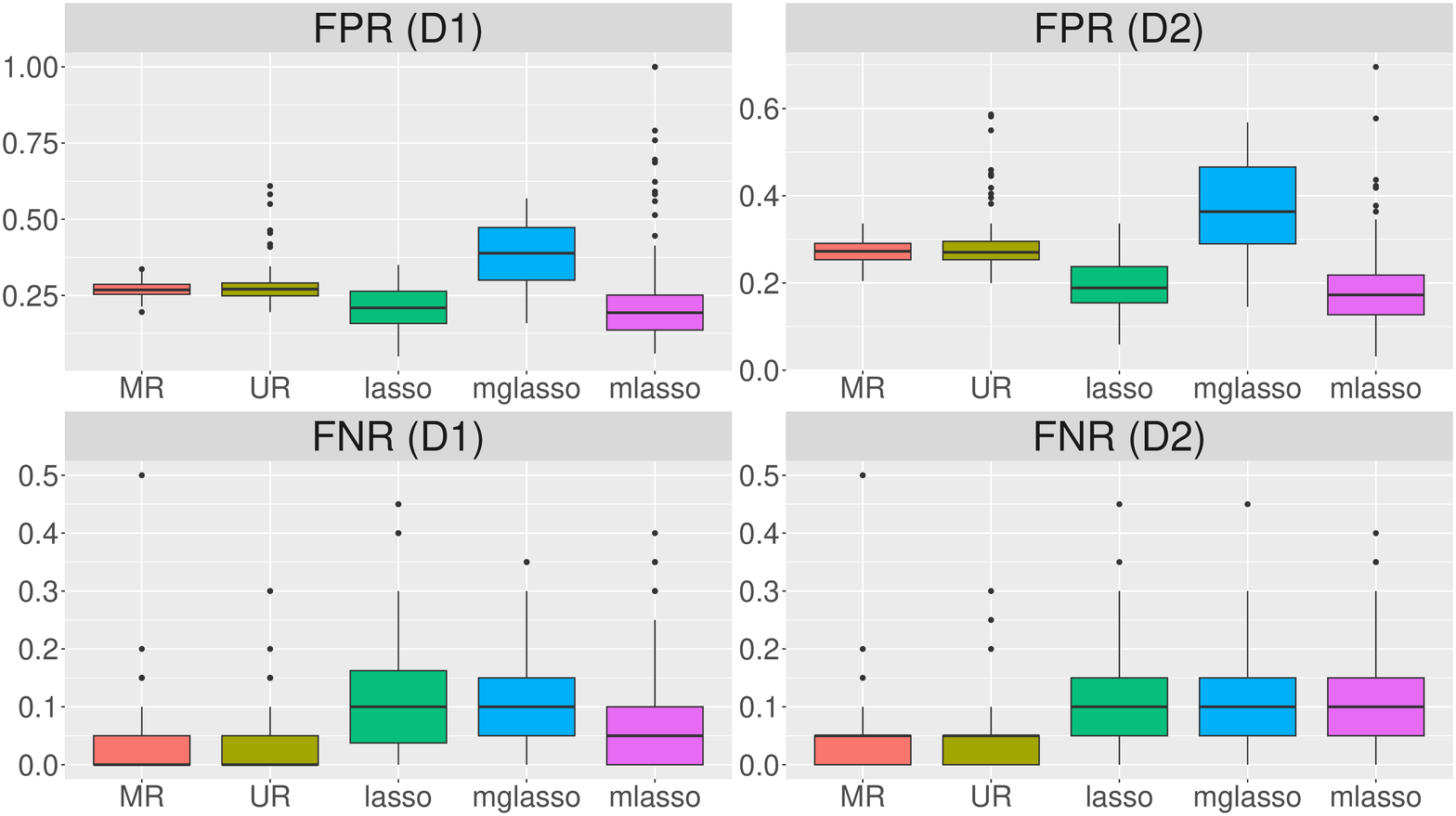} 
\vspace{-3.5mm}
\subcaption{$s=50, \rho_x=0.1, \rho_y=0.1$}
\vspace{2.5mm}
\end{minipage}
\begin{minipage}[b]{0.5\linewidth}
\centering
\includegraphics[width=8cm,height=4.6cm]{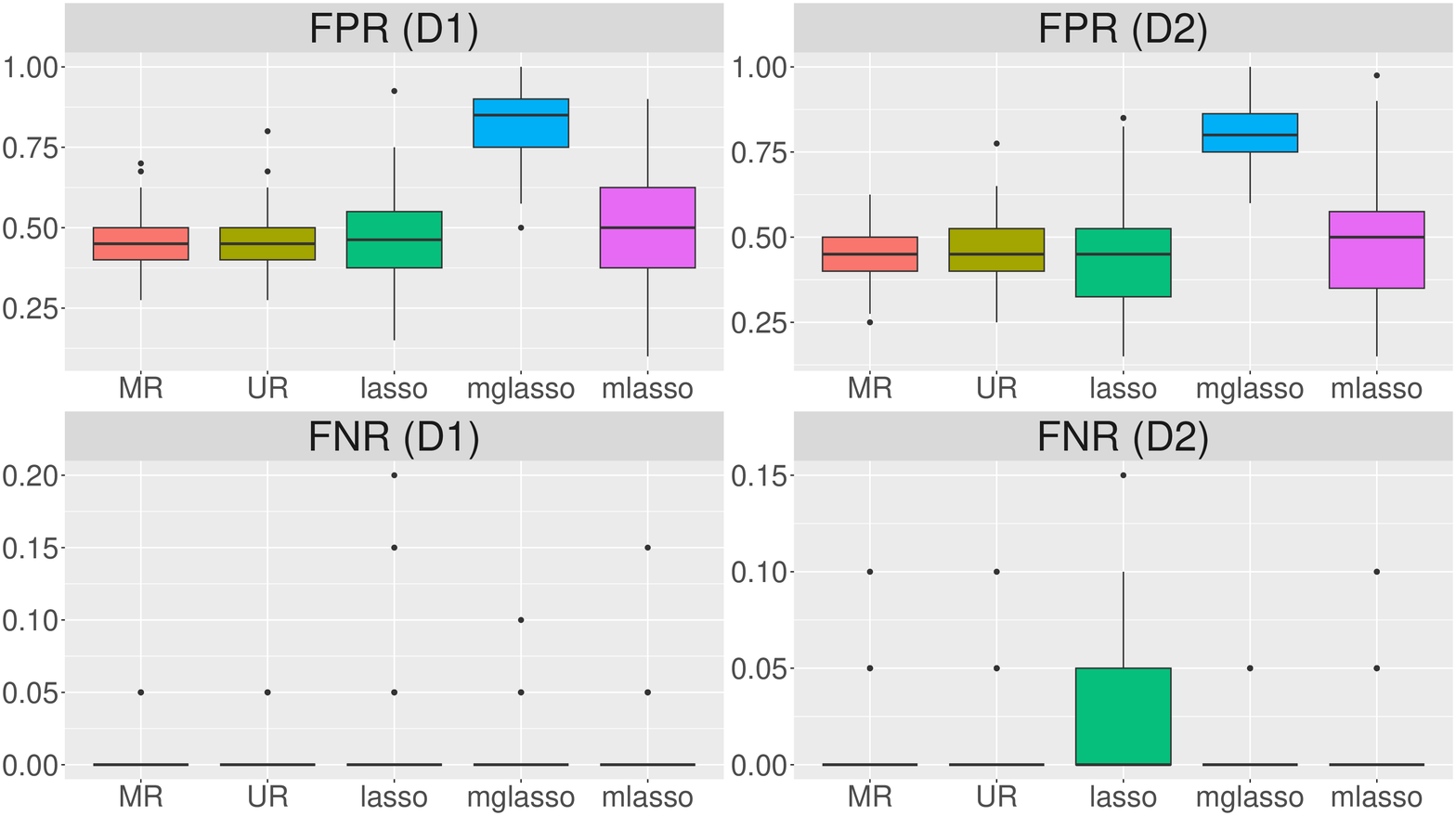} 
\vspace{-3.5mm}
\subcaption{$s=5, \rho_x=0.1, \rho_y=0.9$}
\vspace{2.5mm}
\end{minipage}
\begin{minipage}[b]{0.5\linewidth}
\centering
\includegraphics[width=8cm,height=4.6cm]{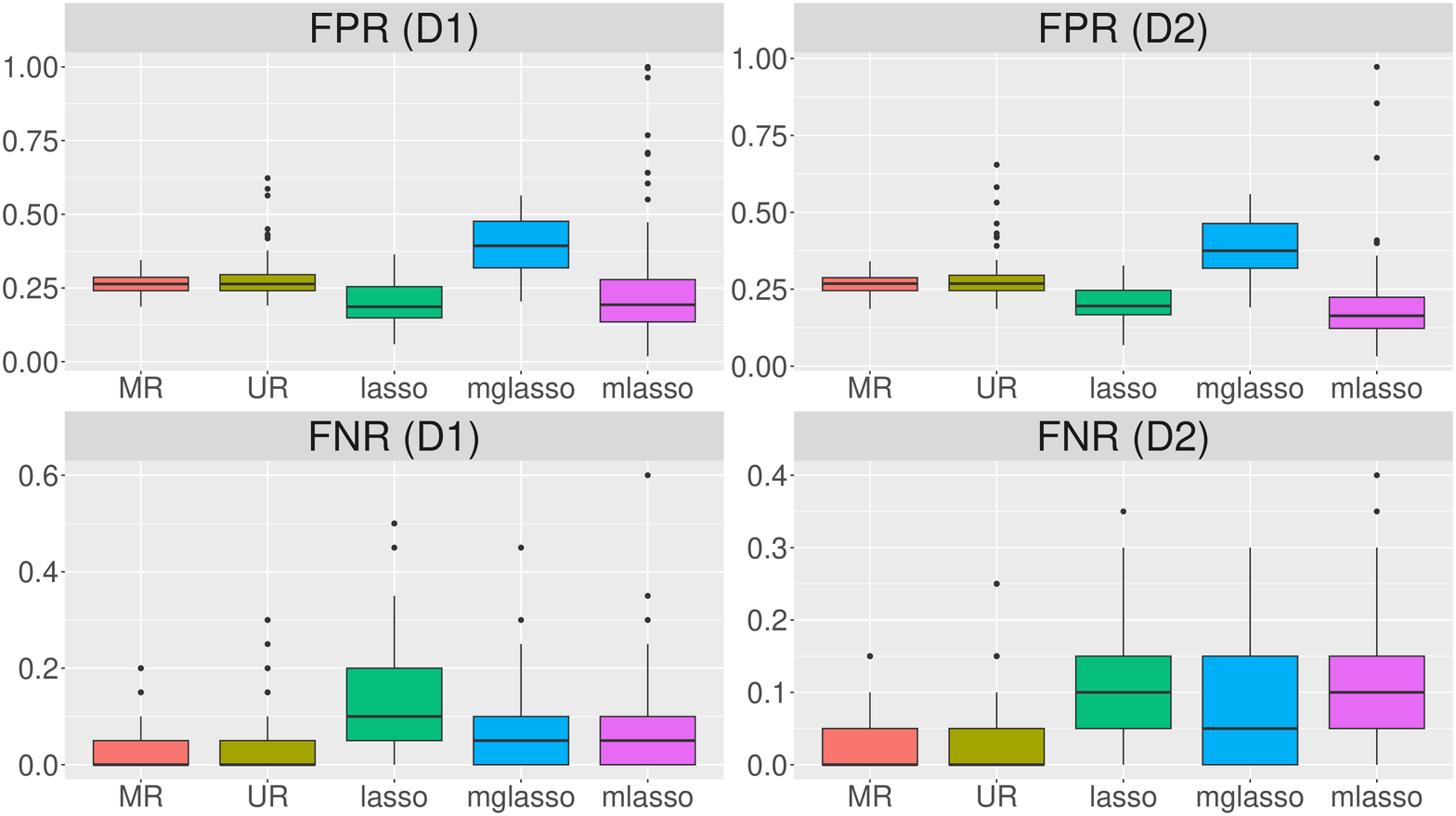}
\vspace{-3.5mm}
\subcaption{$s=50, \rho_x=0.1, \rho_y=0.9$}
\vspace{2.5mm}
\end{minipage}
\begin{minipage}[b]{0.5\linewidth}
\centering
\includegraphics[width=8cm,height=4.6cm]{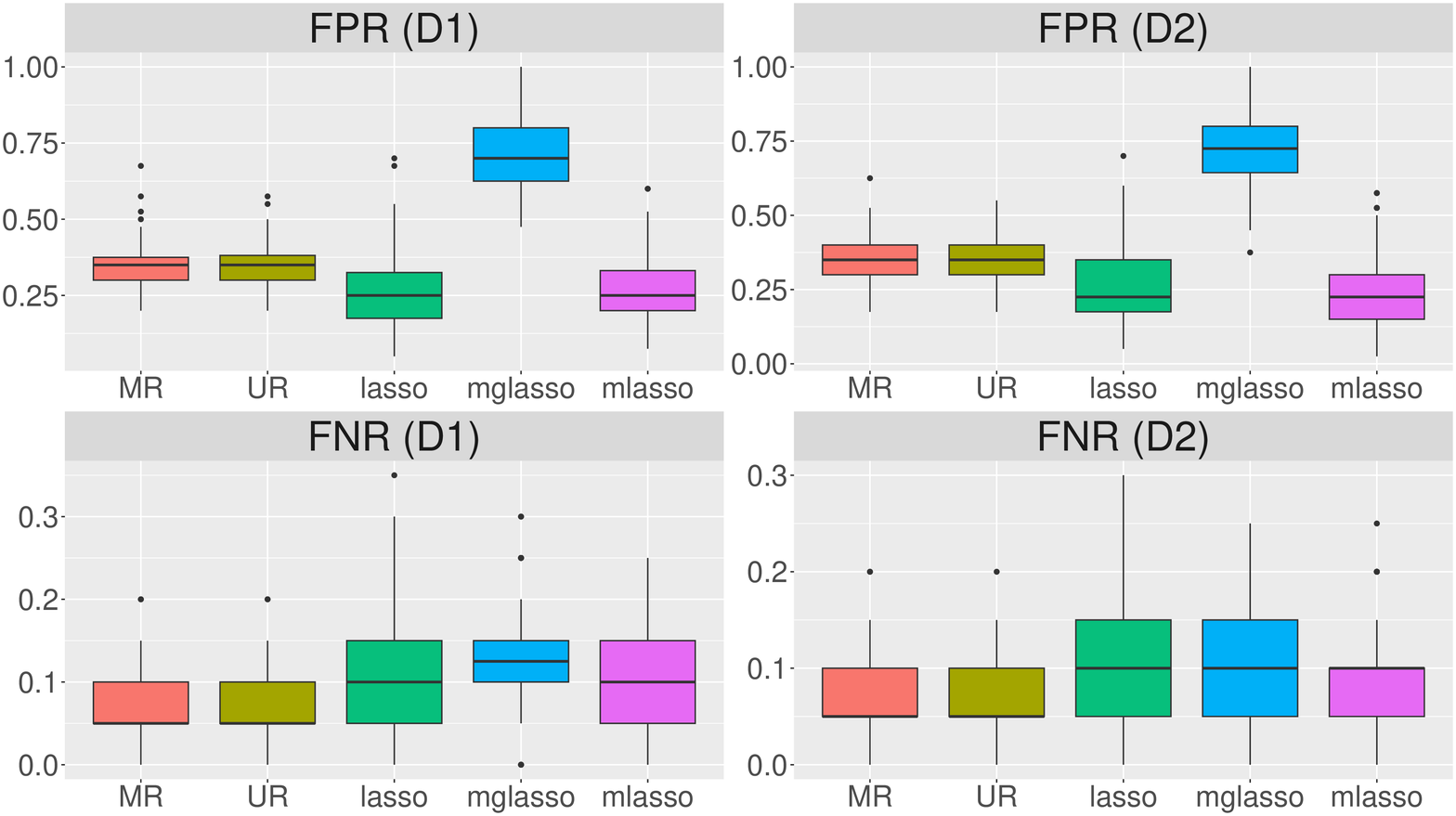}
\vspace{-3.5mm}
\subcaption{$s=5, \rho_x=0.9, \rho_y=0.1$}
\vspace{2.5mm}
\end{minipage}
\begin{minipage}[b]{0.5\linewidth}
\centering
\includegraphics[width=8cm,height=4.6cm]{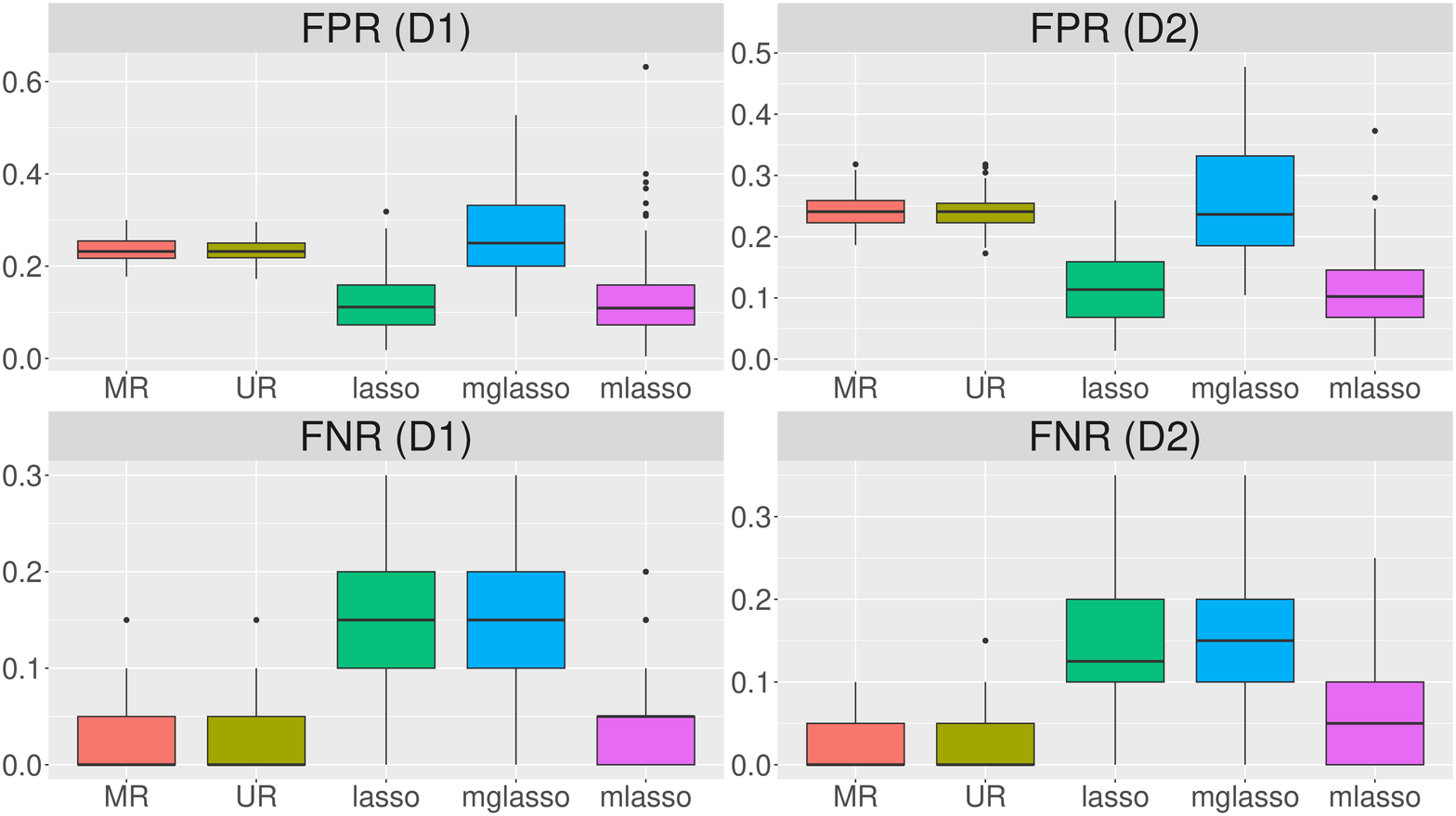}
\vspace{-3.5mm}
\subcaption{$s=50, \rho_x=0.9, \rho_y=0.1$}
\vspace{2.5mm}
\end{minipage}
\begin{minipage}[b]{0.5\linewidth}
\centering
\includegraphics[width=8cm,height=4.6cm]{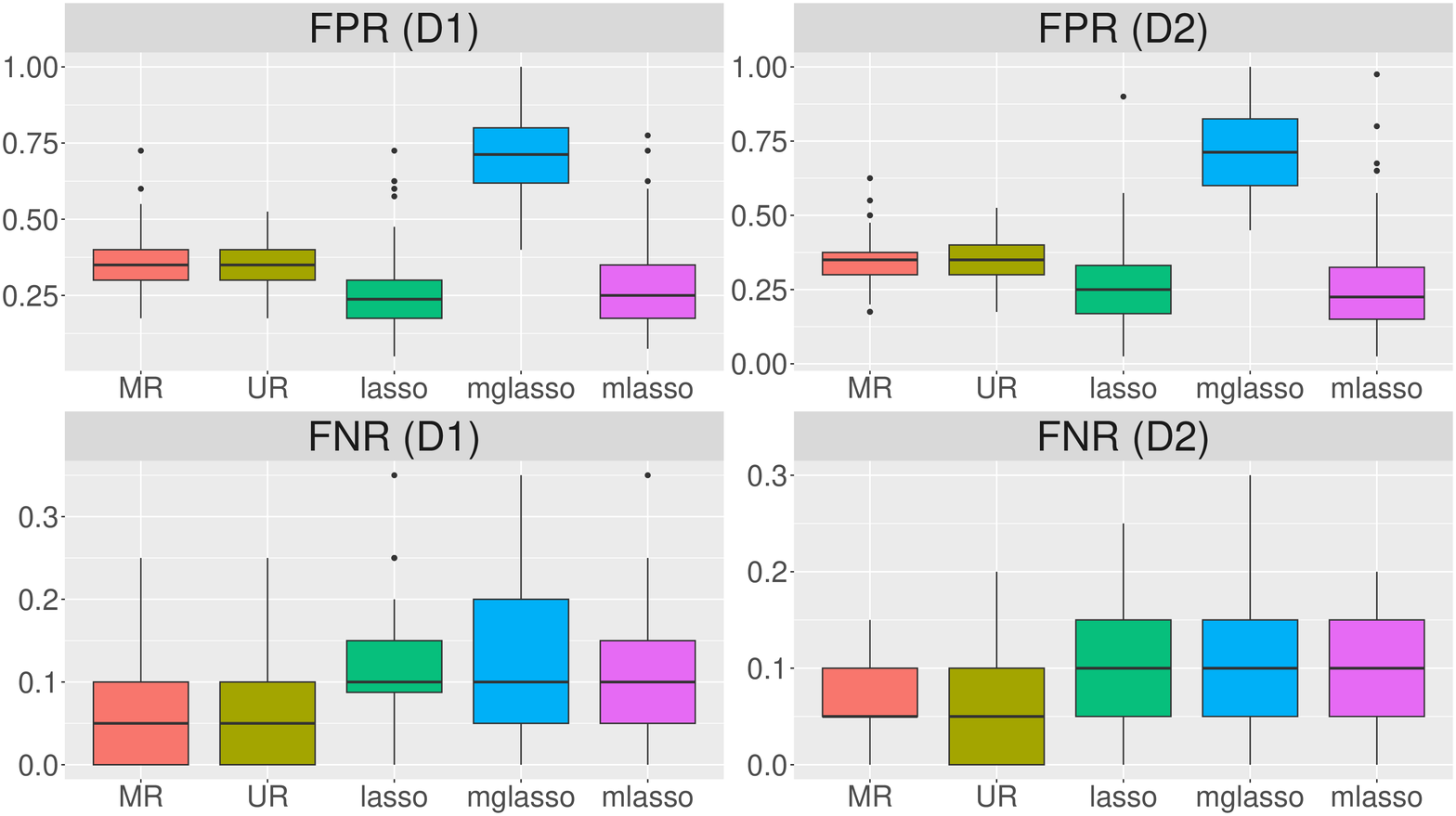}
\vspace{-3.5mm}
\subcaption{$s=5, \rho_x=0.9, \rho_y=0.9$}
\vspace{2.5mm}
\end{minipage}
\begin{minipage}[b]{0.5\linewidth}
\centering
\includegraphics[width=8cm,height=4.6cm]{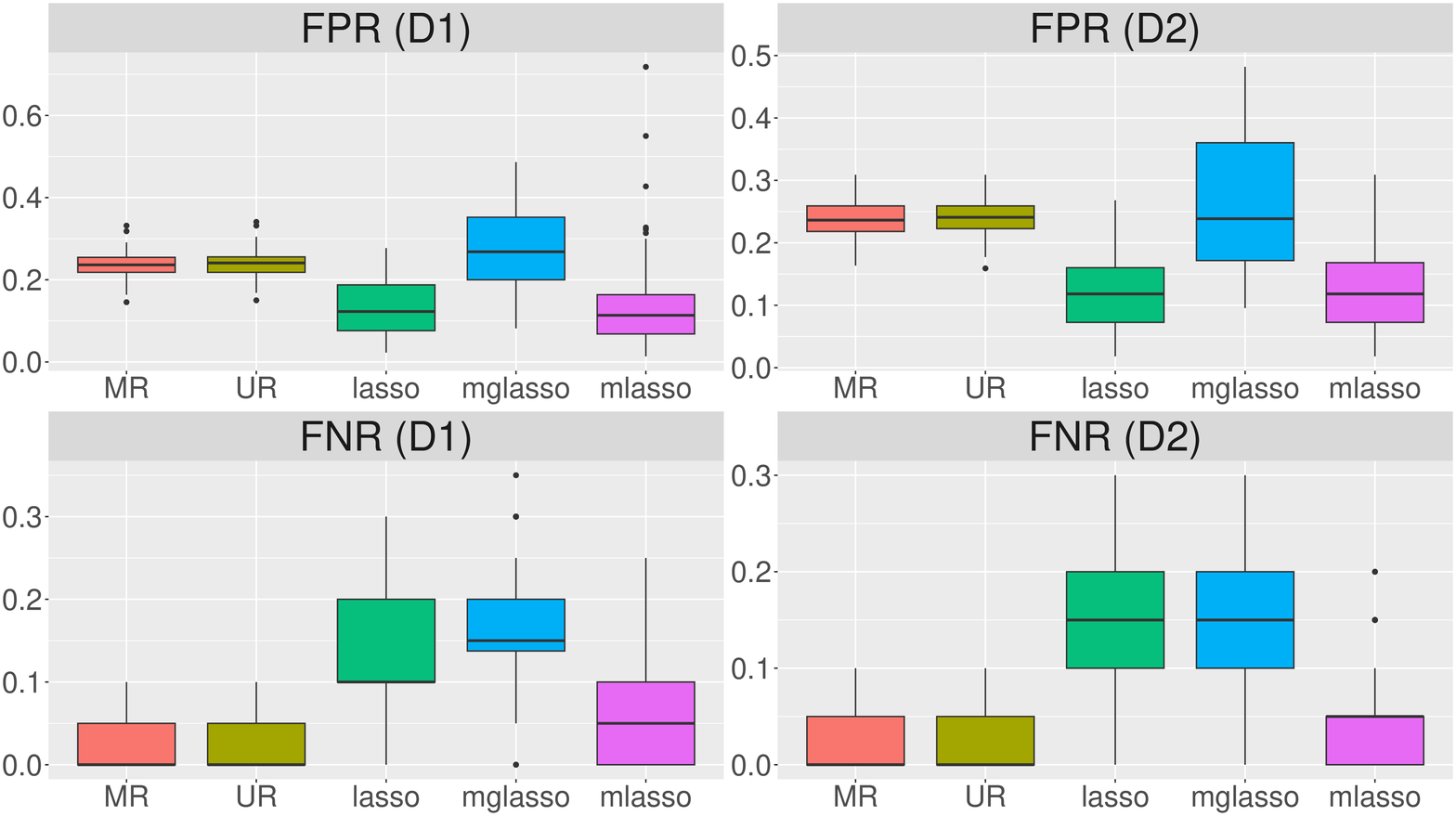}
\vspace{-3.5mm}
\subcaption{$s=50, \rho_x=0.9, \rho_y=0.9$}
\vspace{2.5mm}
\end{minipage}
\caption{Boxplots of FPR and FNR for $n=50$ when the case $M=2$.
}
\label{fig:SimuM2n50_FPRFNR}
\end{figure}

\begin{figure}[htbp]
\begin{minipage}[b]{0.5\linewidth}
\centering
\includegraphics[width=8cm,height=4.6cm]{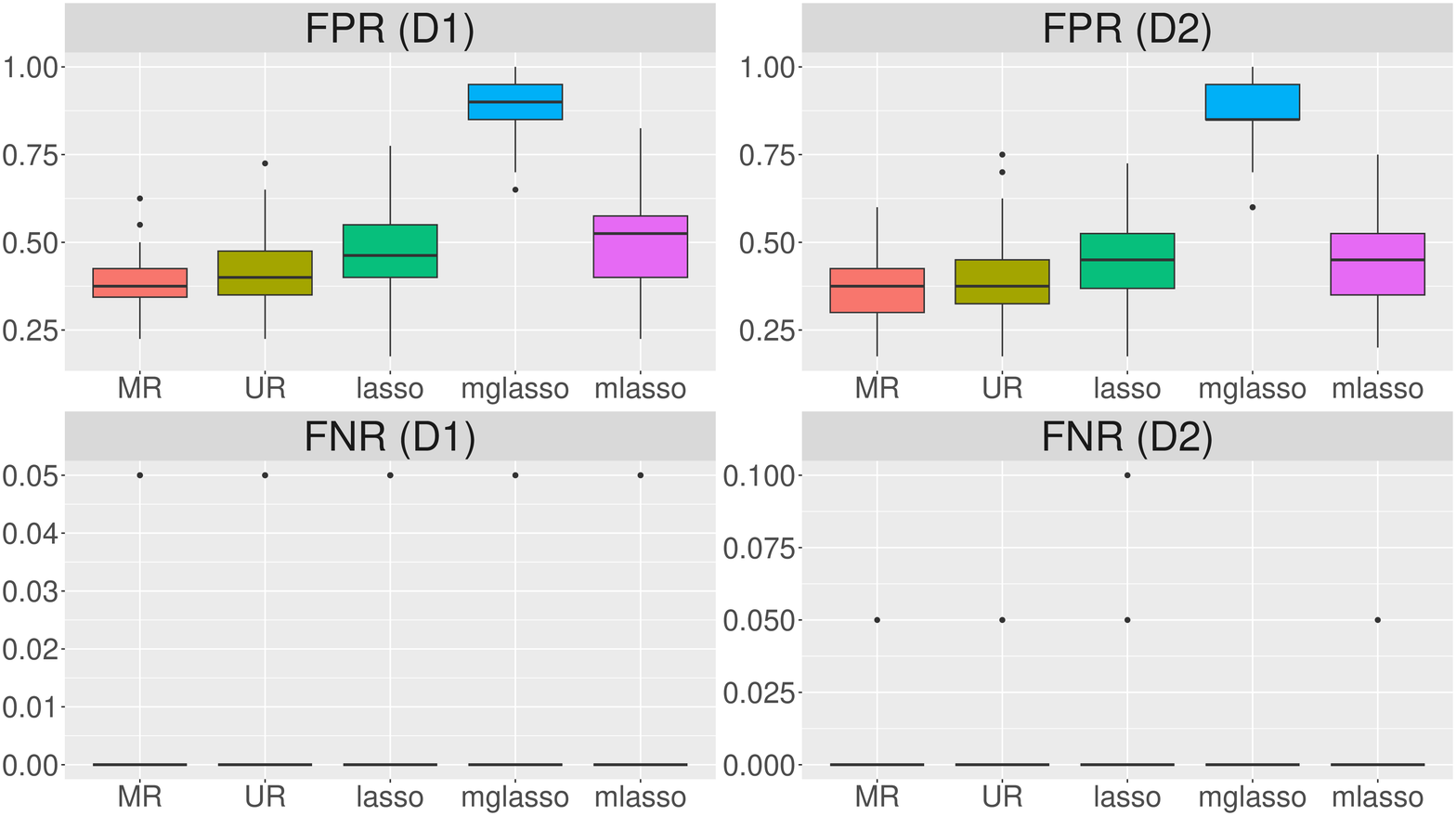}
\vspace{-3.5mm}
\subcaption{$s=5, \rho_x=0.1, \rho_y=0.1$}
\vspace{2.5mm}
\end{minipage}
\begin{minipage}[b]{0.5\linewidth}
\centering
\includegraphics[width=8cm,height=4.6cm]{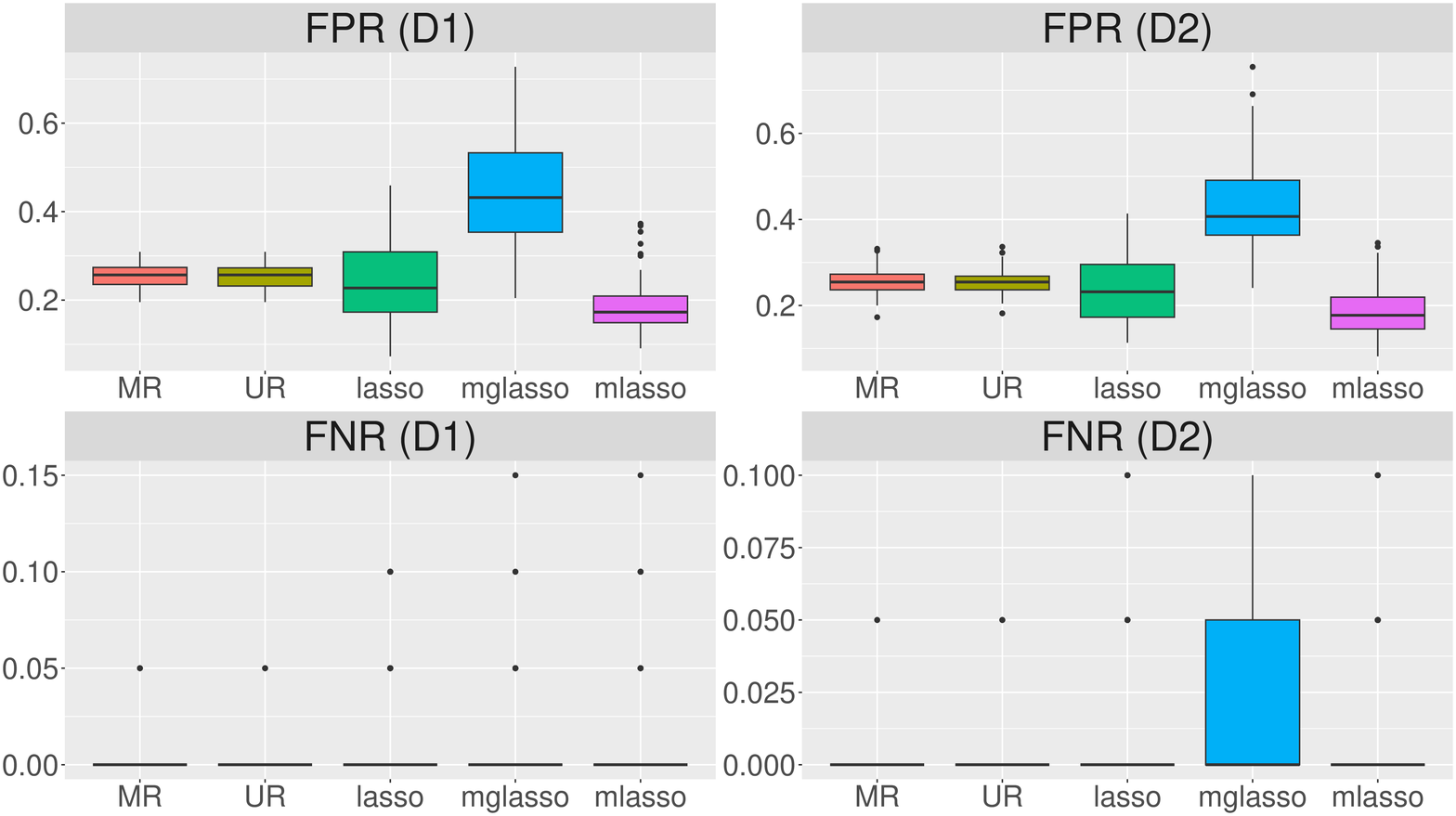} 
\vspace{-3.5mm}
\subcaption{$s=50, \rho_x=0.1, \rho_y=0.1$}
\vspace{2.5mm}
\end{minipage}
\begin{minipage}[b]{0.5\linewidth}
\centering
\includegraphics[width=8cm,height=4.6cm]{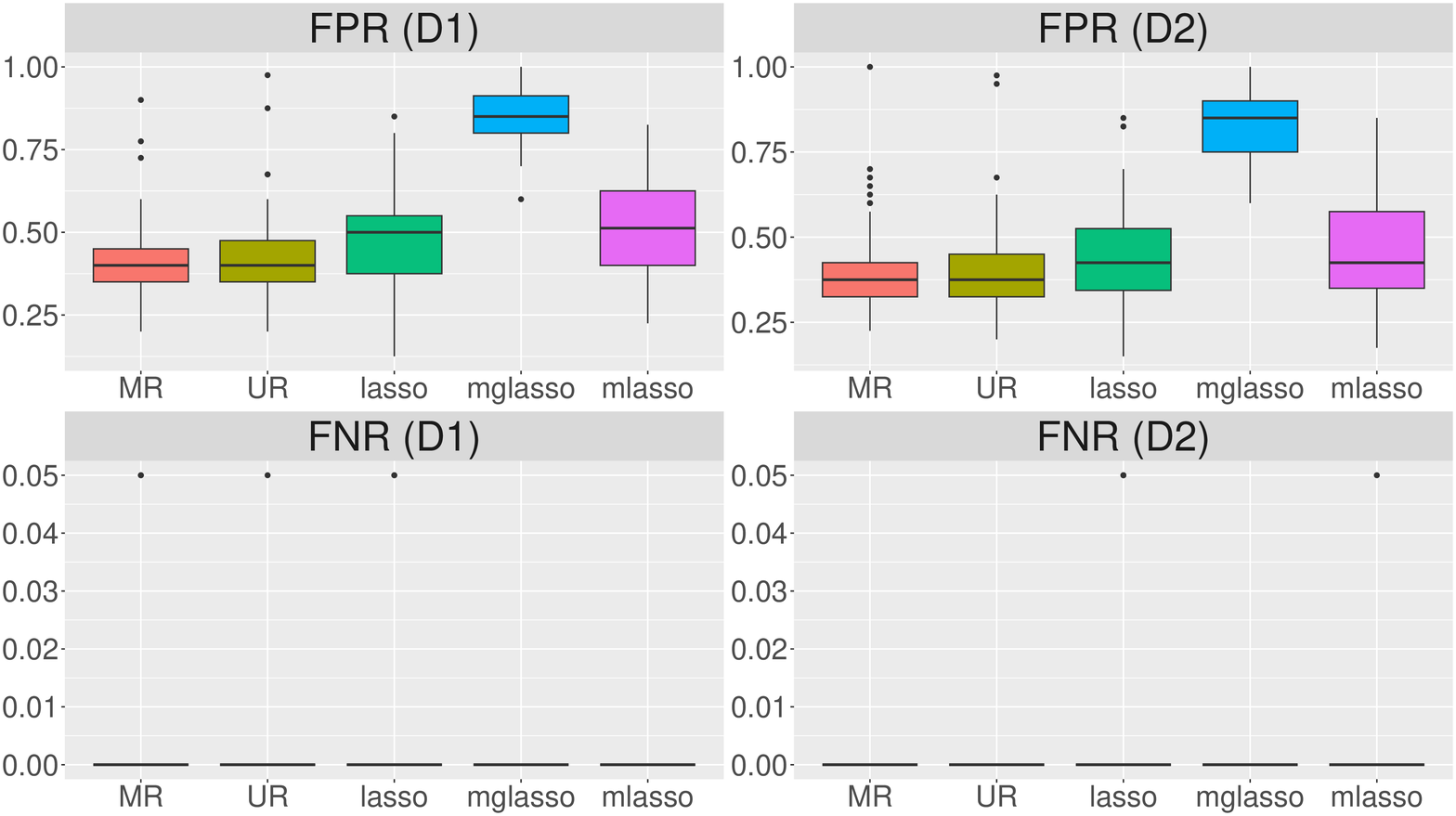} 
\vspace{-3.5mm}
\subcaption{$s=5, \rho_x=0.1, \rho_y=0.9$}
\vspace{2.5mm}
\end{minipage}
\begin{minipage}[b]{0.5\linewidth}
\centering
\includegraphics[width=8cm,height=4.6cm]{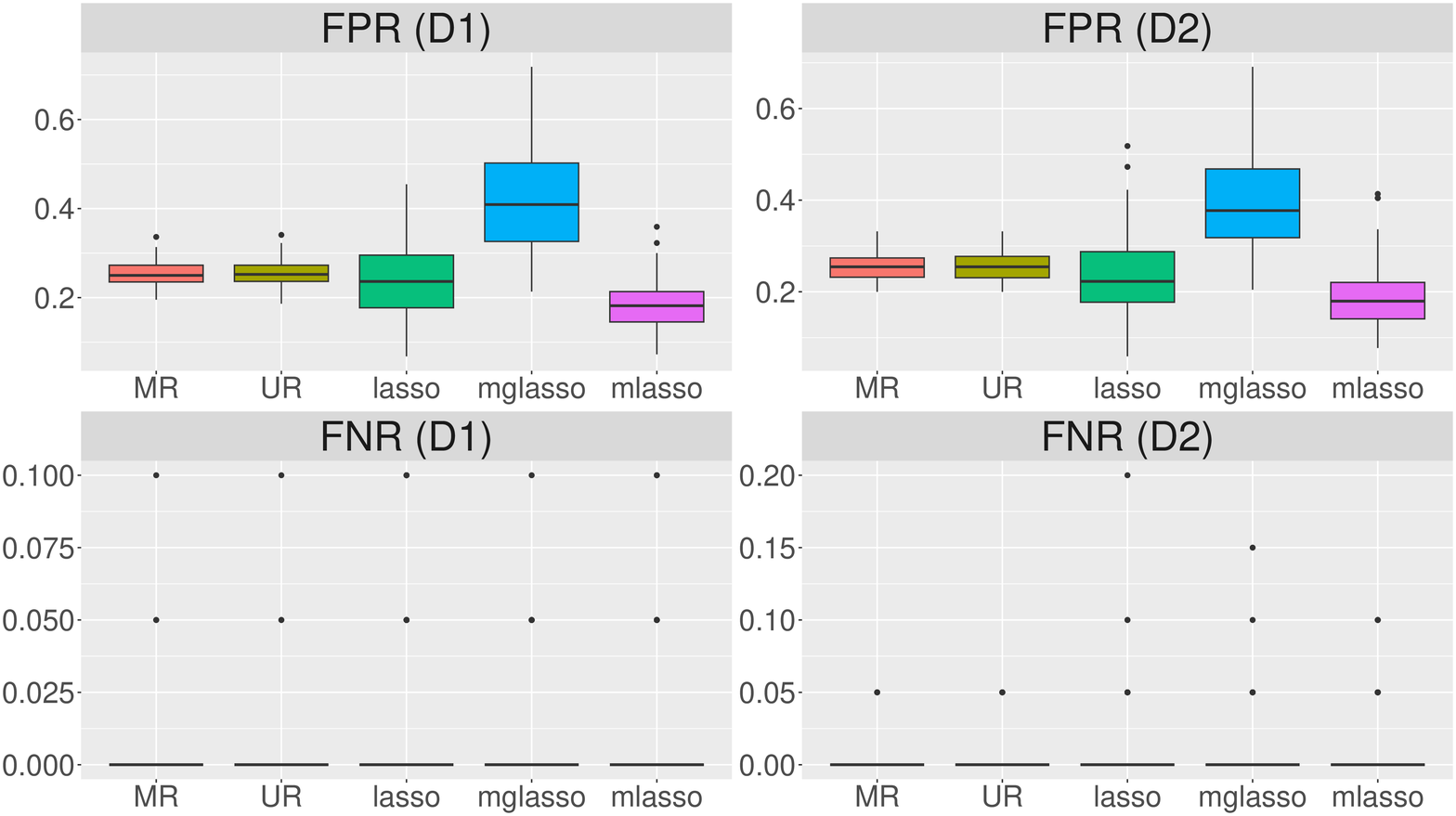}
\vspace{-3.5mm}
\subcaption{$s=50, \rho_x=0.1, \rho_y=0.9$}
\vspace{2.5mm}
\end{minipage}
\begin{minipage}[b]{0.5\linewidth}
\centering
\includegraphics[width=8cm,height=4.6cm]{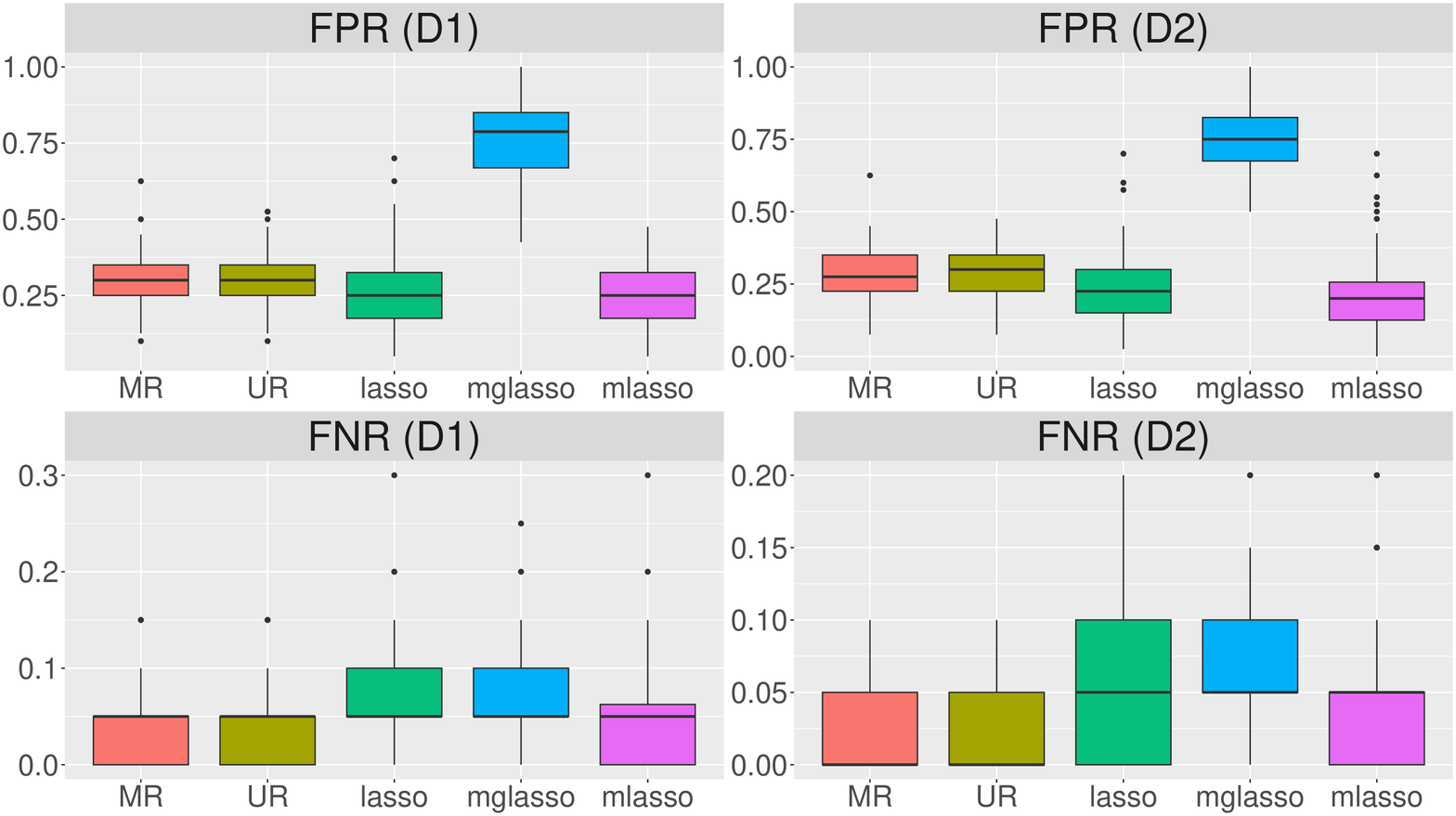}
\vspace{-3.5mm}
\subcaption{$s=5, \rho_x=0.9, \rho_y=0.1$}
\vspace{2.5mm}
\end{minipage}
\begin{minipage}[b]{0.5\linewidth}
\centering
\includegraphics[width=8cm,height=4.6cm]{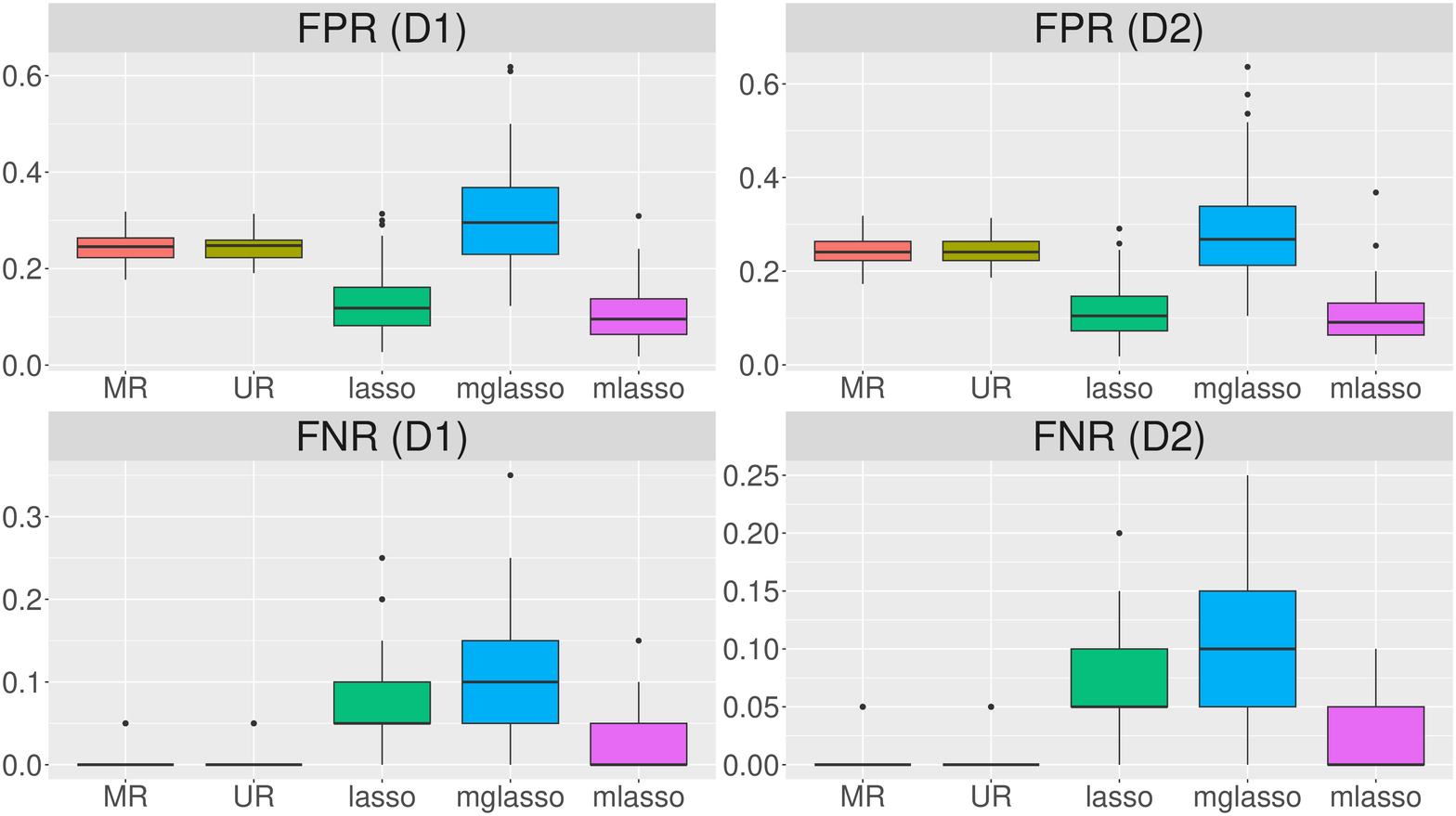}
\vspace{-3.5mm}
\subcaption{$s=50, \rho_x=0.9, \rho_y=0.1$}
\vspace{2.5mm}
\end{minipage}
\begin{minipage}[b]{0.5\linewidth}
\centering
\includegraphics[width=8cm,height=4.6cm]{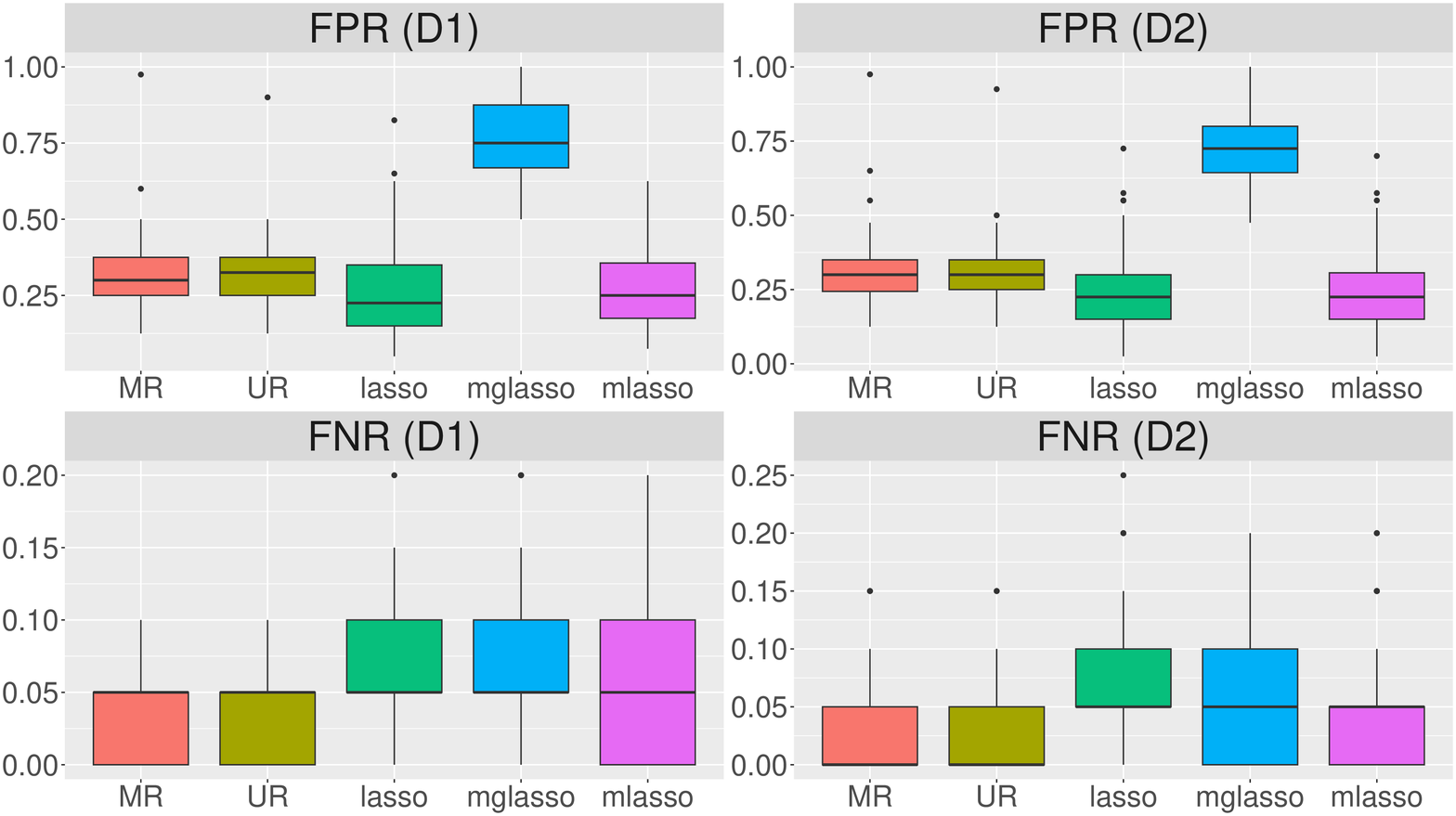}
\vspace{-3.5mm}
\subcaption{$s=5, \rho_x=0.9, \rho_y=0.9$}
\vspace{2.5mm}
\end{minipage}
\begin{minipage}[b]{0.5\linewidth}
\centering
\includegraphics[width=8cm,height=4.6cm]{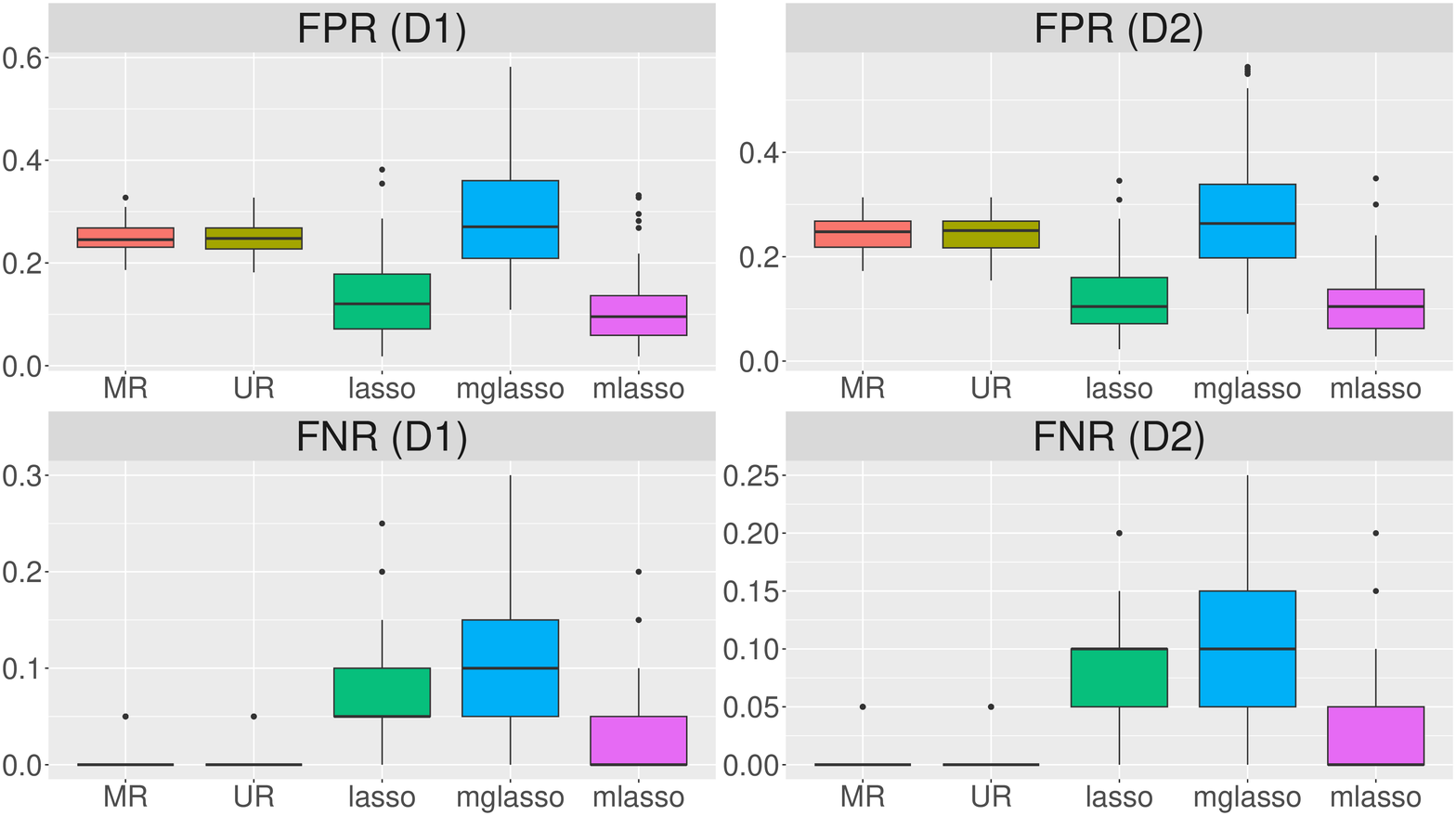}
\vspace{-3.5mm}
\subcaption{$s=50, \rho_x=0.9, \rho_y=0.9$}
\vspace{2.5mm}
\end{minipage}
\caption{Boxplots of FPR and FNR for $n=75$ when the case $M=2$.
}
\label{fig:SimuM2n75_FPRFNR}
\end{figure}

\subsection{Application}
\label{sec:application}

\begin{table}[t]
\begin{center}
\caption{Sample size and the numbers of covariates and responses in the wastewater treatment data.}
\begin{tabular}{lccc}
\hline
Dataset No.    &  Sample size   &  \# of covariates &  \# of responses \\ \hline
1 & 22  & 3,752 & 2  \\ 
2 & 28 & 3,305 & 2 \\ 
3 & 22 &   3,752 & 4 \\ 
4 & 27 & 7,646 & 4 \\ 
\hline
\end{tabular}
\label{table:Datasets}
\end{center}
\end{table}

We applied our proposed method into wastewater treatment data with microbe measurements. 
The data were obtained to find out the relationship between treatment of synthetic industrial wastewater and microbial species in the wastewater treatment \citep{fukushima2022development}. 
As industrial wastewater treatment, the rate (mg/L/day) of nitrite production (NO2$^-$), thiocyanate degradation (SCN$^-$), phenol degradation (C6H6O), and thiosulfate degradation (S2O3$^{2-}$) were observed. 
The microbe data were obtained based on 16S rRNA gene sequencing using next generation sequencer.  
In this experiment, four datasets were given, which is summarized in Table \ref{table:Datasets}. 
Note that the covariates and responses are corresponding to microbial species and industrial wastewater treatment, respectively. 
The responses in the first and second datasets consist of NO2$^-$ and SCN$^-$, while those in the third and fourth datasets do of NO2$^-$, SCN$^-$, C6H6O, and S2O3$^{2-}$. 
For more details of these datasets, we refer to \cite{fukushima2022development}.

We standardized the covariates for each dataset. 
We compared our proposed method with four methods in Section \ref{sec:MonteCarlo}. 
To perform MR and UR, the datasets were preprocessed as follows. 
For Datasets 1 and 2, we extracted 227 microbial species commonly included among these datasets, 161 and 152 microbial species only included in Datasets 1 and 2, respectively. 
For Datasets 3 and 4, we extracted 215 microbial species commonly included among these datasets, 173 and 286 microbial species only included in Datasets 3 and 4, respectively. 
The value of tuning parameters in MR, UR, mlasso was selected by five-fold cross-validation, while its value of lasso and mglasso was done by leave-one-out cross-validation by \textbf{glmnet}.

\begin{table}[t]
\begin{center}
\caption{Cross-validated $R^2$ values for the wastewater treatment data.
The bold value corresponds to the largest $R^2$ for each case. }
\begin{tabular}{llccccc}
\hline
Dataset No. &  Wastewater treatment  &  MR   &  UR &  lasso & mglasso & mlasso \\ \hline
1& NO2$^-$ & 0.753  & 0.645 & 0.501 & 0.199 & \textbf{0.788}  \\ 
& SCN$^-$ & 0.899 &   0.899 & 0.596 &  0.604 & \textbf{0.914} \\  \hline
2& NO2$^-$ & 0.225  & 0.512 & \textbf{0.545} & 0.374 & \textbf{0.545}  \\ 
& SCN$^-$ & \textbf{0.741} &   \textbf{0.741} & 0.403 &  0.444 & 0.733 \\ \hline
3& NO2$^-$ & \textbf{0.792}  & 0.576 & 0.285 & 0.339 & 0.674  \\ 
& SCN$^-$ & \textbf{0.782} &   0.665 & 0.524 &  0.586 & 0.667 \\
& C6H6O & \textbf{0.792}  & 0.741 & 0.489 & 0.607 & 0.681  \\ 
& S2O3$^{2-}$ & \textbf{0.794} &   0.773 & 0.539 &  0.605 & 0.686 \\ \hline
4& NO2$^-$ & \textbf{0.574}  & 0.508 & 0.191 & 0.035 & 0.342  \\ 
& SCN$^-$ & \textbf{0.776} &   0.651 & 0.634 & 0.466  & 0.683 \\
& C6H6O & 0.240  & \textbf{0.449} & $-0.039$ & $-0.014$ & 0.235  \\ 
& S2O3$^{2-}$ & \textbf{0.725} &   0.692 & 0.130 & 0.412 & 0.627 \\ \hline 
\end{tabular}
\label{table:resultApplication}
\end{center}
\end{table}

We computed the leave-one-out cross-validated $R^2$ value for each method. 
Table \ref{table:resultApplication} summarizes the $R^2$ values. 
The mglasso method does not provide the largest value of $R^2$ for all cases. 
Although the lasso method provides the largest value once, it sometimes gives the smallest one for other cases. 
The mlasso has better results in Dataset 1. 
The UR method does not perform well compared to Monte Carlo simulations in Section \ref{sec:MonteCarlo}. 
In many cases, the MR method is better than existing methods.


\section{Conclusion and Discussion}
\label{sec:Conclusions}

We have presented a novel integrative analysis method in the framework of multivariate regression models. 
The integration has been achieved by group regularization. 
We have introduced a computational algorithm to obtain estimates of the parameters via ADMM. 
We have also provided the convergence property of the algorithm. 
Simulation results have showed that our proposed method is competitive or better than competing approaches. 
In the analysis of wastewater treatment datasets, we have found that our proposed method often offers larger $R^2$ values than existing methods.

We note that there are some limitations for our proposed method. 
The squared loss function in the first term in \eqref{eq:mini1_lasso} is simply extended as follows:
\begin{equation*}
\mathrm{tr} \left[ \left\{ Y^m - \bm 1_{n_m} (\bm \alpha^m)^\top -  X^m B^m - Z^m C^m \right\} 
\Sigma 
\left\{ (Y^m - \bm 1_{n_m} (\bm \alpha^m)^\top -  X^m B^m - Z^m C^m \right\}^\top \right].
\end{equation*}
This loss function explicitly includes correlation among responses. 
Thus, using this loss function may be expected to improve accuracy. 
In this article, we assume the homogeneity model. 
Recently, the heterogeneity model, which is defined such that $I(\beta_{jk}^1=0)=\cdots=I(\beta_{jk}^M=0)$ under some $(j,k)$ holds while does not hold under others, has been intensively studied in integrative analysis \citep{huang2017promotingJASA,deng2021integrative,chang2022integrative}.  
It is of interest that our proposed method is extended into heterogeneity models by using sparse group regularization \citep{huang2012selective,simon2013sparse}. 
In Section \ref{sec:application}, we do not interpret the estimated coefficients. 
At present, many coefficient values are estimated as nonzero. 
By using non-convex penalties, e.g., SCAD \citep{fan2001variable} and MCP \citep{zhang2010nearly}, the number of coefficients that are estimated by nonzero needs to be reduced a bit more in order to be interpretable.
We leave them as a future research.

\section*{Acknowledgements}
S. K. was supported by JSPS KAKENHI Grant Number JP19K11854. 



\bibliography{KawanoEtAl}

\clearpage
\renewcommand{\thefigure}{S.\arabic{figure}}
\renewcommand{\thetable}{S.\arabic{table}}
\setcounter{figure}{0}
\appendix




\begin{center}
\textbf{\LARGE Supplementary Material for ``Multivariate regression modeling in integrative analysis via sparse regularization"} 
\end{center}
\begin{center}
{\large by Shuichi Kawano,  Toshikazu Fukushima, 
Junichi Nakagawa,  Mamoru Oshiki}
\end{center}


\renewcommand{\thesection}{S\arabic{section}}

\section{Proof of Theorem 1}

\renewcommand{\theequation}{S.\arabic{equation}}
\setcounter{equation}{0}


Here we prove Theorem 1. 
This proof is basically according to \cite{ye2011split}. 
In this proof, without loss of generality, we set $\bm \alpha^m=\bm 0$ for $m=1,\ldots,M$.

Let $\mathcal B, \mathcal C, \mathcal A$ be
\begin{equation*}
\mathcal B = 
\begin{pmatrix}
\bm \beta_{11} & \cdots & \bm \beta_{p1} \\
\vdots & \ddots & \vdots \\
\bm \beta_{1q} & \cdots & \bm \beta_{pq} \\
\end{pmatrix}, \quad 
\mathcal C =(C^1, \ldots, C^M), \quad
\mathcal A =
\begin{pmatrix}
\textrm{vec}(\mathcal B) \\
\textrm{vec}(\mathcal C) \\
\end{pmatrix},
\end{equation*}
respectively. 
In addition, we define $V(\mathcal A) = \sum_{m=1}^M \| Y^m - X^m B^m - Z^m C^m \|_F^2/(2n_m)$. 
Then, the first order optimality condition of Algorithm 1 provides
\begin{equation}
\begin{split}
\begin{cases}
\bm h_{jk}^{\ell+1} - \rho ( \bm \eta_{jk}^\ell - \bm \beta_{jk}^{\ell+1} + \bm u_{jk}^\ell ) = \bm 0, \\
\lambda \bm s_{jk}^{\ell+1} + \rho (\bm \eta_{jk}^{\ell+1} - \bm \beta_{jk}^{\ell+1} + \bm u_{jk}^{\ell}) = \bm 0, \\
\bm u_{jk}^{\ell+1} = \bm u_{jk}^{\ell} + \bm \eta_{jk}^{\ell+1} - \bm \beta_{jk}^{\ell+1}, \\
\left( \bm f_{k}^m \right)^{\ell+1} + \rho \left\{ (\bm c_k^m)^{\ell+1} - (\bm d_k^m)^{\ell} + (\bm v_k^m)^\ell \right\} = \bm 0, \\
\gamma (\bm t_k^m)^{\ell+1} - \rho \left\{ (\bm c_k^m)^{\ell+1} - (\bm d_k^m)^{\ell+1} + (\bm v_k^m)^{\ell} \right\} = \bm 0, \\
(\bm v_k^m)^{\ell+1} = (\bm v_k^m)^\ell + (\bm c_k^m)^{\ell+1} - (\bm d_k^m)^{\ell+1},
\end{cases}
\end{split}
\label{appendixeq:foo1}
\end{equation}
for $j=1,\ldots,p, \ k=1,\ldots,q,$ and $m=1,\ldots,M$. 
Here, $\left. \bm h_{jk}^{\ell+1} = \frac{\partial V(\mathcal A)}{\partial \bm \beta_{jk}} \right|_{\mathcal A = \mathcal A^{\ell+1}}$, $\bm s_{jk}^{\ell+1} \in \partial \| \bm \eta_{jk}^{\ell+1} \|_2$ satisfying $\| \bm s_{jk}^{\ell+1} \|_2 \le 1$, $\left. (\bm f_{k}^m)^{\ell+1} = \frac{\partial V(\mathcal A)}{\partial \bm c_{k}^m} \right|_{\mathcal A = \mathcal A^{\ell+1}}$, and $(\bm t_k^m)^{\ell+1} \in \partial \| (\bm d_k^m)^{\ell+1} \|_1$.

When we set $\mathcal A^\ast = ( \mathrm{vec} (\mathcal B^\ast)^\top, \mathrm{vec} (\mathcal C^\ast)^\top )^\top$ as a solution of (1), there exist $\bm h_{jk}^\ast$, $\bm s_{jk}^\ast$, $(\bm f_k^m)^\ast$, $(\bm t_k^m)^\ast$ such that
\begin{equation}
\begin{split}
\begin{cases}
\bm h_{jk}^\ast + \lambda \bm s_{jk}^\ast = \bm 0, \\
(\bm f_k^m)^\ast + \gamma (\bm t_k^m)^\ast = \bm 0,
\end{cases}
\end{split}
\label{appendixeq:foo2}
\end{equation}
for $j=1,\ldots,p, \ k=1,\ldots,q,$ and $m=1,\ldots,M$. 
Here, $\left. \bm h_{jk}^\ast = \frac{\partial V(\mathcal A)}{\partial \bm \beta_{jk}} \right|_{\mathcal A = \mathcal A^\ast}$, $\bm s_{jk}^\ast \in \partial \| \bm \beta_{jk}^\ast \|_2$, $\left. (\bm f_{k}^m)^\ast = \frac{\partial V(\mathcal A)}{\partial \bm c_{k}^m} \right|_{\mathcal A = \mathcal A^\ast}$, $(\bm t_k^m)^\ast \in \partial \| (\bm c_k^m)^\ast \|_1$. 
By introducing variables $\bm \eta_{jk}^\ast = \bm \beta_{jk}^\ast$, $\bm u_{jk}^\ast=-\lambda \bm s_{jk}^\ast/\rho$, $(\bm d_k^m)^\ast=(\bm c_k^m)^\ast$, $(\bm v_k^m)^\ast=\gamma ( \bm t_k^m )^\ast/\rho$, we can rewrite \eqref{appendixeq:foo2} into two optimality conditions
\begin{equation}
\begin{split}
\begin{cases}
\bm h_{jk}^\ast - \rho (\bm \eta_{jk}^\ast - \bm \beta_{jk}^\ast + \bm u_{jk}^\ast) = \bm 0, \\
\lambda \bm s_{jk}^\ast + \rho (\bm \eta_{jk}^\ast - \bm \beta_{jk}^\ast + \bm u_{jk}^\ast) = \bm 0, \\
\bm u_{jk}^\ast = \bm u_{jk}^\ast + \bm \eta_{jk}^\ast - \bm \beta_{jk}^\ast,
\end{cases}
\end{split}
\label{appendixeq:foo3}
\end{equation}
for $j=1,\ldots,p, \ k=1,\ldots,q$ and 
\begin{equation}
\begin{split}
\begin{cases}
(\bm f_k^m)^\ast + \rho \left\{ (\bm c_k^m)^\ast - (\bm d_k^m)^\ast + (\bm v_k^m)^\ast \right\} = \bm 0, \\
\gamma (\bm t_k^m)^\ast - \rho \left\{ (\bm c_k^m)^\ast - (\bm d_k^m)^\ast + (\bm v_k^m)^\ast \right\} = \bm 0, \\
(\bm v_k^m)^\ast = (\bm v_k^m)^\ast + (\bm c_k^m)^\ast - (\bm d_k^m)^\ast,
\end{cases}
\end{split}
\label{appendixeq:foo4}
\end{equation}
for $m=1,\ldots,M, \ k=1,\ldots,q$. 
From \eqref{appendixeq:foo1}, \eqref{appendixeq:foo3}, and \eqref{appendixeq:foo4}, we can find that $\bm \eta_{jk}^\ast$, $\bm \beta_{jk}^\ast$, $\bm u_{jk}^\ast$, $(\bm d_k^m)^\ast$, $(\bm c_k^m)^\ast$, $(\bm v_k^m)^\ast$ are a fixed point of Algorithm 1.

First, we discuss \eqref{appendixeq:foo4}. 
In this paragraph, we omit the index $k, m$. 
We denote the errors by
\begin{equation*}
\bm c_e^\ell = \bm c^\ell - \bm c^\ast, \quad \bm d_e^\ell = \bm d^\ell - \bm d^\ast, \quad \bm v_e^\ell = \bm v^\ell - \bm v^\ast. 
\end{equation*}
By subtracting the fourth equation in \eqref{appendixeq:foo1} by the first equation in \eqref{appendixeq:foo4}, we have
\begin{equation*}
\bm f^{\ell+1} - \bm f^\ast + \rho (\bm c_e^{\ell+1} - \bm d_e^\ell + \bm v_e^\ell) = \bm 0. 
\end{equation*}
Taking the inner product for this equality and $\bm c_e^\ell$, we get
\begin{equation}
( \bm f^{\ell+1} - \bm f^\ast )^\top (\bm c^{\ell+1} - \bm c^\ast) + \rho \| \bm c_e^{\ell+1} \|_2^2 - \rho (\bm d_e^{\ell} )^\top (\bm c_e^{\ell+1}) + \rho (\bm v_e^{\ell})^\top \bm c_e^{\ell+1} = 0. 
\label{appendixeq:inner1}
\end{equation}
Similarly, we obtain
\begin{equation}
\gamma ( \bm t^{\ell+1} - \bm t^\ast )^\top (\bm d^{\ell+1} - \bm d^\ast) - \rho (\bm c_e^{\ell+1})^\top \bm d_e^{\ell+1} + \rho \| \bm d_e^{\ell} \|_2^2 - \rho (\bm v_e^{\ell})^\top \bm d_e^{\ell+1} = 0. 
\label{appendixeq:inner2}
\end{equation}
To add the equations \eqref{appendixeq:inner1} and \eqref{appendixeq:inner2} leads to
\begin{align}
\begin{split}
& ( \bm f^{\ell+1} - \bm f^\ast )^\top (\bm c^{\ell+1} - \bm c^\ast) + \gamma ( \bm t^{\ell+1} - \bm t^\ast )^\top (\bm d^{\ell+1} - \bm d^\ast) \\
& + \rho \left\{ \| \bm c_e^{\ell+1} \|_2^2 + \| \bm d_e^{\ell+1} \|_2^2 - ( \bm d_e^\ell + \bm d_e^{\ell+1} )^\top \bm c_e^{\ell+1} + (\bm v_e^\ell)^\top (\bm c_e^{\ell+1} - \bm d_e^{\ell+1}) \right\} = 0.
\label{appendixeq:summing1}
\end{split}
\end{align}
By subtracting the sixth equation in \eqref{appendixeq:foo1} by the third equation in \eqref{appendixeq:foo4}, we have $\bm v_e^{\ell+1} = \bm v_e^{\ell}+\bm c_e^{\ell+1} - \bm d_e^{\ell+1}$. 
Taking square of both sides of this equation leads to
\begin{equation}
(\bm v_e^{\ell})^\top (\bm c_e^{\ell+1} - \bm d_e^{\ell+1} ) = \frac{1}{2} \left( \| \bm v_e^{\ell+1} \|_2^2 - \| \bm v_e^{\ell} \|_2^2 \right) - \frac{1}{2} \| \bm c_e^{\ell+1} - \bm d_e^{\ell+1} \|_2^2.
\label{appendixeq:square1}
\end{equation}
By substituting \eqref{appendixeq:square1} for \eqref{appendixeq:summing1}, we obtain
\begin{align}
\begin{split}
& ( \bm f^{\ell+1} - \bm f^\ast )^\top (\bm c^{\ell+1} - \bm c^\ast) + \gamma ( \bm t^{\ell+1} - \bm t^\ast )^\top (\bm d^{\ell+1} - \bm d^\ast) \\
& + \rho \Big\{ \| \bm c_e^{\ell+1} \|_2^2 + \| \bm d_e^{\ell+1} \|_2^2 - ( \bm d_e^\ell + \bm d_e^{\ell+1} )^\top \bm c_e^{\ell+1} \\
& +  \frac{1}{2} \left( \| \bm v_e^{\ell+1} \|_2^2 - \| \bm v_e^{\ell} \|_2^2 \right) - \frac{1}{2} \| \bm c_e^{\ell+1} - \bm d_e^{\ell+1} \|_2^2 \Big\} = 0.
\label{appendixeq:square2}
\end{split}
\end{align}
Because the equation
\begin{equation*}
\| \bm x \|_2^2 \pm \bm x^\top (\bm y + \bm z) + \|\bm z \|_2^2 = \frac{1}{2} \| \bm x \pm \bm y \|_2^2 + \frac{1}{2} \| \bm x \pm \bm z \|_2^2 + \frac{1}{2} (\|\bm y\|_2^2 - \|\bm z\|_2^2)
\end{equation*}
does hold for any $\bm x, \bm y, \bm z \in \mathbb R^p$, we obtain the following equation from the above equation and \eqref{appendixeq:square2}: 
\begin{align}
\begin{split}
& \frac{\rho}{2} ( \| \bm v_e^\ell \|_2^2 - \| \bm v_e^{\ell+1} \|_2^2 ) + \frac{\rho}{2} ( \| \bm d_e^\ell \|_2^2 - \| \bm d_e^{\ell+1} \|_2^2 ) \\
&= ( \bm f^{\ell+1} - \bm f^\ast )^\top (\bm c^{\ell+1} - \bm c^\ast) + \gamma ( \bm t^{\ell+1} - \bm t^\ast )^\top (\bm d^{\ell+1} - \bm d^\ast) + \frac{\rho}{2} \| \bm c_e^{\ell+1} - \bm d_e^\ell \|_2^2. 
\label{appendixeq:square3}
\end{split}
\end{align}
Since this equation holds for any $m,k$, we have
\begin{align}
\begin{split}
& \frac{\rho}{2} \sum_{m,k} ( \| (\bm v_k^m)_e^\ell \|_2^2 - \| (\bm v_k^m)_e^{\ell+1} \|_2^2 ) + \frac{\rho}{2} \sum_{m,k} ( \| (\bm d_k^m)_e^\ell \|_2^2 - \| (\bm d_k^m)_e^{\ell+1} \|_2^2 ) \\
&= \sum_{m,k} \left\{ (\bm f_k^m)^{\ell+1} - (\bm f_k^m)^\ast \right\}^\top \left\{(\bm c_k^m)^{\ell+1} - (\bm c_k^m)^\ast \right\} \\
&+ \gamma \sum_{m,k} \left\{ (\bm t_k^m)^{\ell+1} - (\bm t_k^m)^\ast \right\}^\top \left\{ (\bm d_k^m)^{\ell+1} - (\bm d_k^m)^\ast \right\}
+ \frac{\rho}{2} \sum_{m,k} \| (\bm c_k^m)_e^{\ell+1} - (\bm d_k^m)_e^\ell \|_2^2. 
\label{appendixeq:square4}
\end{split}
\end{align}

By calculating for \eqref{appendixeq:foo3} in the same way, we can obtain 
\begin{align*}
& \frac{\rho}{2} \sum_{m,k} ( \| (\bm v_k^m)_e^\ell \|_2^2 - \| (\bm v_k^m)_e^{\ell+1} \|_2^2 ) + \frac{\rho}{2} \sum_{m,k} ( \| (\bm d_k^m)_e^\ell \|_2^2 - \| (\bm d_k^m)_e^{\ell+1} \|_2^2 ) \\
&+ \frac{\rho}{2} \sum_{j,k} ( \| (\bm u_{jk})_e^\ell \|_2^2 - \| (\bm u_{jk})_e^{\ell+1} \|_2^2 ) + \frac{\rho}{2} \sum_{j,k} ( \| (\bm \eta_{jk})_e^\ell \|_2^2 - \| (\bm \eta_{jk})_e^{\ell+1} \|_2^2 ) \\
&= \sum_{m,k} \left\{ (\bm f_k^m)^{\ell+1} - (\bm f_k^m)^\ast \right\}^\top \left\{(\bm c_k^m)^{\ell+1} - (\bm c_k^m)^\ast \right\} \\
&+ \gamma \sum_{m,k} \left\{ (\bm t_k^m)^{\ell+1} - (\bm t_k^m)^\ast \right\}^\top \left\{ (\bm d_k^m)^{\ell+1} - (\bm d_k^m)^\ast \right\}
+ \frac{\rho}{2} \sum_{m,k} \| (\bm c_k^m)_e^{\ell+1} - (\bm d_k^m)_e^\ell \|_2^2 \\
&+ \sum_{j,k} ( \bm h_{jk}^{\ell+1} - \bm h_{jk}^\ast )^\top ( \bm \beta_{jk}^{\ell+1} - \bm \beta_{jk}^\ast ) \\
&+ \lambda \sum_{j,k} ( \bm s_{jk}^{\ell+1} - \bm s_{jk}^\ast )^\top (\bm \eta_{jk}^{\ell+1} - \bm \eta_{jk}^\ast)
+ \frac{\rho}{2} \sum_{j,k} \| (\bm \beta_{jk})_e^{\ell+1} - (\bm \eta_{jk})_e^\ell \|_2^2.
\end{align*}
Summing the above equation from $\ell=0$ to $\ell=L$ brings in
\begin{align}
\begin{split}
& \frac{\rho}{2} \sum_{m,k} ( \| (\bm v_k^m)_e^0 \|_2^2 - \| (\bm v_k^m)_e^{L+1} \|_2^2 ) + \frac{\rho}{2} \sum_{m,k} ( \| (\bm d_k^m)_e^0 \|_2^2 - \| (\bm d_k^m)_e^{L+1} \|_2^2 ) \\
&+ \frac{\rho}{2} \sum_{j,k} ( \| (\bm u_{jk})_e^0 \|_2^2 - \| (\bm u_{jk})_e^{L+1} \|_2^2 ) + \frac{\rho}{2} \sum_{j,k} ( \| (\bm \eta_{jk})_e^0 \|_2^2 - \| (\bm \eta_{jk})_e^{L+1} \|_2^2 ) \\
&= \sum_{\ell=1}^L \sum_{m,k} \left\{ (\bm f_k^m)^{\ell+1} - (\bm f_k^m)^\ast \right\}^\top \left\{(\bm c_k^m)^{\ell+1} - (\bm c_k^m)^\ast \right\} \\
&+ \gamma \sum_{\ell=1}^L \sum_{m,k} \left\{ (\bm t_k^m)^{\ell+1} - (\bm t_k^m)^\ast \right\}^\top \left\{ (\bm d_k^m)^{\ell+1} - (\bm d_k^m)^\ast \right\}
+ \frac{\rho}{2} \sum_{\ell=1}^L \sum_{m,k} \| (\bm c_k^m)_e^{\ell+1} - (\bm d_k^m)_e^\ell \|_2^2 \\
&+ \sum_{\ell=1}^L \sum_{j,k} ( \bm h_{jk}^{\ell+1} - \bm h_{jk}^\ast )^\top ( \bm \beta_{jk}^{\ell+1} - \bm \beta_{jk}^\ast ) \\
&+ \lambda \sum_{\ell=1}^L \sum_{j,k} ( \bm s_{jk}^{\ell+1} - \bm s_{jk}^\ast )^\top (\bm \eta_{jk}^{\ell+1} - \bm \eta_{jk}^\ast)
+ \frac{\rho}{2} \sum_{\ell=1}^L \sum_{j,k} \| (\bm \beta_{jk})_e^{\ell+1} - (\bm \eta_{jk})_e^\ell \|_2^2.
\label{appendixeq:combine2}
\end{split}
\end{align}
We note that the second and fifth terms in the right-hand side in \eqref{appendixeq:combine2} satisfy non-negativity, because we have the following inequalities:
\begin{align*}
( \bm s_{jk}^{\ell+1} - \bm s_{jk}^\ast )^\top ( \bm \eta_{jk}^{\ell+1} - \bm \eta_{jk}^\ast )
&=\| \bm \eta_{jk}^{\ell+1} \|_2^2 - \| \bm \eta_{jk}^{\ast} \|_2^2 - (\bm s_{jk}^\ast)^\top (\bm \eta_{jk}^{\ell+1} - \bm \eta_{jk}^\ast) \\
&+ \| \bm \eta_{jk}^{\ast} \|_2^2 - \| \bm \eta_{jk}^{\ell+1} \|_2^2  - (\bm s_{jk}^{\ell+1})^\top (\bm \eta_{jk}^{\ast} - \bm \eta_{jk}^{\ell+1}) \\
& \ge 0, \quad (\because \text{definition of sub-gradient.}) \\
\left\{ (\bm t_k^m)^{\ell+1} - (\bm t_k^m)^\ast \right\}^\top \left\{ (\bm d_k^m)^{\ell+1} - (\bm d_k^m)^\ast \right\}
&= \| (\bm d_k^m)^{\ell+1} \|_1 - \| (\bm d_k^m)^{\ast} \|_1 \\
&- \left\{ (\bm t_k^m)^\ast \right\}^\top \left\{ (\bm d_k^m)^{\ell+1} - (\bm d_k^m)^\ast \right\} \\
& + \| (\bm d_k^m)^{\ast} \|_1 - \| (\bm d_k^m)^{\ell+1} \|_1 \\
&- \left\{ (\bm t_k^m)^{\ell+1} \right\}^\top \left\{ (\bm d_k^m)^{\ast} - (\bm d_k^m)^{\ell+1} \right\} \\
& \ge 0. \quad (\because \text{definition of sub-gradient.})
\end{align*}
In addition, when we set $I^\ell$ as a vector whose $i$-th element is defined by $\left. \frac{\partial V(\mathcal A)}{\partial (\mathcal A)_i} \right|_{\mathcal A = \mathcal A^{\ell}}$, the first and third terms also satisfy non-negativity from
\begin{align}
\begin{split}
& \sum_{j,k} ( \bm h_{jk}^{\ell+1} - \bm h_{jk}^\ast )^\top ( \bm \beta_{jk}^{\ell+1} - \bm \beta_{jk}^\ast )
+ \sum_{m,k} \left\{ (\bm f_k^m)^{\ell+1} - (\bm f_k^m)^\ast \right\}^\top \left\{(\bm c_k^m)^{\ell+1} - (\bm c_k^m)^\ast \right\} \\
&= - (I^\ast)^\top (\mathcal A^{\ell+1} - \mathcal A^\ast) + (I^{\ell+1})^\top (\mathcal A^{\ell+1} - \mathcal A^\ast) \\
&= V(\mathcal A^{\ell+1}) - V(\mathcal A^{\ast}) - (I^\ast)^\top (\mathcal A^{\ell+1} - \mathcal A^\ast)
 + V (\mathcal A^{\ast}) - V(\mathcal A^{\ell+1}) - (I^{\ell+1})^\top (\mathcal A^{\ast} - \mathcal A^{\ell+1}) \\
& \ge 0. \quad (\because \text{definition of gradient and convexity.})
\label{appendixeq:combine3}
\end{split}
\end{align}
These facts conclude that all terms in \eqref{appendixeq:combine2} are nonnegative. 
Thus, we have 
\begin{align*}
&\sum_{\ell=0}^\infty \left[ \sum_{j,k} ( \bm h_{jk}^{\ell+1} - \bm h_{jk}^\ast )^\top ( \bm \beta_{jk}^{\ell+1} - \bm \beta_{jk}^\ast )
+ \sum_{m,k} \left\{ (\bm f_k^m)^{\ell+1} - (\bm f_k^m)^\ast \right\}^\top \left\{(\bm c_k^m)^{\ell+1} - (\bm c_k^m)^\ast \right\} \right] \\
& \le \frac{\rho}{2} \sum_{j,k}  \| (\bm u_{jk})_e^0 \|_2^2  + \frac{\rho}{2} \sum_{j,k}  \| (\bm \eta_{jk})_e^0 \|_2^2 
+ \frac{\rho}{2} \sum_{m,k}  \| (\bm v_k^m)_e^0 \|_2^2  + \frac{\rho}{2} \sum_{m,k}  \| (\bm d_k^m)_e^0 \|_2^2.
\end{align*}
This leads to
\begin{equation*}
\lim_{\ell \to \infty} \left[ \sum_{j,k} ( \bm h_{jk}^{\ell} - \bm h_{jk}^\ast )^\top ( \bm \beta_{jk}^{\ell} - \bm \beta_{jk}^\ast )
+ \sum_{m,k} \left\{ (\bm f_k^m)^{\ell} - (\bm f_k^m)^\ast \right\}^\top \left\{(\bm c_k^m)^{\ell} - (\bm c_k^m)^\ast \right\} \right] = 0.
\end{equation*}
From this convergence and \eqref{appendixeq:combine3}, we can prove
\begin{equation}
\lim_{\ell \to \infty} \left[ V(\mathcal A^{\ell}) - V(\mathcal A^{\ast}) - (I^\ast)^\top (\mathcal A^{\ell} - \mathcal A^\ast) \right] = 0.
\label{appendixeq:convergence1}
\end{equation}
Similarly, we have
\begin{align*}
&\lim_{\ell \to \infty} \lambda \sum_{j,k} \left\{ \| \bm \eta_{jk}^\ell \|_2 - \| \bm \eta_{jk}^\ast \|_2  - (\bm s_{jk}^\ast)^\top ( \bm \eta_{jk}^\ell - \bm \eta_{jk}^\ast ) \right\} = 0, \\
&\lim_{\ell \to \infty} \| \bm \beta_{jk}^\ell - \bm \eta_{jk}^\ell \|_2 = 0, \\
&\lim_{\ell \to \infty} \gamma \sum_{m,k} \left[ \| (\bm d_k^m)^\ell \|_1 - \| (\bm d_k^m)^\ast \|_1 - \{ (\bm t_k^m)^\ast \}^\top \{ (\bm d_k^m)^\ell - (\bm d_k^m)^\ast \} \right] = 0, \\
&\lim_{\ell \to \infty} \| (\bm c_k^m)^\ell - (\bm d_k^m)^\ell \|_2 = 0.
\end{align*}
Because the norms $\| \cdot \|_1, \| \cdot \|_2$ and the inner product are continuous and all norms on a finite dimensional space are equivalent, the equations
\begin{align}
&\lim_{\ell \to \infty} \lambda \sum_{j,k} \left\{ \| \bm \beta_{jk}^\ell \|_2 - \| \bm \beta_{jk}^\ast \|_2  - (\bm s_{jk}^\ast)^\top ( \bm \beta_{jk}^\ell - \bm \beta_{jk}^\ast ) \right\} = 0, \label{appendixeq:convergence2} \\
&\lim_{\ell \to \infty} \gamma \sum_{m,k} \left[ \| (\bm c_k^m)^\ell \|_1 - \| (\bm c_k^m)^\ast \|_1 - \{ (\bm t_k^m)^\ast \}^\top \{ (\bm c_k^m)^\ell - (\bm c_k^m)^\ast \} \right] = 0 \label{appendixeq:convergence3}
\end{align}
hold. 
From the equations \eqref{appendixeq:convergence1}, \eqref{appendixeq:convergence2}, \eqref{appendixeq:convergence3}, we obtain
\begin{align*}
& \lim_{\ell \to \infty}  V (\mathcal A^\ell) - V (\mathcal A^\ast) - (I^\ast)^\top (\mathcal A^\ell - \mathcal A^\ast) \\
& \hspace{5pt} + \lambda \sum_{j,k} \left\{ \| \bm \beta_{jk}^\ell \|_2 - \| \bm \beta_{jk}^\ast \|_2 - (\bm s_{jk}^\ast)^\top (\bm \beta_{jk}^\ell - \bm \beta_{jk}^\ast)  \right\} \\
& \hspace{5pt} + \gamma \sum_{m,k} \left[ \| (\bm c_k^m)^\ell \|_1 - \| (\bm c_k^m)^\ast \|_1 - \{(\bm t_k^m)^\ast \}^\top \{(\bm c_k^m)^\ell - (\bm c_k^m)^\ast) \}  \right] = 0 \\
\iff & \lim_{\ell \to \infty} V (\mathcal A^\ell) + \lambda \sum_{j,k} \| \bm \beta_{jk}^\ell \|_2 + \gamma \sum_{m,k} \| (\bm c_k^m)^\ell \|_1 \\
& \hspace{5pt} - \left\{ V (\mathcal A^\ast) + \lambda \sum_{j,k} \| \bm \beta_{jk}^\ast \|_2 + \gamma \sum_{m,k}  \| (\bm c_k^m)^\ast \|_1 \right\} \\
& \hspace{5pt} - \underbrace{\left[ (I^\ast)^\top (\mathcal A^\ell - \mathcal A^\ast) + \lambda \sum_{j,k} (\bm s_{jk}^\ast)^\top ( \bm \beta_{jk}^\ell - \bm \beta_{jk}^\ast ) + \gamma \sum_{m,k} \{ (\bm t_k^m)^\ast \}^\top \{ (\bm c_k^m)^\ell - (\bm c_k^m)^\ast \} \right]}_{\rm (A)}= 0.
\end{align*}
Formula (A) turns out to be zero because of the first optimality condition of \eqref{appendixeq:foo2}. 
Thus, we can prove
\begin{equation*}
\lim_{\ell \to \infty} \left[ V (\mathcal A^\ell) + \lambda \sum_{j,k} \| \bm \beta_{jk}^\ell \|_2 + \gamma \sum_{m,k} \| (\bm c_k^m)^\ell \|_1 \right]
= V (\mathcal A^\ast) + \lambda \sum_{j,k} \| \bm \beta_{jk}^\ast \|_2 + \gamma \sum_{m,k} \| (\bm c_k^m)^\ast \|_1.
\end{equation*}

Next, we prove that $\lim_{\ell \to \infty} \| \bm \theta^\ell - \bm \theta^\ast \|_2=0$ holds whenever $\bm \theta^\ast$ is a unique solution. 
Since $\mathcal L (\bm \theta)$ is a convex 
function, we can directly apply the proof of \cite{ye2011split} by replacing $\Phi(\bm \beta)$ in \cite{ye2011split} with $\mathcal L (\bm \theta)$. 
This completes the proof of Theorem 1. \qed

\section{Additional figures in the Monte Carlo simulations}

\begin{landscape}
\begin{figure}[htbp]
\begin{minipage}[b]{0.33\linewidth}
\centering
\includegraphics[width=7cm,height=4.6cm]{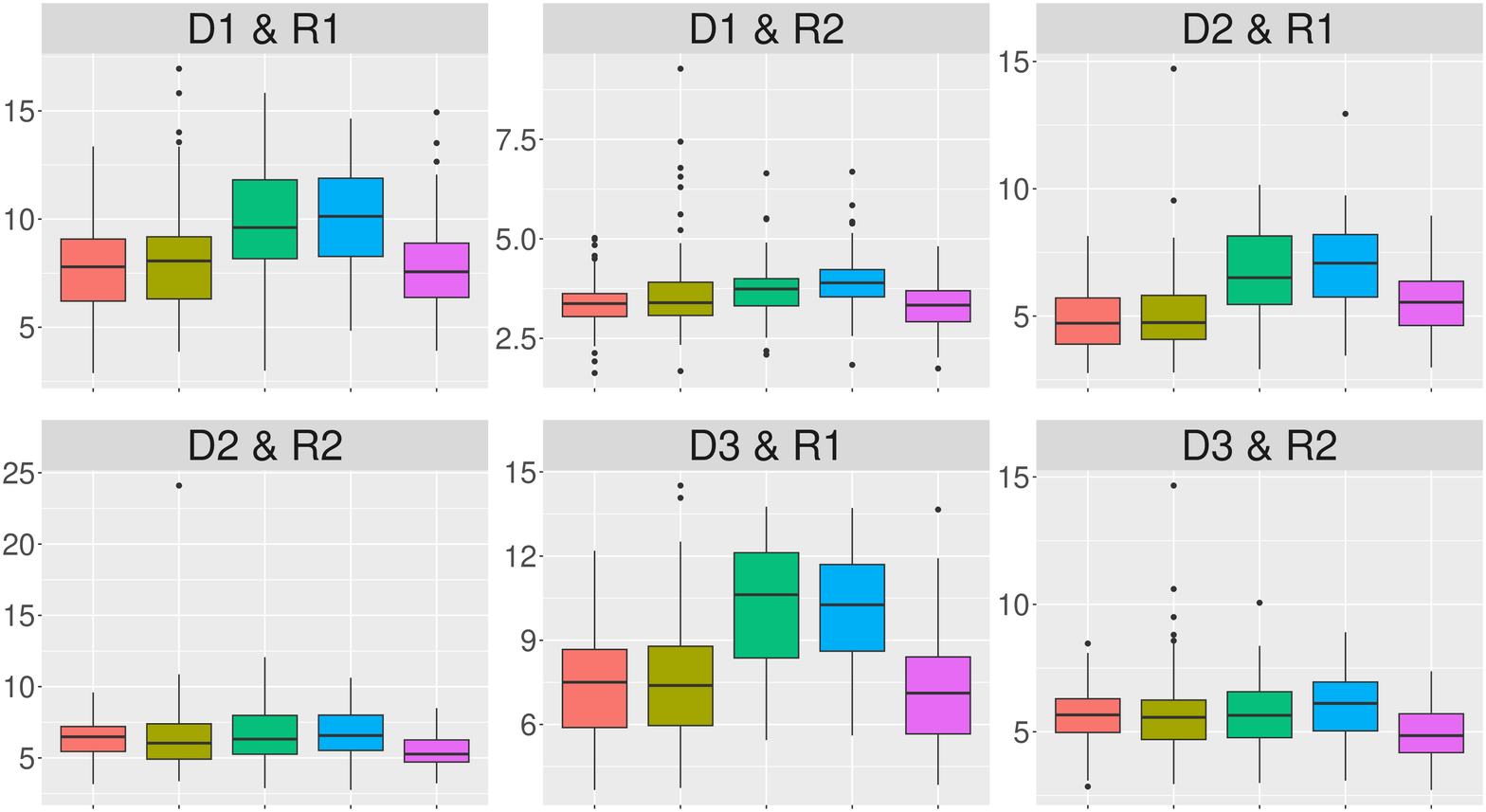}
\vspace{-2.5mm}
\subcaption{$s=5, \rho_x=0.1, \rho_y=0.1$}
\end{minipage}
\begin{minipage}[b]{0.33\linewidth}
\centering
\includegraphics[width=7cm,height=4.6cm]{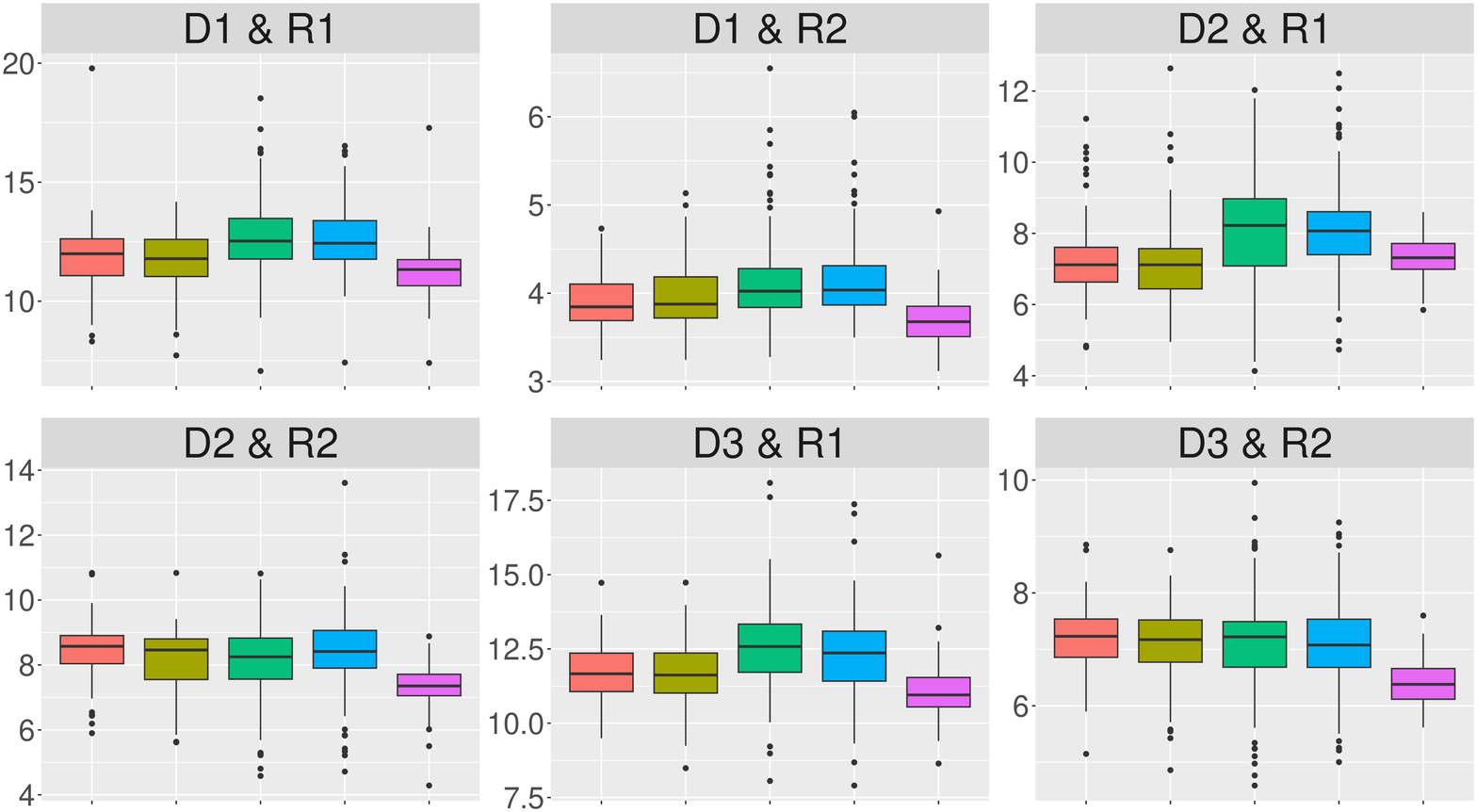} 
\vspace{-2.5mm}
\subcaption{$s=50, \rho_x=0.1, \rho_y=0.1$}
\end{minipage}
\begin{minipage}[b]{0.33\linewidth}
\centering
\includegraphics[width=7cm,height=4.6cm]{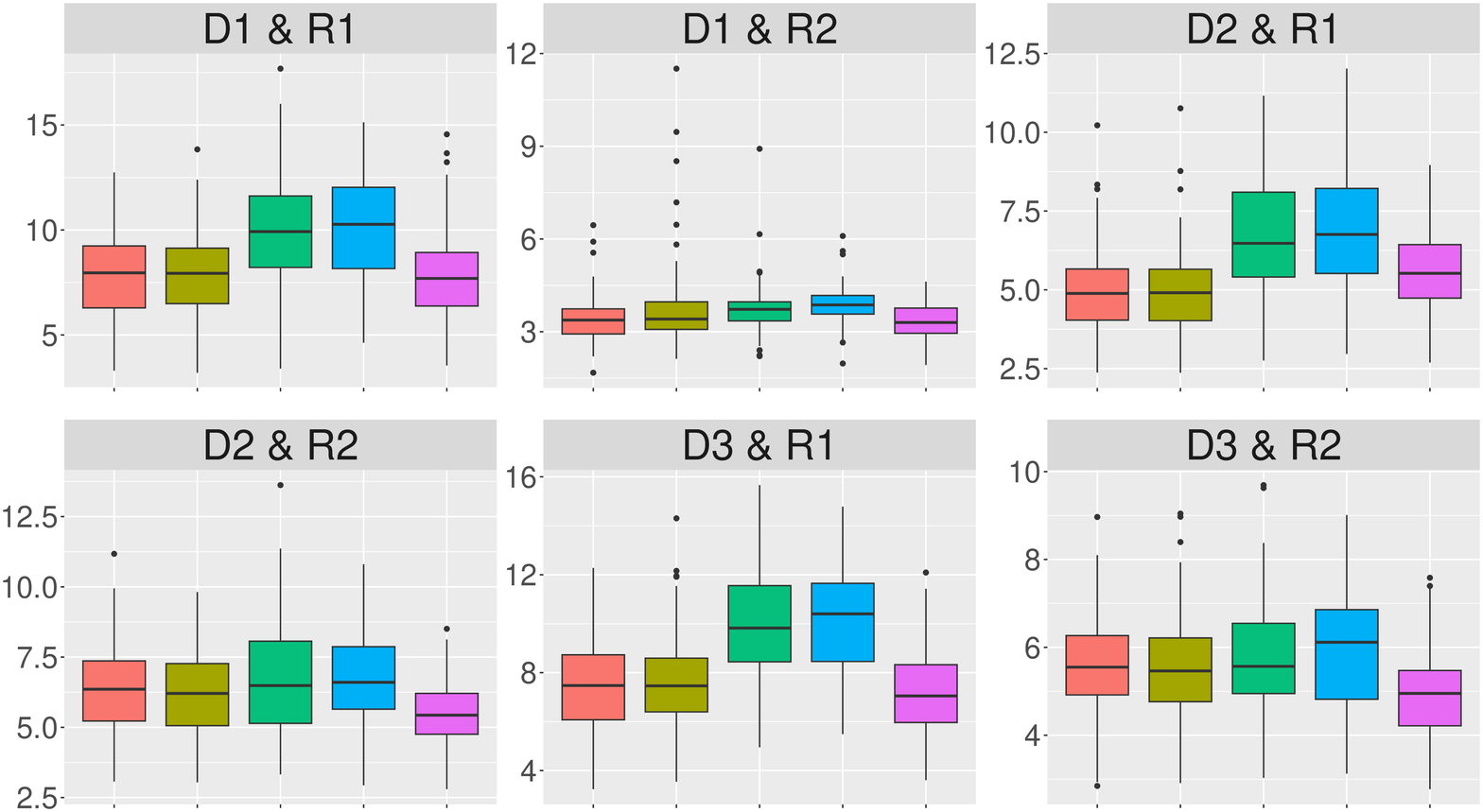} 
\vspace{-2.5mm}
\subcaption{$s=5, \rho_x=0.1, \rho_y=0.9$}
\end{minipage}
\begin{minipage}[b]{0.33\linewidth}
\centering
\includegraphics[width=7cm,height=4.6cm]{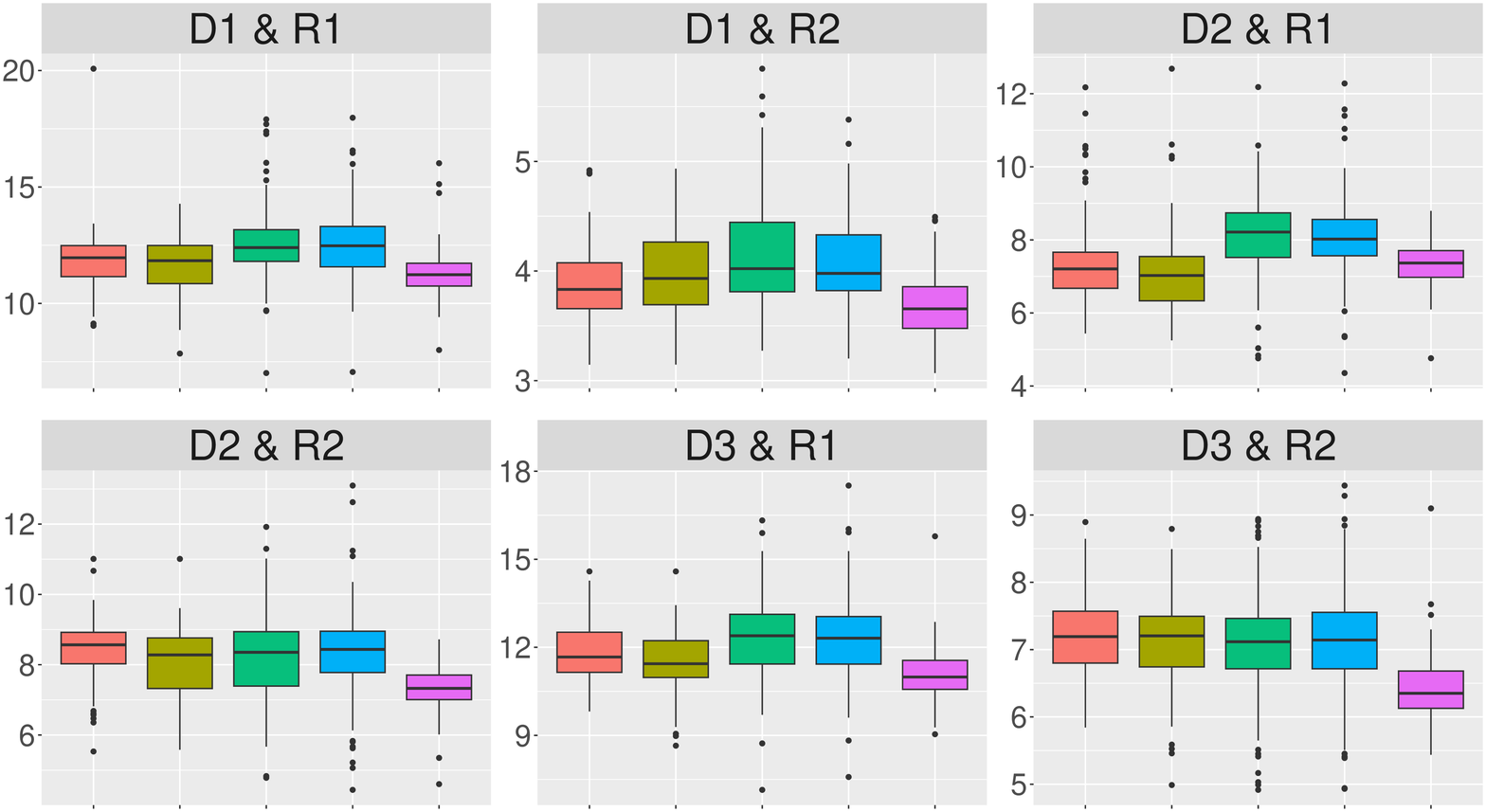}
\vspace{-2.5mm}
\subcaption{$s=50, \rho_x=0.1, \rho_y=0.9$}
\end{minipage}
\begin{minipage}[b]{0.33\linewidth}
\centering
\includegraphics[width=7cm,height=4.6cm]{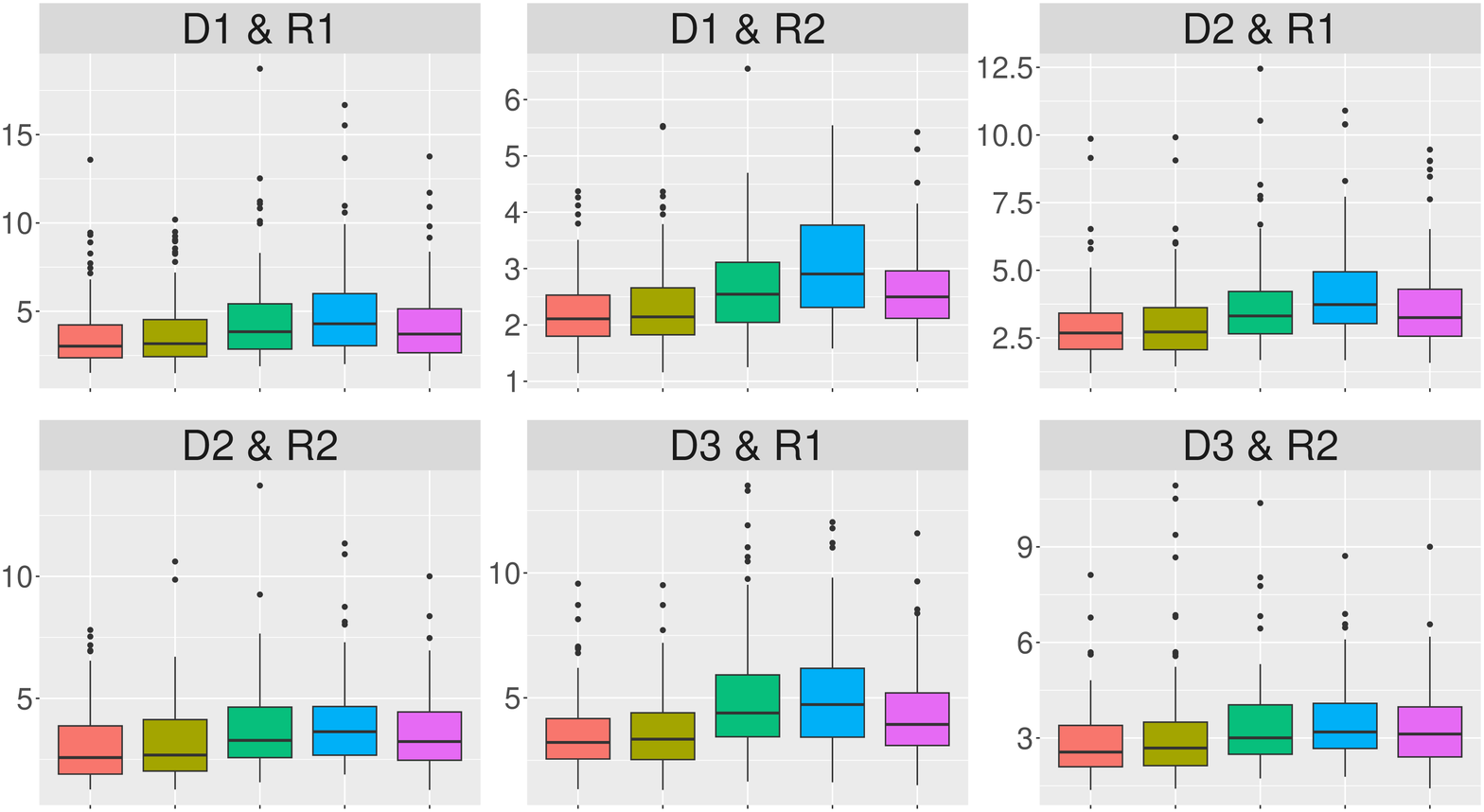}
\vspace{-2.5mm}
\subcaption{$s=5, \rho_x=0.9, \rho_y=0.1$}
\end{minipage}
\begin{minipage}[b]{0.33\linewidth}
\centering
\includegraphics[width=7cm,height=4.6cm]{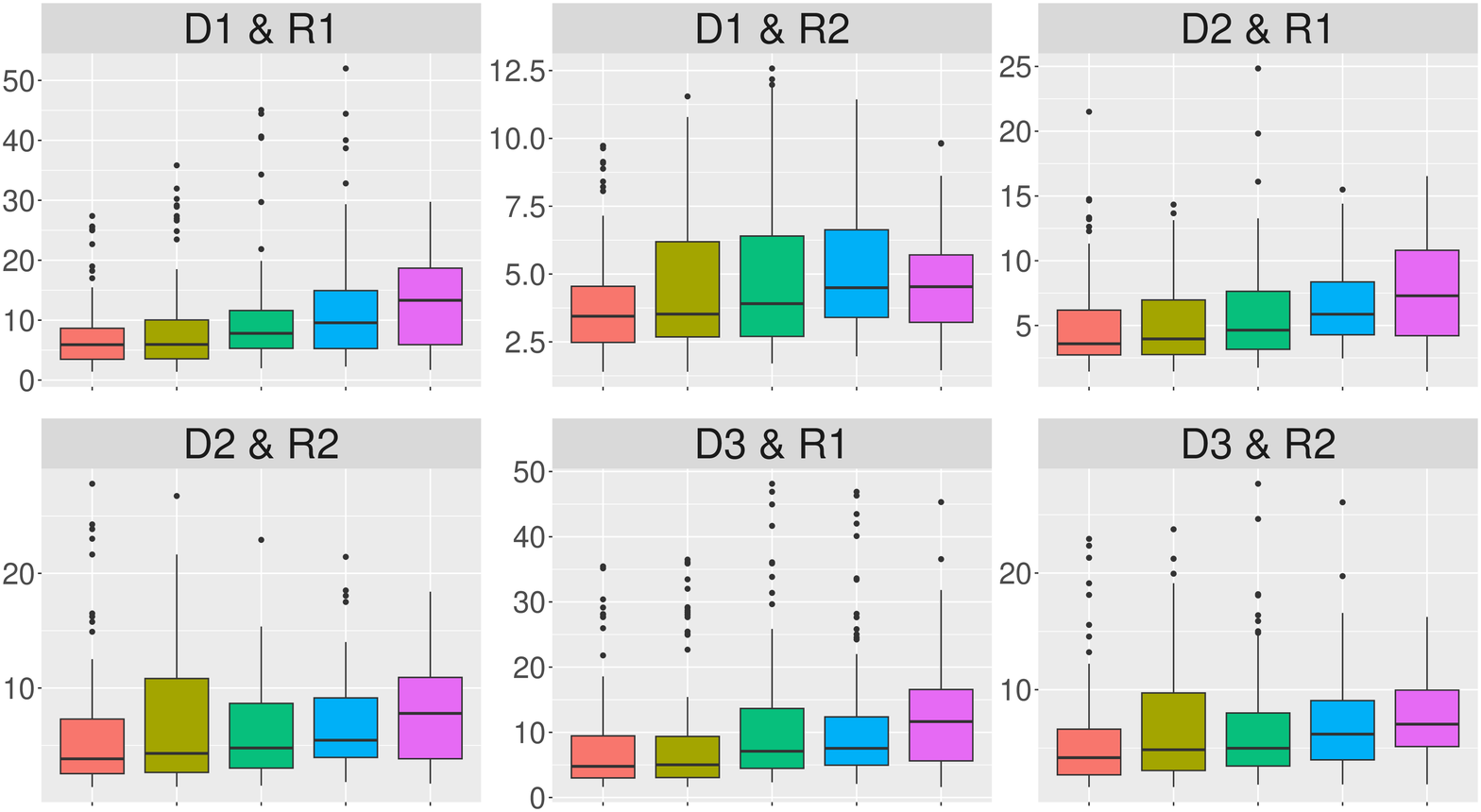}
\vspace{-2.5mm}
\subcaption{$s=50, \rho_x=0.9, \rho_y=0.1$}
\end{minipage}
\begin{minipage}[b]{0.33\linewidth}
\centering
\includegraphics[width=7cm,height=4.6cm]{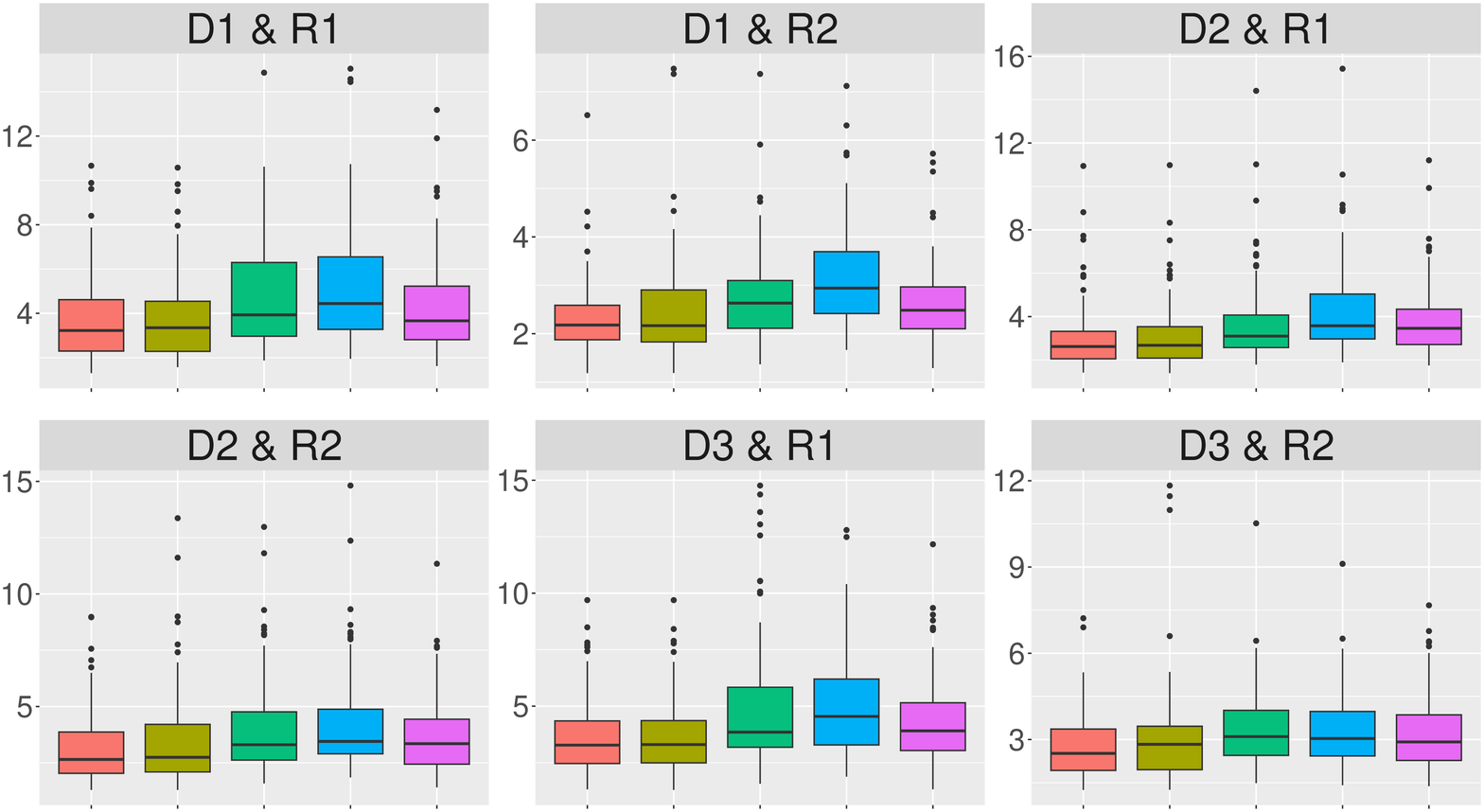}
\vspace{-2.5mm}
\subcaption{$s=5, \rho_x=0.9, \rho_y=0.9$}
\end{minipage}
\begin{minipage}[b]{0.33\linewidth}
\centering
\includegraphics[width=7cm,height=4.6cm]{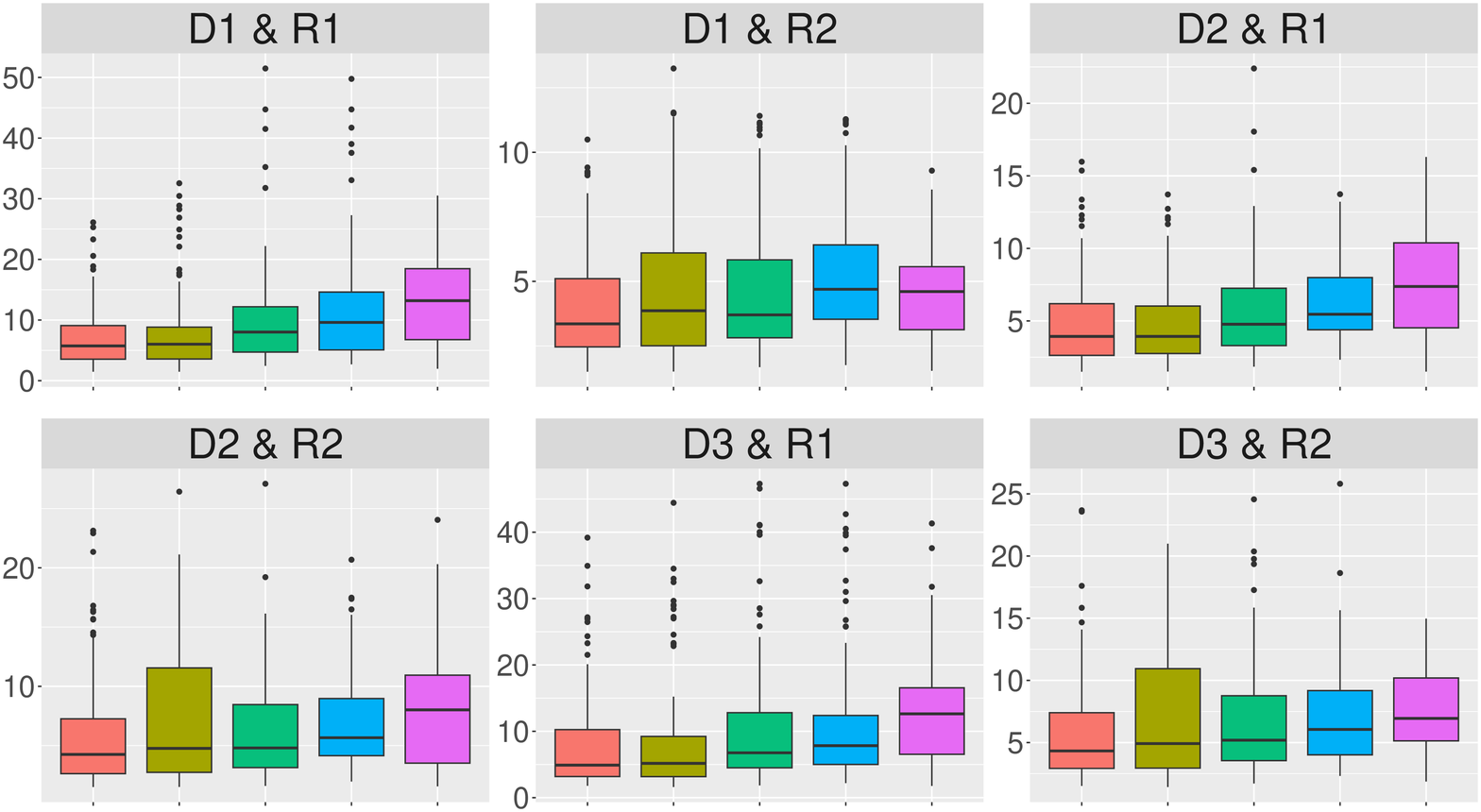}
\vspace{-2.5mm}
\subcaption{$s=50, \rho_x=0.9, \rho_y=0.9$}
\end{minipage}
\caption{Boxplots of MSE for $n=15$ when the case $M=3$.
The red boxplot indicates MR, dark yellow UR, green lasso, blue mglasso, and magenta mlasso. 
}
\label{fig:SimuM3n15}
\end{figure}
\end{landscape}

\begin{landscape}
\begin{figure}[htbp]
\begin{minipage}[b]{0.33\linewidth}
\centering
\includegraphics[width=7cm,height=4.6cm]{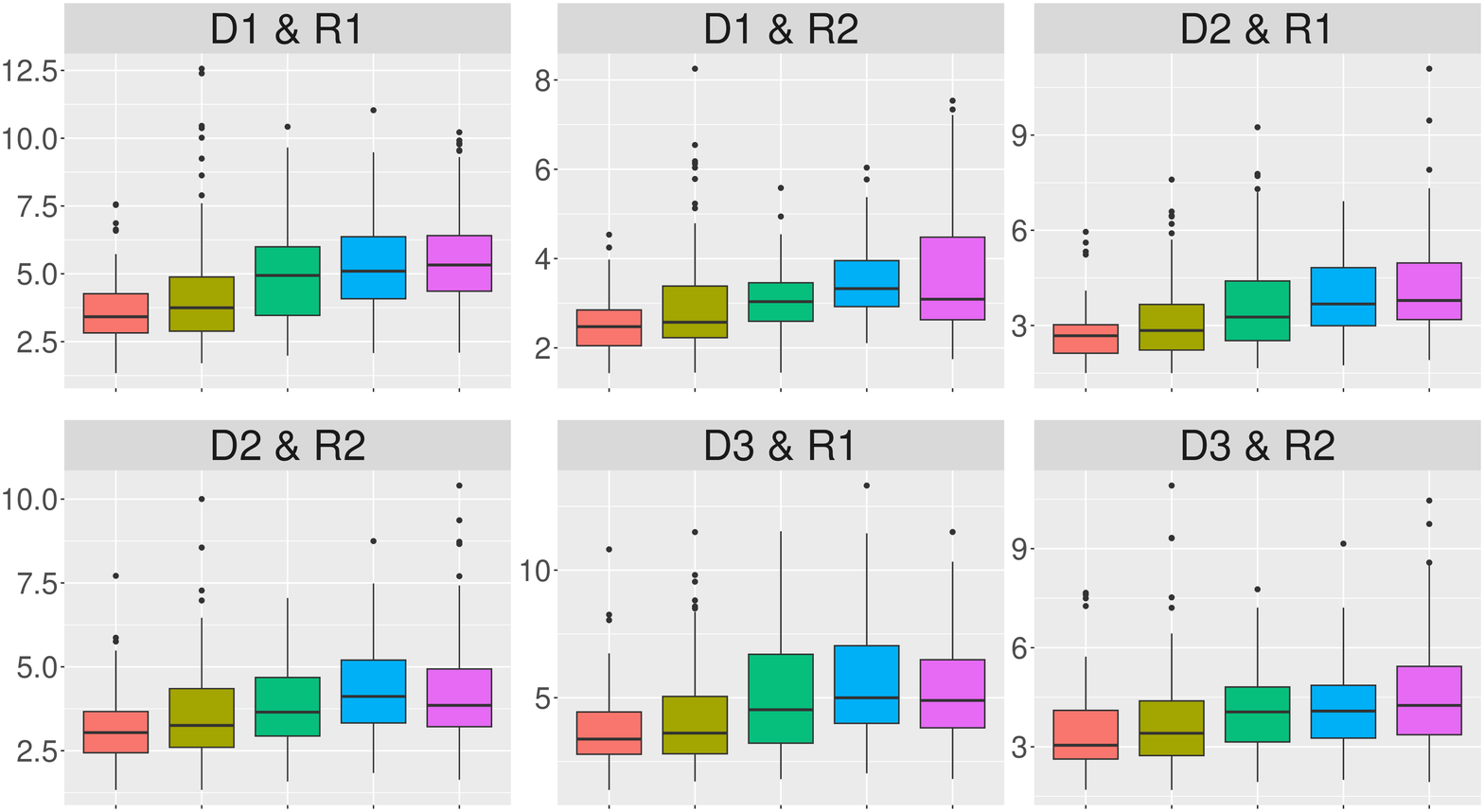}
\vspace{-2.5mm}
\subcaption{$s=5, \rho_x=0.1, \rho_y=0.1$}
\end{minipage}
\begin{minipage}[b]{0.33\linewidth}
\centering
\includegraphics[width=7cm,height=4.6cm]{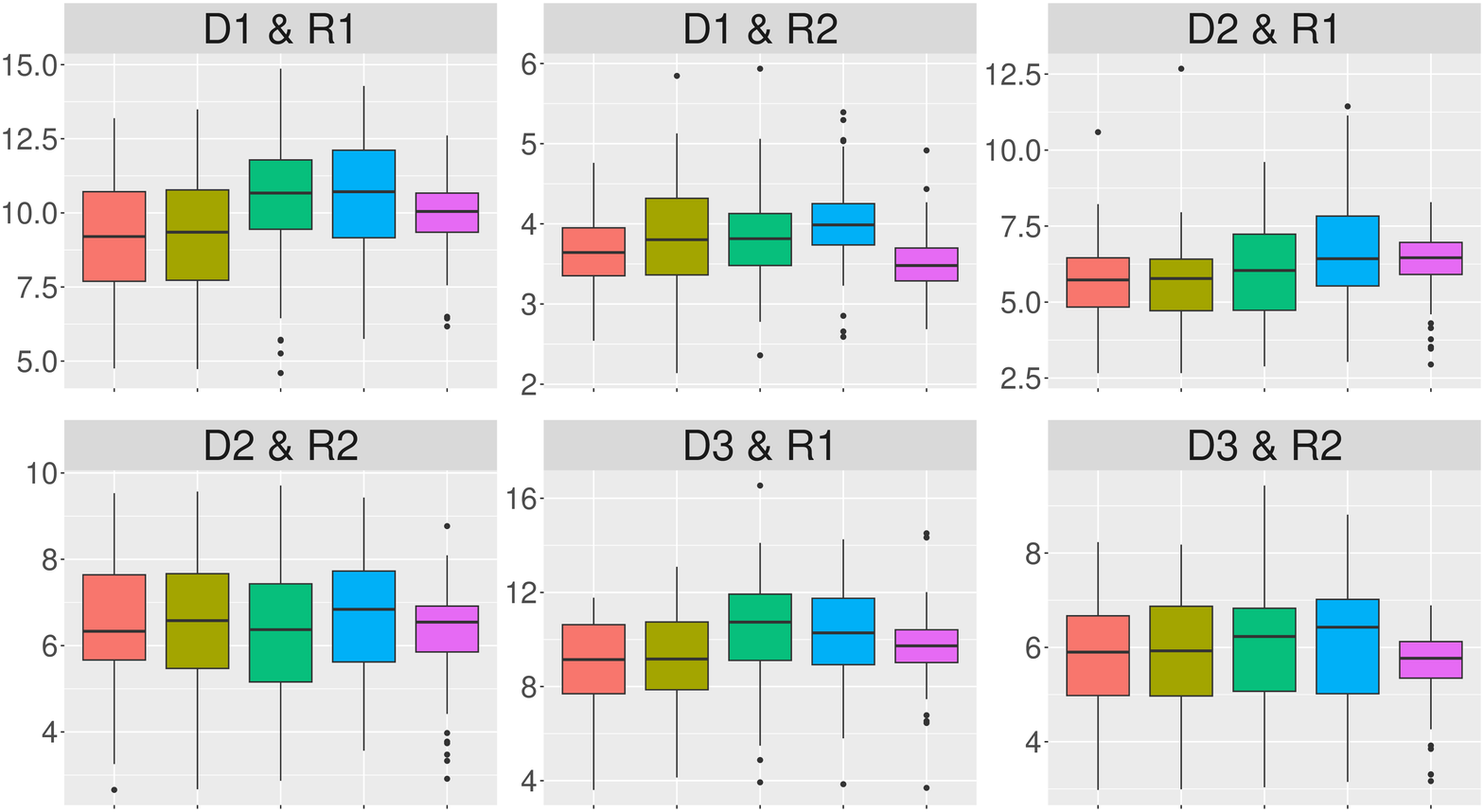} 
\vspace{-2.5mm}
\subcaption{$s=50, \rho_x=0.1, \rho_y=0.1$}
\end{minipage}
\begin{minipage}[b]{0.33\linewidth}
\centering
\includegraphics[width=7cm,height=4.6cm]{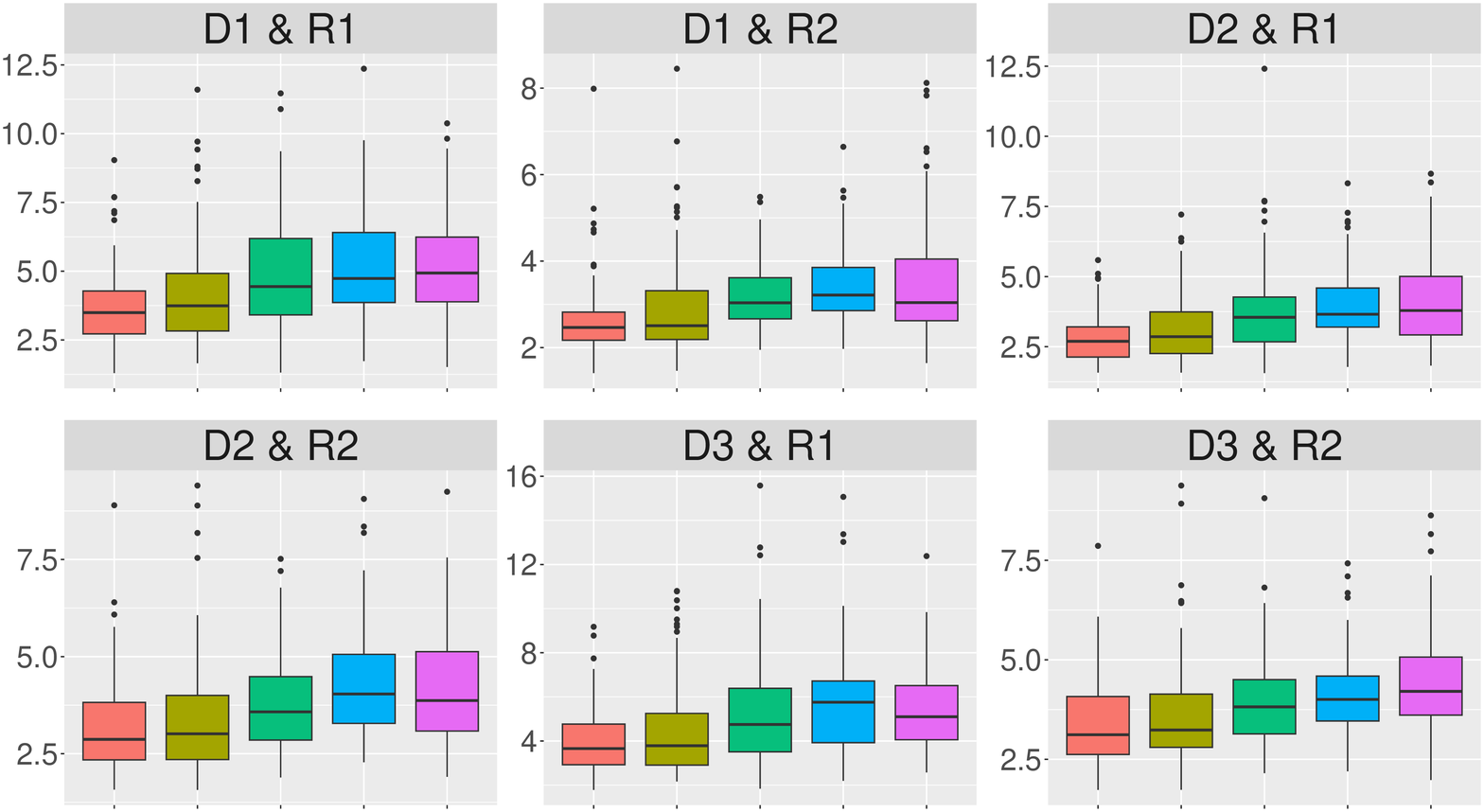} 
\vspace{-2.5mm}
\subcaption{$s=5, \rho_x=0.1, \rho_y=0.9$}
\end{minipage}
\begin{minipage}[b]{0.33\linewidth}
\centering
\includegraphics[width=7cm,height=4.6cm]{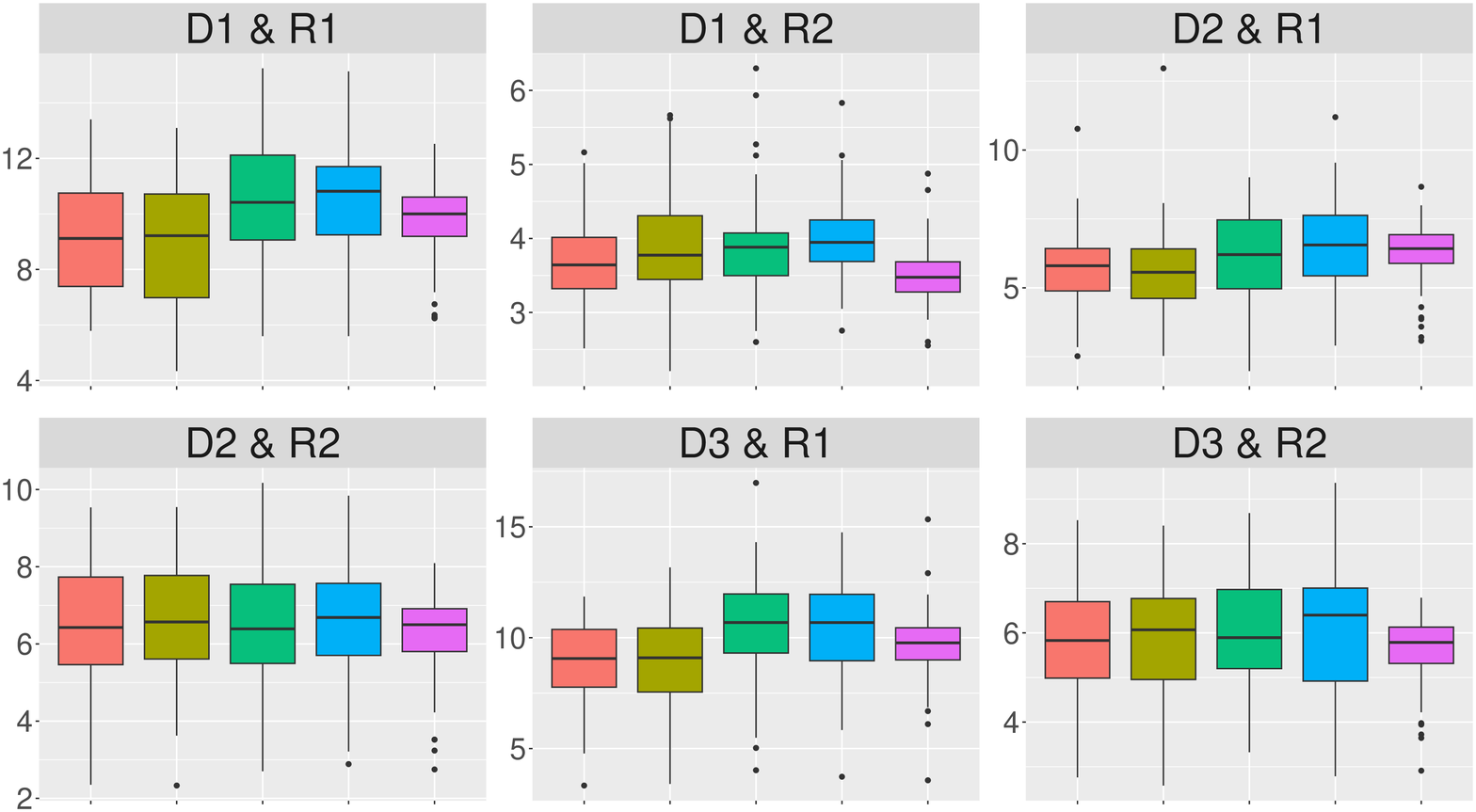}
\vspace{-2.5mm}
\subcaption{$s=50, \rho_x=0.1, \rho_y=0.9$}
\end{minipage}
\begin{minipage}[b]{0.33\linewidth}
\centering
\includegraphics[width=7cm,height=4.6cm]{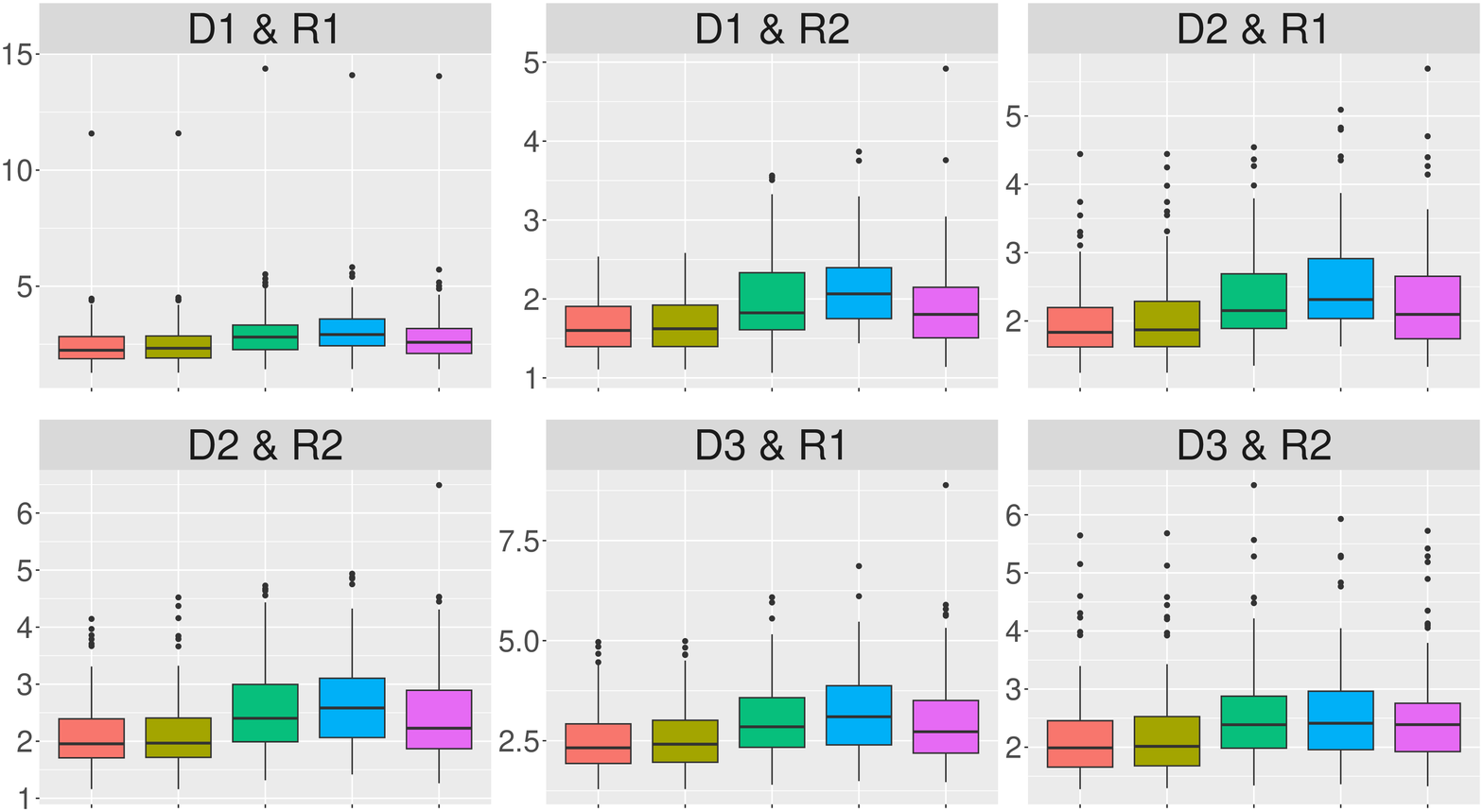}
\vspace{-2.5mm}
\subcaption{$s=5, \rho_x=0.9, \rho_y=0.1$}
\end{minipage}
\begin{minipage}[b]{0.33\linewidth}
\centering
\includegraphics[width=7cm,height=4.6cm]{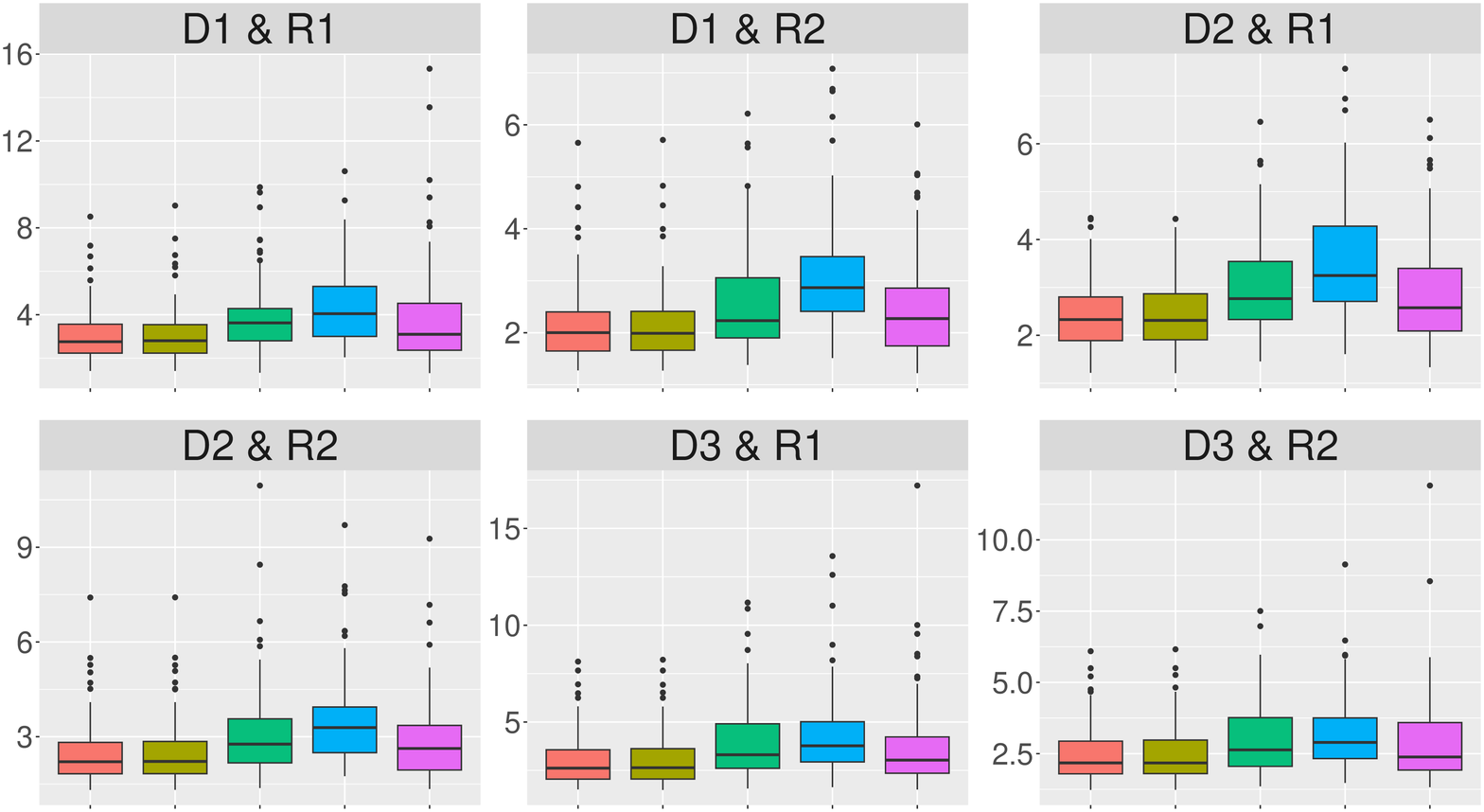}
\vspace{-2.5mm}
\subcaption{$s=50, \rho_x=0.9, \rho_y=0.1$}
\end{minipage}
\begin{minipage}[b]{0.33\linewidth}
\centering
\includegraphics[width=7cm,height=4.6cm]{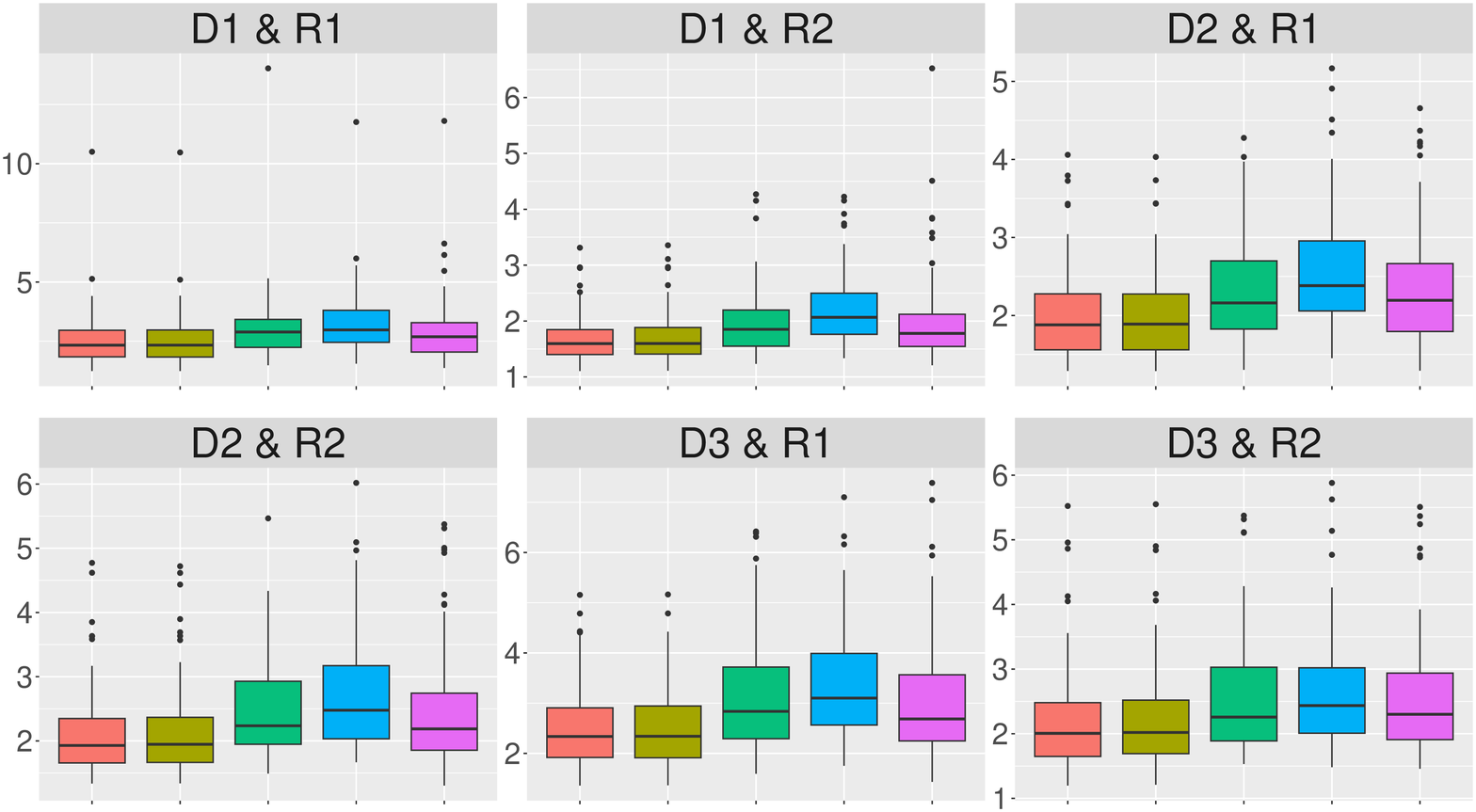}
\vspace{-2.5mm}
\subcaption{$s=5, \rho_x=0.9, \rho_y=0.9$}
\end{minipage}
\begin{minipage}[b]{0.33\linewidth}
\centering
\includegraphics[width=7cm,height=4.6cm]{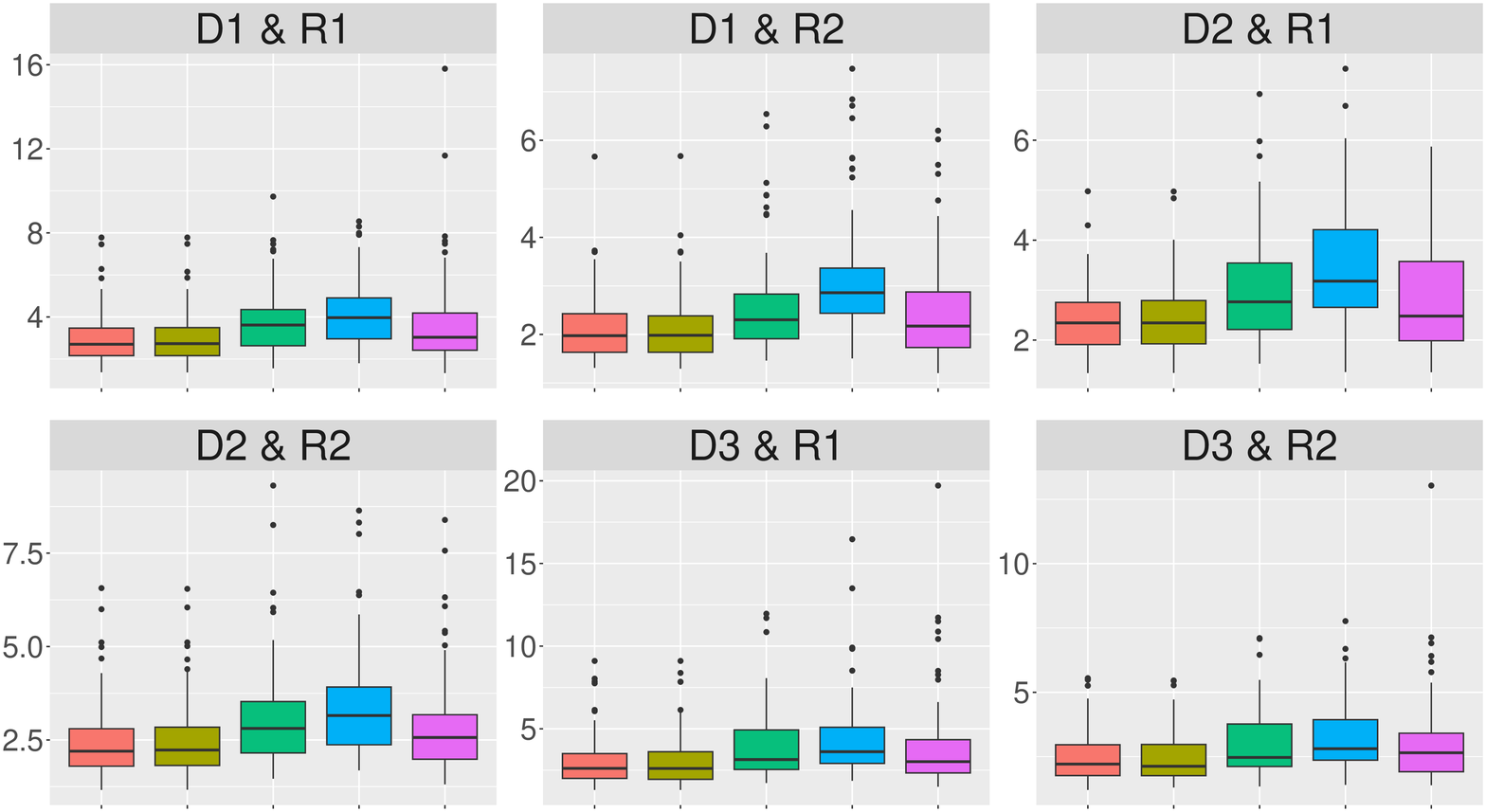}
\vspace{-2.5mm}
\subcaption{$s=50, \rho_x=0.9, \rho_y=0.9$}
\end{minipage}
\caption{Boxplots of MSE for $n=25$ when the case $M=3$.
The red boxplot indicates MR, dark yellow UR, green lasso, blue mglasso, and magenta mlasso. 
}
\label{fig:SimuM3n25}
\end{figure}
\end{landscape}

\begin{landscape}
\begin{figure}[htbp]
\begin{minipage}[b]{0.33\linewidth}
\centering
\includegraphics[width=7cm,height=4.6cm]{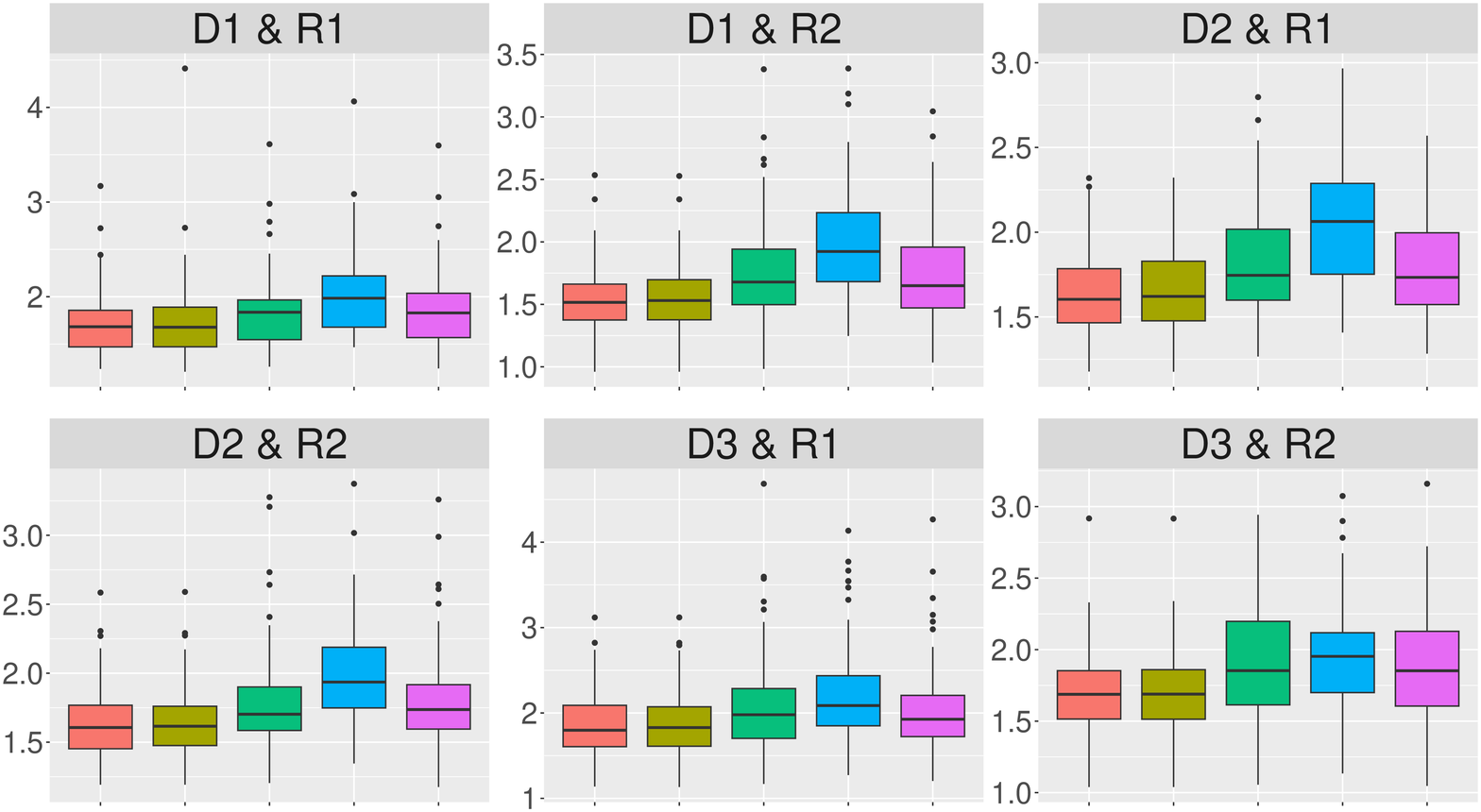}
\vspace{-2.5mm}
\subcaption{$s=5, \rho_x=0.1, \rho_y=0.1$}
\end{minipage}
\begin{minipage}[b]{0.33\linewidth}
\centering
\includegraphics[width=7cm,height=4.6cm]{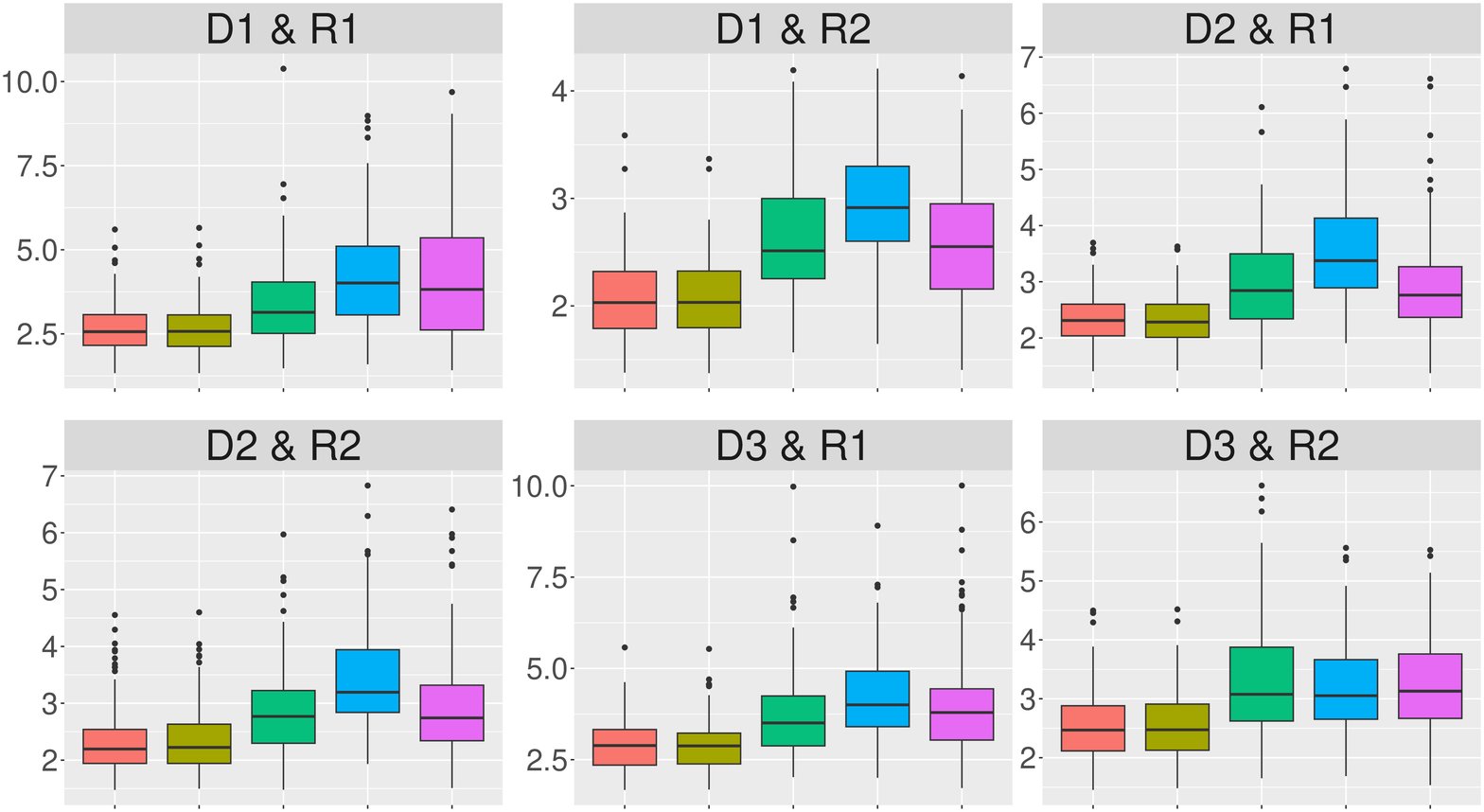} 
\vspace{-2.5mm}
\subcaption{$s=50, \rho_x=0.1, \rho_y=0.1$}
\end{minipage}
\begin{minipage}[b]{0.33\linewidth}
\centering
\includegraphics[width=7cm,height=4.6cm]{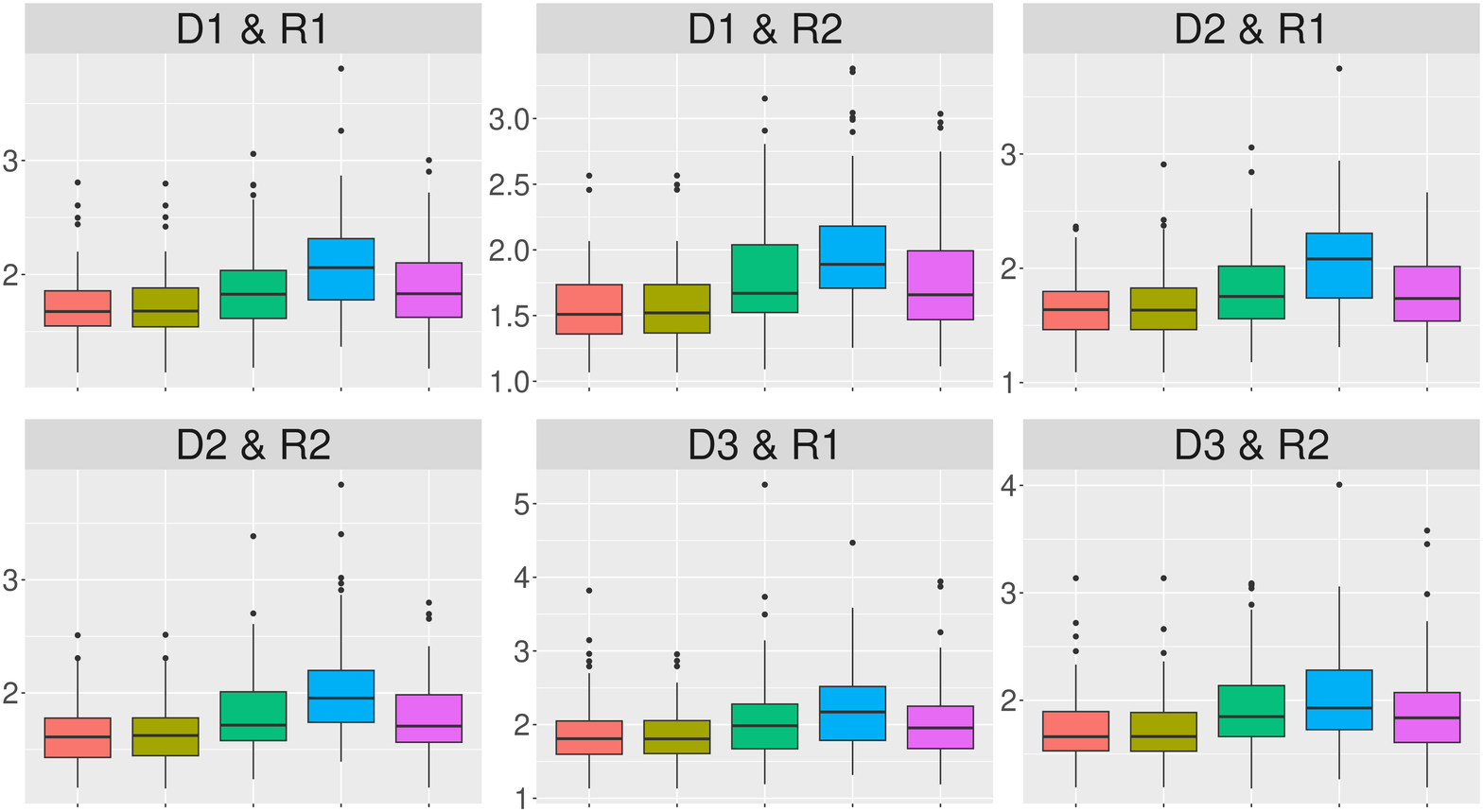} 
\vspace{-2.5mm}
\subcaption{$s=5, \rho_x=0.1, \rho_y=0.9$}
\end{minipage}
\begin{minipage}[b]{0.33\linewidth}
\centering
\includegraphics[width=7cm,height=4.6cm]{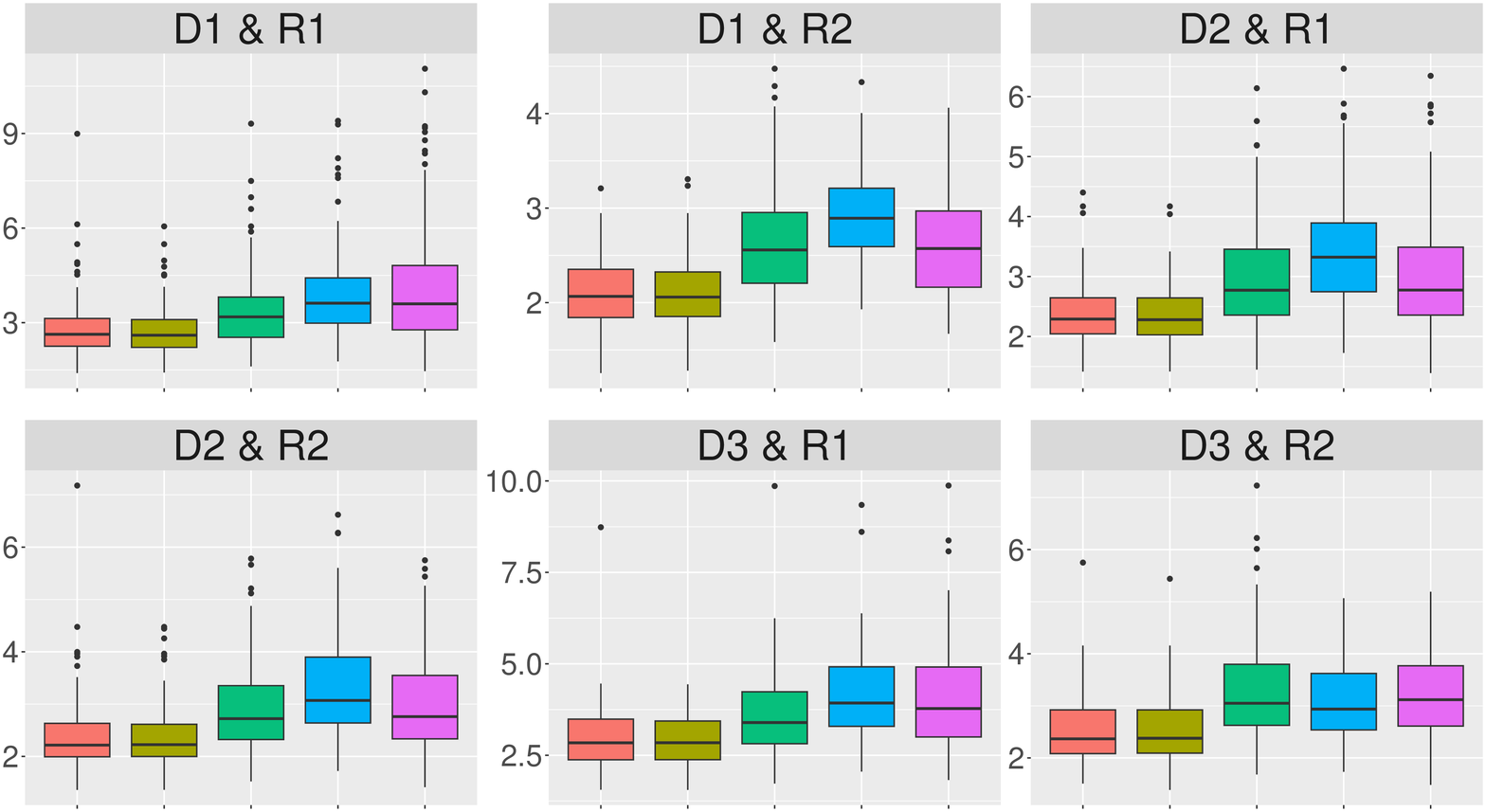}
\vspace{-2.5mm}
\subcaption{$s=50, \rho_x=0.1, \rho_y=0.9$}
\end{minipage}
\begin{minipage}[b]{0.33\linewidth}
\centering
\includegraphics[width=7cm,height=4.6cm]{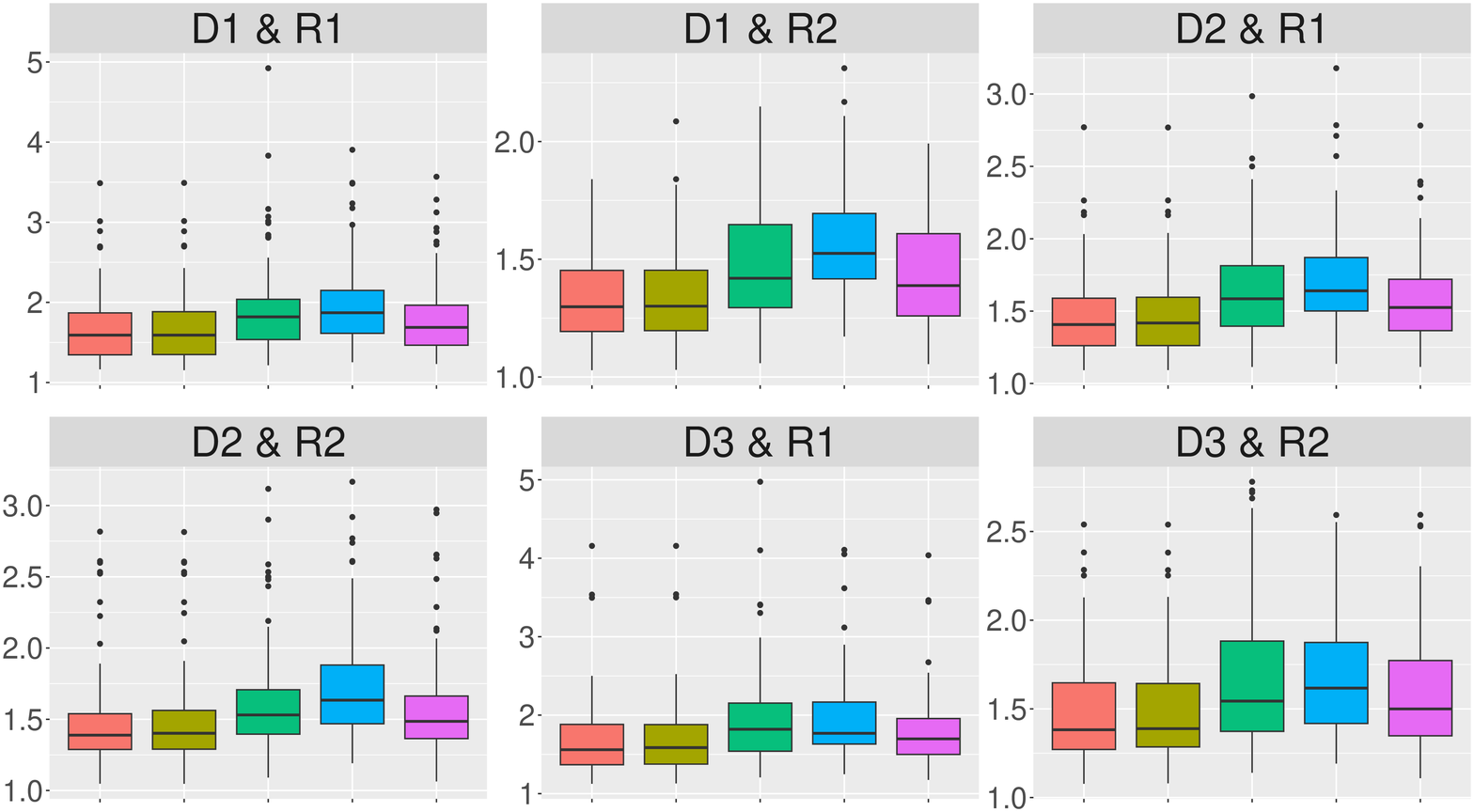}
\vspace{-2.5mm}
\subcaption{$s=5, \rho_x=0.9, \rho_y=0.1$}
\end{minipage}
\begin{minipage}[b]{0.33\linewidth}
\centering
\includegraphics[width=7cm,height=4.6cm]{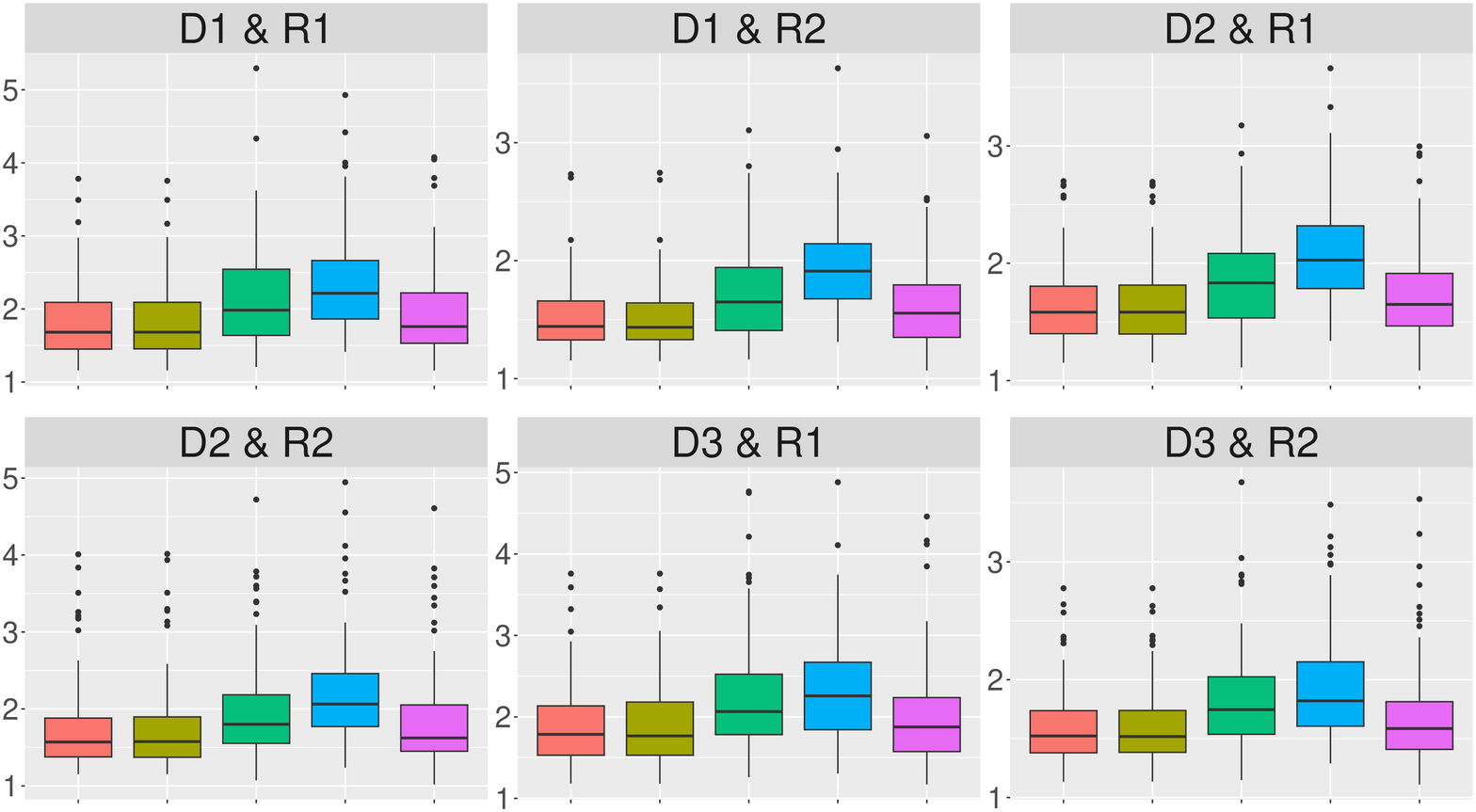}
\vspace{-2.5mm}
\subcaption{$s=50, \rho_x=0.9, \rho_y=0.1$}
\end{minipage}
\begin{minipage}[b]{0.33\linewidth}
\centering
\includegraphics[width=7cm,height=4.6cm]{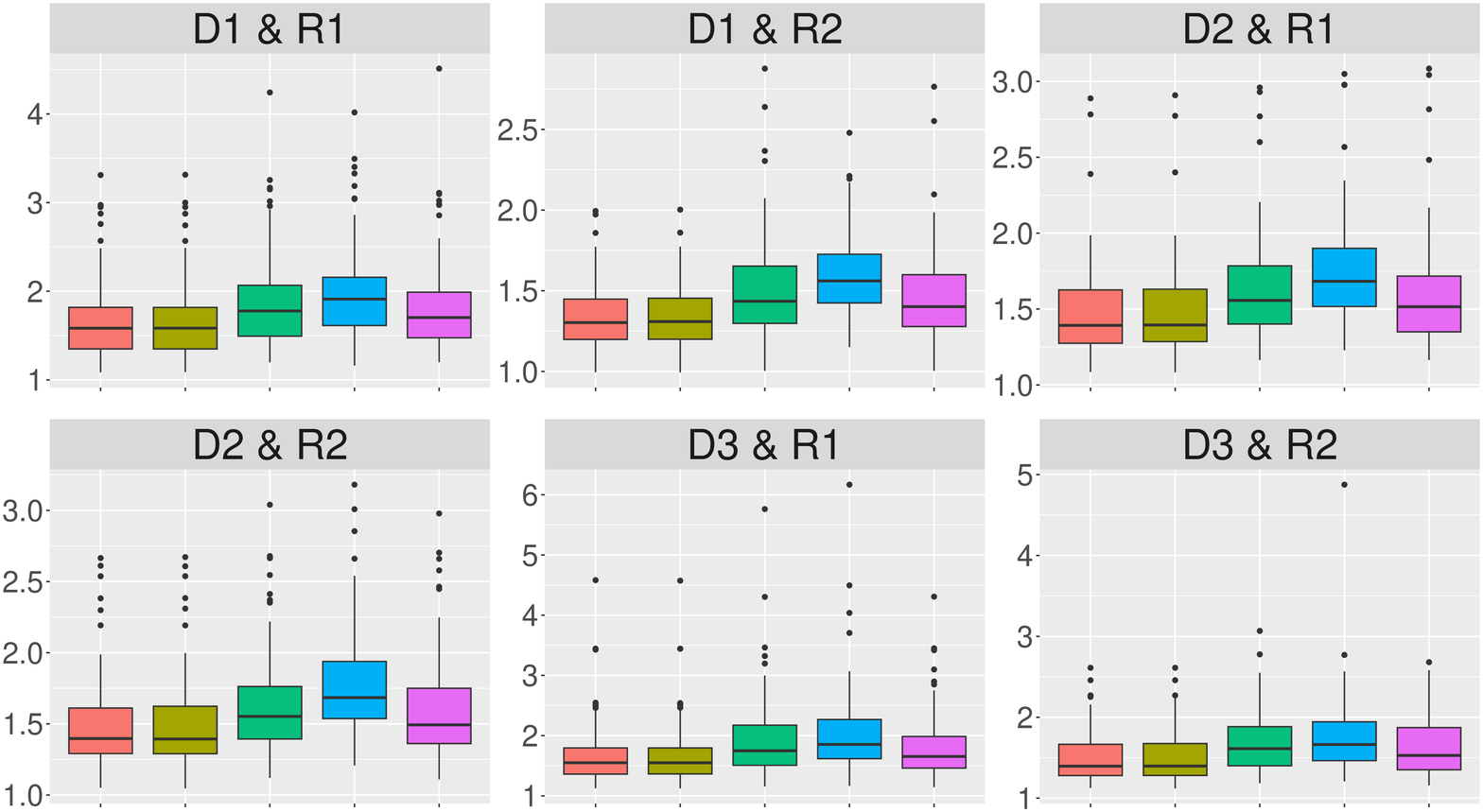}
\vspace{-2.5mm}
\subcaption{$s=5, \rho_x=0.9, \rho_y=0.9$}
\end{minipage}
\begin{minipage}[b]{0.33\linewidth}
\centering
\includegraphics[width=7cm,height=4.6cm]{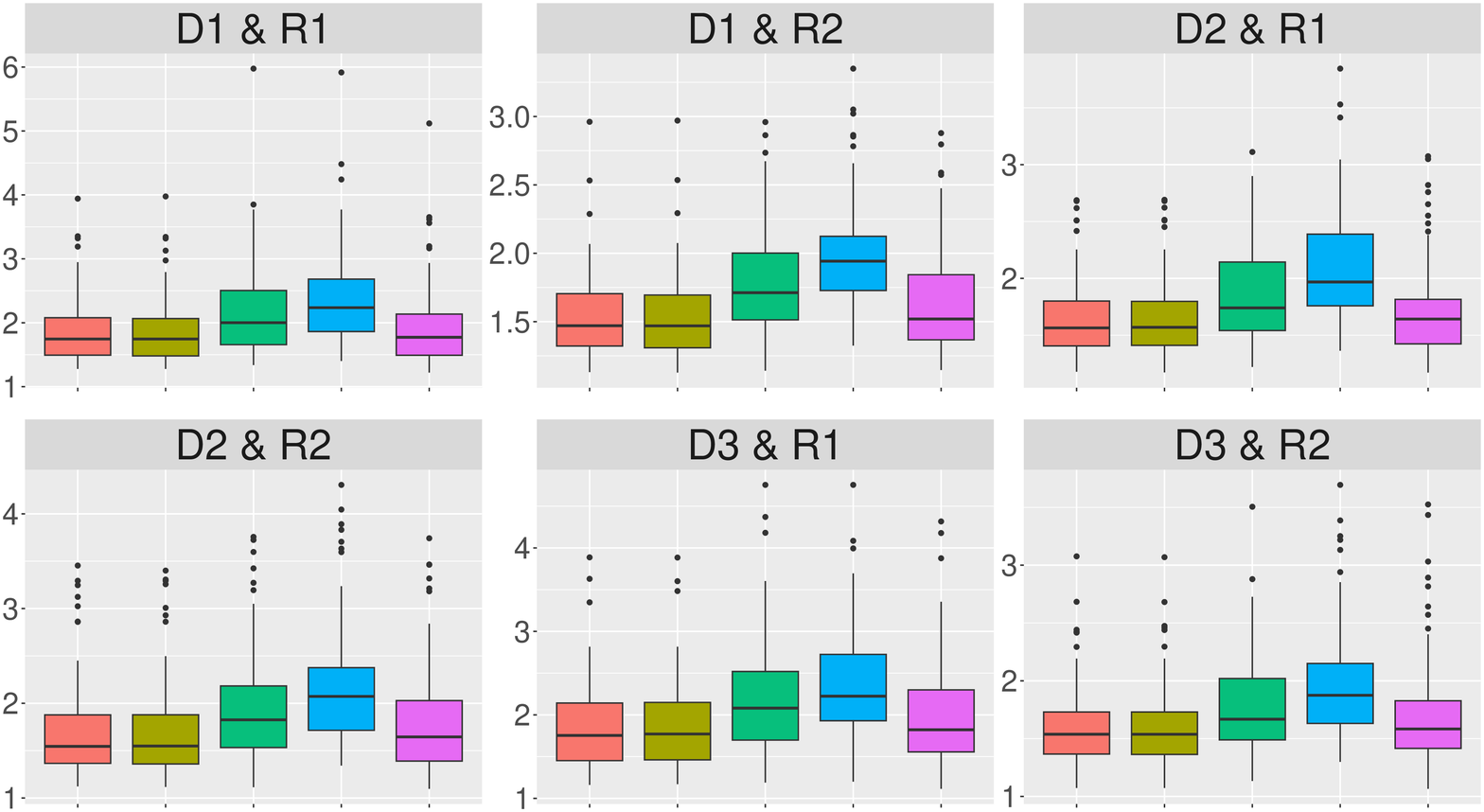}
\vspace{-2.5mm}
\subcaption{$s=50, \rho_x=0.9, \rho_y=0.9$}
\end{minipage}
\caption{Boxplots of MSE for $n=50$ when the case $M=3$.
The red boxplot indicates MR, dark yellow UR, green lasso, blue mglasso, and magenta mlasso. 
}
\label{fig:SimuM3n50}
\end{figure}
\end{landscape}

\begin{landscape}
\begin{figure}[htbp]
\begin{minipage}[b]{0.33\linewidth}
\centering
\includegraphics[width=7cm,height=4.6cm]{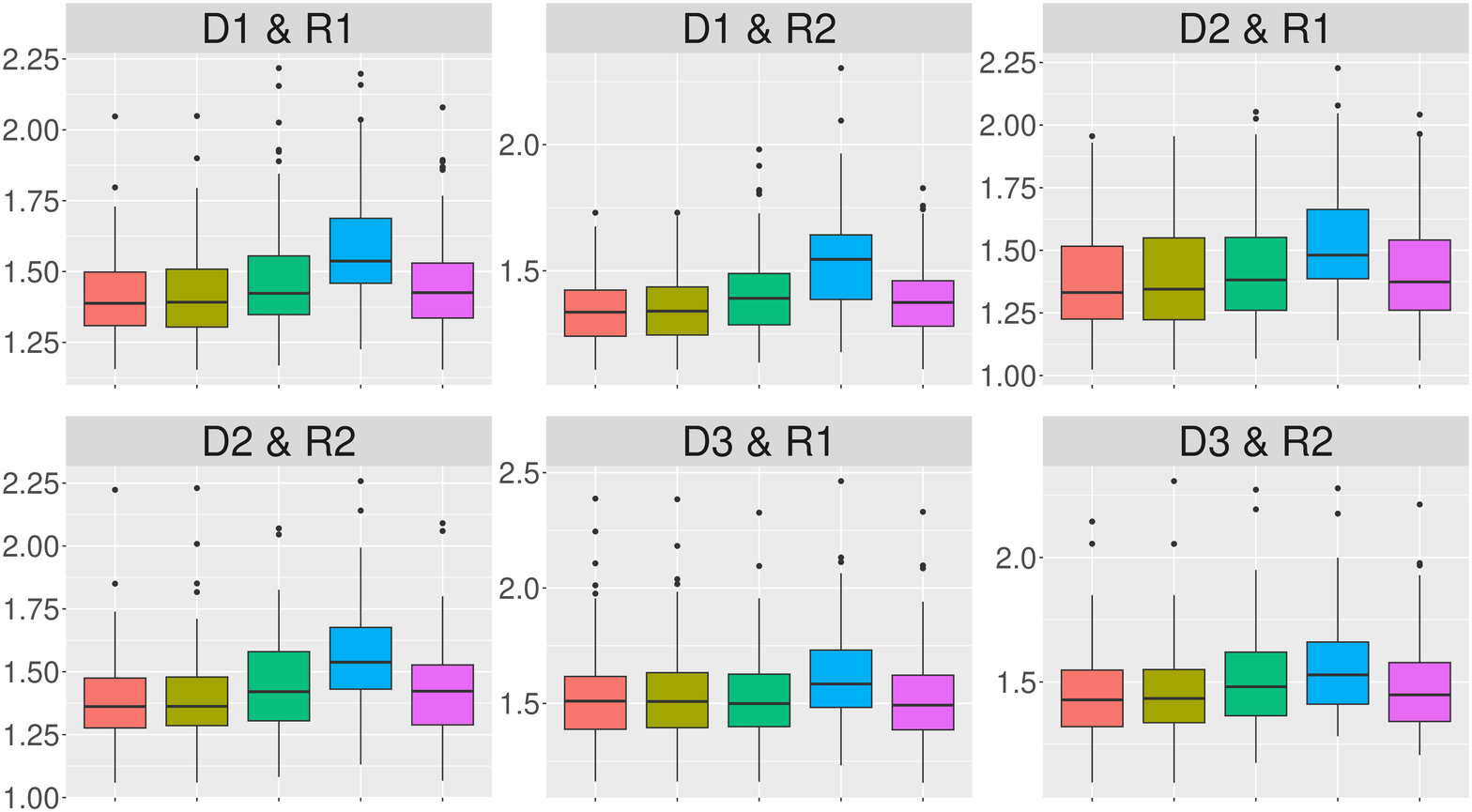}
\vspace{-2.5mm}
\subcaption{$s=5, \rho_x=0.1, \rho_y=0.1$}
\end{minipage}
\begin{minipage}[b]{0.33\linewidth}
\centering
\includegraphics[width=7cm,height=4.6cm]{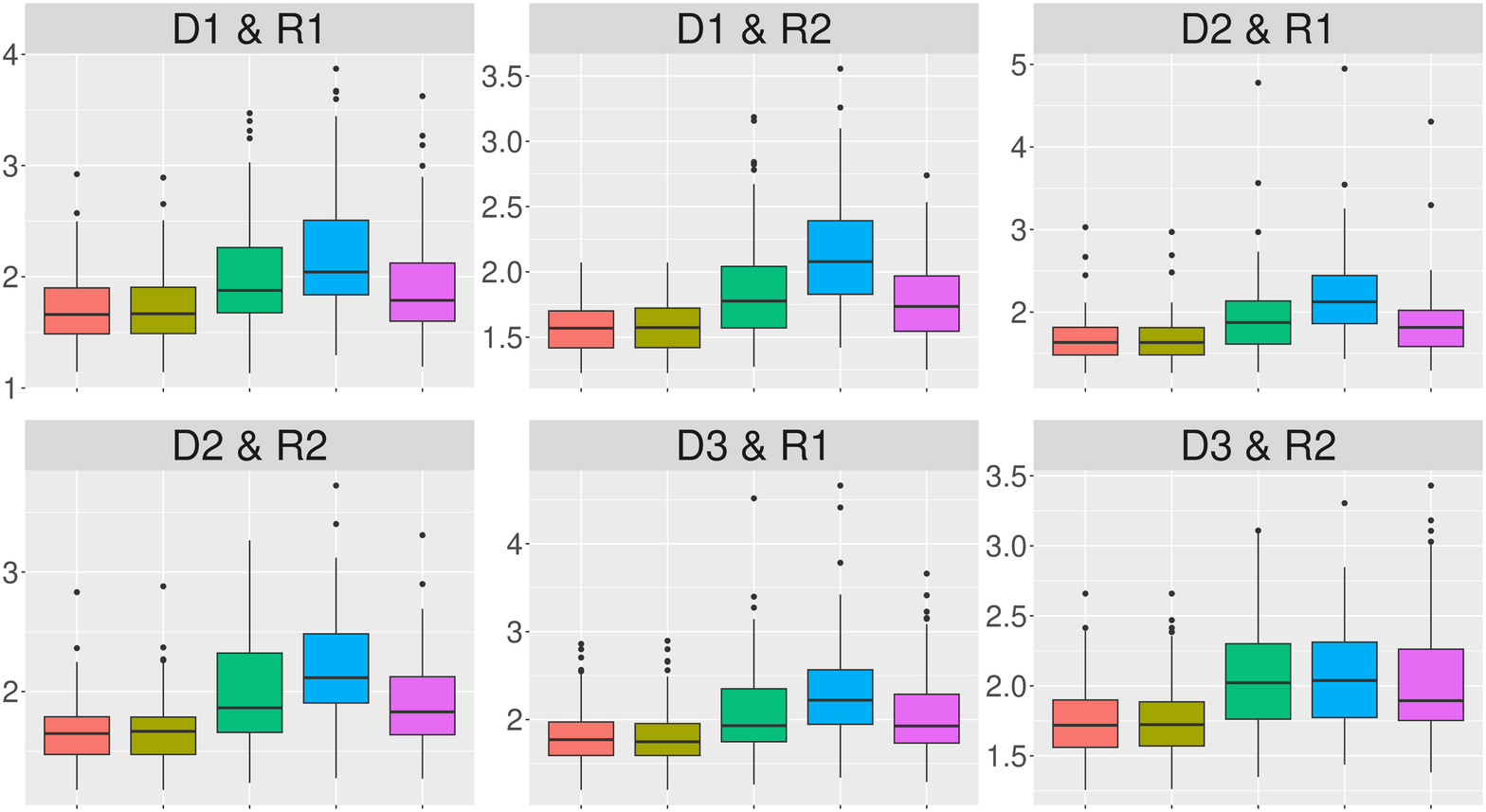} 
\vspace{-2.5mm}
\subcaption{$s=50, \rho_x=0.1, \rho_y=0.1$}
\end{minipage}
\begin{minipage}[b]{0.33\linewidth}
\centering
\includegraphics[width=7cm,height=4.6cm]{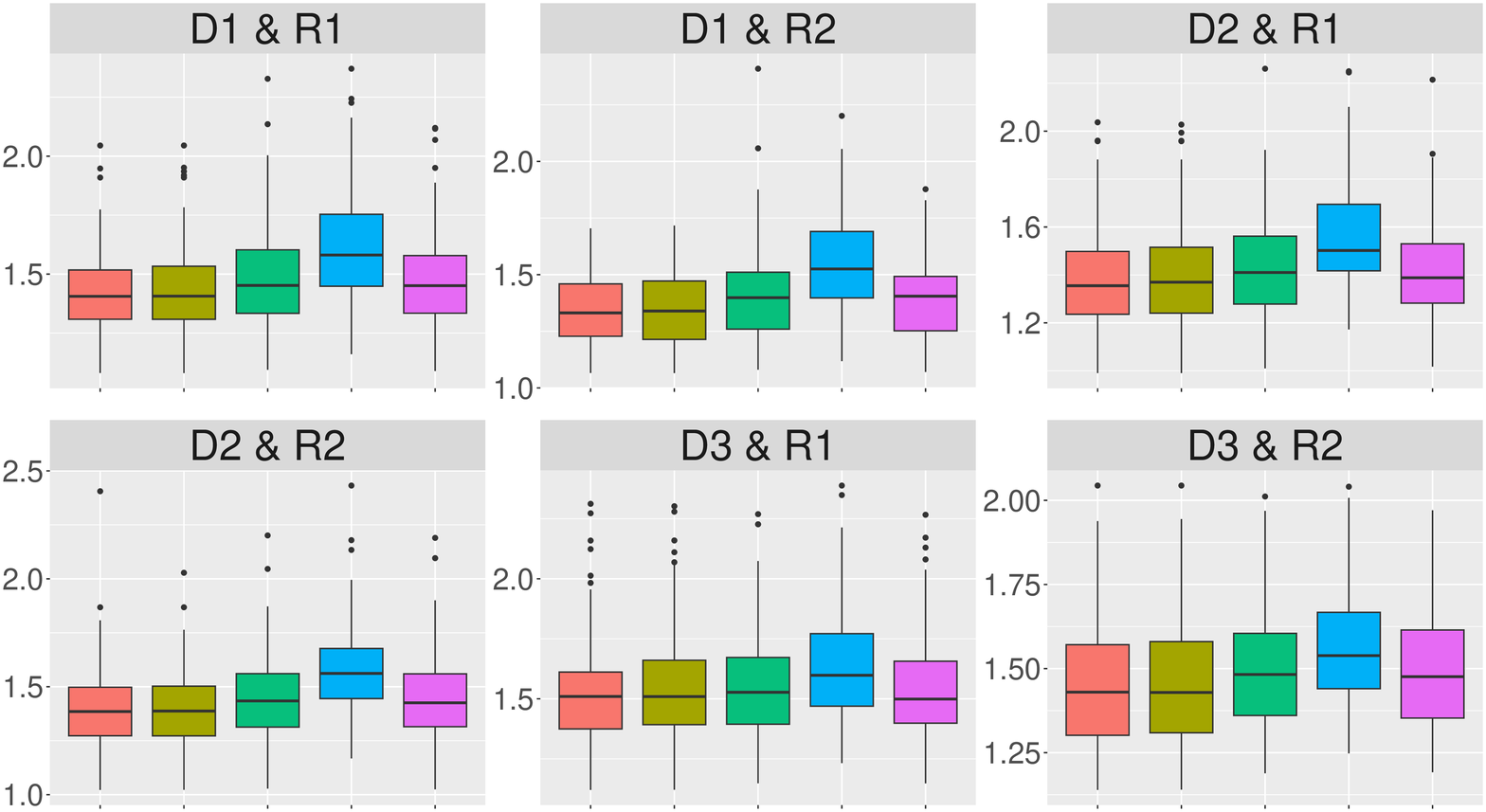} 
\vspace{-2.5mm}
\subcaption{$s=5, \rho_x=0.1, \rho_y=0.9$}
\end{minipage}
\begin{minipage}[b]{0.33\linewidth}
\centering
\includegraphics[width=7cm,height=4.6cm]{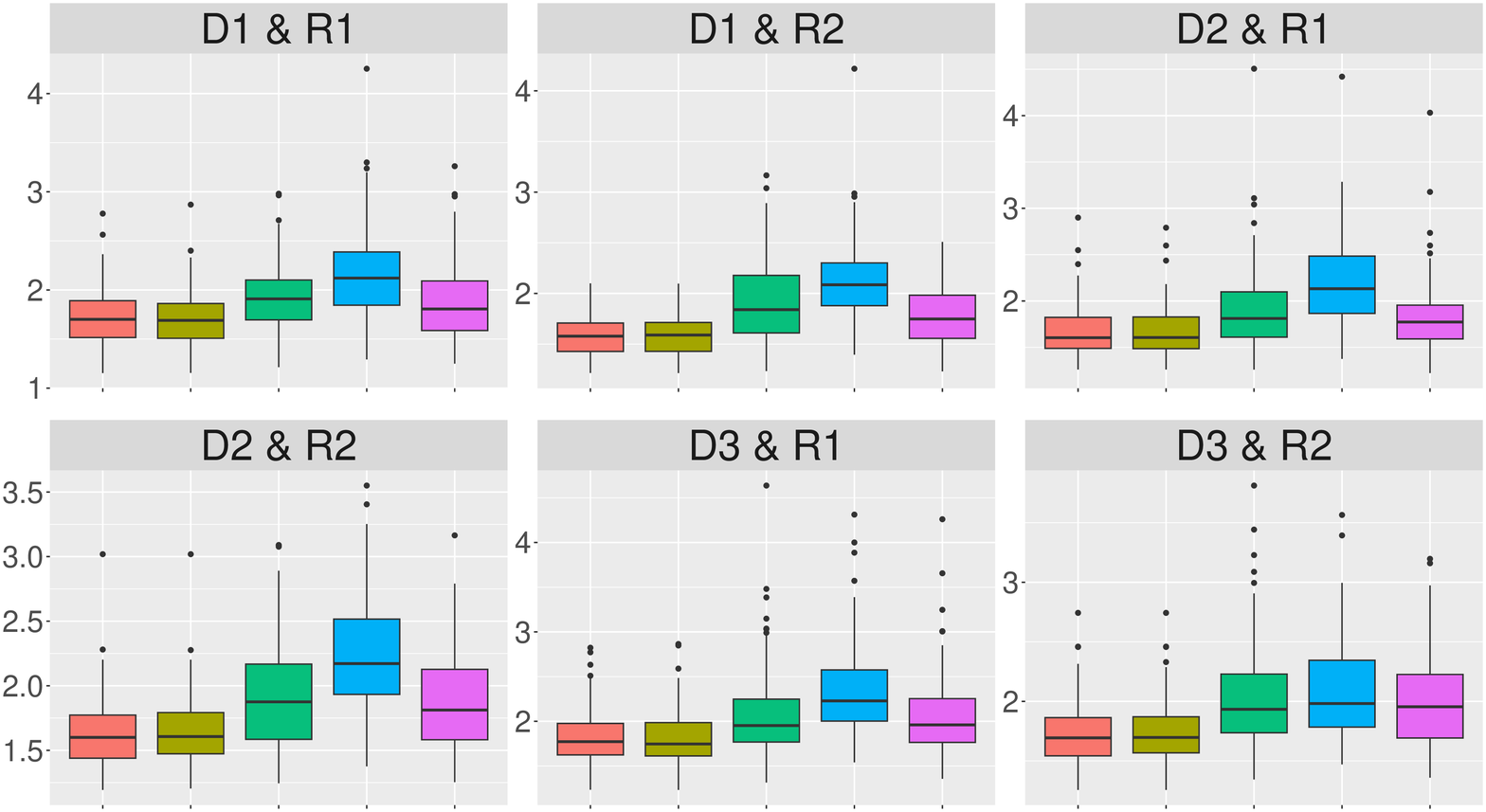}
\vspace{-2.5mm}
\subcaption{$s=50, \rho_x=0.1, \rho_y=0.9$}
\end{minipage}
\begin{minipage}[b]{0.33\linewidth}
\centering
\includegraphics[width=7cm,height=4.6cm]{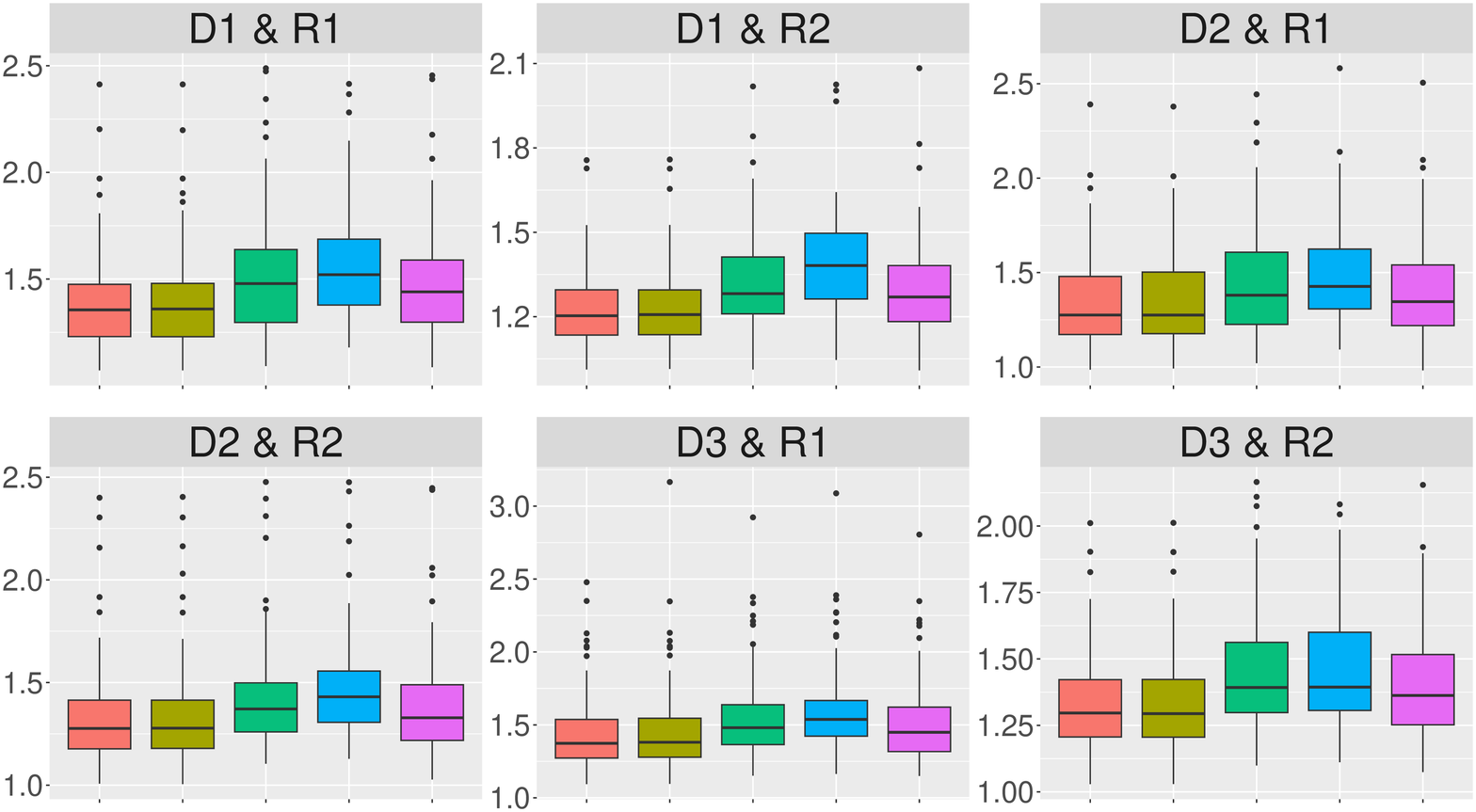}
\vspace{-2.5mm}
\subcaption{$s=5, \rho_x=0.9, \rho_y=0.1$}
\end{minipage}
\begin{minipage}[b]{0.33\linewidth}
\centering
\includegraphics[width=7cm,height=4.6cm]{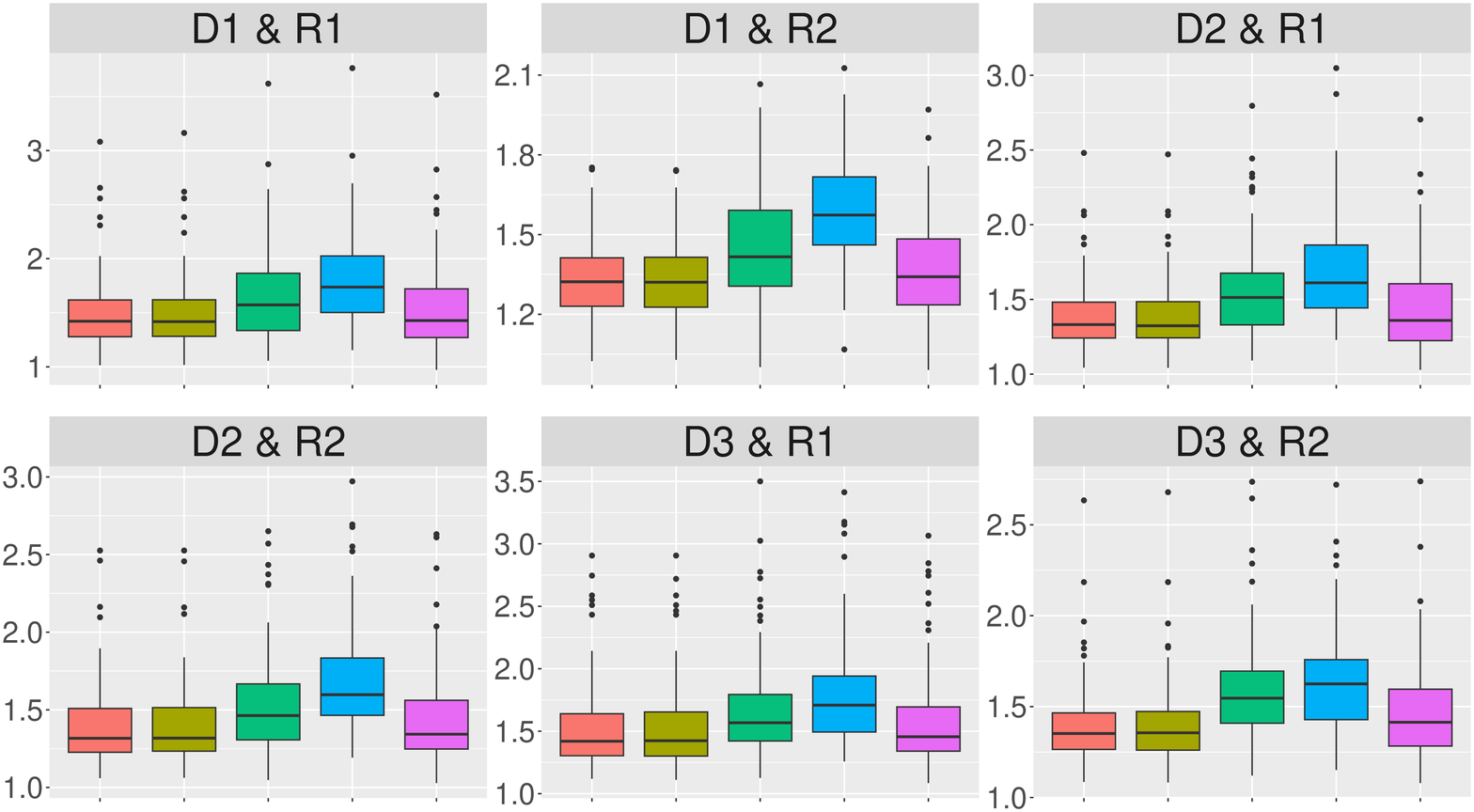}
\vspace{-2.5mm}
\subcaption{$s=50, \rho_x=0.9, \rho_y=0.1$}
\end{minipage}
\begin{minipage}[b]{0.33\linewidth}
\centering
\includegraphics[width=7cm,height=4.6cm]{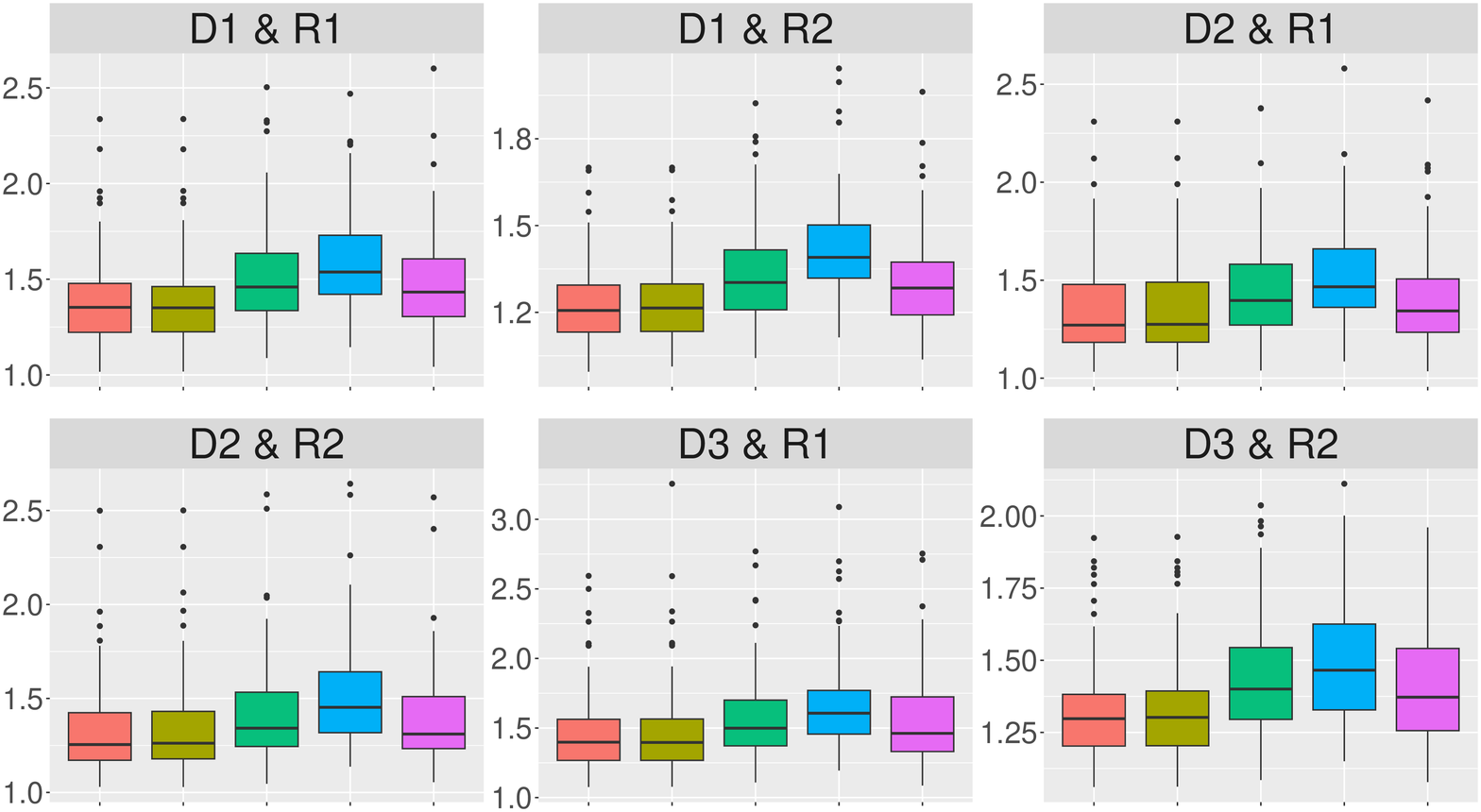}
\vspace{-2.5mm}
\subcaption{$s=5, \rho_x=0.9, \rho_y=0.9$}
\end{minipage}
\begin{minipage}[b]{0.33\linewidth}
\centering
\includegraphics[width=7cm,height=4.6cm]{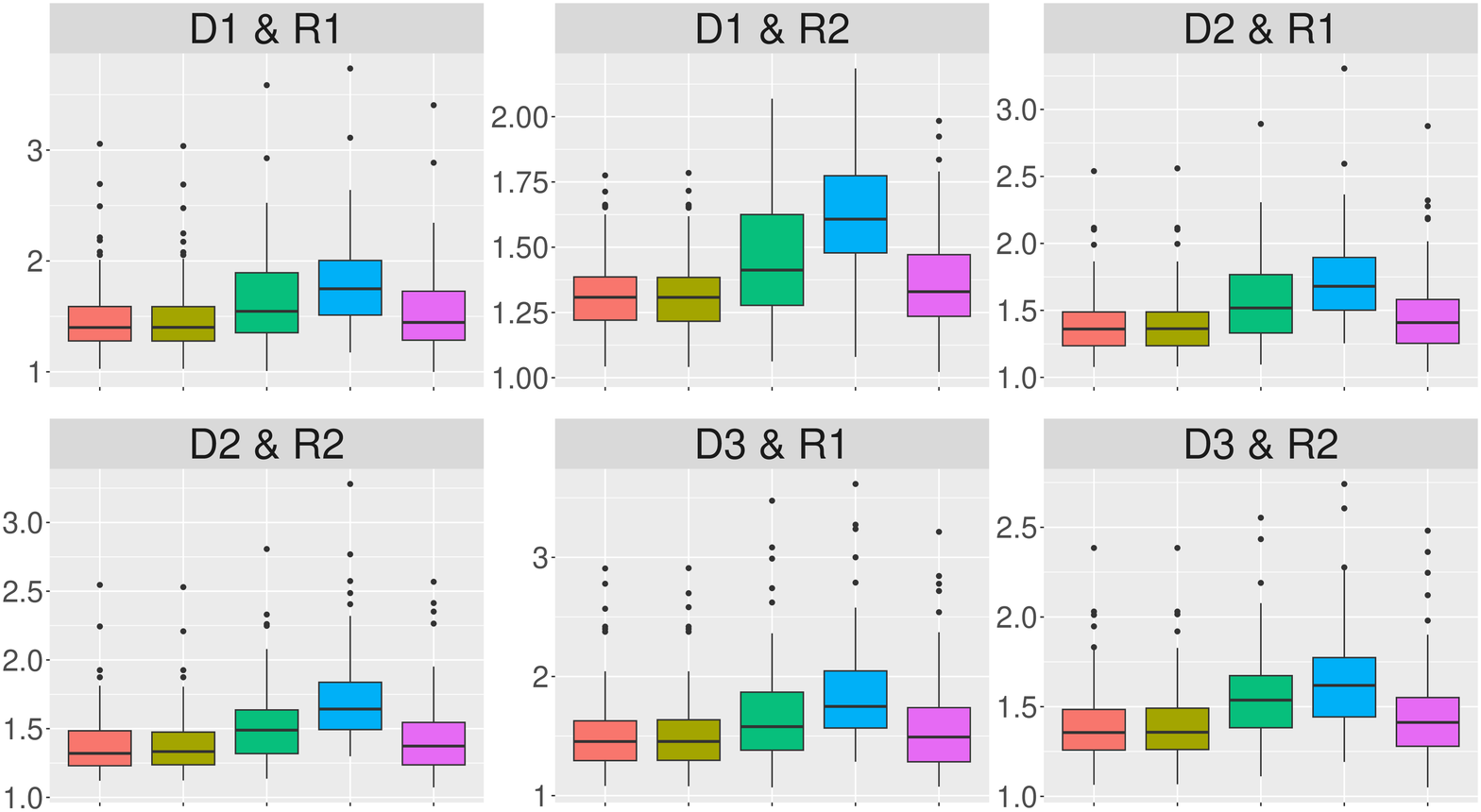}
\vspace{-2.5mm}
\subcaption{$s=50, \rho_x=0.9, \rho_y=0.9$}
\end{minipage}
\caption{Boxplots of MSE for $n=75$ when the case $M=3$.
The red boxplot indicates MR, dark yellow UR, green lasso, blue mglasso, and magenta mlasso. 
}
\label{fig:SimuM3n75}
\end{figure}
\end{landscape}

\begin{landscape}
\begin{figure}[htbp]
\begin{minipage}[b]{0.33\linewidth}
\centering
\includegraphics[width=7cm,height=4.6cm]{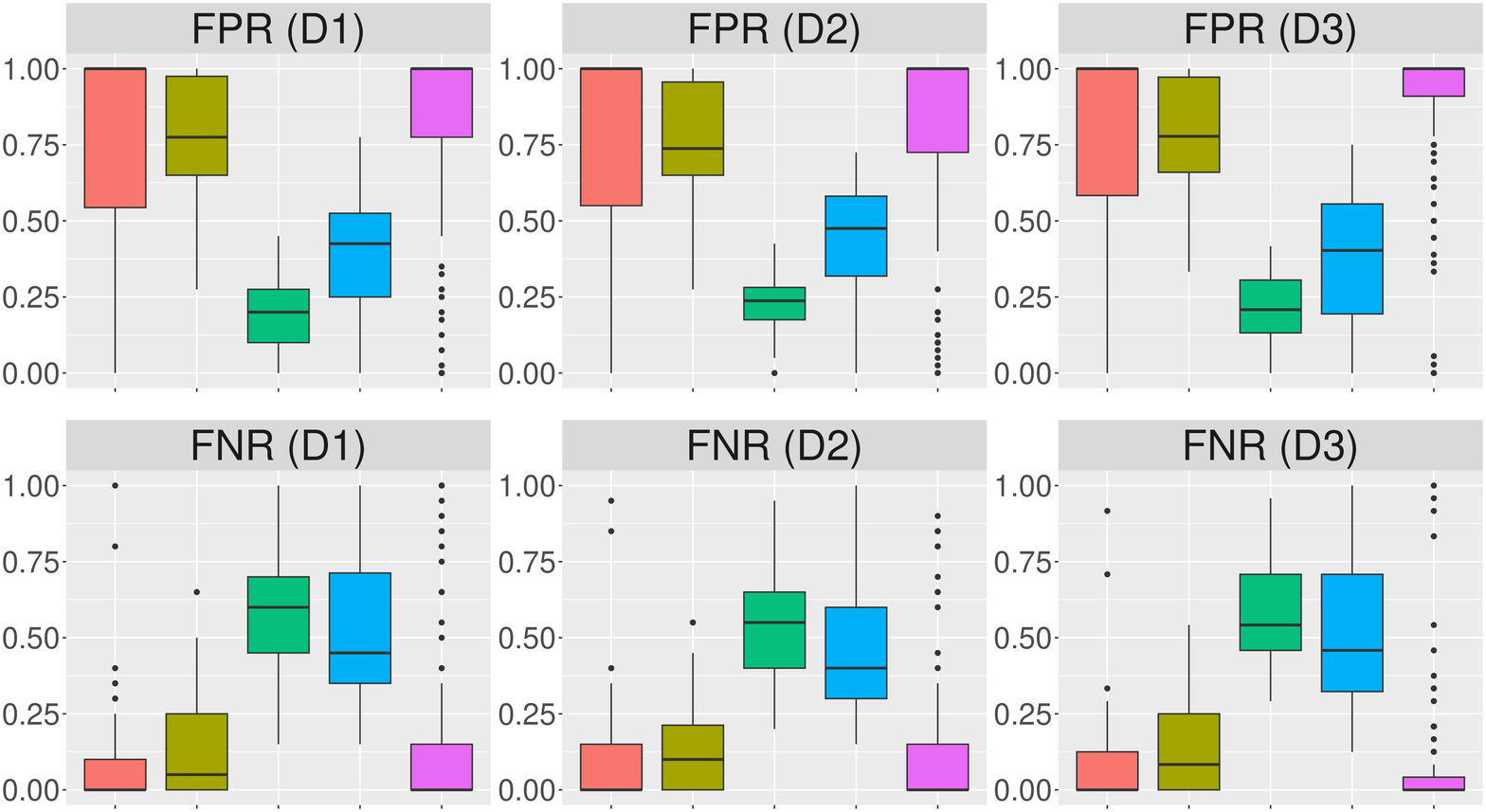}
\vspace{-2.5mm}
\subcaption{$s=5, \rho_x=0.1, \rho_y=0.1$}
\end{minipage}
\begin{minipage}[b]{0.33\linewidth}
\centering
\includegraphics[width=7cm,height=4.6cm]{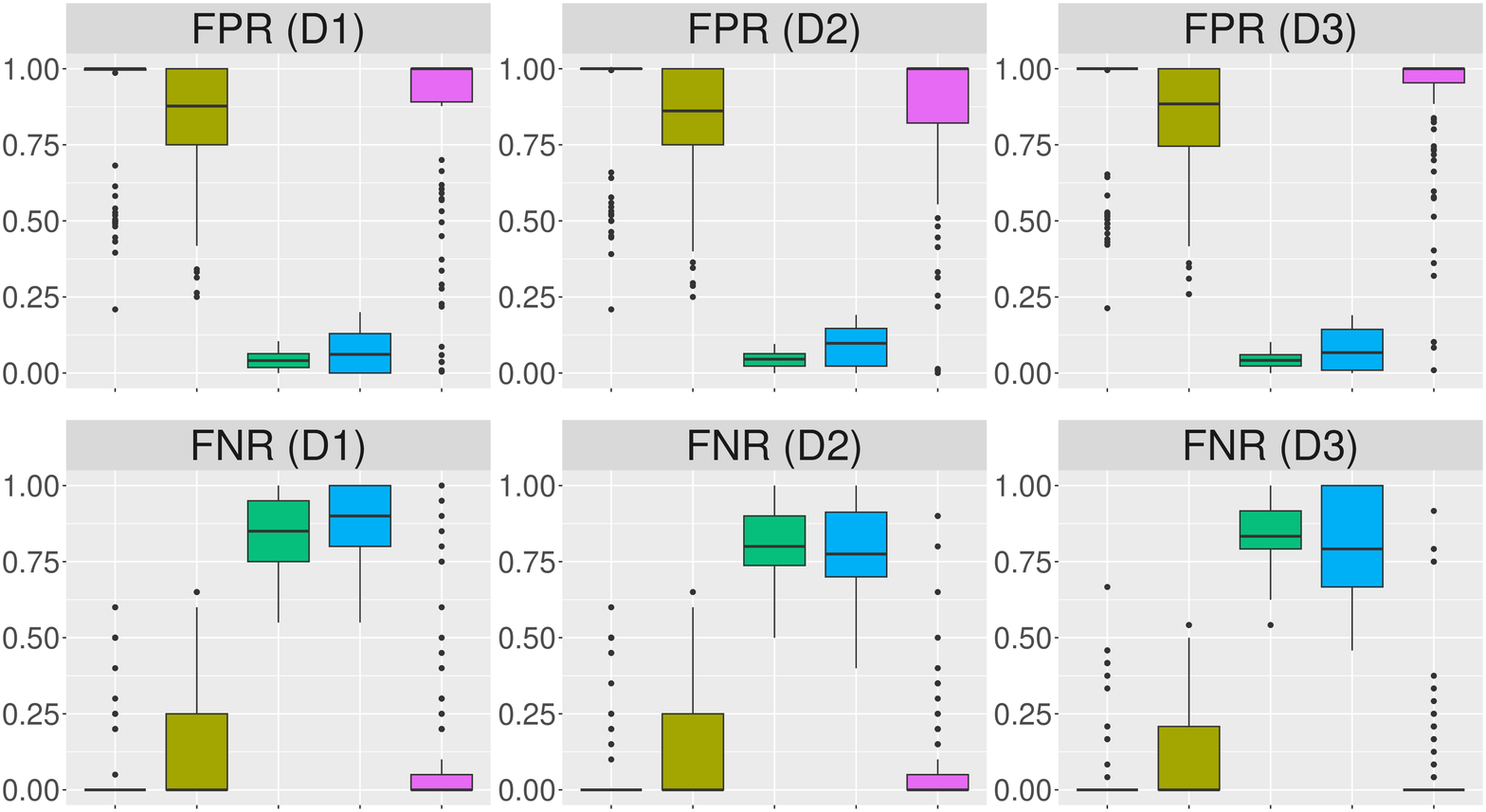} 
\vspace{-2.5mm}
\subcaption{$s=50, \rho_x=0.1, \rho_y=0.1$}
\end{minipage}
\begin{minipage}[b]{0.33\linewidth}
\centering
\includegraphics[width=7cm,height=4.6cm]{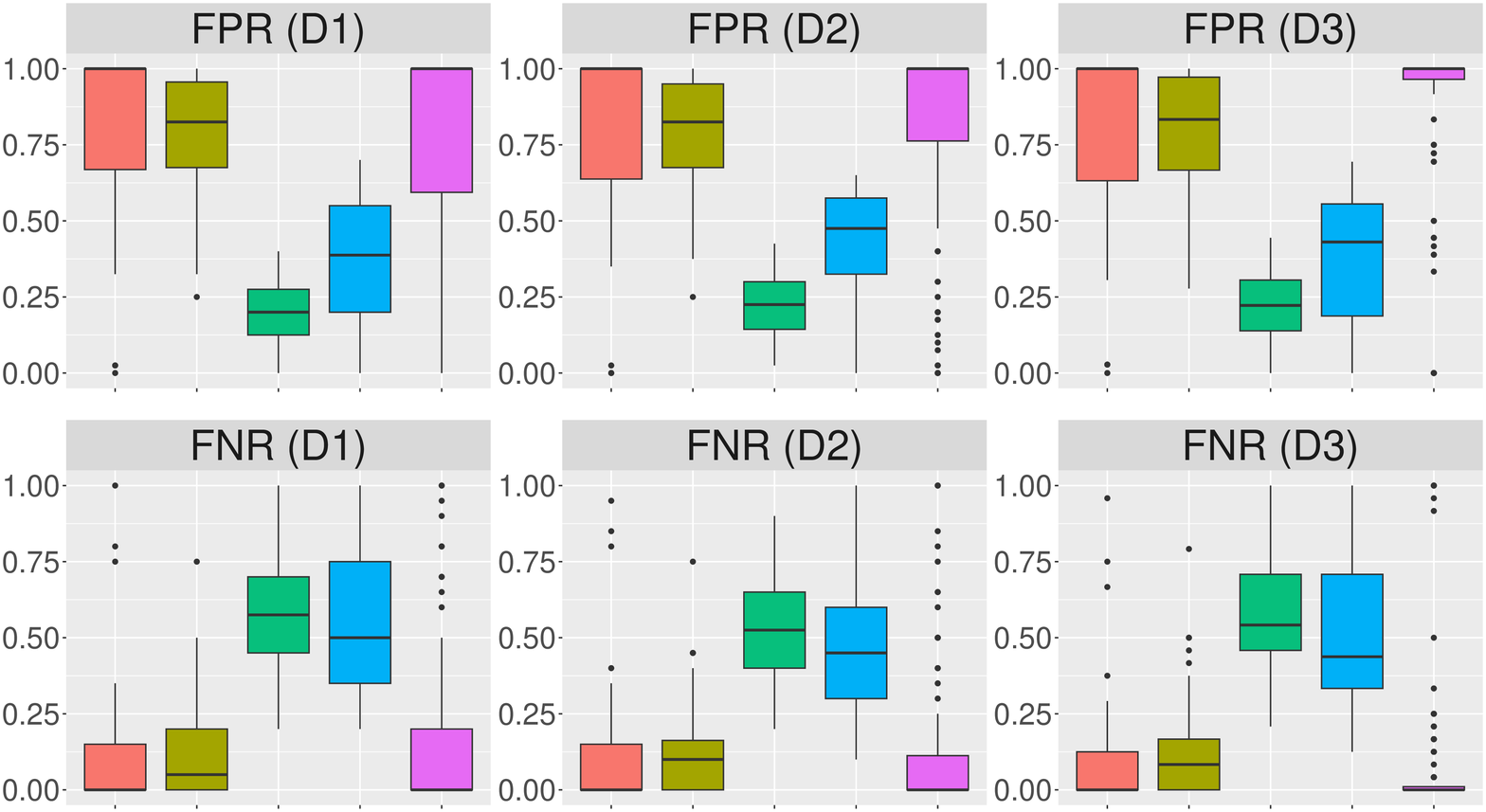} 
\vspace{-2.5mm}
\subcaption{$s=5, \rho_x=0.1, \rho_y=0.9$}
\end{minipage}
\begin{minipage}[b]{0.33\linewidth}
\centering
\includegraphics[width=7cm,height=4.6cm]{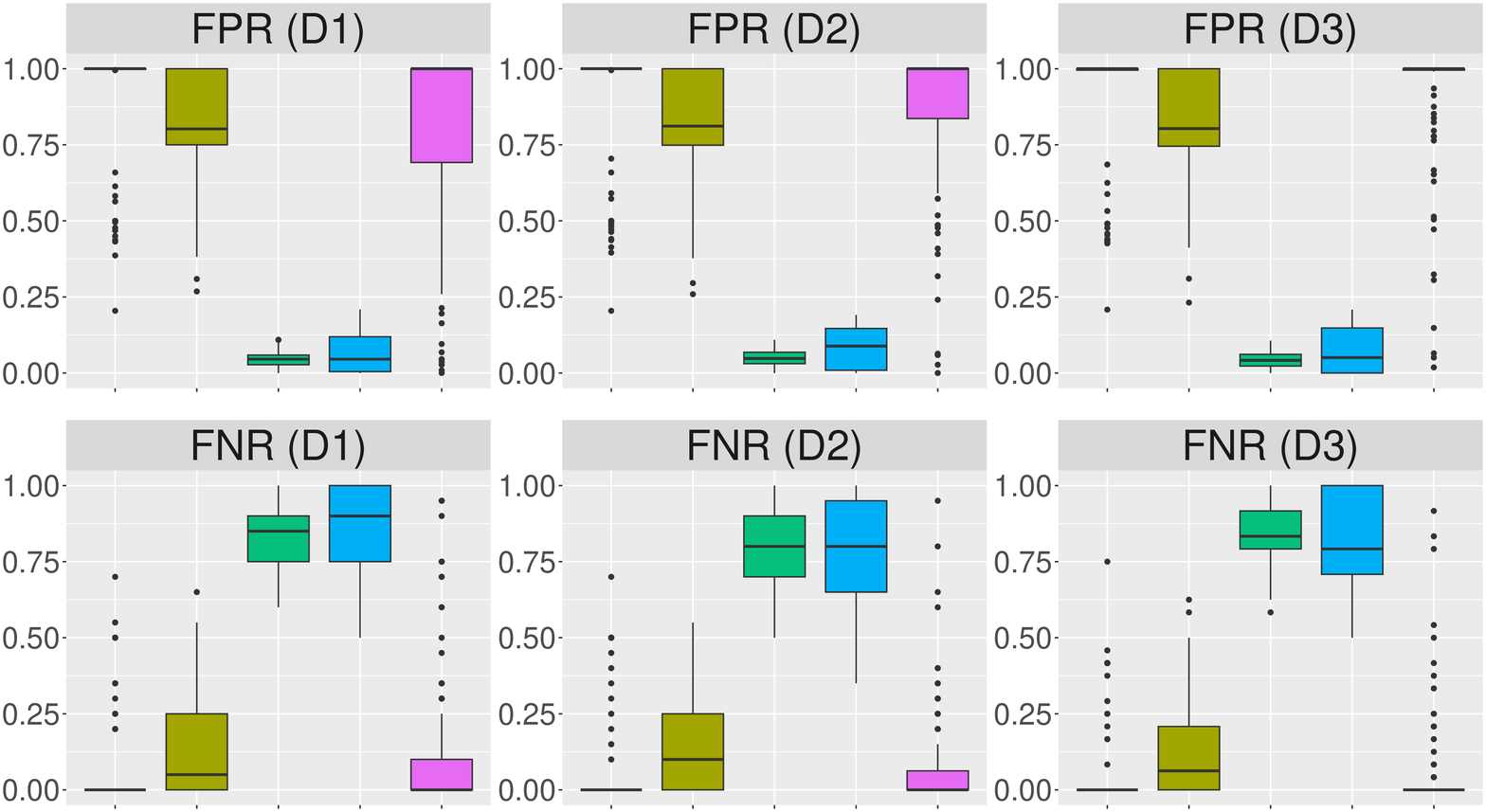}
\vspace{-2.5mm}
\subcaption{$s=50, \rho_x=0.1, \rho_y=0.9$}
\end{minipage}
\begin{minipage}[b]{0.33\linewidth}
\centering
\includegraphics[width=7cm,height=4.6cm]{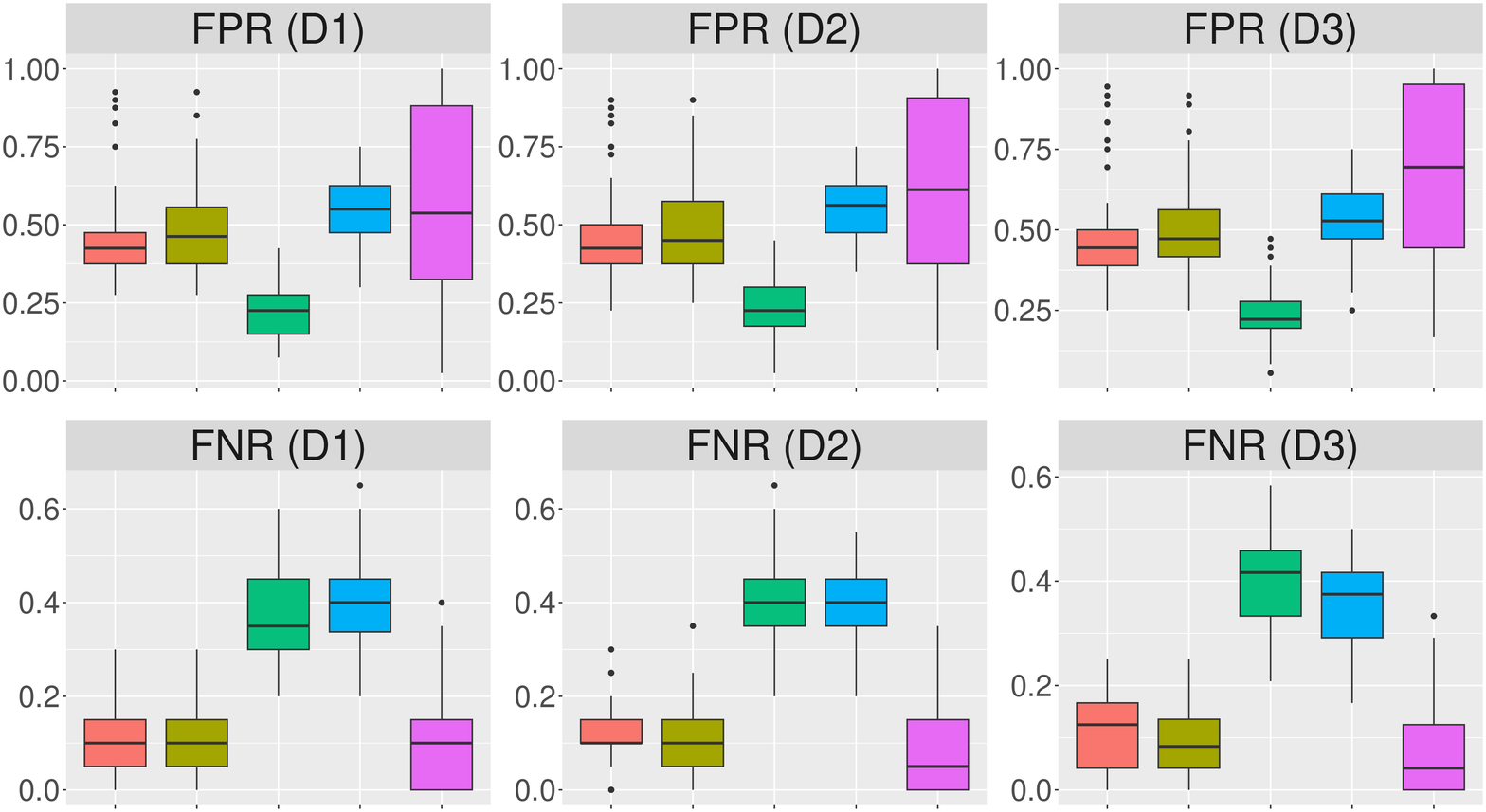}
\vspace{-2.5mm}
\subcaption{$s=5, \rho_x=0.9, \rho_y=0.1$}
\end{minipage}
\begin{minipage}[b]{0.33\linewidth}
\centering
\includegraphics[width=7cm,height=4.6cm]{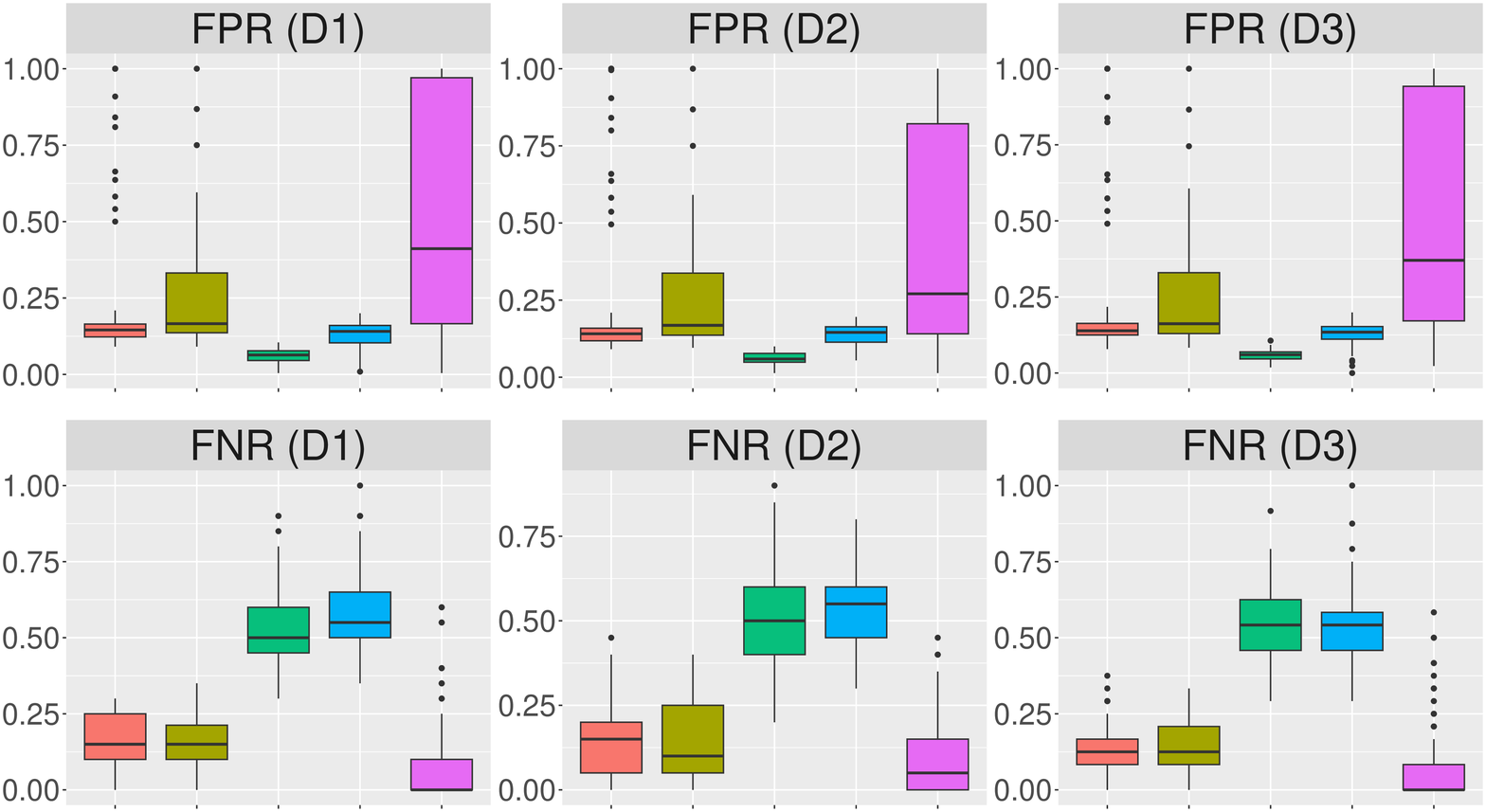}
\vspace{-2.5mm}
\subcaption{$s=50, \rho_x=0.9, \rho_y=0.1$}
\end{minipage}
\begin{minipage}[b]{0.33\linewidth}
\centering
\includegraphics[width=7cm,height=4.6cm]{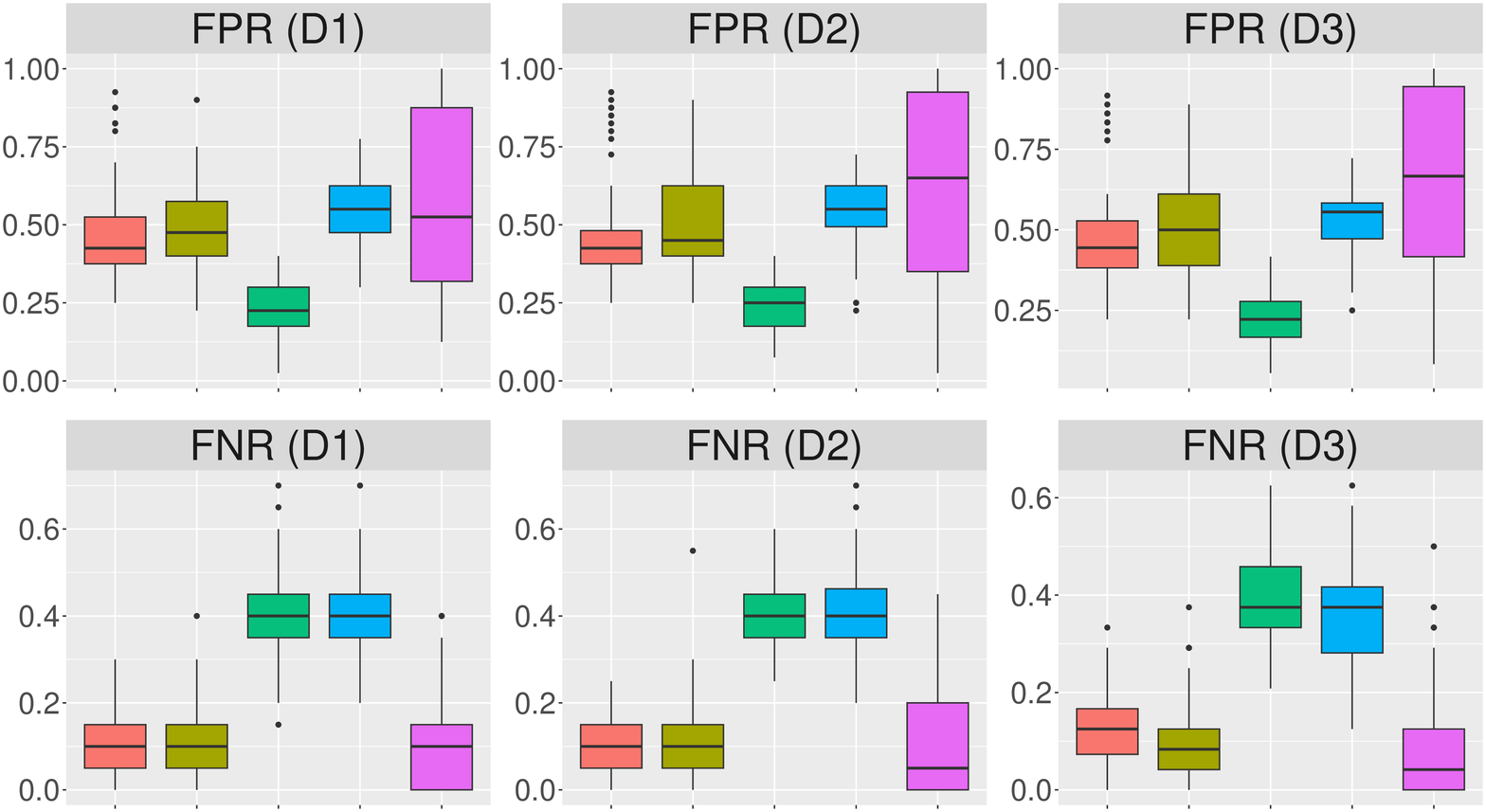}
\vspace{-2.5mm}
\subcaption{$s=5, \rho_x=0.9, \rho_y=0.9$}
\end{minipage}
\begin{minipage}[b]{0.33\linewidth}
\centering
\includegraphics[width=7cm,height=4.6cm]{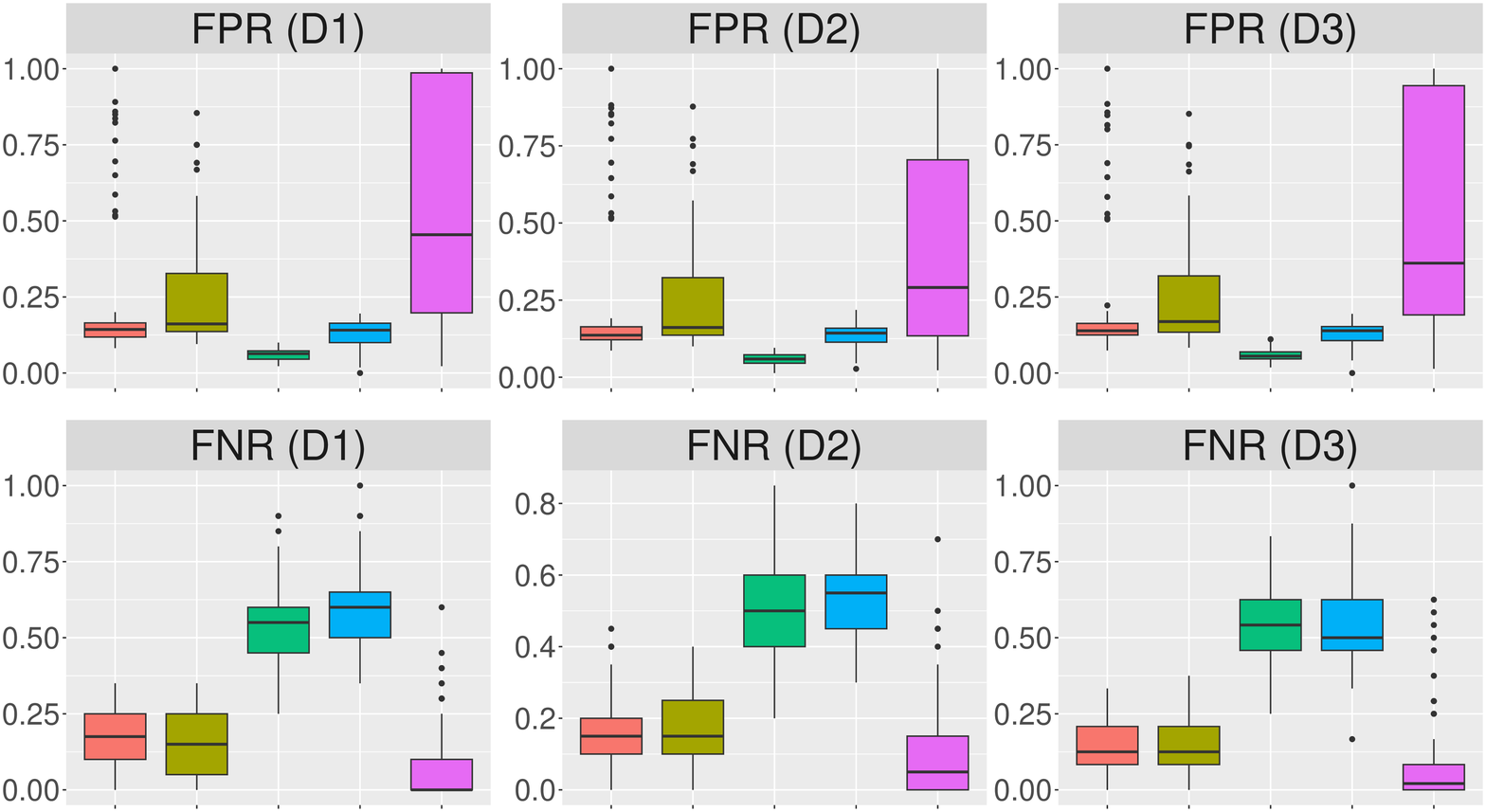}
\vspace{-2.5mm}
\subcaption{$s=50, \rho_x=0.9, \rho_y=0.9$}
\end{minipage}
\caption{Boxplots of FPR and FNR for $n=15$ when the case $M=3$.
The red boxplot indicates MR, dark yellow UR, green lasso, blue mglasso, and magenta mlasso. 
}
\label{fig:SimuFPRFNR_M3n15_FPRFNR}
\end{figure}
\end{landscape}

\begin{landscape}
\begin{figure}[htbp]
\begin{minipage}[b]{0.33\linewidth}
\centering
\includegraphics[width=7cm,height=4.6cm]{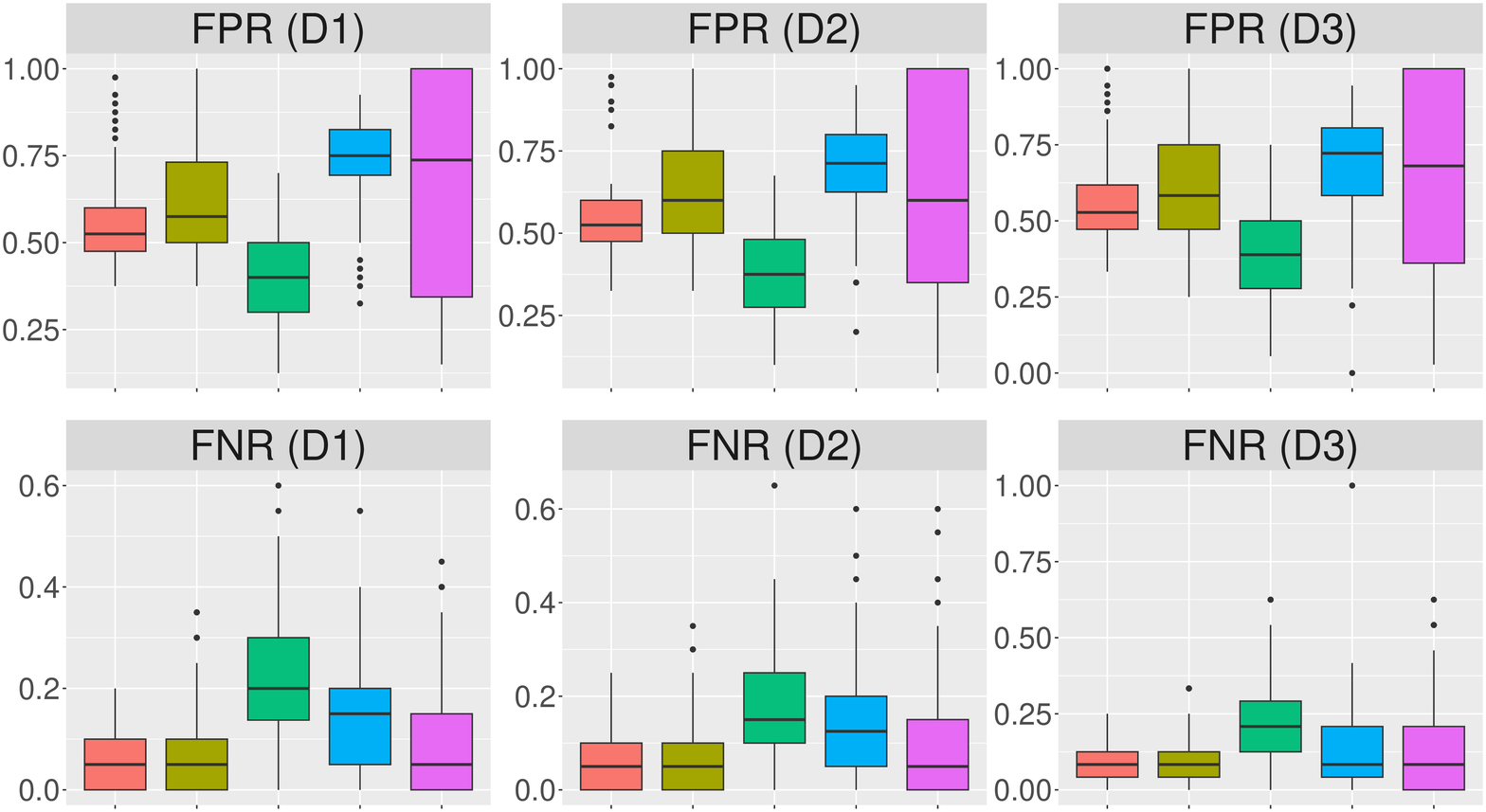}
\vspace{-2.5mm}
\subcaption{$s=5, \rho_x=0.1, \rho_y=0.1$}
\end{minipage}
\begin{minipage}[b]{0.33\linewidth}
\centering
\includegraphics[width=7cm,height=4.6cm]{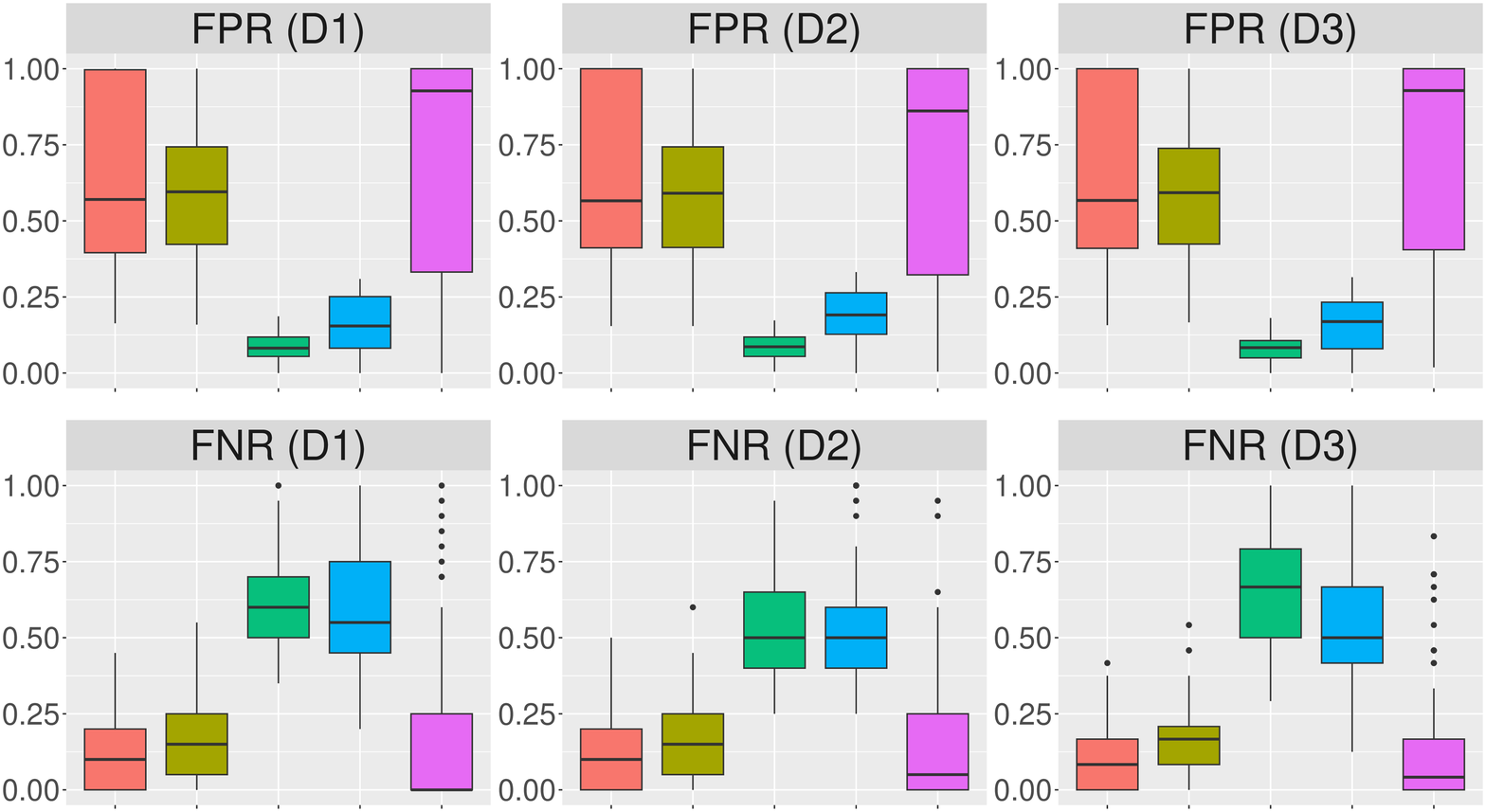} 
\vspace{-2.5mm}
\subcaption{$s=50, \rho_x=0.1, \rho_y=0.1$}
\end{minipage}
\begin{minipage}[b]{0.33\linewidth}
\centering
\includegraphics[width=7cm,height=4.6cm]{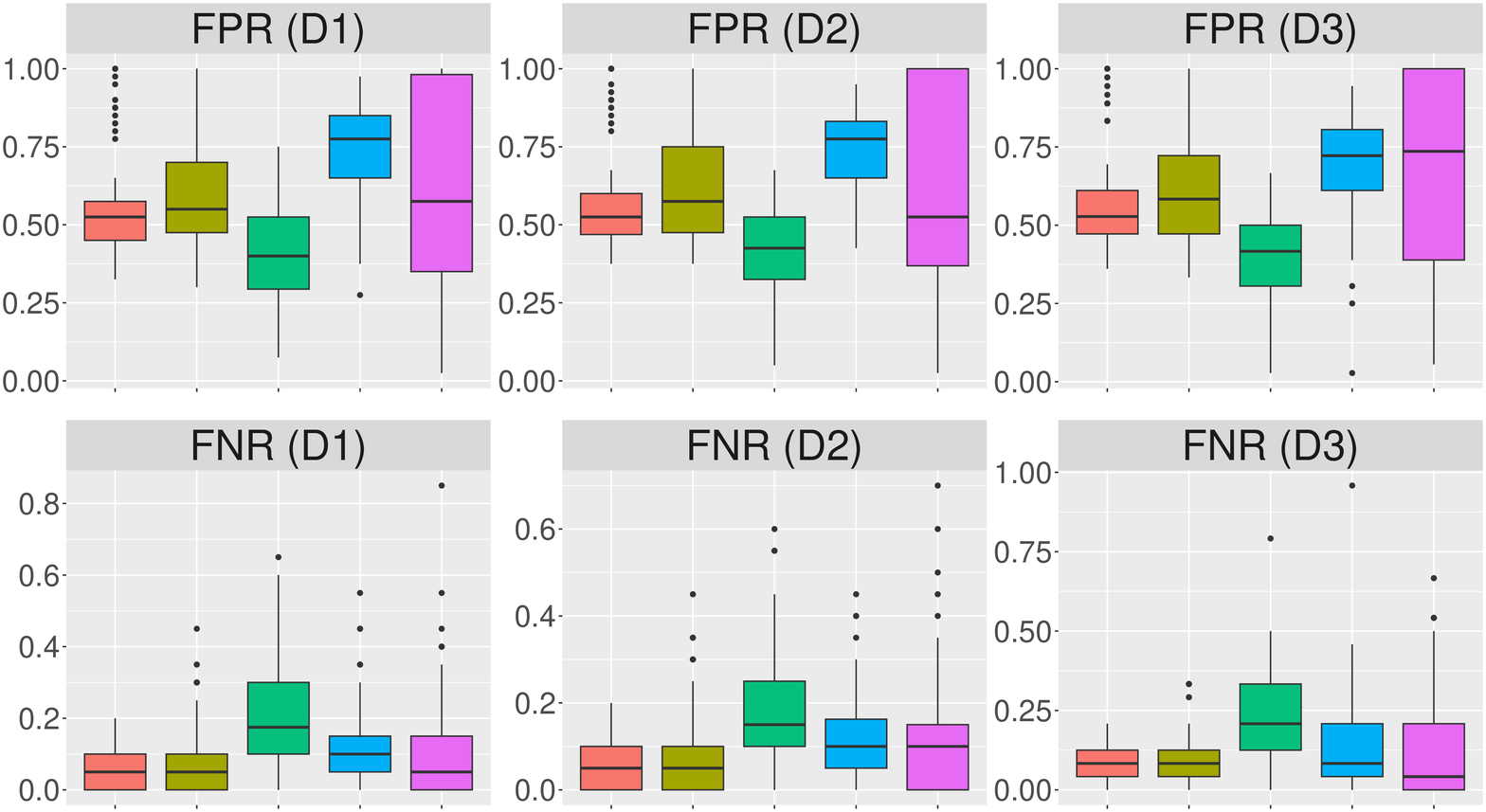} 
\vspace{-2.5mm}
\subcaption{$s=5, \rho_x=0.1, \rho_y=0.9$}
\end{minipage}
\begin{minipage}[b]{0.33\linewidth}
\centering
\includegraphics[width=7cm,height=4.6cm]{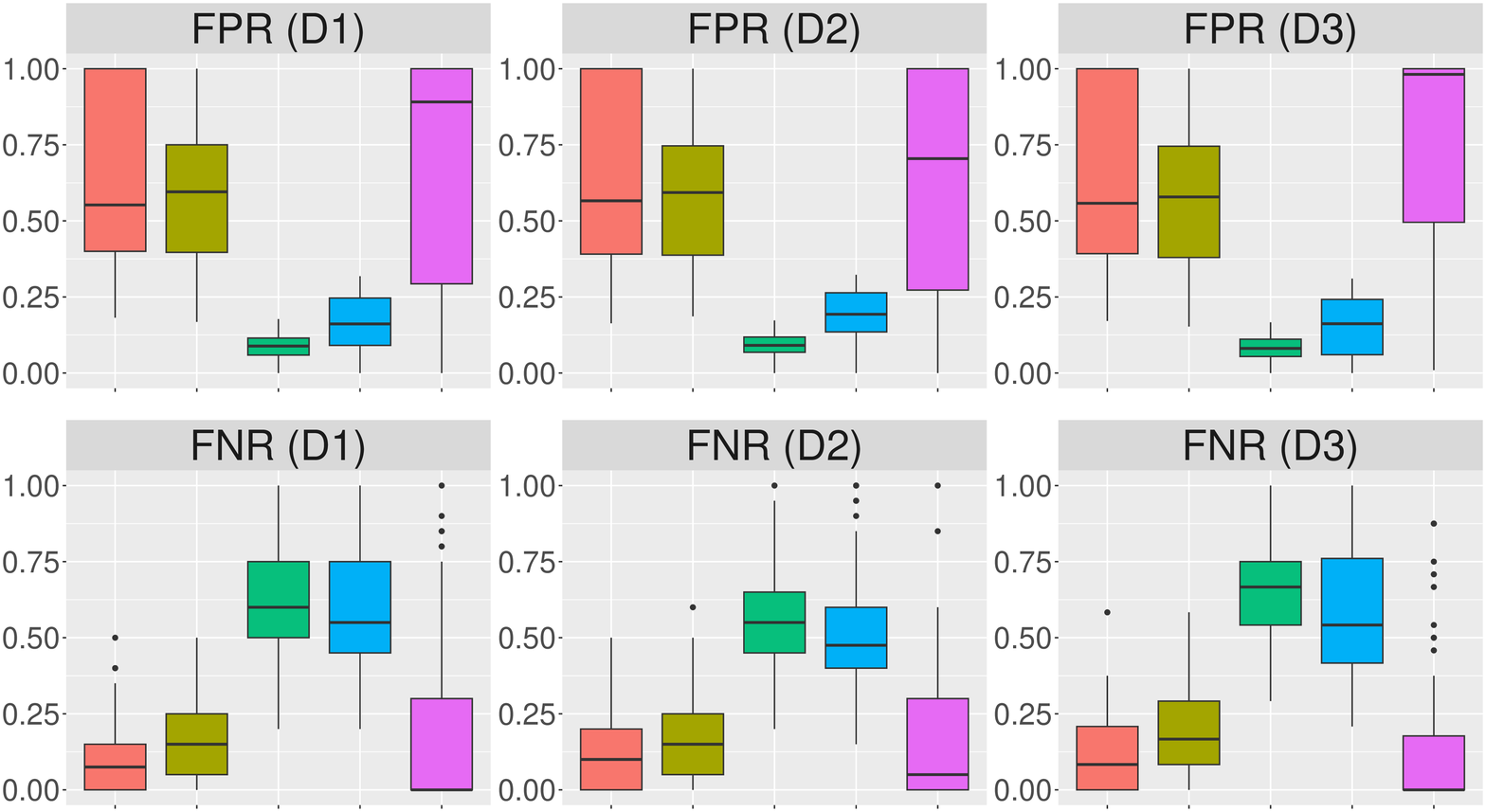}
\vspace{-2.5mm}
\subcaption{$s=50, \rho_x=0.1, \rho_y=0.9$}
\end{minipage}
\begin{minipage}[b]{0.33\linewidth}
\centering
\includegraphics[width=7cm,height=4.6cm]{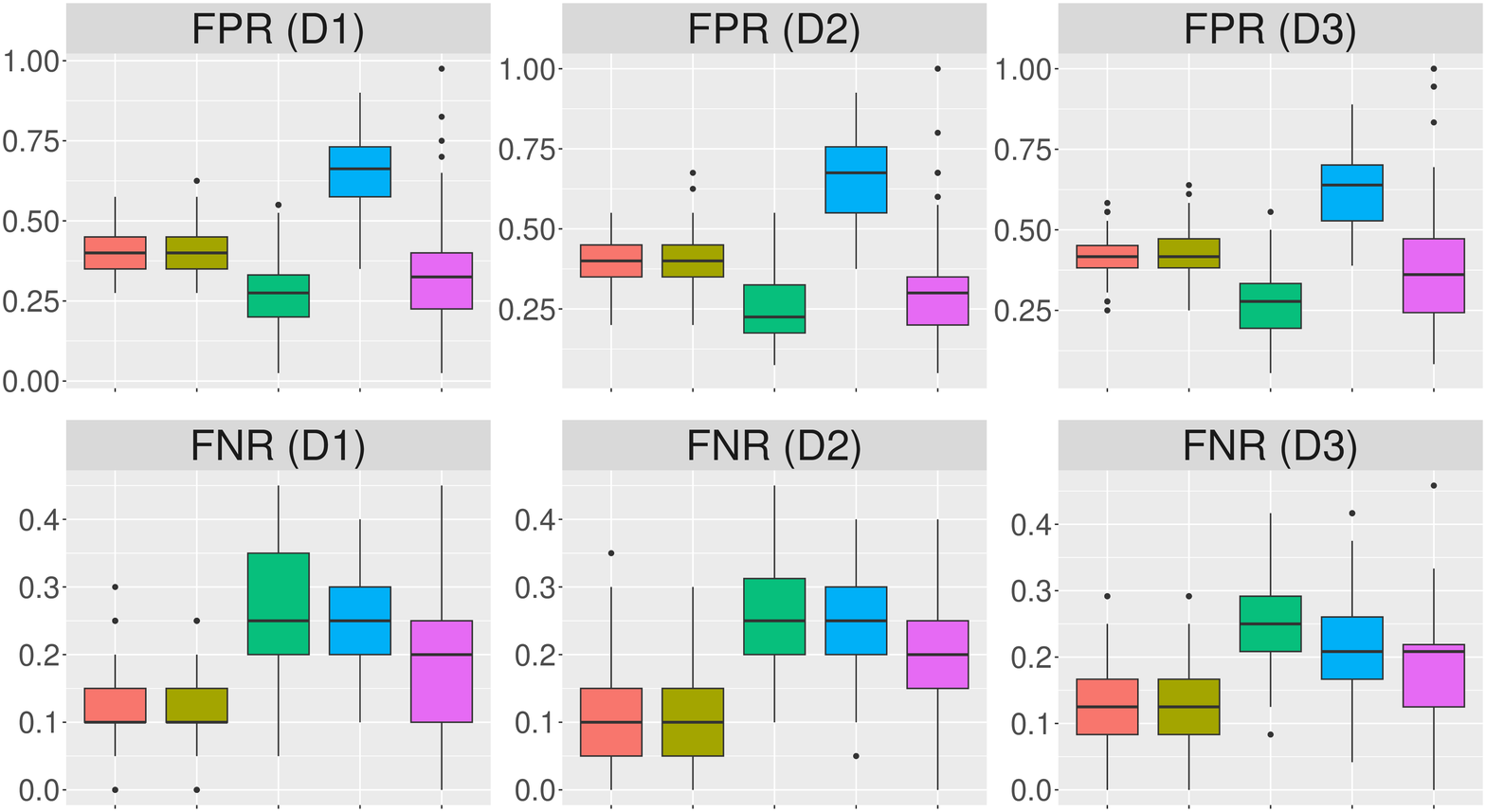}
\vspace{-2.5mm}
\subcaption{$s=5, \rho_x=0.9, \rho_y=0.1$}
\end{minipage}
\begin{minipage}[b]{0.33\linewidth}
\centering
\includegraphics[width=7cm,height=4.6cm]{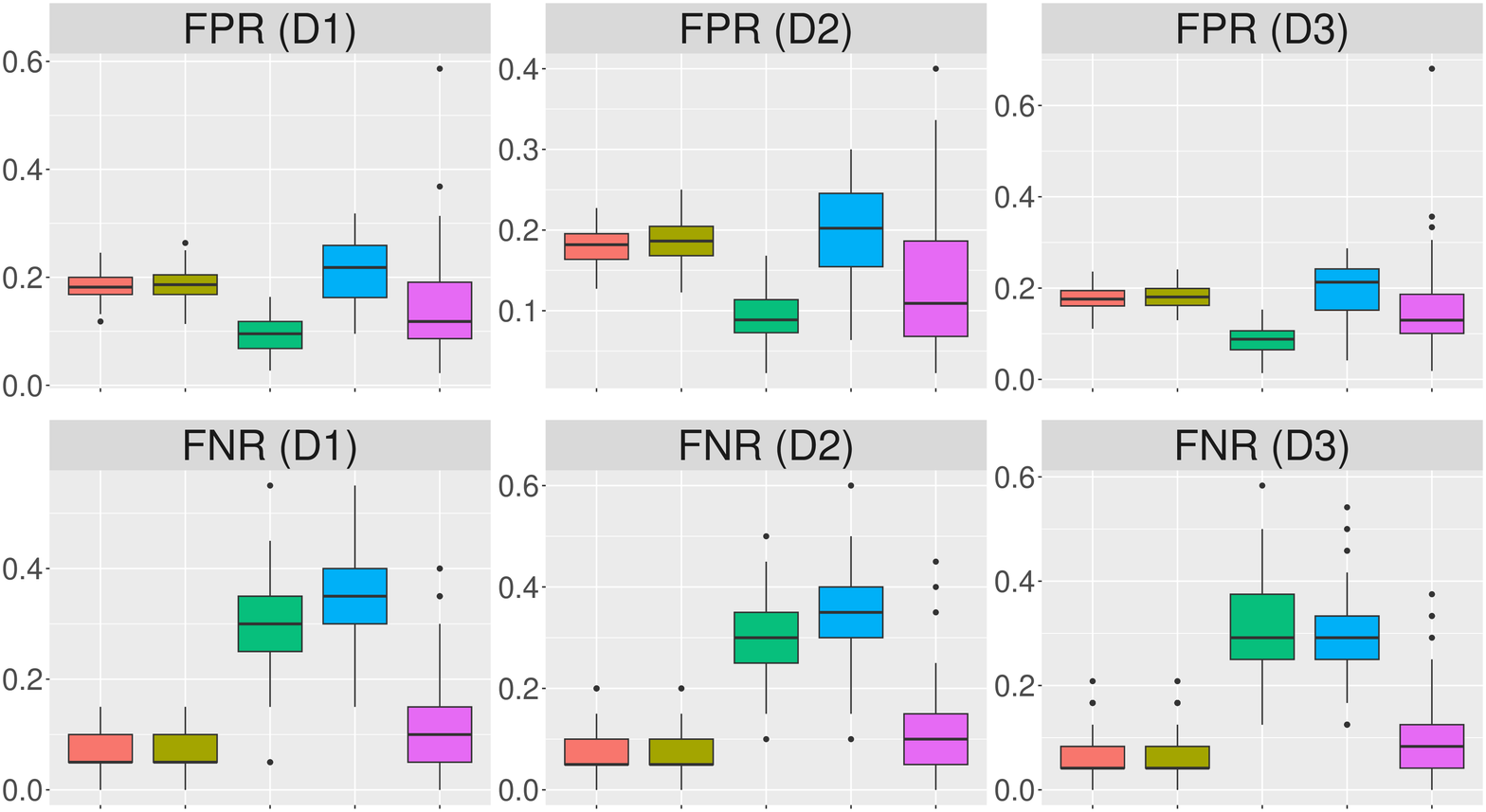}
\vspace{-2.5mm}
\subcaption{$s=50, \rho_x=0.9, \rho_y=0.1$}
\end{minipage}
\begin{minipage}[b]{0.33\linewidth}
\centering
\includegraphics[width=7cm,height=4.6cm]{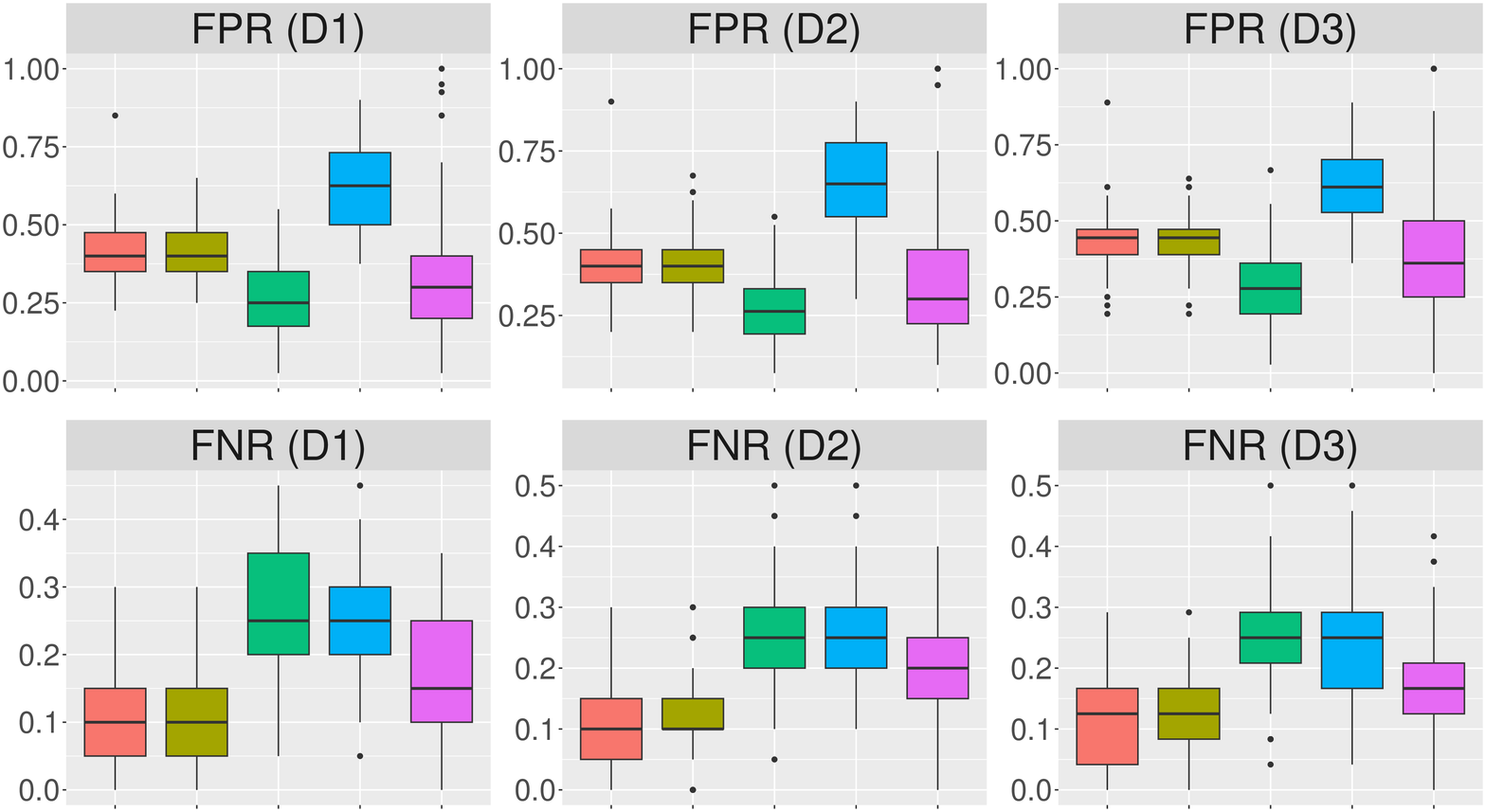}
\vspace{-2.5mm}
\subcaption{$s=5, \rho_x=0.9, \rho_y=0.9$}
\end{minipage}
\begin{minipage}[b]{0.33\linewidth}
\centering
\includegraphics[width=7cm,height=4.6cm]{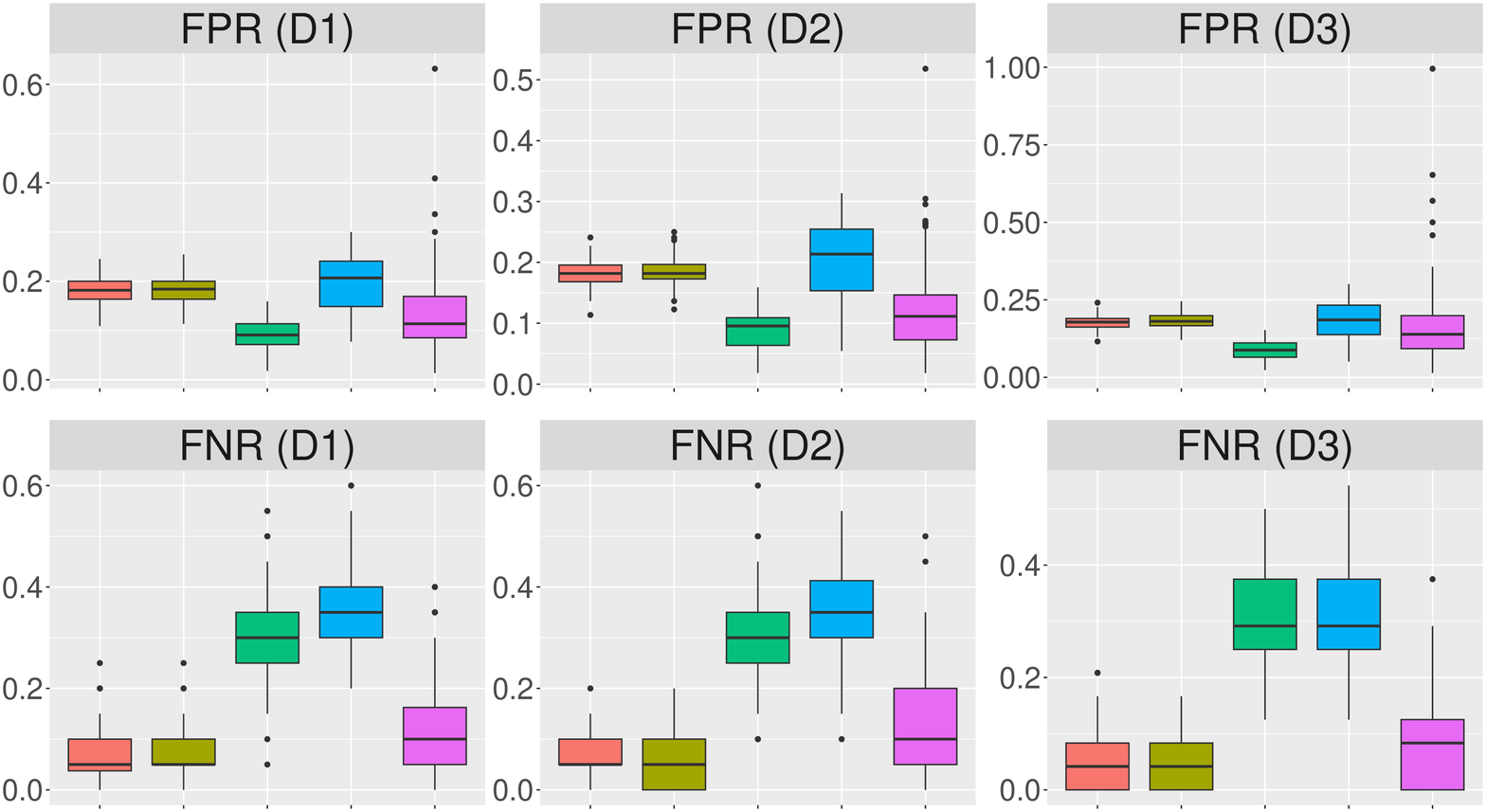}
\vspace{-2.5mm}
\subcaption{$s=50, \rho_x=0.9, \rho_y=0.9$}
\end{minipage}
\caption{Boxplots of FPR and FNR for $n=25$ when the case $M=3$.
The red boxplot indicates MR, dark yellow UR, green lasso, blue mglasso, and magenta mlasso. 
}
\label{fig:SimuFPRFNR_M3n25_FPRFNR}
\end{figure}
\end{landscape}

\begin{landscape}
\begin{figure}[htbp]
\begin{minipage}[b]{0.33\linewidth}
\centering
\includegraphics[width=7cm,height=4.6cm]{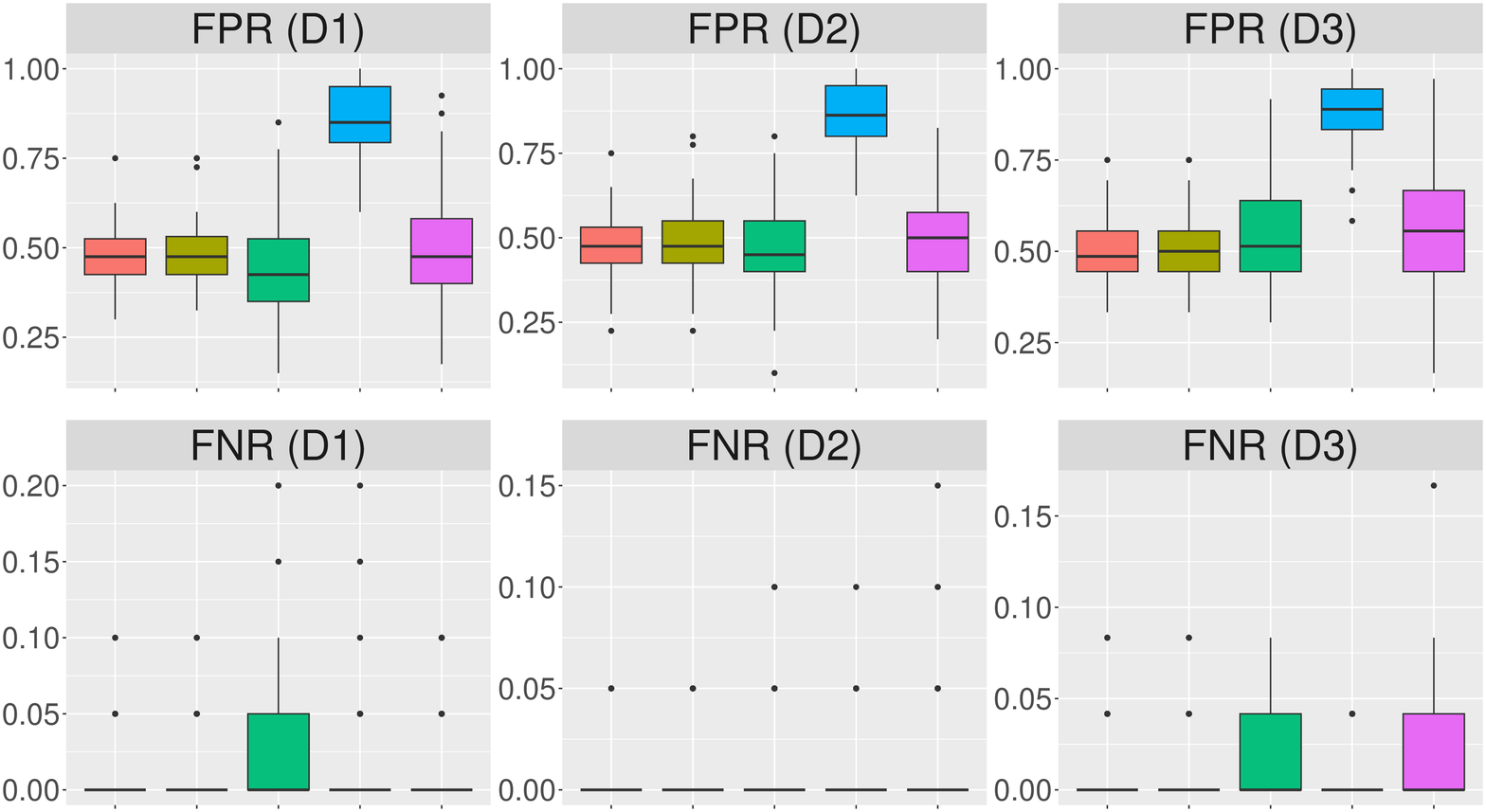}
\vspace{-2.5mm}
\subcaption{$s=5, \rho_x=0.1, \rho_y=0.1$}
\end{minipage}
\begin{minipage}[b]{0.33\linewidth}
\centering
\includegraphics[width=7cm,height=4.6cm]{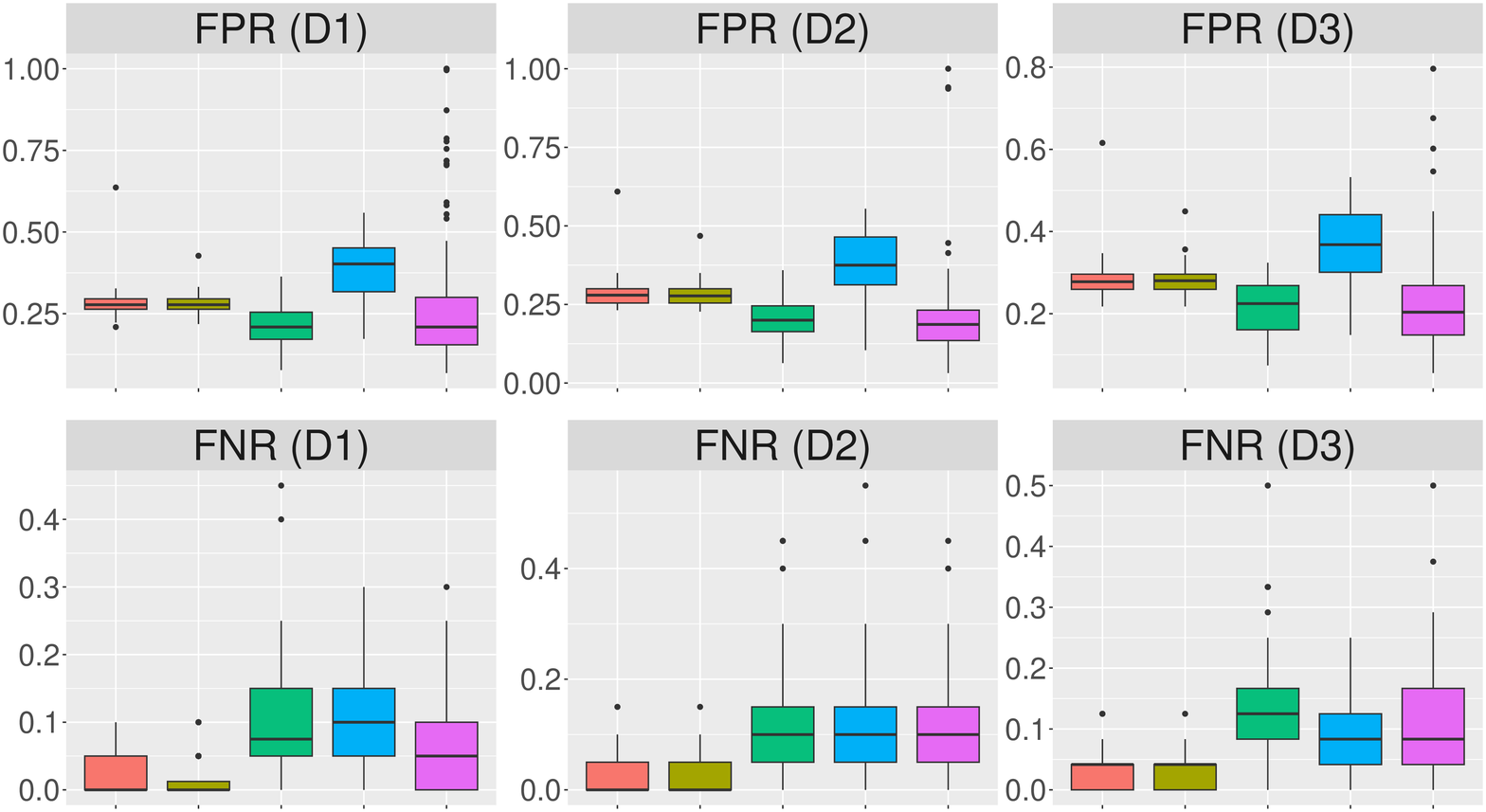} 
\vspace{-2.5mm}
\subcaption{$s=50, \rho_x=0.1, \rho_y=0.1$}
\end{minipage}
\begin{minipage}[b]{0.33\linewidth}
\centering
\includegraphics[width=7cm,height=4.6cm]{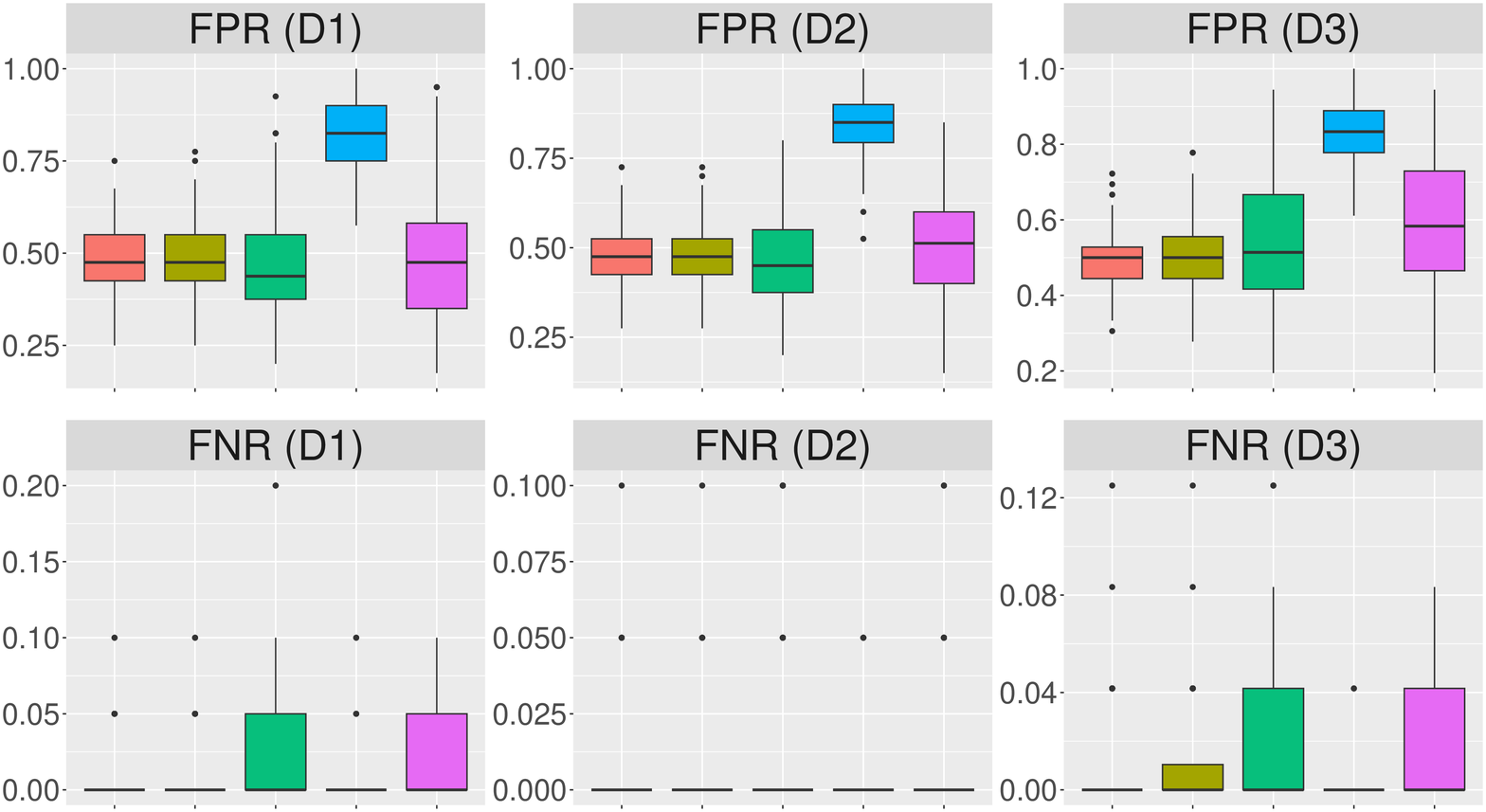} 
\vspace{-2.5mm}
\subcaption{$s=5, \rho_x=0.1, \rho_y=0.9$}
\end{minipage}
\begin{minipage}[b]{0.33\linewidth}
\centering
\includegraphics[width=7cm,height=4.6cm]{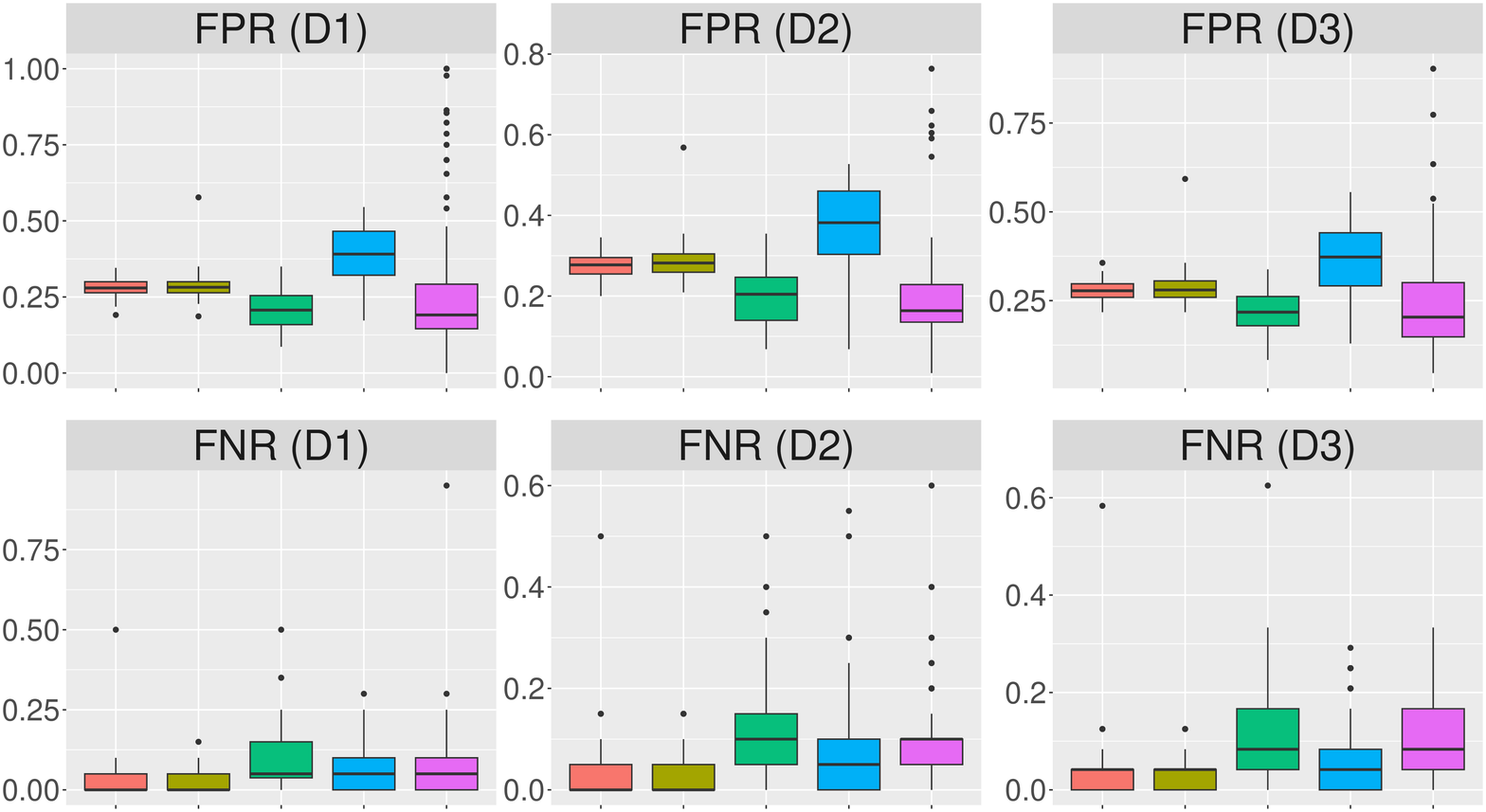}
\vspace{-2.5mm}
\subcaption{$s=50, \rho_x=0.1, \rho_y=0.9$}
\end{minipage}
\begin{minipage}[b]{0.33\linewidth}
\centering
\includegraphics[width=7cm,height=4.6cm]{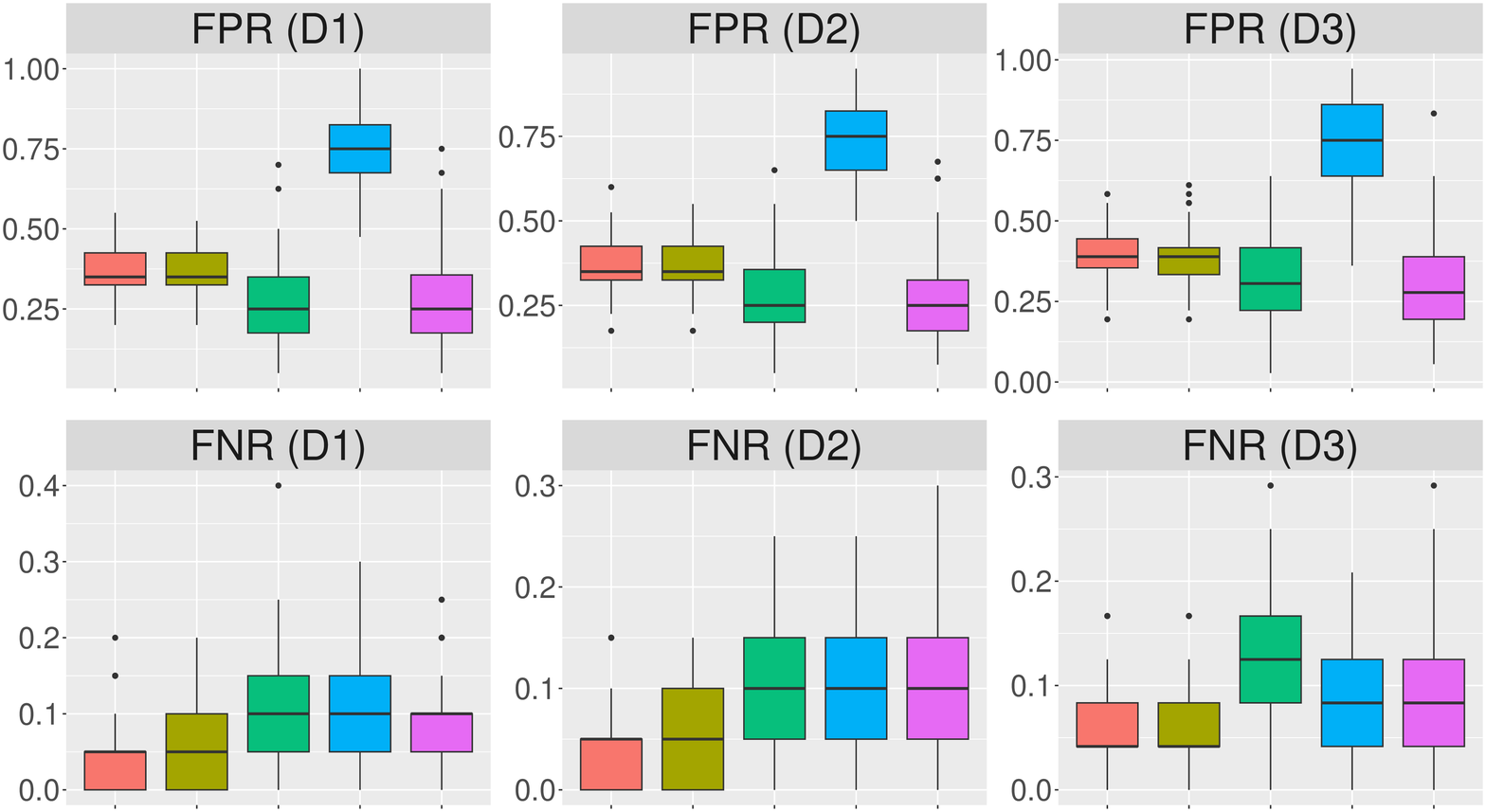}
\vspace{-2.5mm}
\subcaption{$s=5, \rho_x=0.9, \rho_y=0.1$}
\end{minipage}
\begin{minipage}[b]{0.33\linewidth}
\centering
\includegraphics[width=7cm,height=4.6cm]{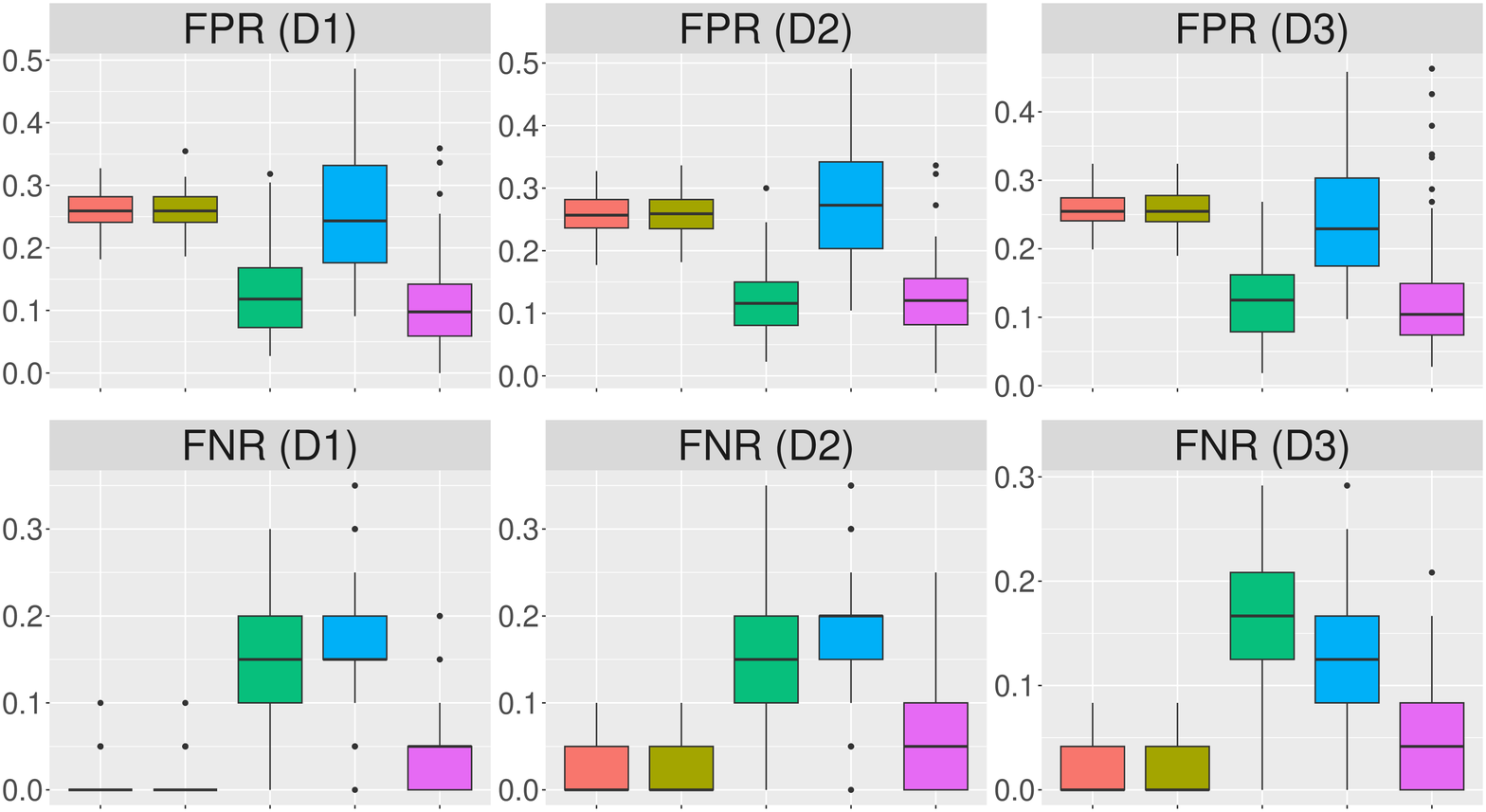}
\vspace{-2.5mm}
\subcaption{$s=50, \rho_x=0.9, \rho_y=0.1$}
\end{minipage}
\begin{minipage}[b]{0.33\linewidth}
\centering
\includegraphics[width=7cm,height=4.6cm]{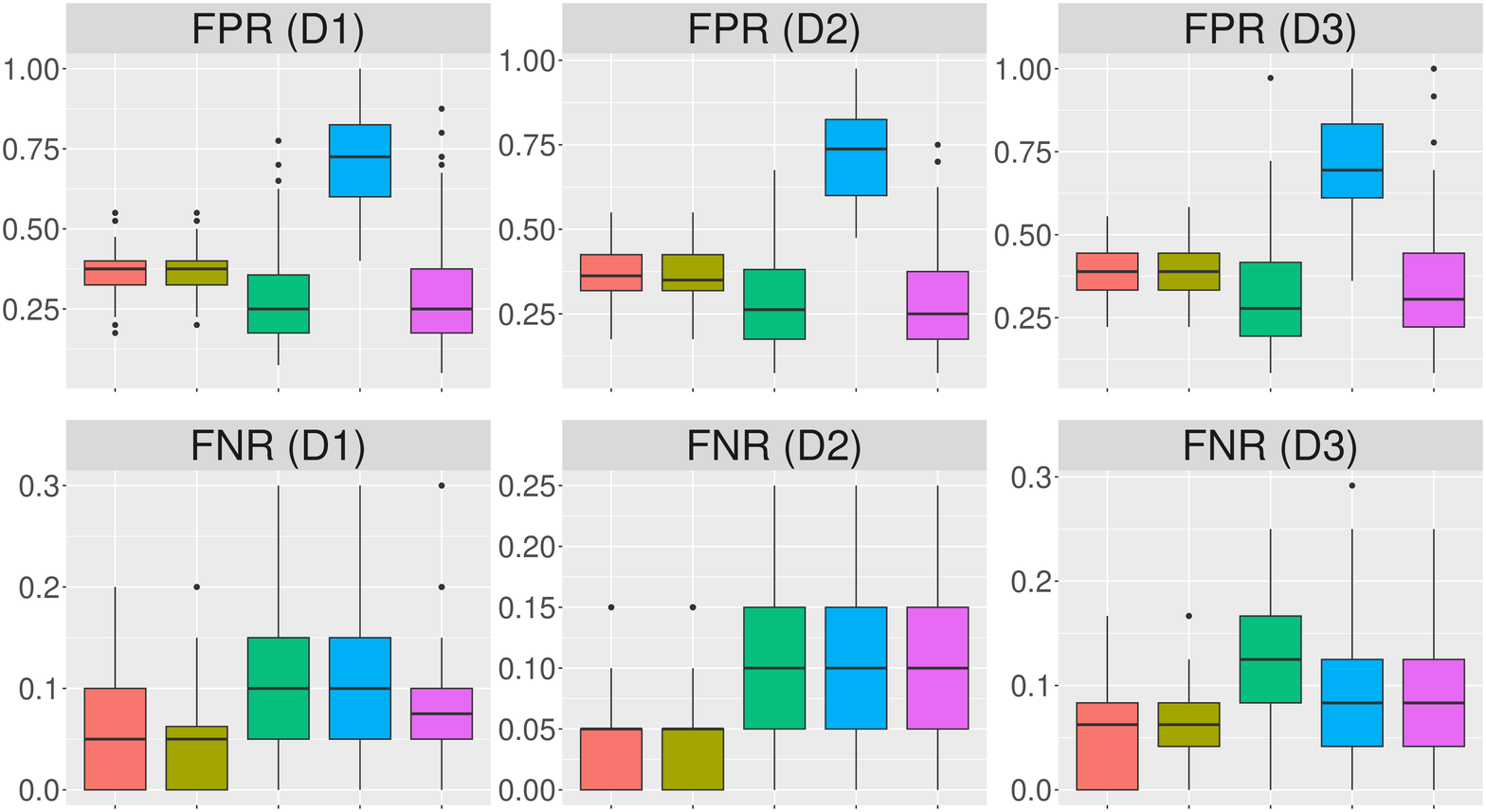}
\vspace{-2.5mm}
\subcaption{$s=5, \rho_x=0.9, \rho_y=0.9$}
\end{minipage}
\begin{minipage}[b]{0.33\linewidth}
\centering
\includegraphics[width=7cm,height=4.6cm]{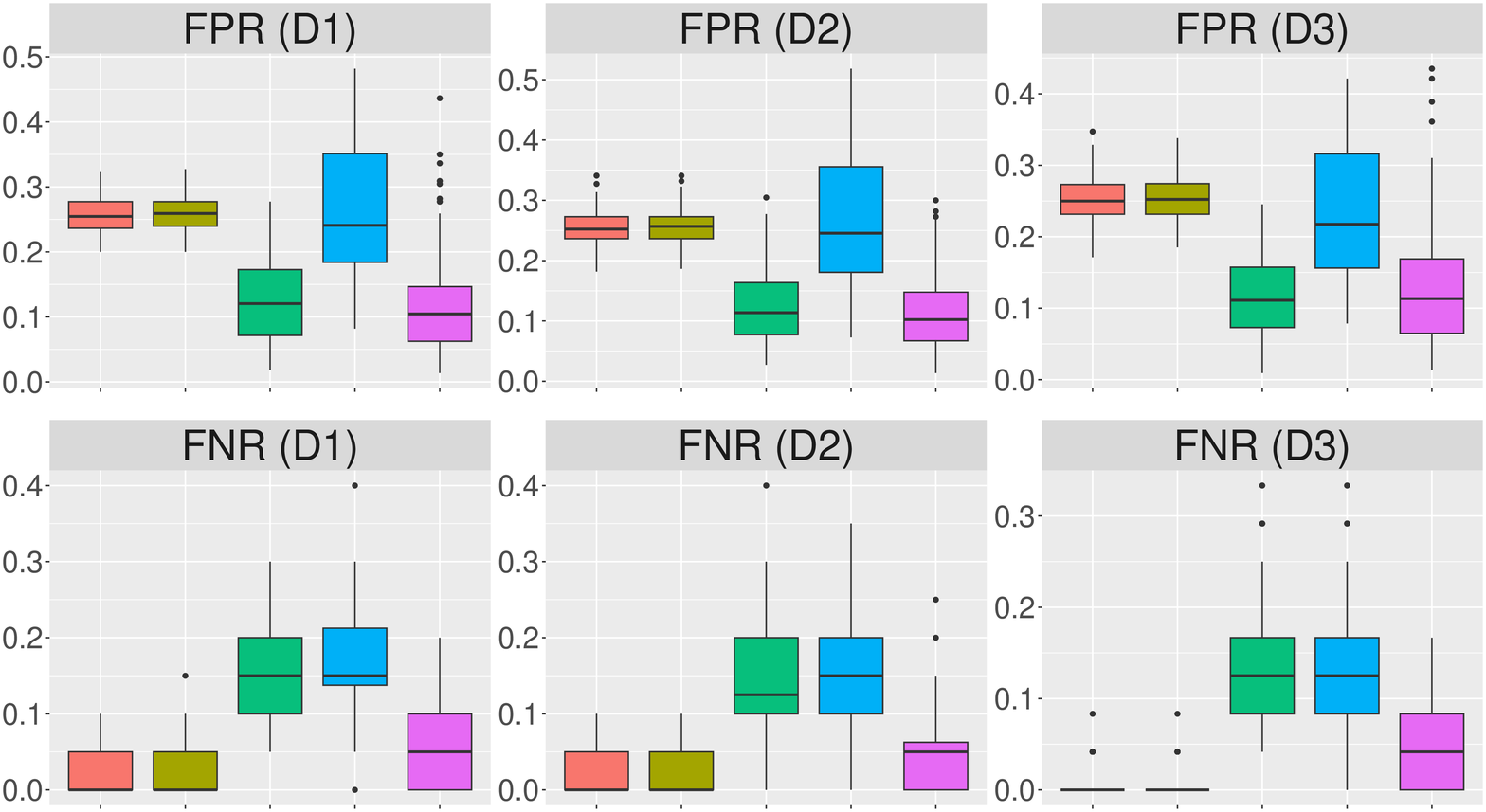}
\vspace{-2.5mm}
\subcaption{$s=50, \rho_x=0.9, \rho_y=0.9$}
\end{minipage}
\caption{Boxplots of FPR and FNR for $n=50$ when the case $M=3$.
The red boxplot indicates MR, dark yellow UR, green lasso, blue mglasso, and magenta mlasso. 
}
\label{fig:SimuFPRFNR_M3n50_FPRFNR}
\end{figure}
\end{landscape}

\begin{landscape}
\begin{figure}[htbp]
\begin{minipage}[b]{0.33\linewidth}
\centering
\includegraphics[width=7cm,height=4.6cm]{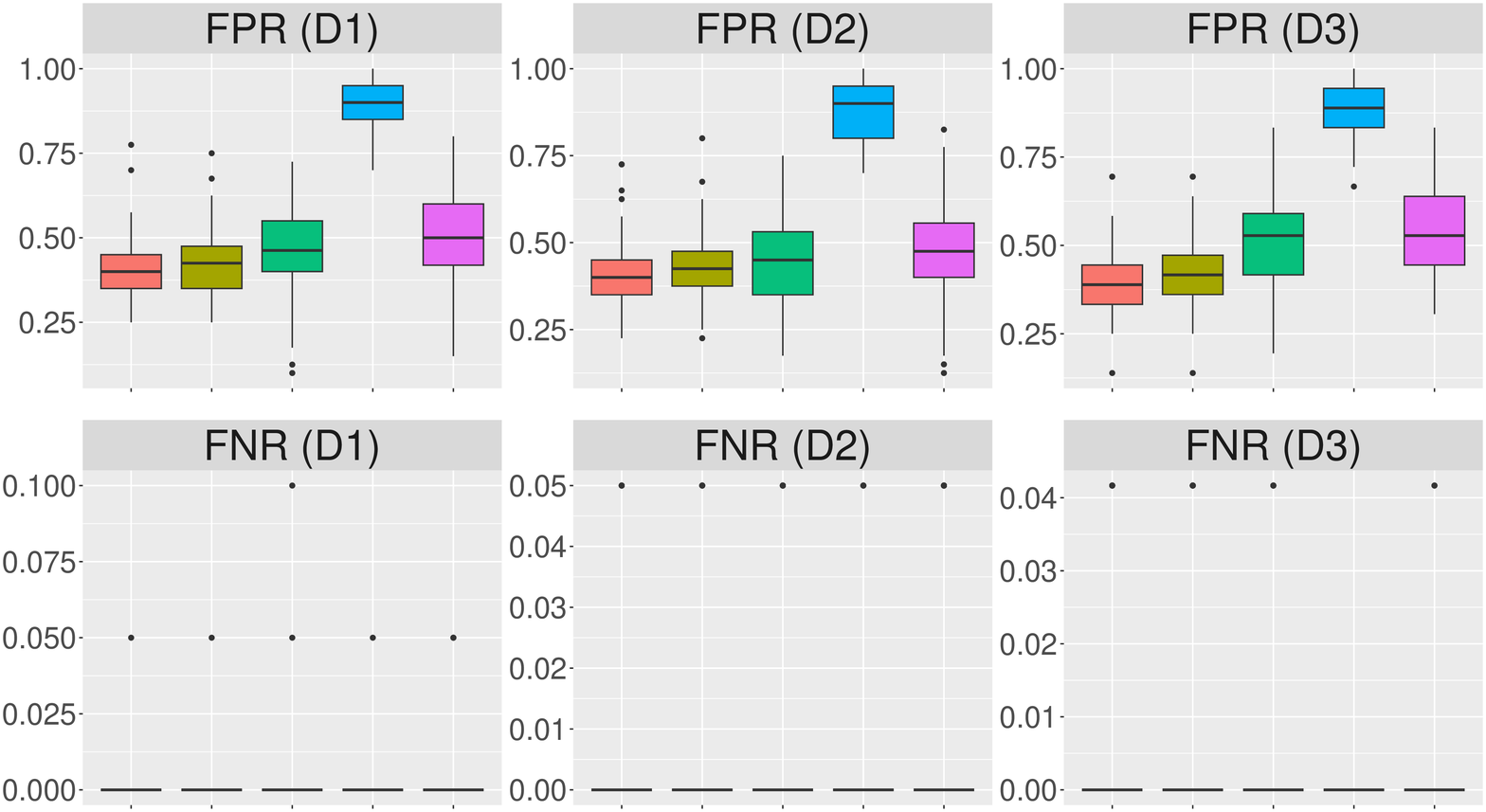}
\vspace{-2.5mm}
\subcaption{$s=5, \rho_x=0.1, \rho_y=0.1$}
\end{minipage}
\begin{minipage}[b]{0.33\linewidth}
\centering
\includegraphics[width=7cm,height=4.6cm]{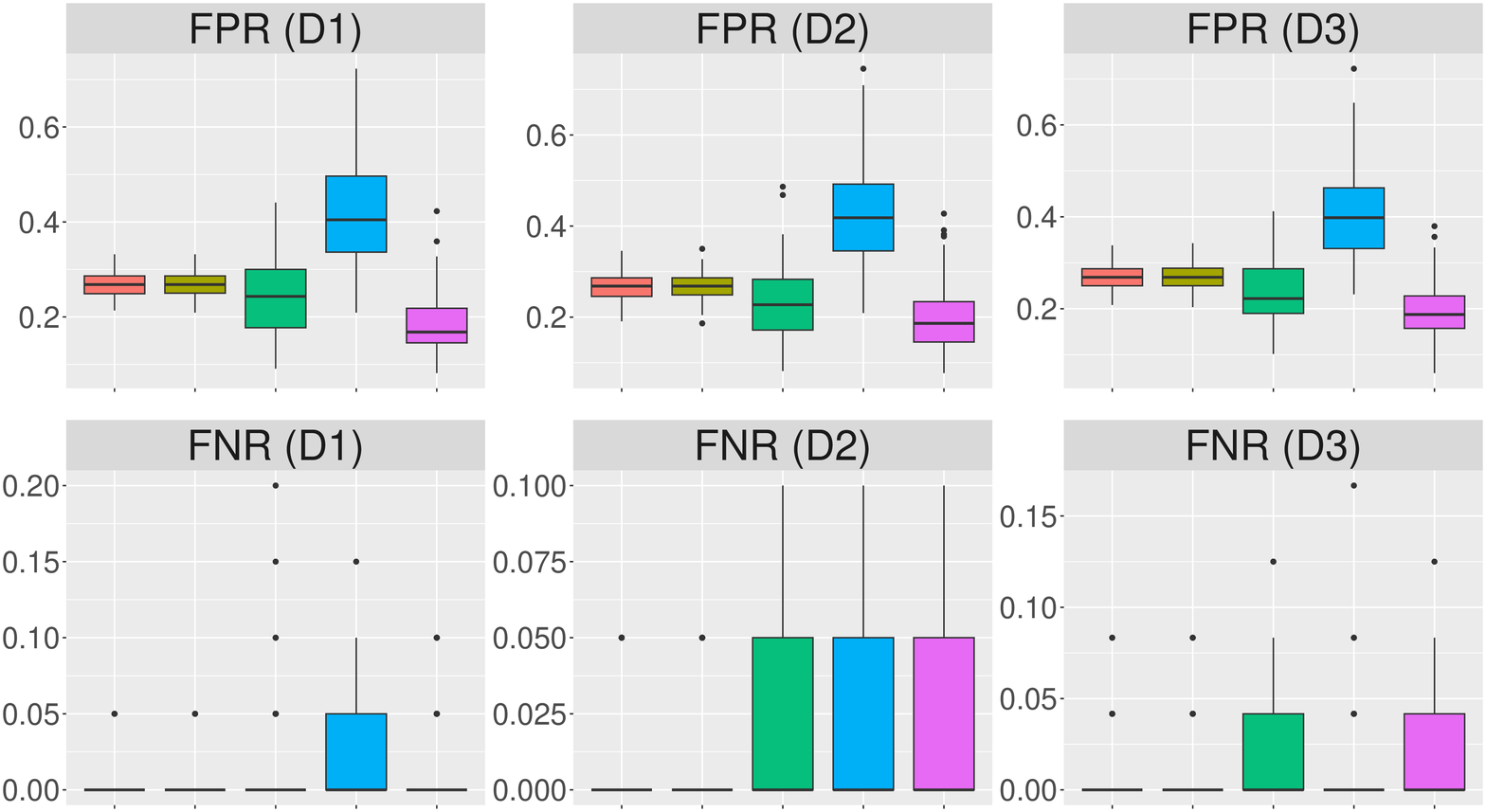} 
\vspace{-2.5mm}
\subcaption{$s=50, \rho_x=0.1, \rho_y=0.1$}
\end{minipage}
\begin{minipage}[b]{0.33\linewidth}
\centering
\includegraphics[width=7cm,height=4.6cm]{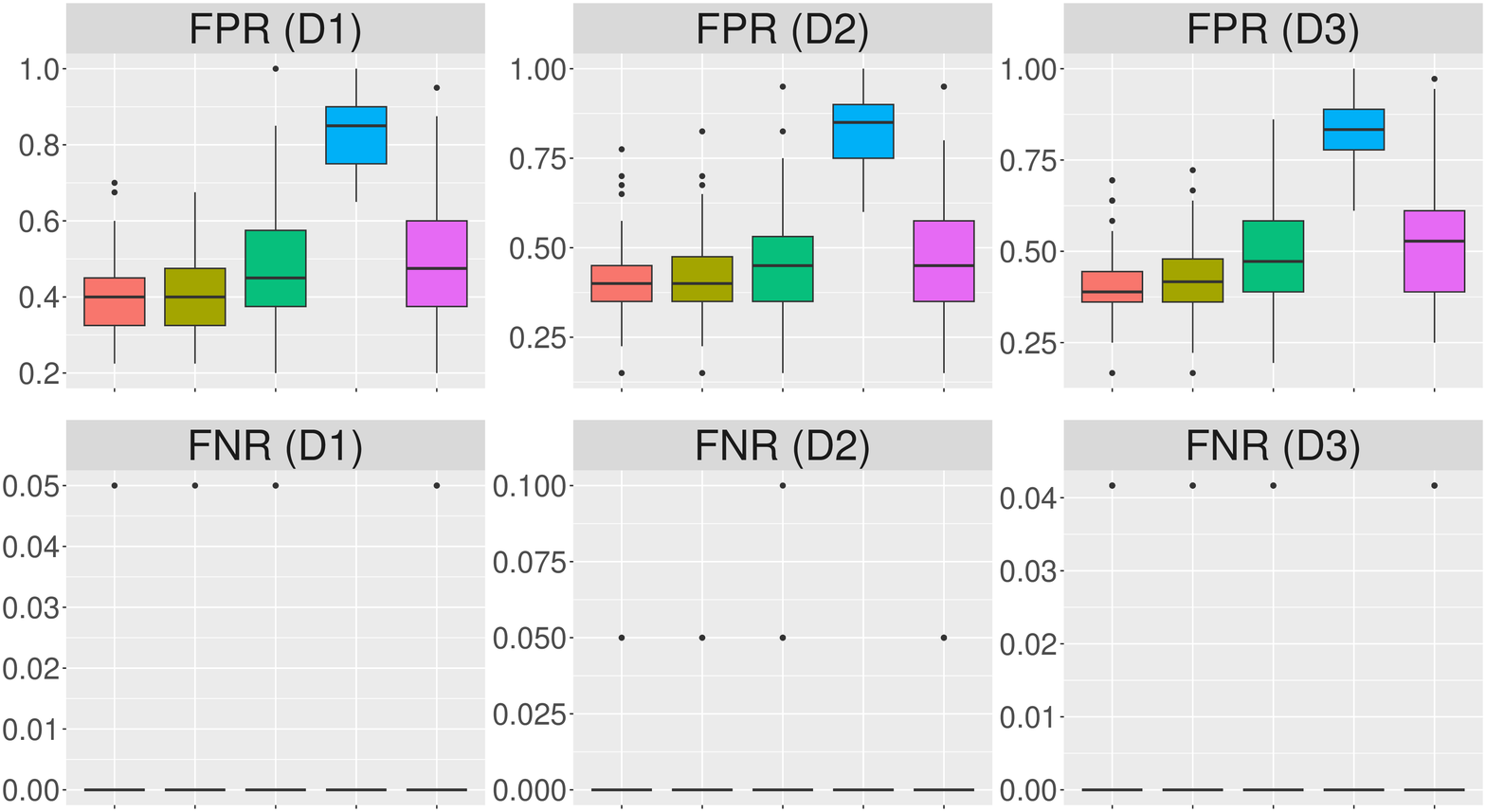} 
\vspace{-2.5mm}
\subcaption{$s=5, \rho_x=0.1, \rho_y=0.9$}
\end{minipage}
\begin{minipage}[b]{0.33\linewidth}
\centering
\includegraphics[width=7cm,height=4.6cm]{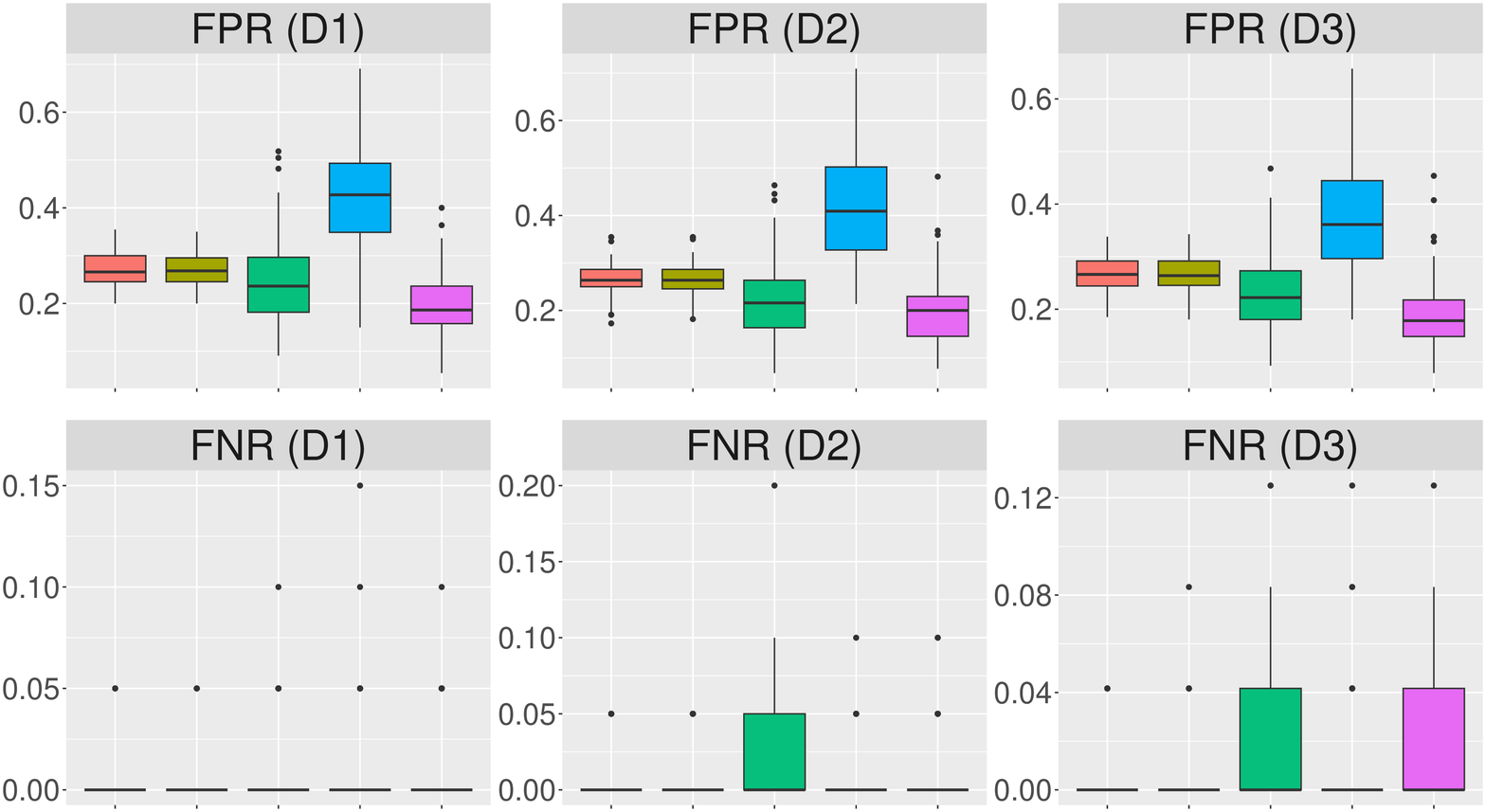}
\vspace{-2.5mm}
\subcaption{$s=50, \rho_x=0.1, \rho_y=0.9$}
\end{minipage}
\begin{minipage}[b]{0.33\linewidth}
\centering
\includegraphics[width=7cm,height=4.6cm]{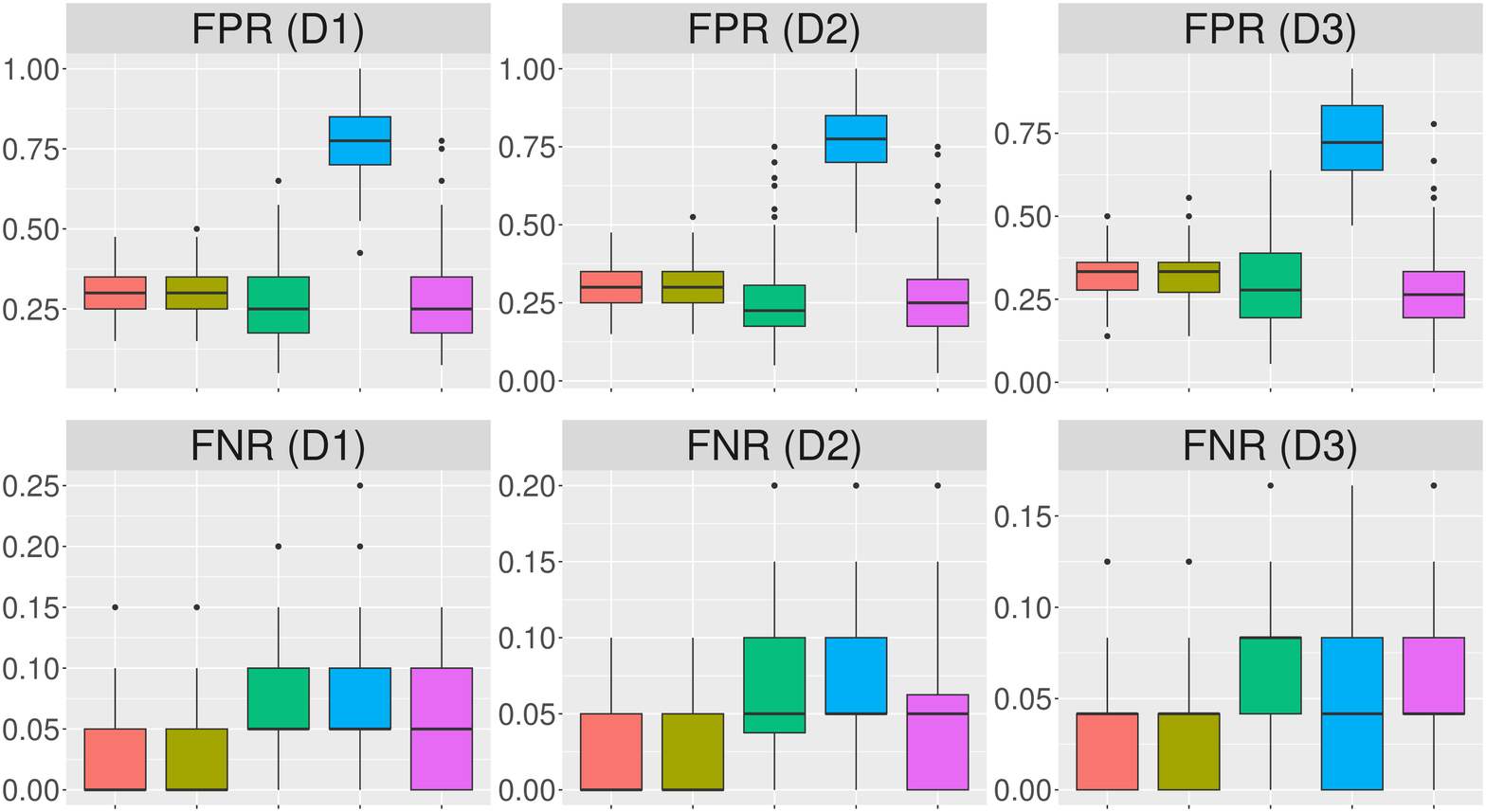}
\vspace{-2.5mm}
\subcaption{$s=5, \rho_x=0.9, \rho_y=0.1$}
\end{minipage}
\begin{minipage}[b]{0.33\linewidth}
\centering
\includegraphics[width=7cm,height=4.6cm]{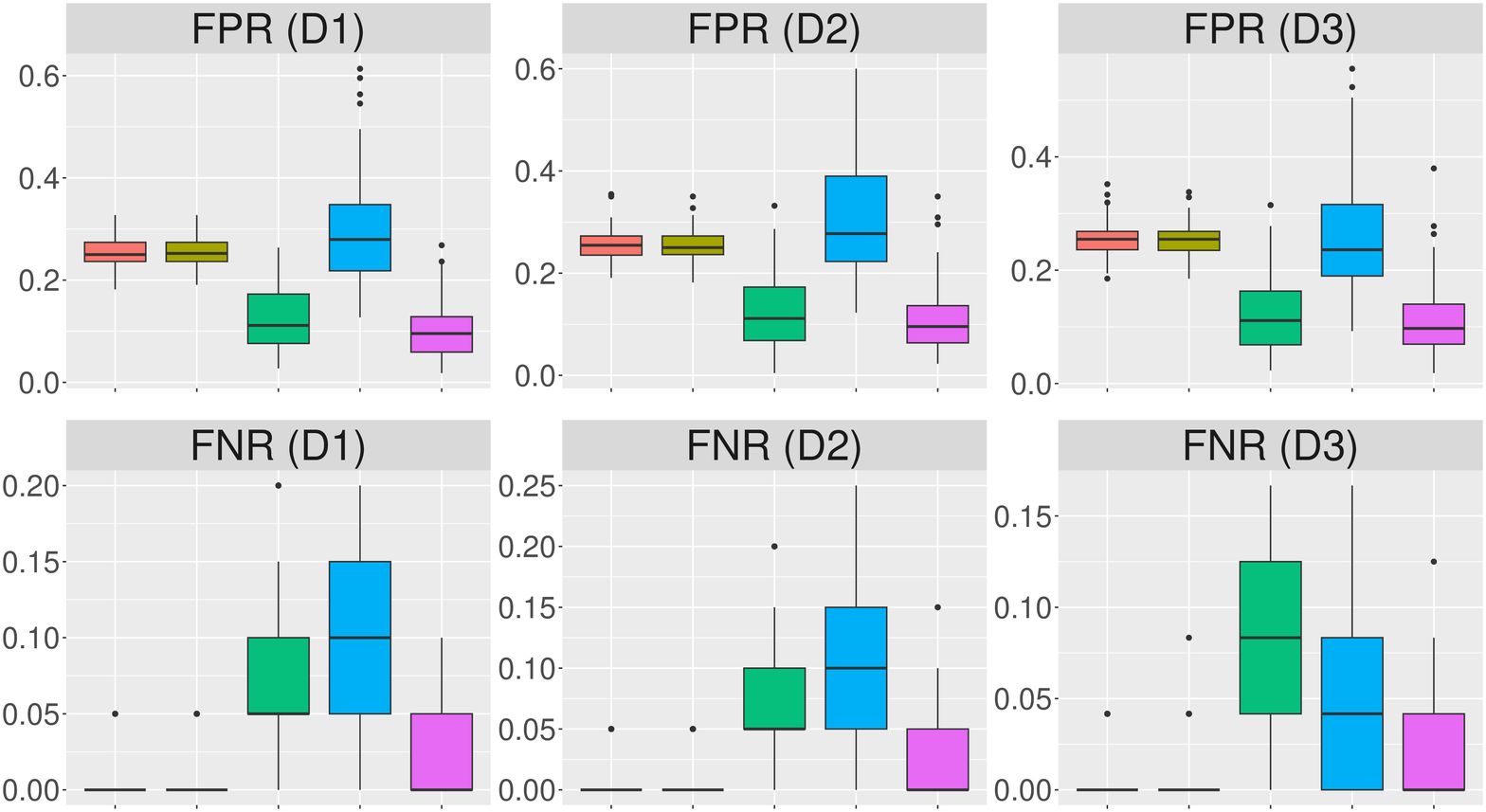}
\vspace{-2.5mm}
\subcaption{$s=50, \rho_x=0.9, \rho_y=0.1$}
\end{minipage}
\begin{minipage}[b]{0.33\linewidth}
\centering
\includegraphics[width=7cm,height=4.6cm]{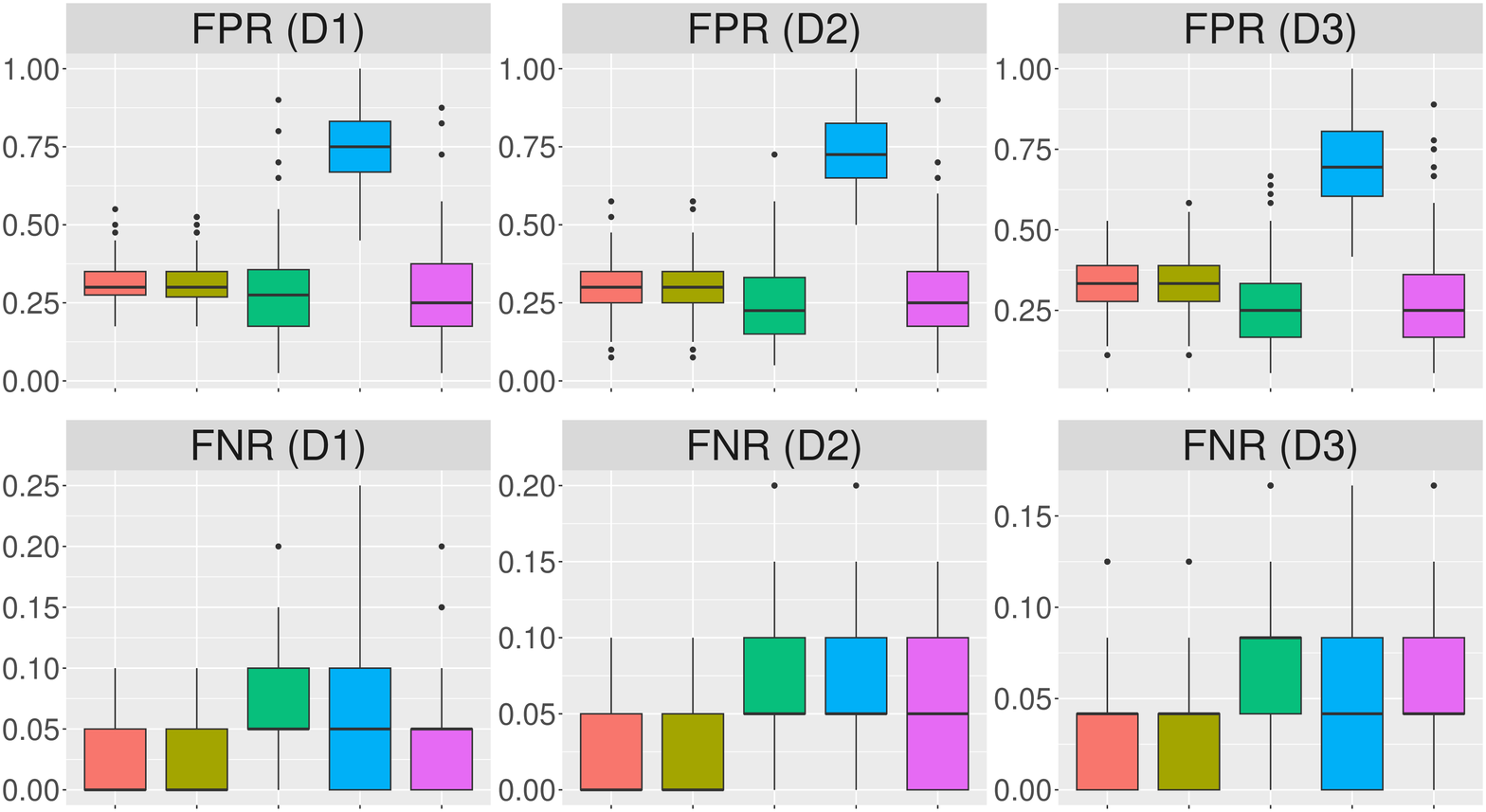}
\vspace{-2.5mm}
\subcaption{$s=5, \rho_x=0.9, \rho_y=0.9$}
\end{minipage}
\begin{minipage}[b]{0.33\linewidth}
\centering
\includegraphics[width=7cm,height=4.6cm]{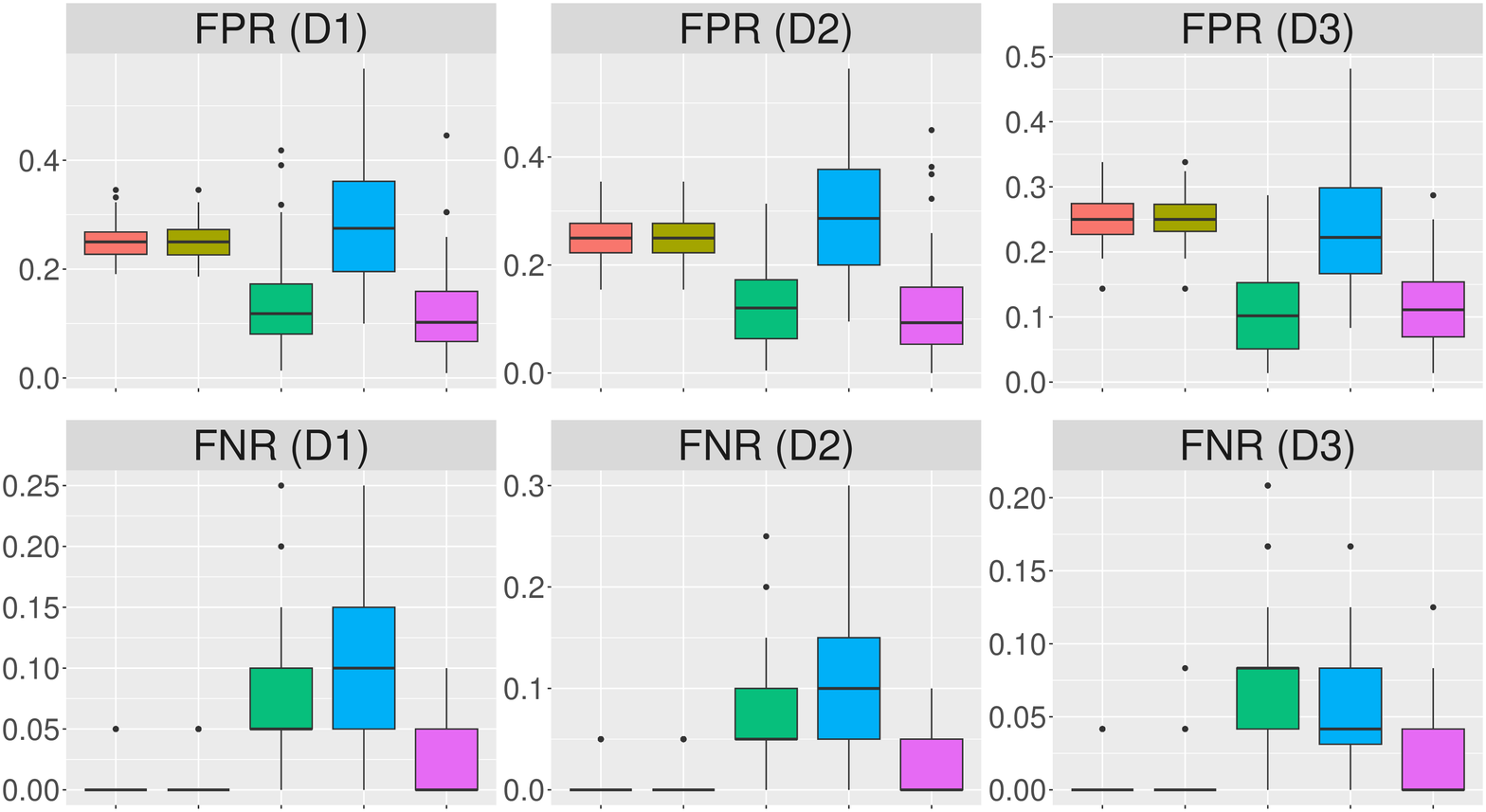}
\vspace{-2.5mm}
\subcaption{$s=50, \rho_x=0.9, \rho_y=0.9$}
\end{minipage}
\caption{Boxplots of FPR and FNR for $n=75$ when the case $M=3$.
The red boxplot indicates MR, dark yellow UR, green lasso, blue mglasso, and magenta mlasso. 
}
\label{fig:SimuFPRFNR_M3n75_FPRFNR}
\end{figure}
\end{landscape}



\end{document}